\pdfoutput=1
\documentclass[a4paper,11pt]{article} 
\usepackage{float}
\usepackage[utf8x]{inputenc} 
\usepackage[english]{babel}
\usepackage{amsmath,amsthm,amsfonts,amssymb,mathrsfs, bbm,
}
\usepackage{latexsym}
\usepackage{authblk}
\usepackage{lineno}
\usepackage{eucal}
\usepackage[dvipsnames]{xcolor}
\usepackage{hyperref}
\usepackage{afterpage}
\usepackage{lmodern}
\usepackage[a4paper,left=2.9cm,right=2.9cm]{geometry}
\usepackage{graphicx}
\usepackage{bm}
\usepackage{mathtools}

\usepackage{cancel}
\definecolor{BeauBlue}{rgb}{0, 0.2, .9}
\definecolor{BeauOrange}{rgb}{.8, .1, 0}
\usepackage{hyperref}
\hypersetup{
	colorlinks = true,
	linkcolor = BeauBlue,
	urlcolor = cyan,
	citecolor = BeauOrange,
}

\usepackage{subfig}

\numberwithin{equation}{section}

\newtheorem{theorem}{Theorem}[section] 
\newtheorem{proposition}[theorem]{Proposition} 

\newtheorem{lemma}[theorem]{Lemma}
\newtheorem{definition}[theorem]{Definition}  
\newtheorem{example}[theorem]{Example}
\newtheorem{remark}[theorem]{Remark}

\DeclareMathOperator{\Tr}{Tr} 
\usepackage{calculus}
\usepackage{calc}

\PassOptionsToPackage{dvipsnames}{xcolor}
\usepackage{graphicx} 
\usepackage{physics}
\usepackage{tikz}
\usepackage{tikz-3dplot}
\usepackage{pgfplots}
\usepackage{xcolor}
\usepackage{subfig}
\usepackage[compat=1.1.0]{tikz-feynman}
\usepackage{amsmath}

\usetikzlibrary{matrix,calc,positioning,decorations.markings,decorations.pathmorphing,decorations.pathreplacing,cd}
\usetikzlibrary{arrows,decorations.markings}
\usetikzlibrary{shapes.misc}
\usetikzlibrary{shapes.geometric}
\usetikzlibrary{shapes.multipart}
\usetikzlibrary{decorations.markings}
\usetikzlibrary{backgrounds}
\usetikzlibrary{shadows}

\usetikzlibrary{automata,positioning}
\tikzset{snake it/.style={decorate, decoration=snake}}

\pgfarrowsdeclaredouble{doubled}{doubled}{stealth}{stealth}
\tikzset{
->-/.style={postaction={decorate,
   decoration={markings,mark=at position .5 with {\arrow{stealth};}}}
   },
   ->>>-/.style={postaction={decorate,
   decoration={markings,mark=at position .75 with {\arrow{stealth};}}}
   },
    ->>>>-/.style={postaction={decorate,
   decoration={markings,mark=at position .55 with {\arrow{stealth};}}}
   },
    ->>>>>-/.style={postaction={decorate,
   decoration={markings,mark=at position .85 with {\arrow{stealth};}}}
   },
->>-/.style={postaction={decorate,
   decoration={markings,mark=at position .50 with {\arrow{doubled};}}}
   },   
   ->>>>>-/.style={postaction={decorate,
   decoration={markings,mark=at position 1 with {\arrow{stealth};}}}
   },   
}

\usetikzlibrary{shapes.geometric} 
\tikzset{cross/.style={cross out, draw=black, minimum size=2*(#1-\pgflinewidth), inner sep=0pt, outer sep=0pt},
cross/.default={4pt}}

\pgfdeclarelayer{bottom}
\pgfdeclarelayer{middleb}
\pgfdeclarelayer{middlet}
\pgfdeclarelayer{top}
\pgfsetlayers{bottom,middleb,main,middlet,top}
\usepgfplotslibrary{colorbrewer}
\pgfplotsset{compat=newest}
\pgfplotsset{colormap/violet}
\pgfdeclarelayer{background layer}
\pgfdeclarelayer{foreground layer}
\pgfsetlayers{background layer,main,foreground layer}

\newcommand{\ep}[1]{\mathrm{e}^{#1}}
\newcommand{\iu}{\mathrm{i}}
\newcommand{\caH}{\mathcal{H}}

\newcommand{\caC}{\mathcal{C}}
\newcommand{\caK}{\mathcal{K}}

\usepackage[deletedmarkup=xout,defaultcolor=ForestGreen]{changes}
\definechangesauthor[name=Sven, color=magenta]{sven}

\title{Anyons in the $\pi$-flux phase of \\
fermionic matter coupled to a $\mathbb{Z}_2$-gauge field}
\author[1]{Sven Bachmann}
\author[2]{Leonardo Goller}
\author[2]{Marcello Porta}
\affil[1]{Department of Mathematics, The University of British Columbia, Vancouver, BC V6T 1Z2, Canada}
\affil[2]{Mathematics Area, SISSA, Via Bonomea 265, 34136 Trieste, Italy}

\date{\today}

\begin{document}

\maketitle

\begin{abstract}
We consider a system of weakly interacting spinful lattice fermions coupled to a dynamical $\mathbb{Z}_2$ gauge field. Using reflection positivity, we prove that the ground state lies in the sector of a uniform $\pi$-flux per plaquette and that the monopoles are massive. In the presence of a staggered mass for the fermions, this yields a fully gapped, four-dimensional ground state space on large tori. It is topologically ordered. By considering adiabatic $\pi$-flux insertion, we construct dressed monopole excitations, show that their self-braiding is proportional to the Hall conductance and hence vanishes, and prove that their braiding with the fermionic excitations is that of the toric code.
\end{abstract}


\tableofcontents

\section{Introduction}

The theoretical possibility of anyons, namely quantum particles in two-dimensional space that exhibit general phases under braiding beyond the boson-fermion dichotomy was pointed out in~\cite{Leinaas:1977fm} in a differential geometric framework, in~\cite{goldin1980particle} in an algebraic one and in~\cite{wilczek1982magnetic,PhysRevLett.51.2250} using a magnetic picture, while it was already implicit in~\cite{frohlich1976new}. A Feynman path integral approach discussing the relation with the braid group appeared in~\cite{wu1984general}. The experimental observation of the fractional quantum Hall effect and its theoretical explanation~\cite{PhysRevLett.50.1395} using an anyonic wavefunction, see also ~\cite{PhysRevLett.52.1583,PhysRevLett.53.722}, validated the theoretical intuition. These exotic ground states that support anyonic excitations in two space dimensions are the prototypical examples of topological order. A first general relativistic quantum field theory of anyons was proposed in~\cite{frohlich2021statistics,frohlich1988quantum}, see also the review~\cite{Froehlich}, and we refer to~\cite{fredenhagen1989superselection} for an operator algebraic approach. When the Hamiltonian is defined on a finite torus, the anyonic properties are realized as ground state degeneracy that is not accompanied with any local order parameter~\cite{WenNiuTQO}. 

The possibility to use topologically ordered ground state spaces as stable quantum memories and braiding as fault tolerant quantum computational method~\cite{freedman2003topological} reignited the interest in the topic. The exactly solvable models of~\cite{kitaev2003,kitaev2006}, as well as the string-net models~\cite{levin2005string}, exhibit both Abelian and non-Abelian anyons, but are arguably very artifical quantum spin systems whose potential large scale realization remains elusive, but see~\cite{ExpAnyons} for recent progress.

In this paper, we analyse a very natural lattice model of complex fermions coupled with a dynamical $\mathbb{Z}_2$-gauge field. It belongs to a class of lattice gauge theories that have been widely studied numerically in the last few years, starting from \cite{AGgauge, GRV, Gazit3}, characterized by a remarkably rich phase diagram, discussed later on in some more detail. Without any uncontrolled assumption, we prove that this model exhibits all the properties of a good anyon theory. On a large but finite torus, the model has a four-dimensional topologically ordered ground state space, with a basis being labelled by the cohomology classes of the torus, or more precisely by spin structures on the torus. Crucially, this is one of the rare lattice models that is not explicitly solvable but where topological quantum order, namely ground state indistinguishability, can be proved, see also~\cite{lucia2024stability}. The almost degenerate ground state energy is separated from the rest of the spectrum by a gap, which can be interpreted as a (lower bound for the) mass of fermions and dressed monopoles. Quasi-particles excitations exhibit non-trivial, although Abelian, braiding properties.

Specifically, we consider the following Hamiltonian on the regular square lattice $\Gamma_L$ wrapped on the torus:
\begin{align}\label{hami}
    H & = -t \sum_{\substack{i,j\in \Gamma_L:\\(i,j) \in \text{E}(\Gamma_L)}} \sum_{\eta = \uparrow,\downarrow} \left(a^+_{i,\eta}\hat{\sigma}^z_{ij} a^-_{j,\eta} + \text{h.c.}\right) + m \sum_{i \in \Gamma_L} \sum_{\eta=\uparrow, \downarrow} (-1)^{i_1+i_2}\, n_{i, \eta}   \\
    &\quad + U \sum_{i\in \Gamma_{L}}\Big(n_{i,\uparrow} - \frac{1}{2}\Big)\Big(n_{i,\downarrow} - \frac{1}{2}\Big) ,
\end{align}
where $\text{E}(\Gamma_L)$ denotes the set of edges, acting on $\mathcal{H}_L =  \mathcal{F}_L \otimes \bigg(\bigotimes_{(i,j)\in \text{E}(\Gamma_L)} \mathbb{C}^2\bigg)$. Here $a^{\pm}_{i,\eta}$ are the fermionic creation/annihilation operators, $n_{i,\eta} =a^+_{i,\eta} a^-_{i,\eta}$, 
$\hat{\sigma}^z_{kl}$ is the third Pauli matrix representing the $\mathbb{Z}_2$-vector potential attached to edges of the lattice, and $t,m>0$. We could also add a plaquette term to the Hamiltonian, involving the sum of products of $\hat \sigma^{z}$ operators associated with the edges of the elementary plaquettes, and all our results would apply to that model as well, for $t$ large enough. The first term in the Hamiltonian describes nearest-neighbour hopping on the square lattice, twisted by the presence of a $\mathbb{Z}_{2}$ gauge field; the second term is a staggered mass term, that introduces a chemical potential imbalance betweem the two sublattices forming $\Gamma_{L}$; while the last term is a Hubbard interaction, that describes the on-site attraction $(U<0)$ or repulsion $(U>0)$ of the fermions. There are two natural local gauge transformations in this system: the `fermionic' $U(1)$-transformation implemented by the unitary $\ep{-\iu\phi n_{i, \eta}}$ and the `pure $\mathbb{Z}_2$' gauge transformation implemented by $\hat{\sigma}^x_{i, i+e_x}  \hat{\sigma}^x_{i, i+e_y}  \hat{\sigma}^x_{i-e_x, i}  \hat{\sigma}^x_{i-e_y, i}$. It is immediate to check that none of them separately is a symmetry of the system, but their combination
\begin{equation}
    Q_i = \hat{\sigma}^x_{i, i+e_x}  \hat{\sigma}^x_{i, i+e_y}  \hat{\sigma}^x_{i-e_x, i}  \hat{\sigma}^x_{i-e_y, i} \ep{-\iu\pi n_{i, \uparrow}-i \pi n_{i, \downarrow}}
\end{equation}
is such that $Q_i H Q_i = H$ for all $i$. Physical states $\ket{\Psi}\in\mathcal{H}_L$ must satisfy $Q_i\ket\Psi = \ket\Psi$, which is a discrete version of Gauss' law. Starting from a product state of the form $\ket\psi\otimes \ket{\bm{\sigma}}\in\mathcal{H}_L$, where $\ket{\bm{\sigma}}$ is a common eigenstate of all $\hat\sigma^z_i$, the physical vector
\begin{equation}\label{eq:IntroPhysical}
\prod_{i \in \Gamma_L} \bigg(\frac{1+Q_i}{2} \bigg)\ket\psi\otimes \ket{\bm{\sigma}}
\end{equation}
describes a loop gas state, exactly as in~\cite{kitaev2006}, which exhibits exotic entanglement properties that are typical of topological order~\cite{PhysRevLett.96.110404,Levin_2006}.

The Hamiltonian \eqref{hami} with $m=0$ and $U=0$ has already been considered in \cite{GP}, where it was shown, using reflection positivity introduced in this context in~\cite{Lieb}, that the ground states of the system have an expression of the form~\eqref{eq:IntroPhysical}, where $\bm{\sigma}$ is a $\pi$-flux configuration, namely the product of $\sigma^z_{ij}$ around any plaquette of the lattice equals $-1$. In fact, \cite{GP} shows that monopoles (which necessarily come in pairs), namely plaquettes where the product equals $+1$, have a finite energy cost so that the $\pi$-flux sector is protected by a gap from the other sectors. However, the fermionic Hamiltonian in this uniform background is gapless as it exhibits Dirac cones~\cite{PhysRevB.50.7526}. In this paper, we gap out these excitations by introducing a staggered mass term which preserves reflection positivity. As a result, we prove that the Hamiltonian with $m>0$ and for $|U|$ small enough has a four-dimensional ground state space with a splitting that is exponentially small in the system size, and the patch is well separated from the rest of the spectrum as displayed in \hyperref[fig:gapf]{Figure 1b}.

\begin{figure}
    \centering   
    \subfloat[Fermion dispersion relation]{ \includegraphics[width=0.45\textwidth]{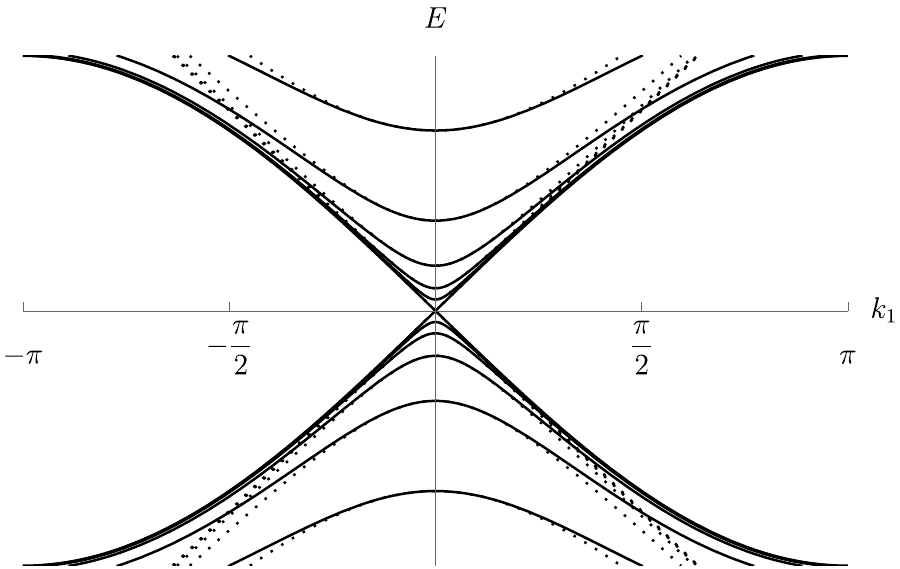}%
        \label{fig:charge gap}%
        }%
    \hspace{1cm}%
    \subfloat[Spectrum of the full Hamiltonian]{%
        \includegraphics[width=0.375\textwidth]{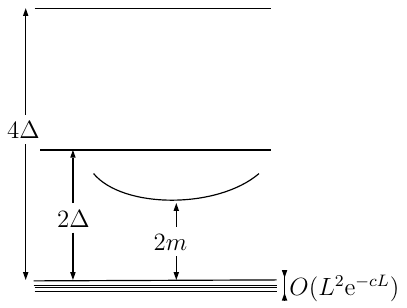}%
        \label{fig:gapf}%
        }%
    \caption{(a) Effect of mass term on $\pi$-flux dispersion relations: the gap opened by the pertubation is proportional to $\frac{m}{t}$. Solid and dotted lines represent respectively lattice and continuum dispersion relations. (b) The spectral structure of the full Hamiltonian: $\Delta$ is the energy necessary to create a monopole, while $m$ is the minimal energy needed to create a fermion (for $U=0$). The ground state splitting is exponentially small. The picture is stable for $|U|$ small enough.}
\end{figure}

The ground state space can be characterized using Wilson loop operators. Different ground states can be distinguished by measuring the $\mathbb{Z}_2$-holonomy along non-contractible cycles: This is given by the operator $\hat Z_\caC$ which is a product of $\hat\sigma^z$ along the edges of each of a non-contractible cycle. In order to connect these orthogonal ground states, we will construct $\pi$-flux threading operators $W_{\caC^*}$ associated with non-trivial cocycles, namely loops in the dual lattice. For this, we consider twisted Hamiltonians, where the hopping terms of \eqref{hami} ({\it i.e.} the terms proportional to $t$) are given an additional $\ep{\pm\iu\phi}$  whenever they cross the non-trivial cycle, see \autoref{FIP}. Since this family of Hamiltonians is gapped, the ground state spaces $P_{\caC^*}(\phi)$ are mapped onto each other by the spectral flow~\cite{Hastings_2005, Bachmann_2011}, which acts non-trivially only in the neighbourhood of the twisting line. At $\phi = \pi$, this non-trivial large gauge transformation can be offset by a $\mathbb{Z}_2$-gauge transformation in the sense that their product leaves the ground state space invariant. In order to determine the effect of this parallel transport on the ground state space, we will show that $\hat Z_{\caC}$ and $W_{\caC^*}$ anticommute for two geometrically orthogonal loops, which shows that $W_{\caC^*}$ permute the topologically ordered ground states~\cite{Wen1990TopologicalOI, PhysRevLett.84.3370, PhysRevB.70.245118}.

\begin{figure}
    \centering
    \begin{tikzpicture}[scale = 0.5]
      \draw[] (0,0) -- (0,8);
      \draw[] (4,0) -- (4,8);
      \draw[] (0.5,-0.24) -- (0.5,7.76);
      \draw[dashed] (0.5,8.24) -- (0.5,7.76);
      \draw[] (1,-0.30) -- (1,7.7);
      \draw[dashed] (1,8.3) -- (1,7.7);
      \draw[] (3.5,-0.24) -- (3.5,7.76);
      \draw[dashed] (3.5, 8.24) -- (3.5, 7.76);
      \draw[] (3,-0.30) -- (3,7.7);
      \draw[dashed] (3,8.3) -- (3,7.7);
      \draw[] (2,-0.35) -- (2,7.65);
      \draw[dashed] (2,8.35) -- (2,7.65);
      \draw[] (1.5,-0.32) -- (1.5,7.68);
      \draw[dashed] (1.5,8.32) -- (1.5,7.68);
      \draw[] (2.5,-0.32) -- (2.5,7.68);
       \draw[dashed] (2.5,8.32) -- (2.5,7.68);
      \draw (0,0) arc(180:360:2cm and 0.35cm);
      \draw[dashed] (4,0) arc(0:180:2cm and 0.35cm);
       \draw (0,1) arc(180:360:2cm and 0.35cm);
      \draw[dashed] (4,1) arc(0:180:2cm and 0.35cm);
       \draw (0,2) arc(180:360:2cm and 0.35cm);
      \draw[dashed] (4,2) arc(0:180:2cm and 0.35cm);
       \draw (0,3) arc(180:360:2cm and 0.35cm);
      \draw[dashed] (4,3) arc(0:180:2cm and 0.35cm);
       \draw (0,4) arc(180:360:2cm and 0.35cm);
      \draw[dashed] (4,4) arc(0:180:2cm and 0.35cm);
       \draw (0,5) arc(180:360:2cm and 0.35cm);
      \draw[dashed] (4,5) arc(0:180:2cm and 0.35cm);
       \draw (0,6) arc(180:360:2cm and 0.35cm);
      \draw[dashed] (4,6) arc(0:180:2cm and 0.35cm);
      \draw (0,7) arc(180:360:2cm and 0.35cm);
      \draw[dashed] (4,7) arc(0:180:2cm and 0.35cm);
      \draw (0,8) arc(180:360:2cm and 0.35cm);
      \draw[] (4,8) arc(0:180:2cm and 0.35cm);
       \draw[dashed] (0,10) arc(180:0:2cm and 0.35cm);
          \draw[dashed] (0,0) arc(180:0:2cm and 0.35cm);
    
      \node[right] (b) at (2.2,9) {{$\Phi(t)$}};
        \draw[dashed] (0,0) -- (0,-2);
         \draw[dashed] (4,0) -- (4,-2);
         \draw[dashed] (0,8) -- (0,10);
         \draw[dashed] (4,8) -- (4,10);
         \draw[dashed]
         (0,10) arc(180:360:2cm and 0.35cm);
          \draw[dashed]
         (0,-2) arc(180:360:2cm and 0.35cm);
         \draw[thick, -Stealth,line width = 0.05cm] (2,7.7) -- (2,10);
         \draw[thick,line width = 0.05cm ] (2,-0.4) -- (2,-2);
         \draw[red, thick] (0.25,-0.12) -- (0.25,7.88);
         \draw[thick] (0,8) arc(180:222:2cm and 0.325cm);
          \draw[thick] (0,7) arc(180:222:2cm and 0.325cm);
          \draw[thick] (0,6) arc(180:222:2cm and 0.325cm);
          \draw[thick] (0,5) arc(180:222:2cm and 0.325cm);
          \draw[thick] (0,4) arc(180:222:2cm and 0.325cm);
          \draw[thick] (0,3) arc(180:222:2cm and 0.325cm);
          \draw[thick] (0,2) arc(180:222:2cm and 0.325cm);
          \draw[thick] (0,1) arc(180:222:2cm and 0.325cm);
          \draw[thick] (0,0) arc(180:222:2cm and 0.325cm);
     \end{tikzpicture}
    \caption{Picture of the flux threading procedure in a portion of the torus: the hoppings on the edges cut by the blue line acquire an extra $e^{\pm i \phi}$ (depending on the orientation).}
    \label{FIP}
\end{figure}
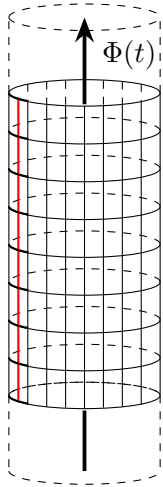

  Anyonic excitations are obtained by opening such lines: open Wilson lines $a^+_{i,\eta} \hat Z_{\caC_{i,j}}a^+_{j,\eta'}$ create a pair of fermionic excitations at sites $i,j$ bound together by a gauge field line along $\caC_{i,j}$ while open $W_{\caC^*}$ operators create a pair of monopoles on the background at the endpoints of $\caC^*$. Braiding of monopoles around fermions is represented schematically in \autoref{braiding1}. The phase acquired by the wave function of two fermion state contains a universal term that is precisely $-1$ as one would expect by Aharonov-Bohm phase of a electric charge $1$ particle and magnetic charge  $\pi$ particle. We shall also show that the braiding of dressed monopoles around each other is trivial; in fact, we will relate the self-braiding phase to the Hall conductance, which vanishes in the $\pi$-flux sector of the model.

  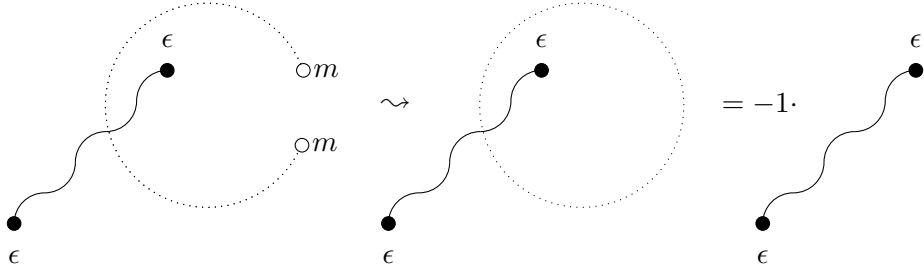
\begin{figure}
      \centering
      \begin{tikzpicture}[scale=0.9]
\node[above] (b) at (0,0.2) {\scalebox{1}{$\epsilon$}};
\node[below] (b) at (-2.25,-2.5) {\scalebox{1}{$\epsilon$}};
\draw[black, ->, dotted] (2,0) arc (20:340:1.5cm);
\node[right] (b) at (2,0) {\scalebox{1}{$m$}};
\node[right] (b) at (1.98,-1.1) {\scalebox{1}{$m$}};
\draw[black, ->, dotted] (2,0) arc (20:340:1.5cm);
\draw[] (0,0) arc (90:180:0.45cm);
\draw[] (-0.45,-0.45) arc (360:270:0.45cm);
\draw[] (-0.9,-0.9) arc (90:180:0.45cm);
\draw[] (-1.35,-1.35) arc (360:270:0.45cm);
\draw[] (-1.80,-1.80) arc (90:180:0.45cm);
\draw[black,fill=black] (-2.25,-2.25) circle (.1 cm);
\draw[black,fill=white] (2,0) circle (.1 cm);
\draw[black,fill=white] (1.98,-1.1) circle (.1 cm);
\draw[black,fill=black] (0,0) circle (.1 cm);
\begin{scope}[xshift=5.5cm]
    \node[right] (b) at (-2.5,-0.5) {\scalebox{1}{$\leadsto$}};
    \node[above] (b) at (0,0.2) {\scalebox{1}{$\epsilon$}};
\node[below] (b) at (-2.25,-2.5) {\scalebox{1}{$\epsilon$}};
\draw[black, dotted] (2,0) arc (20:380:1.5cm);
\draw[] (0,0) arc (90:180:0.45cm);
\draw[] (-0.45,-0.45) arc (360:270:0.45cm);
\draw[] (-0.9,-0.9) arc (90:180:0.45cm);
\draw[] (-1.35,-1.35) arc (360:270:0.45cm);
\draw[] (-1.80,-1.80) arc (90:180:0.45cm);
\draw[black,fill=black] (-2.25,-2.25) circle (.1 cm);
\draw[black,fill=black] (0,0) circle (.1 cm);
\end{scope}
\begin{scope}[xshift=11cm]
    \node[right] (b) at (-3,-0.5) {\scalebox{1}{$=-1 \cdot$}};
     \node[above] (b) at (0,0.2) {\scalebox{1}{$\epsilon$}};
\node[below] (b) at (-2.25,-2.5) {\scalebox{1}{$\epsilon$}};
\draw[] (0,0) arc (90:180:0.45cm);
\draw[] (-0.45,-0.45) arc (360:270:0.45cm);
\draw[] (-0.9,-0.9) arc (90:180:0.45cm);
\draw[] (-1.35,-1.35) arc (360:270:0.45cm);
\draw[] (-1.80,-1.80) arc (90:180:0.45cm);
\draw[black,fill=black] (-2.25,-2.25) circle (.1 cm);
\draw[black,fill=black] (0,0) circle (.1 cm);
\end{scope}
      \end{tikzpicture}
      \caption{Braiding of a monopole around a fermion, its annihilation into the vacuum, and the resulting anyonic phase $-1$.}
      \label{braiding1}
  \end{figure}

As mentioned above, the model we study is closely related to the $\mathbb{Z}_{2}$ gauge theory coupled to fermionic matter whose numerical study has been initiated in \cite{AGgauge, GRV} for spinful fermions, see also \cite{Borla} for the spinless case. In these works it has been observed that, in the absence of a mass term, the model has a semimetallic phase at half-filling (studied rigorously in \cite{GP}), while away from half-filling, and depending on the model parameters, it shows a BCS-BEC crossover and a superconducting phase. See also \cite{Gazit3} for further numerical study of the semimetallic phase. Among the recent developments, we mention the proposed experimental realization of such model using cold atoms \cite{Barbieroexp, Homeierexp}, the observation of a confinement transition of the gauge theory in $1+1$ dimensions via quantum computing \cite{Mildenberger}, and in higher dimensions using quantum simulators \cite{kebric}.

To conclude this introduction, we discuss the relation of our model with Kitaev's honeycomb model introduced in~\cite{kitaev2006}, in the gapped phase. Although the emergent anyon theories (superselection sectors) are the same as that of the toric code, they are microscopically different. The present model is made up of complex fermions so there is a full $U(1)$ symmetry. The presence of a continuous symmetry is crucial for our construction of local flux insertion unitaries and loop operators. Following~\cite{Lieb}, reflection positivity allows us to rigorously identify the $\pi$-flux sector as that of the ground state (incidentally, in Kitaev's honeycomb model, the ground state lies in the uniform $0$-flux sector). What is more, proceeding as in \cite{GP}, the chessboard estimate yields a proof of the existence of a mass gap for the monopoles and its extensivity in the number of monopoles: This is conjectured to be true in Kitaev's model and analysed numerically in Appendix~A of \cite{kitaev2006}. Interestingly, the chessboard configuration appears in this numerical study  as that which minimizes the energy per volume. Here, the monopole excitations are a dressed version of `pure $\mathbb{Z}_2$' vortices, which calls upon the $U(1)$-symmetry and quasi-adiabatic technology. The fermionic excitations are simpler, being only related by a bare string operator. Because of the presence of the staggered mass term, they are gapped, thus yielding a fully gapped ground state for the model. Unlike in the toric code, our quasi-particles are dynamical: They are not eigenstates of the Hamiltonian but rather localized wavepackets. From that point of view, there are similarities with~\cite{bachmann2023dynamical} but we consider here the full model, while the analysis there was restricted to the Hilbert space spanned by two quasi-particle states. Finally, we explicitly include a small  Hubbard interaction, thereby exhibiting the robustness of topological order, even when explicit diagonalization is not possible.

The paper is organized as follows. We define the model in \hyperref[sec1]{Section \ref*{sec1}} and introduce its symmetries. In Section \ref{sec:res} we state our main theorems: \autoref{thm1} on the stability of the $\pi$-flux sector and the mass of the monopoles, \autoref{thm2} on properties of the ground state, in particular topological order, and \autoref{thm3} exhibiting the braiding properties. \autoref{thm1} is proved in \hyperref[sec2]{Section \ref*{sec2}}. We adapt the argument of \cite{GP} based on reflection positivity and the chessboard estimate, to show the optimality of the $\pi$-flux phase and to estimate the mass of the monopoles. We determine the four dimensional ground state manifold of the system, parametrized by the values of the magnetic fluxes across the two non-contractible loops of the torus, and we prove the exponential closeness of the energies. The effect of small Hubbard interaction is taken into account using fermionic cluster expansion. In \hyperref[sec:topord]{Section \ref*{sec:topord}} we use the information about the ground state space derived in \hyperref[sec2]{Section \ref*{sec2}}, to prove the topological order of the ground state space. Finally, in \hyperref[sec:TQO]{Section \ref*{sec:TQO}} we prove \autoref{thm3}. We develop the flux threading procedure, compute commutators of loops operators on the ground state space and analyse the braiding.

\paragraph{Notations.}The lattice $\Gamma_L$ is bipartite, meaning that it consists of two sublattices $\Gamma^{\text{A}}_L$ and $\Gamma^{\text{B}}_L$ which are represented with white and black dots respectively. 

\begin{figure}[h]
    \centering
   \begin{tikzpicture}[scale=0.6]
\draw[black, ->>>>-] (1,1) -- (2,1);
\draw[black, ->>>>-] (0,0) -- (0,1);
\draw[black, ->>>>-] (0,0) -- (1,0);
\draw[black, ->>>>-] (1,0) -- (1,1);
\draw[black, ->>>>-] (1,0) -- (2,0);
\draw[black, ->>>>-] (2,0) -- (2,1);
\draw[black, ->>>>-] (2,0) -- (2,1);
\draw[black, ->>>>-] (3,0) -- (4,0);
\draw[black, ->>>>-] (3,0) -- (3,1);
\draw[black, ->>>>-] (4,0) -- (5,0);
\draw[black, ->>>>-] (2,0) -- (3,0); 
\draw[black, ->>>>-] (4,0) -- (4,1);
\draw[black, ->>>>-] (5,0) -- (6,0);
\draw[black, ->>>>-] (5,0) -- (5,1);
\draw[black, ->>>>-] (6,0) -- (7,0);
\draw[black, ->>>>-] (6,0) -- (6,1);
\draw[black, ->>>>-] (7,0) -- (7,1);

\draw[black, ->>>>-] (0,1) -- (1,1);
\draw[black, ->>>>-] (0,1) -- (0,2);
\draw[black, ->>>>-] (1,1) -- (1,2);
\draw[black, ->>>>-] (2,1) -- (2,2);

\draw[black, ->>>>-] (2,1) -- (3,1);
\draw[black, ->>>>-] (3,1) -- (3,2);
\draw[black, ->>>>-] (3,1) -- (4,1);
\draw[black, ->>>>-] (4,1) -- (4,2);

\draw[black, ->>>>-] (4,1) -- (5,1);
\draw[black, ->>>>-] (5,1) -- (5,2);
\draw[black, ->>>>-] (5,1) -- (6,1);
\draw[black, ->>>>-] (6,1) -- (6,2);

\draw[black, ->>>>-] (6,1) -- (7,1);
\draw[black, ->>>>-] (7,1) -- (7,2);
\draw[black, ->>>>-] (1,2) -- (2,2);
\draw[black, ->>>>-] (1,2) -- (1,3);

\draw[black, ->>>>-] (0,2) -- (1,2);
\draw[black, ->>>>-] (0,2) -- (0,3);

\draw[black, ->>>>-] (2,2) -- (3,2);
\draw[black, ->>>>-] (2,2) -- (2,3);
\draw[black, ->>>>-] (3,2) -- (4,2);
\draw[black, ->>>>-] (3,2) -- (3,3);
\draw[black, ->>>>-] (4,2) -- (5,2);
\draw[black, ->>>>-] (4,2) -- (4,3);
\draw[black, ->>>>-] (5,2) -- (6,2);
\draw[black, ->>>>-] (5,2) -- (5,3);
\draw[black, ->>>>-] (5,2) -- (6,2);
\draw[black, ->>>>-] (5,2) -- (5,3);
\draw[black, ->>>>-] (6,2) -- (7,2);
\draw[black, ->>>>-] (6,2) -- (6,3);
\draw[black, ->>>>-] (4,2) -- (5,2);
\draw[black, ->>>>-] (5,3) -- (6,3);
\draw[black, ->>>>-] (7,2) -- (7,3);
\draw[black, ->>>>-] (0,3) -- (1,3);
\draw[black, ->>>>-] (0,3) -- (0,4);
\draw[black, ->>>>-] (1,3) -- (2,3);
\draw[black, ->>>>-] (2,3) -- (3,3);
\draw[black, ->>>>-] (3,3) -- (4,3);
\draw[black, ->>>>-] (4,3) -- (5,3);
\draw[black, ->>>>-] (0,4) -- (0,5);
\draw[black, ->>>>-] (5,3) -- (6,3);
\draw[black, ->>>>-] (6,3) -- (7,3);
\draw[black, ->>>>-] (0,4) -- (0,5);
\draw[black, ->>>>-] (0,4) -- (1,4);
\draw[black, ->>>>-] (1,4) -- (2,4);
\draw[black, ->>>>-] (2,4) -- (3,4);
\draw[black, ->>>>-] (3,4) -- (4,4);
\draw[black, ->>>>-] (4,4) -- (5,4);
\draw[black, ->>>>-] (0,5) -- (0,6);
\draw[black, ->>>>-] (5,4) -- (6,4);
\draw[black, ->>>>-] (6,4) -- (7,4);
\draw[black, ->>>>-] (0,5) -- (1,5);
\draw[black, ->>>>-] (1,5) -- (2,5);
\draw[black, ->>>>-] (2,5) -- (3,5);
\draw[black, ->>>>-] (3,5) -- (4,5);
\draw[black, ->>>>-] (4,5) -- (5,5);
\draw[black, ->>>>-] (0,6) -- (0,7);
\draw[black, ->>>>-] (5,5) -- (6,5);
\draw[black, ->>>>-] (6,5) -- (7,5);
\draw[black, ->>>>-] (0,6) -- (1,6);
\draw[black, ->>>>-] (1,6) -- (2,6);
\draw[black, ->>>>-] (2,6) -- (3,6);
\draw[black, ->>>>-] (3,6) -- (4,6);
\draw[black, ->>>>-] (4,6) -- (5,6);
\draw[black, ->>>>-] (7,6) -- (7,7);
\draw[black, ->>>>-] (5,6) -- (6,6);
\draw[black, ->>>>-] (6,6) -- (7,6);
\draw[black, ->>>>-] (0,7) -- (1,7);
\draw[black, ->>>>-] (1,7) -- (2,7);
\draw[black, ->>>>-] (2,7) -- (3,7);
\draw[black, ->>>>-] (3,7) -- (4,7);
\draw[black, ->>>>-] (4,7) -- (5,7);
\draw[black, ->>>>-] (7,5) -- (7,6);
\draw[black, ->>>>-] (5,7) -- (6,7);
\draw[black, ->>>>-] (6,7) -- (7,7);
\draw[black, ->>>>-] (7,3) -- (7,4);
\draw[black, ->>>>-] (7,4) -- (7,5);
\draw[black, ->>>>-] (1,4) -- (1,5);
\draw[black, ->>>>-] (1,3) -- (1,4);
\draw[black, ->>>>-] (1,5) -- (1,6);
\draw[black, ->>>>-] (1,6) -- (1,7);
\draw[black, ->>>>-] (2,4) -- (2,5);
\draw[black, ->>>>-] (2,3) -- (2,4);
\draw[black, ->>>>-] (2,5) -- (2,6);
\draw[black, ->>>>-] (2,6) -- (2,7);
\draw[black, ->>>>-] (3,4) -- (3,5);
\draw[black, ->>>>-] (3,3) -- (3,4);
\draw[black, ->>>>-] (3,5) -- (3,6);
\draw[black, ->>>>-] (3,6) -- (3,7);
\draw[black, ->>>>-] (4,4) -- (4,5);
\draw[black, ->>>>-] (4,3) -- (4,4);
\draw[black, ->>>>-] (4,5) -- (4,6);
\draw[black, ->>>>-] (4,6) -- (4,7);
\draw[black, ->>>>-] (5,4) -- (5,5);
\draw[black, ->>>>-] (5,3) -- (5,4);
\draw[black, ->>>>-] (5,5) -- (5,6);
\draw[black, ->>>>-] (5,6) -- (5,7);
\draw[black, ->>>>-] (6,4) -- (6,5);
\draw[black, ->>>>-] (6,3) -- (6,4);
\draw[black, ->>>>-] (6,5) -- (6,6);
\draw[black, ->>>>-] (6,6) -- (6,7);
\draw[black, dotted, ->-] (-0.5,-0.5) -- (7.5,-0.5);
\draw[black, dotted, ->-] (-0.5,7.5) -- (7.5,7.5);
\draw[black, dotted, ->>-] (-0.5,-0.5) -- (-0.5,7.5);
\draw[black, dotted, ->>-] (7.5,-0.5) -- (7.5,7.5);
\draw[black, ->>>>-] (0, -0.5) --(0,0);
\draw[black, ->>>>-] (-0.5,0) -- (0,0);
\draw[black, ->>>>-] (1, -0.5) --(1,0);
\draw[black, ->>>>-] (-0.5,1) -- (0,1);
\draw[black, ->>>>-] (2, -0.5) --(2,0);
\draw[black, ->>>>-] (-0.5,2) -- (0,2);
\draw[black, ->>>>-] (3, -0.5) --(3,0);
\draw[black, ->>>>-] (-0.5,3) -- (0,3);
\draw[black, ->>>>-] (4, -0.5) --(4,0);
\draw[black, ->>>>-] (-0.5,4) -- (0,4);
\draw[black, ->>>>-] (5, -0.5) --(5,0);
\draw[black, ->>>>-] (-0.5,6) -- (0,6);
\draw[black, ->>>>-] (6,-0.5) -- (6,0);
\draw[black, ->>>>-] (-0.5,5) -- (0,5);
\draw[black, ->>>>-] (7, -0.5) --(7,0);
\draw[black, ->>>>-] (-0.5,7) -- (0,7);

\draw[black, ->>>-] (0,7) -- (0,7.5);
\draw[black, ->>>-] (7,7) -- (7.5,7);
\draw[black, ->>>-] (1,7) -- (1,7.5);
\draw[black, ->>>-] (7,6) -- (7.5,6);
\draw[black, ->>>-] (2,7) -- (2,7.5);
\draw[black, ->>>-] (7,5) -- (7.5,5);
\draw[black, ->>>-] (3,7) -- (3,7.5);
\draw[black, ->>>-] (7,4) -- (7.5,4);
\draw[black, ->>>-] (4,7) -- (4,7.5);
\draw[black, ->>>-] (7,0) -- (7.5,0);
\draw[black, ->>>-] (5,7) -- (5,7.5);
\draw[black, ->>>-] (7,3) -- (7.5,3);
\draw[black, ->>>-] (6,7) -- (6,7.5);
\draw[black, ->>>-] (7,2) -- (7.5,2);
\draw[black, ->>>-] (7,7) -- (7,7.5);
\draw[black, ->>>-] (7,1) -- (7.5,1);

\draw[black,fill=white] (0,0) circle (.1 cm);
\draw[black,fill=white] (1,1) circle (.1 cm);
\draw[black,fill=white] (3,3) circle (.1 cm);
\draw[black,fill=white] (2,2) circle (.1 cm);
\draw[black,fill=white] (4,4) circle (.1 cm);
\draw[black,fill=white] (5,5) circle (.1 cm);
\draw[black,fill=white] (6,6) circle (.1 cm);
\draw[black,fill=white] (7,7) circle (.1 cm);
\draw[black,fill=white] (3,5) circle (.1 cm);
\draw[black,fill=white](2,0) circle (.1 cm);
\draw[black,fill=white] (0,2) circle (.1 cm);
\draw[black,fill=white] (4,0) circle (.1 cm);
\draw[black,fill=white] (0,4) circle (.1 cm);
\draw[black,fill=white] (6,0) circle (.1 cm);
\draw[black,fill=white] (0,6) circle (.1 cm);
\draw[black,fill=white] (1,3) circle (.1 cm);
\draw[black,fill=white] (1,5) circle (.1 cm);
\draw[black,fill=white] (1,7) circle (.1 cm);
\draw[black,fill=white] (7,1) circle (.1 cm);
\draw[black,fill=white] (5,1) circle (.1 cm);
\draw[black,fill=white] (3,1) circle (.1 cm);

\draw[black,fill=white] (4,2) circle (.1 cm);
\draw[black,fill=white] (6,2) circle (.1 cm);
\draw[black,fill=white] (2,6) circle (.1 cm);
\draw[black,fill=white] (2,4) circle (.1 cm);
\draw[black,fill=white] (6,4) circle (.1 cm);
\draw[black,fill=white] (5,3) circle (.1 cm);
\draw[black,fill=white] (7,3) circle (.1 cm);
\draw[black,fill=white] (3,5) circle (.1 cm);
\draw[black,fill=white] (3,7) circle (.1 cm);
\draw[black,fill=white] (7,5) circle (.1 cm);
\draw[black,fill=white] (4,6) circle (.1 cm);
\draw[black,fill=white] (5,7) circle (.1 cm);

\draw[black,fill=black] (1,0) circle (.1 cm);
\draw[black,fill=black] (0,1) circle (.1 cm);
\draw[black,fill=black] (3,0) circle (.1 cm);
\draw[black,fill=black] (0,3) circle (.1 cm);
\draw[black,fill=blue] (5,0) circle (.1 cm);
\draw[blue,fill=black] (0,5) circle (.1 cm);
\draw[black,fill=black] (7,0) circle (.1 cm);
\draw[black,fill=black] (0,7) circle (.1 cm);
\draw[black,fill=black] (1,0) circle (.1 cm);
\draw[black,fill=black] (0,1) circle (.1 cm);
\draw[black,fill=black] (3,0) circle (.1 cm);
\draw[black,fill=black] (0,3) circle (.1 cm);
\draw[black,fill=black] (5,0) circle (.1 cm);
\draw[black,fill=black] (0,5) circle (.1 cm);
\draw[black,fill=black] (7,0) circle (.1 cm);
\draw[black,fill=black] (0,7) circle (.1 cm);

\draw[black,fill=black] (1,2) circle (.1 cm);
\draw[black,fill=black] (2,1) circle (.1 cm);
\draw[black,fill=black] (3,2) circle (.1 cm);
\draw[black,fill=black] (2,3) circle (.1 cm);
\draw[black,fill=black] (5,2) circle (.1 cm);
\draw[black,fill=black] (2,5) circle (.1 cm);
\draw[black,fill=black] (7,2) circle (.1 cm);
\draw[black,fill=black] (2,7) circle (.1 cm);
\draw[black,fill=black] (1,4) circle (.1 cm);
\draw[black,fill=black] (4,1) circle (.1 cm);
\draw[black,fill=black] (3,4) circle (.1 cm);
\draw[black,fill=black] (4,3) circle (.1 cm);
\draw[black,fill=black] (5,4) circle (.1 cm);
\draw[black,fill=black] (4,5) circle (.1 cm);
\draw[black,fill=black] (7,4) circle (.1 cm);
\draw[black,fill=black] (4,7) circle (.1 cm);

\draw[black,fill=black] (1,6) circle (.1 cm);
\draw[black,fill=black] (6,1) circle (.1 cm);
\draw[black,fill=black] (3,6) circle (.1 cm);
\draw[black,fill=black] (6,3) circle (.1 cm);
\draw[black,fill=black] (5,6) circle (.1 cm);
\draw[black,fill=black] (6,5) circle (.1 cm);
\draw[black,fill=black] (7,6) circle (.1 cm);
\draw[black,fill=black] (6,7) circle (.1 cm);

\draw[black, dotted] (1.5,-0.5) -- (1.5,7.5);
\draw[black, dotted] (2.5,-0.5) -- (2.5,7.5);
\draw[black, dotted] (3.5,-0.5) -- (3.5,7.5);
\draw[black, dotted] (4.5,-0.5) -- (4.5,7.5);
\draw[black, dotted] (5.5,-0.5) -- (5.5,7.5);
\draw[black, dotted] (6.5,-0.5) -- (6.5,7.5);
\draw[black, dotted] (7.5,-0.5) -- (7.5,7.5);
\draw[black, dotted] (-0.5,0.5) -- (7.5,0.5);
\draw[black, dotted] (-0.5,1.5) -- (7.5,1.5);
\draw[black, dotted] (-0.5,2.5) -- (7.5,2.5);
\draw[black, dotted] (-0.5,3.5) -- (7.5,3.5);
\draw[black, dotted] (-0.5,4.5) -- (7.5,4.5);
\draw[black, dotted] (0,5.5) -- (7.5,5.5);
\draw[black, dotted] (-0.5,6.5) -- (7.5,6.5);
\draw[black, dotted]
(0.5,-0.5) -- (0.5, 7.5);
   \end{tikzpicture}
   \caption{Oriented lattice $\Gamma_L$ and its dual $\Gamma^*_L$}
   \label{lattice}
\end{figure}
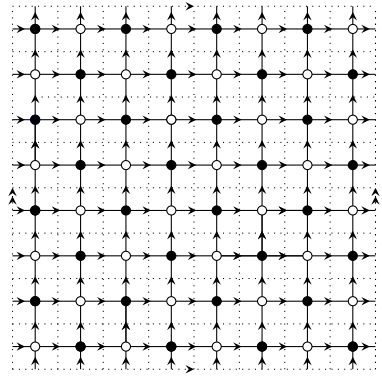

We will denote with $\text{V}(\Gamma_L)$, $\text{E}(\Gamma_L)$ and $\text{F}(\Gamma_L)$ respectively the set of vertices, edges and faces of the lattice $\Gamma_L$. Their elements will be written as $i \in \text{V}(\Gamma_L)$, $(i,j) \in \text{E}(\Gamma_L)$ (with the convention that $j$ and $i$ are respectively the starting and arrival point according to the orientation shown in the picture) and $p = (i,j,k,l) \in \text{F}(\Gamma_L)$. Often, with a slight abuse of notation, we will identify the set of vertices $\text{V}(\Gamma_L)$ with the lattice $\Gamma_L$ itself. The set of vertices, edges and faces of the dual lattice $\Gamma^*_L$ of $\Gamma_L$ (sketched with dotted lines) will be likewise denoted with $\text{V}(\Gamma^*_L)$, $\text{E}(\Gamma_L^*)$ and $\text{F}(\Gamma_L^*)$. Of course, there is a bijective correspondence between $\text{F}(\Gamma_L) \simeq \text{V}(\Gamma_L^*)$ and $\text{V}(\Gamma_L) \simeq \text{F}(\Gamma_L^*)$, which we shall sometimes use without further mention. In particular, the four edges sharing one vertex can be thought of as the dual edges bounding the corresponding dual face.

\paragraph{Acknowledgments.} The authors would like to thank the anonymous referees for their very thorough proofreading of this manuscript. SB \& LG would like to thank GSSI in L'Aquila, where this work originated, for their hospitality. LG would like to thank Fabrizio Caragiulo, Simone Fabbri, Tommaso Pedroni, Bruno Renzi and Harman Preet Singh for support and useful discussions throughout the realization of the paper. SB acknowledges financial support of NSERC of Canada. LG \& MP acknowledge financial support of the European Research Council through the
grant ERCStG MaMBoQ, n. 802901. MP acknowledges support from the MUR, PRIN
2022 project MaIQuFi cod. 20223J85K3. This work has been carried
out under the auspices of the GNFM of INdAM.

\section{The model}\label{sec1}

We consider a system of spinful fermions hopping on the vertices of a two-dimensional square lattice $\Gamma_L = \mathbb{Z}^2_L$ with periodic boundary conditions (and length $L \in 4 \mathbb{N}$) in the background of (deconfined) $\mathbb{Z}_2$-valued gauge fields living on the edges of the lattice.

The fermionic matter is described by the usual fermionic Fock space:
\begin{equation}
\mathcal{F}_L= \mathbb{C} \oplus \bigoplus_{n\geq 1} \ell^{2}(\Gamma_{L} \times \{\uparrow,\downarrow\})^{\wedge n}.
\end{equation}
Such space can be conveniently thought as the Hilbert space that arises by successive applications of fermionic creation and annihilation operators $a^{+}_{i,\eta}$ and $a^{-}_{i,\eta}$ (for any $i,\eta\in \Gamma_{L}\times \{\uparrow,\downarrow\}$) on the fermionic vacuum state $\ket{0}$. Antisymmetry of wavefunctions is naturally enforced by the canonical anticommutation relations:
\begin{equation}\label{CAR}
\{a^{+}_{i,\eta}, a^{-}_{j,\eta'} \} = \delta_{ij}\delta_{\eta,\eta'}\;,\qquad  \{ a^{+}_{i,\eta}, a^{+}_{j,\eta'} \} = \{ a^{-}_{i,\eta}, a^{-}_{j,\eta'} \}= 0\;,
\end{equation}
with the understanding that $a^{+}_{i,\eta} = (a^{-}_{i,\eta})^{*}$. The algebra $\mathcal{A}_{\text{Fer}}$ of fermionic observables is given by the (self-adjoint) polynomials in the creation and annihilation operators. A simple example is the number operator,
\begin{equation}
N = \sum_{i\in \Gamma_{L}}\sum_{\eta = \uparrow\downarrow} a^{+}_{i,\eta} a^{-}_{i,\eta}\;.
\end{equation}
The parity automorphism $\mathcal{P}$ of $\mathcal{A}_{\text{Fer}}$ is defined by
\begin{equation}
\mathcal{P}(\mathcal{O}) = (-1)^{N} \mathcal{O} (-1)^{N}\;.
\end{equation}  
Being an involution, the eigenvalues of $\mathcal{P}$ are $\pm 1$. We denote by $\mathcal{A}^{+}_{\text{Fer}}$ the algebra of polynomials which are even under $\mathcal{P}$ (eigenvalue $+1$) and by $\mathcal{A}^{-}_{\text{Fer}}$ the set of the polynomials that are odd under $\mathcal{P}$ (eigenvalue $-1$). In the following, we shall always consider physical observables that belong to the even subalgebra $\mathcal{A}^{+}_{\text{Fer}}$. In other words, all physical fermionic observables we shall consider in the following satisfy the global $\mathbb{Z}_2$-symmetry:
\begin{equation}\label{eq:thetaO}
\mathcal{O} = \mathcal{P} (\mathcal{O}).
\end{equation}

Even fermionic observables display the following locality property. Given two non-intersecting subsets $X, Y \subset \text{V}(\Gamma_L)$, then any pair of fermion-even observables $\mathcal{O}_X$ and $\mathcal{O}_Y$ localized in $X$ and $Y$, satisfy:
\begin{equation}
    [\mathcal{O}_X, \mathcal{O}_Y] =0.
\end{equation}
A simple example of fermionic Hamiltonian acting on $\mathcal{F}_L$ is the tight binding model:
\begin{equation}\label{hamif}
\begin{split}
    H^{\text{Fer}} &= -t \sum_{\substack{i,j\in \Gamma_L:\\(i,j) \in \text{E}(\Gamma_L)}} \sum_{\eta=\uparrow,\downarrow} (a^+_{i,\eta} a^-_{j,\eta} + \text{h.c.}) + m \sum_{i \in \Gamma_L} \sum_{\eta = \uparrow,\downarrow}(-1)^{i_1+i_2}\, n_{i,\eta} \\ &\quad + U \sum_{i\in \Gamma_{L}} \Big( n_{i,\uparrow} - \frac{1}{2} \Big)\Big( n_{i,\downarrow} - \frac{1}{2} \Big),
    \end{split}
\end{equation}
where $t,m>0$. The first term is the kinetic energy where the hopping is diagonal in spin space, the second describes a (staggered) mass, while the last term describes a Hubbard interaction. 

\begin{remark}
\begin{enumerate}
    \item The name mass will be understood when we diagonalize the Hamiltonian in the $\pi$-flux background: the corresponding term opens up a gap between the energy bands (see \emph{\hyperref[fig:charge gap]{Figure 1a}}).
    \item In general, a staggered mass term is only well-defined if the lattice is bipartite. For the current square lattice on the torus, this requires the length $L$ to be even.
    \item The Hubbard term is a standard, strictly on-site interaction for lattice fermions. Its effect will be studied via a convergent cluster expansion, which will allow us to prove the stability of all the properties exhibited by the non-interacting system. One of the key features of the Hubbard interaction is that it is compatible with reflection positivity \cite{Lieb}. As discussed later, see Remark \ref{rem:26}, our results can be extended in a straightforward way to a larger class of interactions, that preserve reflection positivity.
\end{enumerate}
\end{remark}

\subsection{Gauging fermion parity}
We now gauge the fermion parity to produce a topologically ordered state of matter. We introduce new degrees of freedom on each edge $(i,j)\in \text{E}(\Gamma_L)$ of the lattice, $\mathcal{H}^{\text{gauge}}_{ij} = \mathbb{C}[\mathbb{Z}_2] \simeq \mathbb{C}^2$, and the total Hilbert space of the gauge sector is 
\begin{equation}
\mathcal{H}^{\text{gauge}}_{\Gamma_L} = \bigotimes_{(i,j) \in \text{E}(\Gamma_L)} \mathcal{H}^{\text{gauge}}_{ij}.
\end{equation}
Note that we will slighlty abuse notations and often write $ij$ for an edge --- in the present case of a $\mathbb{Z}_2$ gauge field, the orientation plays no role. The algebra of pure gauge observables is generated by Pauli matrices $\{\hat\sigma^z_{ij}, \hat\sigma^x_{ij}:(i,j)\in \text{E}(\Gamma_L)\}$ satisfying the usual algebraic relations. They shall be respectively interpreted as a magnetic vector potential and an electric field.

The total Hilbert space for the gauged matter is given by:
\begin{equation}\label{eq:total H}
\mathcal{H}_{L} = \mathcal{F}_L \otimes \mathcal{H}^{\text{gauge}}_{L}\;.
\end{equation} 
where the tensor product is understood to be symmetric, and, accordingly, the fermionic and gauge observables act on each factor independently and they commute. We denote the corresponding observable algebra by $\mathcal{A}$, namely, it is the algebra of polynomials in creation and annihilation operators and magnetic vector potentials and electric fields. The global $\mathbb{Z}_2$ parity transformation is now promoted to a local $\mathbb{Z}_2$-transformation acting on each vertex $i \in \text{V}(\Gamma_L)$ and edges touching the vertex:

\begin{definition}[$\mathbb{Z}_2$-charges]
    For each vertex $i \in \text{V}(\Gamma_L)$, we define the $\mathbb{Z}_2$ charge operators as:
    \begin{equation}\label{eq:Def of charge}
        {Q}_i = {A}_i (-1)^{n_{i,\uparrow} +n_{i,\downarrow}},\qquad 
    \end{equation}
    where $A_i = \prod_{j:(i,j)\in \text{E}(\Gamma_L)}\hat\sigma^x_{ij}$.
    \begin{figure}[H]
\centering
\begin{tikzpicture}[scale=0.8]
\draw (1,1) grid (3,3);
\draw[black] (1,1) -- (1,0.5);
\draw[black] (2,1) -- (2,0.5);
\draw[black] (3,1) -- (3,0.5);
\draw[black] (3,1) -- (3.5,1);
\draw[black] (3,2) -- (3.5,2);
\draw[black] (3,3) -- (3.5,3);
\draw[black] (3,3) -- (3,3.5);
\draw[black] (2,3) -- (2,3.5);
\draw[black] (1,3) -- (1,3.5);
\draw[black] (1,3) -- (0.5,3);
\draw[black] (1,2) -- (0.5,2);
\draw[black] (1,1) -- (0.5,1);
\draw[black, dotted] (0.5,0.5)--(3.5,0.5) --(3.5, 3.5) -- (0.5,3.5)--(0.5,0.5);
\draw[red,fill=red] (1.5,2) circle (.09 cm);
\draw[red,fill=red] (2.5,2) circle (.09 cm);
\draw[red,fill=red] (2,1.5) circle (.09 cm);
\draw[red,fill=red] (2,2.5) circle (.09 cm);
\draw[red, thick] (2,1) -- (2,3);
\draw[red, thick] (1,2) -- (3,2);
\draw[red, thick] (1.5,1.5) -- (2.5,1.5) -- (2.5, 2.5) -- (1.5, 2.5) -- (1.5,1.5);
\node[above right] (a) at (2,2) {\scalebox{0.8}{$i$}};
\draw[black,fill=white] (1,1) circle (.1 cm);
\draw[black,fill=white] (3,1) circle (.1 cm);
\draw[black,fill=white] (3,3) circle (.1 cm);
\draw[black,fill=white] (2,2) circle (.1 cm);
\draw[black,fill=white] (1,3) circle (.1 cm);
\draw[black,fill=black] (1,2) circle (.1 cm);
\draw[black,fill=black] (2,1) circle (.1 cm);
\draw[black,fill=black] (3,2) circle (.1 cm);
\draw[black,fill=black] (2,3) circle (.1 cm);
\end{tikzpicture}
\caption{The operator $A_i$ acts on all edges corresponding to the lattice site $i\in\Gamma_L$.}\label{fig:star}
\end{figure}
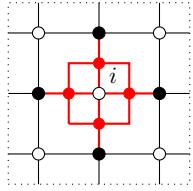
\end{definition}

The charge operators are unitary and self-adjoint. They mutually commute. Moreover,
\begin{equation}
    \prod_{i \in \Gamma_{L}} {Q}_{i} = (-1)^{N}.
\end{equation}
It follows in particular that $\text{spec}({Q}_{i}) =\{\pm 1\}$. The operator $Q_{i}$ implements the following local gauge transformation:
\begin{align}
Q_{i} \hat{\sigma}^z_{kl} Q_{i} &= 
\begin{cases}
     -\hat{\sigma}^z_{kl}& \text{if either $k$ or $l=i$}\\
     \hat{\sigma}^z_{kl} & \text{otherwise.}
\end{cases} \\
Q_{i} \hat{\sigma}^x_{kl} Q_{i} &= \hat{\sigma}^x_{kl} \\
Q_{i} a^{\pm}_k Q_{i} &=
\begin{cases}
     -a^{\pm}_k & \text{if $k=i$ }\\
     a^{\pm}_k  & \text{otherwise.}
\end{cases}
\end{align}

If we denote by $\mathcal{Q} = \{ Q_i:i\in\Gamma_L \}$, we have the following:

\begin{definition}[Physical observables]\label{def:gi} The algebra of physical observables is the commutant $\cal C_{\cal A}(\cal Q)$ of $\cal Q$.
\end{definition}

In other words, the physical observables (such as the Hamiltonian) are those that commute with any $\mathbb{Z}_{2}$-charge operator. We now define `string observables' associated with chains and cochains (see Appendix~\ref{homo} for notations and basic definitions of cellular homology and cohomology):
\begin{align}
     \hat Z: C_1(\Gamma_L) &\to\cal A\\
        \mathcal{C} &\mapsto \hat Z_{\mathcal{C}}=\prod_{(i,j)\in\mathcal{C}} \hat{\sigma}^z_{ij}\\
      \hat X:  C_1(\Gamma_L^*) &\to\cal A \\
        \mathcal{C}^* &\mapsto \hat X_{\mathcal{C}^*}=\prod_{(i,j)\cap\mathcal{C}^*\neq \emptyset} \hat{\sigma}^x_{ij}
\end{align}
These maps are group homomorphisms
\begin{equation}\label{eq:boundaries and Wilson loops}
\hat Z_{\mathcal{C}_1+\mathcal{C}_2} = \hat Z_{\mathcal{C}_1} \hat Z_{\mathcal{C}_2},\qquad
\hat X_{\mathcal{C}^*_1+\mathcal{C}^*_2} = \hat X_{\mathcal{C}^*_1} \hat X_{\mathcal{C}^*_2}
\end{equation}
where the addition in the chain group is understood to be$\mod 2$. This implies that, if $\mathcal{C}\in B_1(\Gamma_L)$, namely $\mathcal{C} = \partial \sum_i p_i$ for some $p_i \in \text{F}(\Gamma_L)$, and likewise if  $\mathcal{C}^*\in B_1(\Gamma_L^*)$, namely $\mathcal{C}^*=  \partial \sum_i p_i^*$ for some $p_i^* \in \text{F}(\Gamma_L^*)$, then:

\begin{equation}
     \hat Z_{\mathcal{C}} = \prod_{\substack{p_i\in \text{F}(\Gamma_L)\\\mathcal{C} = \partial \sum_i p_i}} \hat Z_{\partial p_i},\qquad      \hat X_{\mathcal{C}^*} = \prod_{\substack{i\in \text{V}(\Gamma_L)\\\mathcal{C}^* = \partial \sum_i p_i^*}} A_i
\end{equation}
where we identified a face $p^*$ of the dual lattice with a vertex $i$ of the primal one and noted that $\hat X_{p^*} = A_i$, see \autoref{fig:star}. If we denote
\begin{equation}
    B_{p} = \prod_{(j,l) \in \partial p} \hat{\sigma}^z_{jl}
\end{equation}
the magnetic field operator of the plaquette $p$, we have that $\hat Z_{\partial p} = B_p$:
\\

\begin{figure}[H]
\centering
\begin{tikzpicture}[scale=0.75]
\scalebox{1.2}{
\draw[black] (1,1) -- (1,0.5);
\draw[black] (2,1) -- (2,0.5);
\draw[black] (1,2) -- (0.5,2);
\draw[black] (1,1) -- (0.5,1);
\draw[black, dotted] (0.5,0.5)--(2.5,0.5) --(2.5, 2.5) -- (0.5,2.5)--(0.5,0.5);
\draw[black] (1,2) -- (1,2.5);
\draw[black] (2,2) -- (2,2.5);
\draw[black] (2,2) -- (2.5,2);
\draw[black] (2,1) -- (2.5,1);
\draw[blue,fill=blue] (1.5,2) circle (.09 cm);
\draw[blue,fill=blue] (1.5,1) circle (.09 cm);
\draw[blue,fill=blue] (2,1.5) circle (.09 cm);
\draw[blue,fill=blue] (1,1.5) circle (.09 cm);
\draw[blue, thick] (2,1) -- (2,2);
\draw[blue, thick] (2,1) -- (1,1);
\draw[blue, thick] (1,1) -- (1,2);
\draw[blue, thick] (1,2) -- (2,2);
\node[] (h) at (1.5,1.4) {$p$};
\draw[black,fill=white] (1,1) circle (.1 cm);

\draw[black,fill=white] (2,2) circle (.1 cm);
 
\draw[black,fill=black] (1,2) circle (.1 cm);
\draw[black,fill=black] (2,1) circle (.1 cm);}
 
\end{tikzpicture}
\caption{The magnetic field operator $B_p$ associated with a plaquette $p$.}\label{plaquette}
\end{figure}
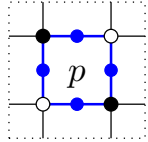

The `electric' string operators $\hat X_{\caC^*}\in \cal C_{\cal A}(\cal Q)$ are physical observables for any $\caC^*\in C_1(\Gamma_L^*)$. For the `magnetic' ones, this is the case only for cycles, namely `closed loop':
\begin{align}
    \mathcal{C}^*\in C_1(\Gamma_L^*)\quad\Longrightarrow\quad & \hat X_{\mathcal{C}^*}\in \cal C_{\cal A}(\cal Q). \label{eq:X cycle} \\
    \mathcal{C}\in Z_1(\Gamma_L)\quad\Longrightarrow\quad & \hat Z_{\mathcal{C}}\in \cal C_{\cal A}(\cal Q), \label{eq:Z cycle}
\end{align}
These observables will play an important role in the following.

Since all the charges mutually commute and their spectrum is $\{\pm 1\}$, the total Hilbert space $\mathcal{H}_{L}$ can be written as a direct sum of supersection sectors, labeled by the simultaneous eigenvalues $\{q_i\}$ of all ${Q}_i$ operators:
\begin{equation}\label{eq:totH}
    \mathcal{H}_L = \bigoplus_{\bar q\in \{\pm 1\}^{\Gamma_L}} \mathcal{H}_{L}^{\bar q}.
\end{equation}
\begin{remark}
\begin{enumerate}
    \item Any supersection sector $\mathcal{H}_{L}^{\bar q}$ can be obtained from the total Hilbert space $\mathcal{H}_L$ by acting with a projector:
    \begin{equation}
        \mathcal{H}_{L}^{\bar q} = \bigotimes_{i \in \Gamma_L} \bigg(\frac{\mathbbm{1} + q_i Q_i}{2}\bigg) \mathcal{H}_L.
    \end{equation}
    Indeed, $\ket{\Psi} = \frac{1}{2}(\mathbbm{1} + q Q)\ket{\Psi}$ implies $Q\ket{\Psi} = q\ket{\Psi}$ since $q^2=1$.
    \item By definition, physical observables leave every superselection sector invariant. If $\mathcal{O}\in\cal C_{\cal A}(\cal Q)$, then for any $\ket{\Psi} \in \mathcal{H}_{L}^{\bar q}$ and any $i \in \Gamma_L$,
\begin{equation}
    Q_i \mathcal{O} \ket{\Psi} = \mathcal{O} Q_i\ket{\Psi} = q_i \mathcal{O}\ket{\Psi}.
\end{equation}
In other words, the algebra $\cal C_{\cal A}(\cal Q)$ decomposes into irreducible blocks labeled by $\bar q$. 
\end{enumerate}    
\end{remark}

We now define the sector that can be interpreted as that of states without background charges.
\begin{definition}[Gauss' law]\label{def:gauss} A state $\ket{\Psi}\in\caH_L$ satisfies Gauss' law if $\ket{\Psi}\in\mathcal{H}_L^{\bar{1}}$, namely $Q_i\ket{\Psi} = \ket{\Psi}$ for all $i\in\Gamma_L$. We shall denote it $\caH^{\rm phys}$ and refer to it as the physical Hilbert space.
\end{definition}

We note that if $\vert\Psi\rangle$ satisfies Gauss' law, then
\begin{equation}\label{eq:Gauss for stars}
    A_i\ket{\Psi} = (-1)^{n_{i, \uparrow}+n_{i,\downarrow}} \ket{\Psi}
\end{equation}
for all $i\in\Gamma_L$.

Finally, the gauging procedure is completed by prescribing a gauge-invariant Hamiltonian obtained by \eqref{hamif} by a lattice analogue of the minimal-coupling procedure:
\begin{equation}
\begin{split}
    H &= -t \sum_{\substack{i,j\in \Gamma_L:\\(i,j) \in \text{E}(\Gamma_L)}}\sum_{\eta = \uparrow,\downarrow} (a^+_{i,\eta} \hat{\sigma}^z_{ij} a^-_{j,\eta} + \text{h.c.}) + m \sum_{i \in \Gamma_L} \sum_{\eta=\uparrow, \downarrow}(-1)^{ i_1+i_2 }\, n_{i,\eta} \\&\quad + U \sum_{i\in \Gamma_{L}} \Big( n_{i,\uparrow} - \frac{1}{2} \Big)\Big( n_{i,\downarrow} - \frac{1}{2} \Big) 
    \end{split}
\end{equation}
\begin{remark}\label{rem:26}
\begin{enumerate}
    \item This Hamiltonian is a spinful-fermion analogue of the one studied in~\emph{\cite{GP}} with the addition of a mass term for the fermions and of a Hubbard-like interaction term. One could have as well added an additional pure gauge term:
\begin{equation}
   H^{\text{gauge}} = \lambda \sum_{p \in \text{F}(\Gamma_L)} B_p\;;
\end{equation}
as we will see, the kinetic term of the model is minimized by a $\pi$-flux background, while the pure gauge term is minimized by a $0$-flux background, introducing an energetic competition between the two terms. However, using the methods described in \emph{\cite{GP}}, one can prove that, for $\lambda \ll t$, such term would not affect the low-energy physics and for that reason it will not be considered it in this work. 
\item Due to the Gauss's Law constraint, the Hubbard term can be thought as a star operator. Indeed:
    \begin{equation}
       A_i = (-1)^{n_{i, \uparrow}} (-1)^{n_{i, \downarrow}} = (1-2 n_{i, \downarrow})(1-2 n_{i, \uparrow}) = 4  \Big( n_{i,\uparrow} - \frac{1}{2} \Big)\Big( n_{i,\downarrow} - \frac{1}{2} \Big)  \end{equation}
       Thus, we can interpret our model as the coupling of the Toric Code with a system of lattice fermions, in the absence of background charges.
\item Our results can be extended to a larger class of interactions that are compatible with reflection positivity \cite{Lieb, NM}, see Section \ref{sec:RP}. For example, repulsive nearest-neighbor interactions: 
\begin{equation}
w \sum_{\substack{i,j \in \Gamma_L:\\
(i,j) \in \text{E}(\Gamma_L)}} \Big(n_i - \frac{1}{2}\Big)\Big(n_j - \frac{1}{2}\Big)\;.
\end{equation} 
For $w>0$, this interaction term is compatible with reflection positivity \cite{Lieb}. More generally, one can allow for two-body interactions of the form \cite{Lieb}:
\begin{equation*}
    \sum_{\substack{i,j \in \Gamma_L\\i \neq j}} V_{ij} \Big(n_i - \frac{1}{2}\Big)\Big(n_j - \frac{1}{2}\Big)
\end{equation*}
where the matrix $V \in M_{|\Gamma_L|}(\mathbb{C})$ is (semi)-positive definite. In particular, if we restrict to the case of translationally invariant interactions, meaning that there exists some function $v: \Gamma_L \to \mathbb{R}$ (with $v(0)=0$) such that $V_{ij} = v(|i-j|)$, the positivity condition is implied by pointwise positivity of $\hat{v}(k)$. Physically relevant cases compatible with reflection positivity are the Coulomb potential and the Yukawa potential. In order to ensure the convergence of the cluster expansion, used in particular to prove the stability of the spectral gap, we need to assume summability of the interaction potential, see {\it e.g.} \cite{GMPhall, DRS}.
\end{enumerate}   
\end{remark}
Finally, if we let $Q_\Lambda = \prod_{i\in\Lambda}Q_i$ for any $\Lambda\subset\Gamma_L$, then
\begin{equation}\label{eq:H gauge invariance}
    Q_\Lambda H Q_\Lambda = H.
\end{equation}
Indeed, $Q_\Lambda = (-1)^{N_\Lambda} \hat X_{\partial \Lambda}$ so that $Q_\Lambda$ commutes with all terms that are supported either completely inside or completely outside of $\Lambda$. For those hopping terms on the boundary both the gauge term and the fermionic term yield a negative sign, so they commute as well. For later purposes, we note that this gauge invariance can also be written as
\begin{equation}\label{eq:H gauge invariance 1}
    A_\Lambda H A_\Lambda = (-1)^{N_\Lambda} H (-1)^{N_\Lambda}
\end{equation}
\begin{example}\label{ex:op}
We conclude this section with a brief discussion of physical pure gauge observables. See \emph{\autoref{fig:stringOps}}. 
\begin{enumerate}
\item Any polynomial in the electric field operators is gauge invariant. This includes the operators $\hat\sigma_{ij}^x$ themselves, but also the string operators $\hat X_{\caC^*}$ for any co-chain $\caC^*\in C_1(\Gamma_L^*)$. If $\partial \caC^*=\{p_1,p_2\}$ is made up of just two plaquettes $p_1,p_2$, we shall refer to $\hat X_{\caC^*}$ as the (bare) monopole pair creation operator. If $\partial\mathcal{C}^* = \emptyset$, namely $\caC^*$ is a cycle in the dual lattice, the operator $\hat X_{\caC^*}$ is called the ’t Hooft magnetic loop operator.
\item The observables $\hat Z_{\caC}$ are gauge-invariant only if $\partial \caC = \emptyset$. In this case, $\hat Z_{\caC}$ measures the $\mathbb{Z}_2$ magnetic flux piercing the path $\mathcal{C}$ and thus, for $\caC\in B_1(\Gamma_L)$, namely $\caC$ is a contractible path, the number (mod $2$) of monopoles in the interior of the path.
\item We have already introduced the elementary star operators $A_i$ and plaquette operators $B_p$ and pointed out their relation to the loop operators. We further note that, on the torus,
\begin{equation}\label{stokesb}
\prod_{p\in \text{F}(\Gamma_L)} B_{p} = \mathbbm{1} = \prod_{i\in \text{V}(\Gamma_L)} A_i.
\end{equation}
\end{enumerate}
\end{example}

\begin{figure}
\centering
\begin{tikzpicture}[scale=0.7]
\draw (0,0) grid (4,4);
\draw[->-] (0,4) -- (4,4);
\draw[->-] (0,0) -- (4,0);
\draw[->>-] (0,0) -- (0,4);
\draw[->>-] (4,0) -- (4,4);
\draw[red,fill=red] (1.5,2) circle (.05 cm);
\draw[thick] (1.5,1.5) node[cross=0.5ex,red] {};
\draw[thick] (1.5,2.5) node[cross=0.5ex,red] {};
\draw[red, thick] (1.5,1.5) -- (1.5,2.5);
\node[below] (a) at (2,0) {$\hat{\sigma}^x_{ij}$};
\begin{scope}[xshift=5.5 cm]
\draw (0,0) grid (4,4);
\draw[->-] (0,4) -- (4,4);
\draw[->-] (0,0) -- (4,0);
\draw[->>-] (0,0) -- (0,4);
\draw[->>-] (4,0) -- (4,4);
\draw[red,fill=red] (1.5,0) circle (.05 cm);
\draw[red,fill=red] (1.5,1) circle (.05 cm);
\draw[red,fill=red] (1.5,2) circle (.05 cm);
\draw[red,fill=red] (1.5,3) circle (.05 cm);
\draw[red,fill=red] (1.5,4) circle (.05 cm);
\draw[red, thick] (1.5,0) -- (1.5,4);
\node[right] (a) at (1.45,2.5) {{$\mathcal{C}^*$}};
\node[below] (a) at (2,0) {$\hat X_{\mathcal{C}^*}$};
\end{scope}
\begin{scope}[xshift=11cm]
\draw (0,0) grid (4,4);
\draw[->-] (0,4) -- (4,4);
\draw[->-] (0,0) -- (4,0);
\draw[->>-] (0,0) -- (0,4);
\draw[->>-] (4,0) -- (4,4);
\draw[red,fill=red] (0.5,3) circle (.05 cm);
\draw[red,fill=red] (1,2.5) circle (.05 cm);
\draw[red,fill=red] (2,2.5) circle (.05 cm);
\draw[red,fill=red] (3,2.5) circle (.05 cm);
\draw[red,fill=red] (3.5,2) circle (.05 cm);
\draw[red, thick] (0.5, 3.5) -- (0.5,2.5) -- (3.5,2.5) -- (3.5,1.5);
\node[above] (b) at (0.5,3.55) {{$p_1$}};
\node[below] (b) at (3.5,1.47) {{$p_2$}};
\draw[thick] (0.5,3.5) node[cross=0.5ex,red] {};
\draw[thick] (3.5,1.5) node[cross=0.5ex,red] {};
\node[below] (a) at (2,0) {$\hat X_{\mathcal{C}^*}$};
\end{scope}
\begin{scope}[yshift=-5.5cm]
\draw (0,0) grid (4,4);
\draw[->-] (0,4) -- (4,4);
\draw[->-] (0,0) -- (4,0);
\draw[->>-] (0,0) -- (0,4);
\draw[->>-] (4,0) -- (4,4);
\draw[blue,fill=blue] (1.5,2) circle (.05 cm);
\draw[blue,fill=blue] (1.5,3) circle (.05 cm);
\draw[blue,fill=blue] (2,2.5) circle (.05 cm);
\draw[blue,fill=blue] (1,2.5) circle (.05 cm);
\draw[thick, blue] (1,2)--(2,2)--(2,3)--(1,3)--(1,2);
\node (a) at (1.5,2.5)  {{$p$}};
\node[below] (a) at (2,0) {{$B_{p}$}};
\end{scope}
\begin{scope}[yshift=-5.5cm, xshift = 5.5cm]
\draw (0,0) grid (4,4);
\draw[->-] (0,4) -- (4,4);
\draw[->-] (0,0) -- (4,0);
\draw[->>-] (0,0) -- (0,4);
\draw[->>-] (4,0) -- (4,4);
\draw[blue,fill=blue] (1.5,1) circle (.05 cm);
\draw[blue,fill=blue] (2.5,1) circle (.05 cm);
\draw[blue,fill=blue] (3,1.5) circle (.05 cm);
\draw[blue,fill=blue] (2.5,2) circle (.05 cm);
\draw[blue,fill=blue] (2,2.5) circle (.05 cm);
\draw[blue,fill=blue] (1.5,3) circle (.05 cm);
\draw[blue,fill=blue] (1,2.5) circle (.05 cm);
\draw[blue,fill=blue] (1,1.5) circle (.05 cm);
\draw[thick, blue] (1,1)--(3,1)--(3,2)--(2,2)--(2,3)--(1,3)--(1,1);
\node (a) at (1.5,1.5)  {{$\Omega$}};
\node[below] (b) at (1.8,1) {{$\partial\Omega = \mathcal{C}$}};
\node[below] (a) at (2,0) {{$\hat Z_{\mathcal{C}}$}};
\end{scope}
\begin{scope}[yshift=-5.5cm, xshift = 11cm]
\draw (0,0) grid (4,4);
\draw[->-] (0,4) -- (4,4);
\draw[->-] (0,0) -- (4,0);
\draw[->>-] (0,0) -- (0,4);
\draw[->>-] (4,0) -- (4,4);
\draw[blue,fill=blue] (1,0.5) circle (.05 cm);
\draw[blue,fill=blue] (1,1.5) circle (.05 cm);
\draw[blue,fill=blue] (1,2.5) circle (.05 cm);
\draw[blue,fill=blue] (1,3.5) circle (.05 cm);
\draw[thick, blue] (1,0)--(1,4);
\node[right] (a) at (1,2.5)  {{$\mathcal{C}$}};
\node[below] (a) at (2,0) {$\hat Z_{\mathcal{C}}$};
\end{scope}
\end{tikzpicture}
\caption{Graphical representation of operators introduced in \hyperref[ex:op]{Example \ref*{ex:op}}.}\label{fig:stringOps}
\end{figure}
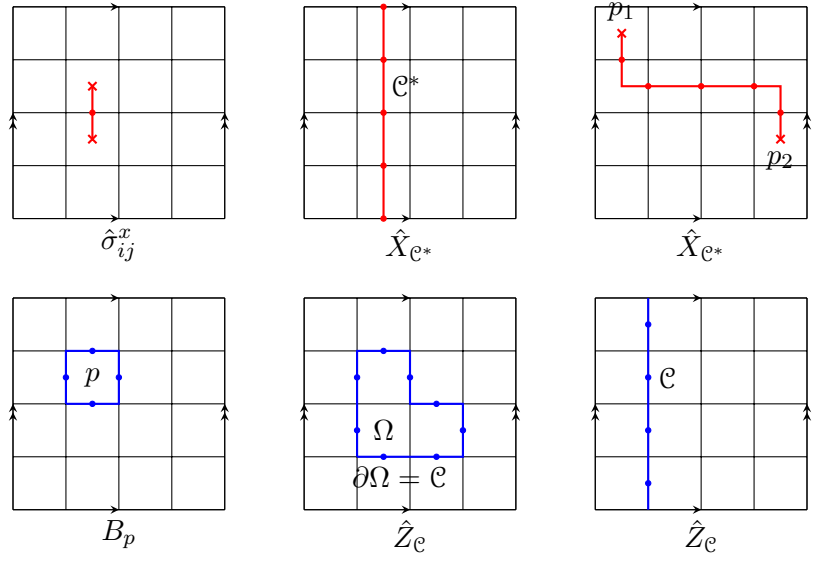

\subsection{Properties of the Hamiltonian}

Let us review some useful properties of the Hamiltonian \eqref{hami} which have been extensively discussed and proved in \cite{GP} (and references therein). The Hamiltonian $H$ commutes with any $\hat{\sigma}^z_{ij}$ operator, namely the $\mathbb{Z}_2$-gauge field background is frozen, so the total Hilbert space $\mathcal{H}_{L}$, see \eqref{eq:total H}, splits into a direct sum over common eigenvalues of $\{\hat{\sigma}^z_{ij}\}$:
\begin{equation}
    \mathcal{H}_{L} = \bigoplus_{\bm{\sigma}} \mathcal{H}_{\bm{\sigma}}
\end{equation}
where $\mathcal{H}_{\bm{\sigma}} = \mathcal{F}_L \otimes \ket{\bm{\sigma}}\bra{\bm{\sigma}}$. The Hamiltonian $H$ fibers over the background sectors $\mathcal{H}_{\bm{\sigma}}$ as $H = \bigoplus_{\bm{\sigma}} H(\bm{\sigma})$, where
\begin{equation}
\begin{split}
    H(\bm{\sigma})  &=    -t \sum_{\substack{i,j\in \Gamma_L:\\(i,j) \in \text{E}(\Gamma_L)}} \sum_{\eta = \uparrow,\downarrow}a^+_{i,\eta} \sigma^z_{ij} a^-_{j,\eta} + \text{h.c.} + m \sum_{i \in \Gamma_L} (-1)^{i_1+i_2}\, n_i \\&\quad + U \sum_{i\in \Gamma_{L}} \Big( n_{i,\uparrow} - \frac{1}{2} \Big)\Big( n_{i,\downarrow} - \frac{1}{2} \Big)
    \end{split}
\end{equation}
acts on a copy of the fermionic Fock space. Note that we have dropped the hat on $\hat{\sigma}^z$, meaning that we are substituting the operator with its eigenvalues: we can think of a background in each fixed sector as a classical $\mathbb{Z}_2$ gauge field.

Fixing the background $\bm\sigma$ determines the value of the pure gauge part $A_i$ of the charges $Q_i$. There is however a large redundency in doing so since the action of any $A_i$ will trivially  not change the charge but flips the value of all four $\sigma$'s around the site $i$. In other words, two backgrounds are gauge equivalent if there exists a subset $\Lambda \subset \Gamma$ such that:
\begin{equation}
    A_\Lambda\ket{\bm{\sigma}}= \prod_{i \in \Lambda} A_i \ket{\bm{\sigma}} = \ket{\bm{\sigma}'} 
\end{equation}
Two gauge equivalent backgrounds yields unitary equivalent Hamiltonians since
\begin{equation}
H(\bm{\sigma}') = \langle \bm{\sigma'}\vert H \bm{\sigma'}\rangle 
= \langle \bm{\sigma} \vert A_{\Lambda}H  A_{\Lambda}\bm{\sigma}\rangle
 = (-1)^{N_\Lambda} H(\bm{\sigma}) (-1)^{N_\Lambda},
\end{equation}
see \eqref{eq:H gauge invariance 1}.

Spectral properties depend thus only on the gauge equivalence class $[\bm{\sigma}]$ of a given background $\bm{\sigma}$. Since the group generated by $\{A_i:i\in\Gamma_L\}$ is isomorphic to the group of $1$-boundaries, see \eqref{eq:boundaries and Wilson loops}, we have that each equivalence class has $2^{|\Gamma_L|-1}$ elements (the $-1$ arising from~(\ref{stokesb})). Since the action of any $A_i$ does not change the flux on any plaquette, gauge equivalence classes are completely characterized by assigning a flux $\pm1$ to each plaquette, yielding $2^{|\Gamma_L|-1}$ choices again by~(\ref{stokesb}), and to two representatives of non-contractible cycles of the torus. There are therefore $2^{|\Gamma_L|+1}$ classes. One checks that this counting yields a total of $2^{2|\Gamma_L|}=2^{|\text{E}(\Gamma_L)|}$ background configurations, as one should expect.

In the following, we refer to a plaquette $p$ such that $\prod_{(i,j)\in\partial p}\sigma^z_{ij} = -1$ as carrying a $\pi$-flux, while a plaquette such that $\prod_{(i,j)\in\partial p}\sigma^z_{ij} = 1$ will be said to carry no flux, or a $0$-flux. One may ask which sector $[\bm{\sigma}]$ corresponds to the lowest energy. It was proved by Lieb~\cite{Lieb} and subsequently refined in~\cite{NM} that the energy-minizing sector is the $\pi$-flux sector. Moreover, the magnetic monopoles (namely, $0$-flux plaquettes) are massive excitations \cite{GP}. More precisely, if $\bm{\sigma}$ is a background with $2k$ $0$-fluxes, then the ground state energy $E_{0,L}(\bm{\sigma})$ of $H(\bm{\sigma})$ (at $m=0$) satisfies the following bound:
\begin{equation}\label{bound}
       E_{0,L}(\bm{\sigma}) \geq  2 k \Delta + E_{0,L}(\bm{-1})
    \end{equation}
    where $E_{0,L}(\bm{-1})$ is the ground state energy of the Hamiltonian with all holonomies set to $-1$ (namely the $\pi$-flux Hamiltonian with antiperiodic boundary conditions in both directions) and $\Delta>0$ is an explicit constant. 

Within the $\pi$-flux phase, explicit diagonalization shows a gapless spectrum with two Dirac cones. Other choices of holonomies correspond to ground state energies that differ from $E_{0,L}(\bm{-1})$ by corrections that vanish as $L\to \infty$, as an inverse power law.

We will show below that the mass term considered in the present paper opens a gap in the fermionic excitation spectrum within the $\pi$-flux phase, see \hyperref[fig:charge gap]{Figure \ref*{fig:charge gap}}, thus yielding a fully gapped theory. That the analysis remains possible hangs on the fact that a staggered mass does not spoil reflection positivity.

\section{Results}\label{sec:res}
\subsection{Spectral structure of the model and local topological order}

From now on, we shall refer to a \emph{flux} for plaquettes and more generally (contractible) boundaries and a \emph{holonomy} for (non-contractible) cycles that are not boundaries. For any configuration $\bm{\sigma}$, we can add additonal holonomies $(e^{i \theta}, e^{i \phi}) \in U(1) \times U(1)$ with the replacement $\sigma^z_{ij}\to \sigma^z_{ij}e^{i \theta}$, respectively $\sigma^z_{ij}\to \sigma^z_{ij}e^{i \phi}$, along any two non-homologous cocycles that are not coboundaries. We denote the corresponding Hamiltonian $H(\bm{\sigma},e^{i \theta}, e^{i \phi})$, and we shall denote by $E_{0,L}(\bm{\sigma},e^{i\theta},e^{i\phi})$ its ground state energy which can be expressed as
\begin{equation}\label{eq:Elim}
E_{0,L}(\bm{\sigma}, e^{i\theta}, e^{i\phi}) = -\lim_{\beta \to \infty} \frac{1}{\beta} \log \tr_{\mathcal{F}_{L}} e^{-\beta H(\bm{\sigma}, e^{i\theta}, e^{i\phi})}.
\end{equation}

\begin{theorem}[Spectral structure of the model]\label{thm1} Let $L\in 4\mathbb{N}$. For $|U|$ small enough uniformly in $L$ the following is true.
\begin{enumerate}
\item \label{thm1.1} For any holonomies $(e^{i \theta}, e^{i \phi}) \in U(1) \times U(1)$, the ground state of $H(\bm{-1}, e^{i \theta}, e^{i \phi})$ on $\mathcal{F}_{L}$ is unique and it is at half-filling. 
\item \label{thm1.2} There exist constants $C,c>0$ depending on $t,m,$ but not on the size of the system, for which:
\begin{equation}\label{eq:expdeg}
    \Big|E_{0,L}(\bm{-1},e^{i\theta}, e^{i\phi})-E_{0,L}(\bm{-1},e^{i\theta'}, e^{i\phi'})\Big| \leq C L^2 e^{-c L}.
\end{equation}

\item \label{thm1.3} Let $\bm{\sigma}$ be a background with $2k$ $0$-fluxes and holonomies $(e^{i \theta}, e^{i \phi}) \in U(1) \times U(1)$. Then, there is a constant $\Delta_{\beta,L}>0$, depending on $t,m,U$, such that:
\begin{equation}\label{eq:chess}
        -\frac{1}{\beta}  \log(\tr_{\mathcal{F}_L}(e^{-\beta H(\bm{\sigma},e^{i \theta}, e^{i \phi})})) \geq 2 k \Delta_{\beta,L}  -  \frac{1}{\beta}   \log(\tr_{\mathcal{F}_L}(e^{-\beta H(\bm{-1},-1,-1)})).
    \end{equation}
The constant $\Delta_{\beta,L}$ is given by:
    \begin{equation}\label{eq:Delta}
        \Delta_{\beta,L} = -\frac{1}{\beta L^2} \max_{(a,b)\in \mathbb{Z}_2 \times \mathbb{Z}_2 } \log\bigg(\frac{\tr_{\mathcal{F}_L} e^{-\beta H(\bm{\sigma^*},a,b)}}{\tr_{\mathcal{F}_L} e^{-\beta H(\bm{-1},-1,-1)}}\bigg) = \Delta_\infty + o(1)
    \end{equation}
  as $\beta,L\to\infty$, where $\Delta_\infty>0$. Here, $\bm{\sigma^*}$ is a particular the chessboard flux configuration (see \emph{\autoref{fig:chessboard}}) with holonomies $(a,b) \in \mathbb{Z}_2 \times \mathbb{Z}_2$. 
    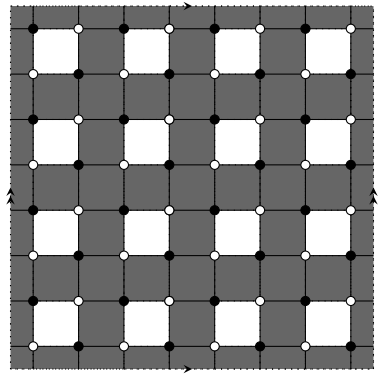
\begin{figure}
\centering
\begin{tikzpicture}[scale=0.6]
\draw[dotted, fill = black, fill opacity = 0.6] (1,-0.5)--(1,7.5) -- (2,7.5) -- (2,-0.5) -- (1,-0.5); 
\draw[dotted, fill = black, fill opacity = 0.6] (3,-0.5)--(3,7.5) -- (4,7.5) -- (4,-0.5) -- (3,-0.5); 
\draw[dotted, fill = black, fill opacity = 0.6] (5,-0.5)--(5,7.5) -- (6,7.5) -- (6,-0.5) -- (5,-0.5);
\draw[dotted, fill = black, fill opacity = 0.6] (7,-0.5)--(7,7.5) -- (7.5,7.5) -- (7.5,-0.5) -- (7,-0.5);
\draw[dotted, fill = black, fill opacity = 0.6] (-0.5,-0.5)--(-0.5,7.5) -- (0,7.5) -- (0,-0.5) -- (-0.5,-0.5);
\draw[dotted, fill = black, fill opacity = 0.6] (0,1)--(1,1) -- (1,2) -- (0,2) -- (0,1);
\draw[dotted, fill = black, fill opacity = 0.6] (2,1)--(3,1) -- (3,2) -- (2,2) -- (2,1);
\draw[dotted, fill = black, fill opacity = 0.6] (4,1)--(5,1) -- (5,2) -- (4,2) -- (4,1);
\draw[dotted, fill = black, fill opacity = 0.6] (6,1)--(7,1) -- (7,2) -- (6,2) -- (6,1);
\draw[dotted, fill = black, fill opacity = 0.6] (0,3)--(1,3) -- (1,4) -- (0,4) -- (0,3);
\draw[dotted, fill = black, fill opacity = 0.6] (2,3)--(3,3) -- (3,4) -- (2,4) -- (2,3);
\draw[dotted, fill = black, fill opacity = 0.6] (4,3)--(5,3) -- (5,4) -- (4,4) -- (4,3);
\draw[dotted, fill = black, fill opacity = 0.6] (6,3)--(7,3) -- (7,4) -- (6,4) -- (6,3);

\draw[dotted, fill = black, fill opacity = 0.6] (0,5)--(1,5) -- (1,6) -- (0,6) -- (0,5);
\draw[dotted, fill = black, fill opacity = 0.6] (2,5)--(3,5) -- (3,6) -- (2,6) -- (2,5);
\draw[dotted, fill = black, fill opacity = 0.6] (4,5)--(5,5) -- (5,6) -- (4,6) -- (4,5);
\draw[dotted, fill = black, fill opacity = 0.6] (6,5)--(7,5) -- (7,6) -- (6,6) -- (6,5);

\draw[dotted, fill = black, fill opacity = 0.6] (0,-0.5)--(1,-0.5) -- (1,0) -- (0,0) -- (0,-0.5);
\draw[dotted, fill = black, fill opacity = 0.6] (2,-0.5)--(3,-0.5) -- (3,0) -- (2,0) -- (2,-0.5);
\draw[dotted, fill = black, fill opacity = 0.6] (4,-0.5)--(5,-0.5) -- (5,0) -- (4,0) -- (4,-0.5);
\draw[dotted, fill = black, fill opacity = 0.6] (6,-0.5)--(7,-0.5) -- (7,0) -- (6,0) -- (6,-0.5);

\draw[dotted, fill = black, fill opacity = 0.6] (0,7.5)--(1,7.5) -- (1,7) -- (0,7) -- (0,7.5);
\draw[dotted, fill = black, fill opacity = 0.6] (2,7.5)--(3,7.5) -- (3,7) -- (2,7) -- (2,7.5);
\draw[dotted, fill = black, fill opacity = 0.6] (4,7.5)--(5,7.5) -- (5,7) -- (4,7) -- (4,7.5);
\draw[dotted, fill = black, fill opacity = 0.6] (6,7.5)--(7,7.5) -- (7,7) -- (6,7) -- (6,7.5);

\draw[dotted, ->-] (-0.5,-0.5) -- (7.5,-0.5);
\draw[dotted, ->-] (-0.5,7.5) -- (7.5,7.5);
\draw[dotted, ->>-] (-0.5,-0.5) -- (-0.5,7.5);
\draw[dotted, ->>-] (7.5,-0.5) -- (7.5,7.5);
\draw (-0.5,0) -- (7.5,0);
\draw (-0.5,1) -- (7.5,1);
\draw (-0.5,2) -- (7.5,2);
\draw (-0.5,3) -- (7.5,3);
\draw (-0.5,4) -- (7.5,4);
\draw (-0.5,5) -- (7.5,5);
\draw (-0.5,6) -- (7.5,6);
\draw (-0.5,7) -- (7.5,7);

\draw (0,-0.5) -- (0,7.5);
\draw (1,-0.5) -- (1,7.5);
\draw (2,-0.5) -- (2,7.5);
\draw (3,-0.5) -- (3,7.5);
\draw (4,-0.5) -- (4,7.5);
\draw (5,-0.5) -- (5,7.5);
\draw (6,-0.5) -- (6,7.5);
\draw (7,-0.5) -- (7,7.5);

\draw[black,fill=black] (0,1) circle (.1 cm);
\draw[black,fill=black] (0,3) circle (.1 cm);
\draw[black,fill=black] (0,5) circle (.1 cm);
\draw[black,fill=black] (0,7) circle (.1 cm);
\draw[black,fill=black] (1,0) circle (.1 cm);
\draw[black,fill=black] (1,2) circle (.1 cm);
\draw[black,fill=black] (1,4) circle (.1 cm);
\draw[black,fill=black] (1,6) circle (.1 cm);
\draw[black,fill=black] (2,1) circle (.1 cm);
\draw[black,fill=black] (2,3) circle (.1 cm);
\draw[black,fill=black] (2,5) circle (.1 cm);
\draw[black,fill=black] (2,7) circle (.1 cm);
\draw[black,fill=black] (4,1) circle (.1 cm);
\draw[black,fill=black] (4,3) circle (.1 cm);
\draw[black,fill=black] (4,5) circle (.1 cm);
\draw[black,fill=black] (4,7) circle (.1 cm);
\draw[black,fill=black] (6,1) circle (.1 cm);
\draw[black,fill=black] (6,3) circle (.1 cm);
\draw[black,fill=black] (6,5) circle (.1 cm);
\draw[black,fill=black] (6,7) circle (.1 cm);
\draw[black,fill=black] (3,0) circle (.1 cm);
\draw[black,fill=black] (3,2) circle (.1 cm);
\draw[black,fill=black] (3,4) circle (.1 cm);
\draw[black,fill=black] (3,6) circle (.1 cm);
\draw[black,fill=black] (5,0) circle (.1 cm);
\draw[black,fill=black] (5,2) circle (.1 cm);
\draw[black,fill=black] (5,4) circle (.1 cm);
\draw[black,fill=black] (5,6) circle (.1 cm);
\draw[black,fill=black] (7,0) circle (.1 cm);
\draw[black,fill=black] (7,2) circle (.1 cm);
\draw[black,fill=black] (7,4) circle (.1 cm);
\draw[black,fill=black] (7,6) circle (.1 cm);

\draw[black,fill=white] (0,0) circle (.1 cm);
\draw[black,fill=white] (0,2) circle (.1 cm);
\draw[black, fill=white] (0,4) circle (.1 cm);
\draw[black, fill=white] (0,6) circle (.1 cm);
\draw[black, fill=white] (2,0) circle (.1 cm);
\draw[black, fill=white] (2,2) circle (.1 cm);
\draw[black, fill=white] (2,4) circle (.1 cm);
\draw[black, fill=white] (2,6) circle (.1 cm);
\draw[black, fill=white] (1,1) circle (.1 cm);
\draw[black, fill=white] (1,3) circle (.1 cm);
\draw[black, fill=white] (1,5) circle (.1 cm);
\draw[black, fill=white] (1,7) circle (.1 cm);
\draw[black, fill=white] (3,1) circle (.1 cm);
\draw[black, fill=white] (3,3) circle (.1 cm);
\draw[black, fill=white] (3,5) circle (.1 cm);
\draw[black, fill=white] (3,7) circle (.1 cm);
\draw[black, fill=white] (5,1) circle (.1 cm);
\draw[black, fill=white] (5,3) circle (.1 cm);
\draw[black, fill=white] (5,5) circle (.1 cm);
\draw[black, fill=white] (5,7) circle (.1 cm);
\draw[black, fill=white] (3,1) circle (.1 cm);
\draw[black, fill=white] (3,3) circle (.1 cm);
\draw[black, fill=white] (3,5) circle (.1 cm);
\draw[black, fill=white] (3,7) circle (.1 cm);
\draw[black, fill=white] (4,0) circle (.1 cm);
\draw[black, fill=white] (4,2) circle (.1 cm);
\draw[black, fill=white] (4,4) circle (.1 cm);
\draw[black, fill=white] (4,6) circle (.1 cm);
\draw[black, fill=white] (6,0) circle (.1 cm);
\draw[black, fill=white] (6,2) circle (.1 cm);
\draw[black, fill=white] (6,4) circle (.1 cm);
\draw[black, fill=white] (6,6) circle (.1 cm);
\draw[black, fill=white] (7,1) circle (.1 cm);
\draw[black, fill=white] (7,3) circle (.1 cm);
\draw[black, fill=white] (7,5) circle (.1 cm);
\draw[black, fill=white] (7,7) circle (.1 cm);
\end{tikzpicture}
\caption{The flux configuration corresponding to $\bm{\sigma}^*$ is that of a lattice of monopoles, represented here in white, carved on the uniform background of $\pi$-fluxes, the dark plaquettes.}
\label{fig:chessboard}
\end{figure}
In particular, as $\beta\to \infty$:
\begin{equation}\label{bound2}
       E_{0,L}(\bm{\sigma},e^{i\theta}, e^{i\phi}) \geq  2 k \Delta_L + E_{0,L}(\bm{-1},-1,-1)
\end{equation}
where $\Delta_L = \Delta_{\infty} + o_{L}(1)$. 
\end{enumerate}
\end{theorem}

In other words, the global ground state lies in $\pi$-flux sector, and monopoles are again massive, as in the case $m=0$, provided $|U|$ is small enough uniformly in the system's size (but nonuniformly in $m$). These results are closely related to~\cite[Theorem~2.12]{GP}. While the monopoles were already shown to be massive there, the presence of the fermionic mass term in the Hamiltonian (i) yields a global spectral gap above the ground state energy, and (ii) causes the splitting~(\ref{eq:expdeg}) of the ground state energies to be exponentially small in the system size (rather than power law, as in \cite{GP}), and indeed the proof will show that $c\to0$ as $m\to0$, see~(\ref{eq:CT C and c}).

We now turn to the ground states and the spectral gap within the $\pi$-flux sector. 

\begin{theorem}[Ground state topological order]\label{thm2} 
Let $L\in 4 \mathbb{N}$, and let $\vert U\vert$ be small enough so that $\Delta_\infty>0$. Let $P$ be the projection onto the span of the four states
\begin{equation}\label{eq:GSs}
    \ket{\Omega_{ab}} = \frac{1}{\mathcal{N}} 
    \prod_{i \in \text{V}(\Gamma_L)} \bigg(\frac{1+Q_i}{2}\bigg) \ket{\psi_{\bm{-1},a,b}}\otimes \ket{\bm{-1},a,b},\qquad (a,b)\in\mathbb{Z}_2\times\mathbb{Z}_2,
\end{equation}
where $\ket{\psi_{\bm{-1},a,b}}$ is the ground state of the fermionic Hamiltonian~$H(\bm{-1},a,b)$ on $\mathcal{F}_{L}$, and $\mathcal{N}$ is a normalisation factor. Then:

\begin{enumerate}
    \item \label{thm2.1}  Let $g_L = \inf\mathrm{Spec}(P^\perp H P^\perp ) - \sup\mathrm{Spec}(P H P)$. Then \begin{equation}\label{eq:gap}
    g_L\geq\min\{2\delta_{L},2\Delta_{L}\} - CL^{2}e^{-cL}
\end{equation}
where $\delta_{L} \geq m - C|U|^{1/3}$.
\item \label{thm2.2} For any fixed gauge-invariant observable $\mathcal{O}$, there are constants $C,c>0$ such that
\begin{equation}\label{eq:Exp LTQO}
    \left \vert \langle \Omega_{ab}\vert \mathcal{O} \Omega_{a'b'}\rangle  - \delta_{aa'} \delta_{bb'} \frac{\Tr(P\mathcal{O})}{\Tr(P)}\right\vert\leq C_{\mathcal{O}}\mathrm{e}^{-cL}.
\end{equation} 
\end{enumerate}
\end{theorem}

In \emph{2.}~above, a `fixed' observable is meant in the following sense. The tori $\Gamma_L$ are identified with squares $[-L/2,L/2]\times[-L/2,L/2]$ (recall that $L\in 2\mathbb{N}$) with periodic boundary conditions. There is $L_0$ such that $\mathcal{O}\in\mathcal{A}_{L_0}$ and therefore $\mathcal{O}\in\mathcal{A}_{L}$ for all $L\geq L_0$ using the natural embedding $\mathcal{A}_{L_0}\hookrightarrow\mathcal{A}_{L}$ given by tensoring with the identity. One may also call such an observable a `local' observable, since its distance to the boundary of the cut torus is of order $L$.

\noindent We note that by \autoref{thm1} - \autoref{thm1.2}, the vectors $\ket{\Omega_{ab}}$ correspond to an almost degenerate eigenvalue. From now on, we shall refer to $P$ as the \emph{ground state projection} and the theorem shows that $P$ is \emph{gapped}, this is equation~(\ref{eq:gap}),  and \emph{topologically ordered}, which refers to the property~(\ref{eq:Exp LTQO}). The ground state energy is (almost) fourfold degenerate and isolated from the rest the spectrum by a gap, which remains open in the infinite-volume limit. This degeneracy is not associated with a local order parameter since all local observables are multiples of the identity when restricted to the ground state space. In order to have a concrete picture in mind, we draw in \autoref{fig:pifluxes} below one representative of $([\bm{-1}],a,b)$ for all choices $(a,b)$. Here, $[\bm{-1}]$ is the equivalence class of $\pi$-flux backgrounds.
\begin{figure}[h]
\centering
\begin{tikzpicture}[scale=0.7]
\draw[dotted, ->-] (-0.5,-0.5) -- (3.5,-0.5);
\draw[dotted, ->>-] (3.5,-0.5) -- (3.5,3.5);
\draw[dotted, ->-] (-0.5,3.5) -- (3.5,3.5);
\draw[dotted, ->>-] (-0.5,-0.5) -- (-0.5,3.5);
\draw[line width = 0.05 cm] (-0.5, 0) -- (3.5,0);
\draw[line width = 0.05 cm] (-0.5, 2) -- (3.5,2);
\draw (-0.5,1) -- (3.5,1);
\draw (-0.5,3) -- (3.5,3);
\draw (0,-0.5) -- (0,3.5);
\draw (1,-0.5) -- (1,3.5);
\draw (2,-0.5) -- (2,3.5);
\draw (3,-0.5) -- (3,3.5);

\draw[black,fill=white] (0,0) circle (.1 cm);
\draw[black,fill=white] (1,1) circle (.1 cm);
\draw[black,fill=white] (3,3) circle (.1 cm);
\draw[black,fill=white] (2,2) circle (.1 cm);
\draw[black,fill=white] (0,2) circle (.1 cm);
\draw[black,fill=white] (1,3) circle (.1 cm);
\draw[black,fill=white] (3,1) circle (.1 cm);
\draw[black,fill=white] (2,0) circle (.1 cm);

\draw[black,fill=black] (0,1) circle (.1 cm);
\draw[black,fill=black] (0,3) circle (.1 cm);
\draw[black,fill=black] (1,0) circle (.1 cm);
\draw[black,fill=black] (3,0) circle (.1 cm);
\draw[black,fill=black] (2,1) circle (.1 cm);

\draw[black,fill=black] (1,2) circle (.1 cm);
\draw[black,fill=black] (3,2) circle (.1 cm);
\draw[black,fill=black] (2,3) circle (.1 cm);

\node[below] (a) at (1.5,-1) {{$(1,1)$}};
\begin{scope}[xshift = 5.5cm]
\draw[red, thick] (3.5,3.5)-- (3.5,-0.5);
\draw[red, thick] (-0.5,3.5)-- (-0.5,-0.5);
\draw[dotted, ->-] (-0.5,-0.5) -- (3.5,-0.5);
\draw[dotted, ->>-] (3.5,-0.5) -- (3.5,3.5);
\draw[dotted, ->-] (-0.5,3.5) -- (3.5,3.5);
\draw[dotted, ->>-] (-0.5,-0.5) -- (-0.5,3.5);
\draw[line width = 0.05 cm] (0, 0) -- (3,0);
\draw[line width = 0.05 cm] (0, 2) -- (3,2);
\draw[line width = 0.05 cm] (3,1) -- (3.5,1);
\draw[line width = 0.05 cm] (-0.5,1) -- (0,1);
\draw[line width = 0.05 cm] (-0.5,3) -- (0,3);
\draw[line width = 0.05 cm] (3,3) -- (3.5,3);
\draw (0,3) -- (3,3);
\draw (0,1) -- (3,1);
\draw (0,-0.5) -- (0,3.5);
\draw (1,-0.5) -- (1,3.5);
\draw (2,-0.5) -- (2,3.5);
\draw (-0.5,2) -- (0,2);
\draw (3,2) -- (3.5,2);
\draw (-0.5,0) -- (0,0);
\draw (3,0) -- (3.5,0);
\draw (3,-0.5) -- (3,3.5);

\draw[black,fill=white] (0,0) circle (.1 cm);
\draw[black,fill=white] (1,1) circle (.1 cm);
\draw[black,fill=white] (3,3) circle (.1 cm);
\draw[black,fill=white] (2,2) circle (.1 cm);
\draw[black,fill=white] (0,2) circle (.1 cm);
\draw[black,fill=white] (1,3) circle (.1 cm);
\draw[black,fill=white] (3,1) circle (.1 cm);
\draw[black,fill=white] (2,0) circle (.1 cm);

\draw[black,fill=black] (0,1) circle (.1 cm);
\draw[black,fill=black] (0,3) circle (.1 cm);
\draw[black,fill=black] (1,0) circle (.1 cm);
\draw[black,fill=black] (3,0) circle (.1 cm);
\draw[black,fill=black] (2,1) circle (.1 cm);

\draw[black,fill=black] (1,2) circle (.1 cm);
\draw[black,fill=black] (3,2) circle (.1 cm);
\draw[black,fill=black] (2,3) circle (.1 cm);

\node[below] (a) at (1.5,-1) {{$(-1,1)$}};

\node[above] (b) at (3.5,3.5) {{$\caC_2^*$}};

\end{scope}
\begin{scope}[xshift=11 cm]
\draw[red, thick] (3.5,3.5)-- (3.5,-0.5);
\draw[red, thick] (-0.5,3.5)-- (-0.5,-0.5);

\draw[red, thick] (-0.5,-0.5) -- (3.5,-0.5);
\draw[red, thick] (-0.5,3.5) -- (3.5,3.5);
\draw[dotted, ->-] (-0.5,-0.5) -- (3.5,-0.5);
\draw[dotted, ->>-] (3.5,-0.5) -- (3.5,3.5);
\draw[dotted, ->-] (-0.5,3.5) -- (3.5,3.5);
\draw[dotted, ->>-] (-0.5,-0.5) -- (-0.5,3.5);
\draw[line width = 0.05 cm] (0, 0) -- (3,0);
\draw[line width = 0.05 cm] (0, 2) -- (3,2);
\draw[line width = 0.05 cm] (3,1) -- (3.5,1);
\draw[line width = 0.05 cm] (-0.5,1) -- (0,1);
\draw[line width = 0.05 cm] (-0.5,3) -- (0,3);
\draw[line width = 0.05 cm] (3,3) -- (3.5,3);
\draw (0,3) -- (3,3);
\draw (0,1) -- (3,1);
\draw (0,-0.5) -- (0,3.5);
\draw (1,-0.5) -- (1,3.5);
\draw (2,-0.5) -- (2,3.5);
\draw (-0.5,2) -- (0,2);
\draw (3,2) -- (3.5,2);
\draw (-0.5,0) -- (0,0);
\draw (3,0) -- (3.5,0);
\draw (3,-0.5) -- (3,3.5);

\draw[line width = 0.05 cm] (0,-0.5) -- (0,0);
\draw[line width = 0.05 cm] (1,-0.5) -- (1,0);
\draw[line width = 0.05 cm] (2,-0.5) -- (2,0);
\draw[line width = 0.05 cm] (3,-0.5) -- (3,0);

\draw[line width = 0.05 cm] (0,3.5) -- (0,3);
\draw[line width = 0.05 cm] (1,3.5) -- (1,3);
\draw[line width = 0.05 cm] (2,3.5) -- (2,3);
\draw[line width = 0.05 cm] (3,3.5) -- (3,3);

\draw[black,fill=white] (0,0) circle (.1 cm);
\draw[black,fill=white] (1,1) circle (.1 cm);
\draw[black,fill=white] (3,3) circle (.1 cm);
\draw[black,fill=white] (2,2) circle (.1 cm);
\draw[black,fill=white] (0,2) circle (.1 cm);
\draw[black,fill=white] (1,3) circle (.1 cm);
\draw[black,fill=white] (3,1) circle (.1 cm);
\draw[black,fill=white] (2,0) circle (.1 cm);

\draw[black,fill=black] (0,1) circle (.1 cm);
\draw[black,fill=black] (0,3) circle (.1 cm);
\draw[black,fill=black] (1,0) circle (.1 cm);
\draw[black,fill=black] (3,0) circle (.1 cm);
\draw[black,fill=black] (2,1) circle (.1 cm);

\draw[black,fill=black] (1,2) circle (.1 cm);
\draw[black,fill=black] (3,2) circle (.1 cm);
\draw[black,fill=black] (2,3) circle (.1 cm);

\node[below] (a) at (1.5,-1) {{$(-1,-1)$}};

\node[above] (b) at (3.5,3.5) {{$\caC_2^*$}};
\node[left] (b) at (-0.5,-0.5) {{$\caC_1^*$}};

\end{scope}
\begin{scope}[xshift=16.5cm]

\draw[red, thick] (-0.5,-0.5) -- (3.5,-0.5);
\draw[red, thick] (-0.5,3.5) -- (3.5,3.5);
\draw[dotted, ->-] (-0.5,-0.5) -- (3.5,-0.5);
\draw[dotted, ->>-] (3.5,-0.5) -- (3.5,3.5);
\draw[dotted, ->-] (-0.5,3.5) -- (3.5,3.5);
\draw[dotted, ->>-] (-0.5,-0.5) -- (-0.5,3.5);

\draw[line width = 0.05 cm] (-0.5, 0) -- (3.5,0);
\draw[line width = 0.05 cm] (-0.5, 2) -- (3.5,2);
\draw (-0.5,1) -- (3.5,1);
\draw (-0.5,3) -- (3.5,3);
\draw (0,-0.5) -- (0,3.5);
\draw (1,-0.5) -- (1,3.5);
\draw (2,-0.5) -- (2,3.5);
\draw (3,-0.5) -- (3,3.5);

\draw[line width = 0.05 cm] (0,-0.5) -- (0,0);
\draw[line width = 0.05 cm] (1,-0.5) -- (1,0);
\draw[line width = 0.05 cm] (2,-0.5) -- (2,0);
\draw[line width = 0.05 cm] (3,-0.5) -- (3,0);

\draw[line width = 0.05 cm] (0,3.5) -- (0,3);
\draw[line width = 0.05 cm] (1,3.5) -- (1,3);
\draw[line width = 0.05 cm] (2,3.5) -- (2,3);
\draw[line width = 0.05 cm] (3,3.5) -- (3,3);

\draw[black,fill=white] (0,0) circle (.1 cm);
\draw[black,fill=white] (1,1) circle (.1 cm);
\draw[black,fill=white] (3,3) circle (.1 cm);
\draw[black,fill=white] (2,2) circle (.1 cm);
\draw[black,fill=white] (0,2) circle (.1 cm);
\draw[black,fill=white] (1,3) circle (.1 cm);
\draw[black,fill=white] (3,1) circle (.1 cm);
\draw[black,fill=white] (2,0) circle (.1 cm);

\draw[black,fill=black] (0,1) circle (.1 cm);
\draw[black,fill=black] (0,3) circle (.1 cm);
\draw[black,fill=black] (1,0) circle (.1 cm);
\draw[black,fill=black] (3,0) circle (.1 cm);
\draw[black,fill=black] (2,1) circle (.1 cm);

\draw[black,fill=black] (1,2) circle (.1 cm);
\draw[black,fill=black] (3,2) circle (.1 cm);
\draw[black,fill=black] (2,3) circle (.1 cm);

\node[below] (a) at (1.5,-1) {{$(1,-1)$}};

\node[left] (b) at (-0.5,-0.5) {{$\caC_1^*$}};
\end{scope}
\end{tikzpicture}
\caption{A representative $[\bm{-1}]$ for all possible $\pi$-flux backgrounds, for $L=4$. A thick line represents an edge with $\sigma_{ij}^z=-1$, while thin lines correspond to $\sigma_{ij}^z=1$.}
\label{fig:pifluxes}
\end{figure}
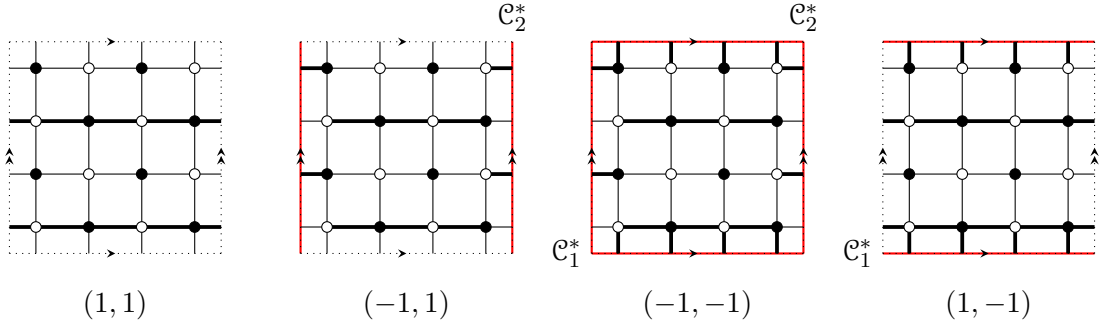

\begin{remark}\label{rem: after thm2}
\begin{enumerate}
\item The spectral gap $g_L$ above the ground state energy is twice the minimum of the (renormalized) monopole mass $\Delta_L$ and of a quantity $\delta_L$, which plays the role of renormalized mass for the fermions. For $|U|$ small enough, it is strictly positive for all $L$ large enough. While the former quasi-particles are neutral, the latter are charged. Following\emph{~\cite{ChargedNeutral}} (and the vast literature therein on gaps in the quantum Hall and anomalous Hall effect) one may ask for the relationship between the two gaps. The exact expression \eqref{boundint} in the non-interacting limit $U=0$ indicates a crossover in our model (see \emph{\autoref{crossover}}): while the monopole mass remains finite as $m\to 0$, see\emph{~\cite{GP}}, $\Delta_\infty\to 0$ as $m\to \infty$ since the three terms of \eqref{boundint} cancel out in the limit. This is not in contradiction with~\emph{\cite{ChargedNeutral}} since our model does not have an analog of the dipole symmetry.
\item Unlike in the toric code, the map $\caC\mapsto \hat Z_\caC$ does not descend to $H_1(\Gamma)$. Indeed, the condition $B_p=-1$ for any plaquette $p$ means that the $\mathbb{Z}_2$ background is not flat: ground states have different eigenvalues with respect to $\hat Z_\caC$ and $\hat Z_{\caC+\partial\Lambda}$ since
    \begin{equation}\label{eq:notflat}
    \hat Z_{\caC+\partial\Lambda}\ket{\Omega_{ab}} =(-1)^{|\Lambda|} \hat Z_\caC \ket{\Omega_{ab}}.
    \end{equation}
As pointed out in~\emph{\cite{Cimasoni_2007}}, any $\pi$-flux background defines in a canonical way  a spin structure on the torus. The spin structure space on a given Riemann surface $\Sigma_g$ is an affine space whose translation space is given by $H_1(\Sigma_g, \mathbb{Z}_2)$: in other words, any two spin structures differ by a $\mathbb{Z}_2$ flat connection. Similarly, here two ground state backgrounds differ by the insertion of a $\mathbb{Z}_2$ holonomy around a non-contractible cycle, but it cannot be defined as the background with a given parallel transport independently on the choice of the representative of the homology cycle.
\end{enumerate}

    \begin{figure}[ht]
        \centering
\includegraphics[width=0.5\linewidth]{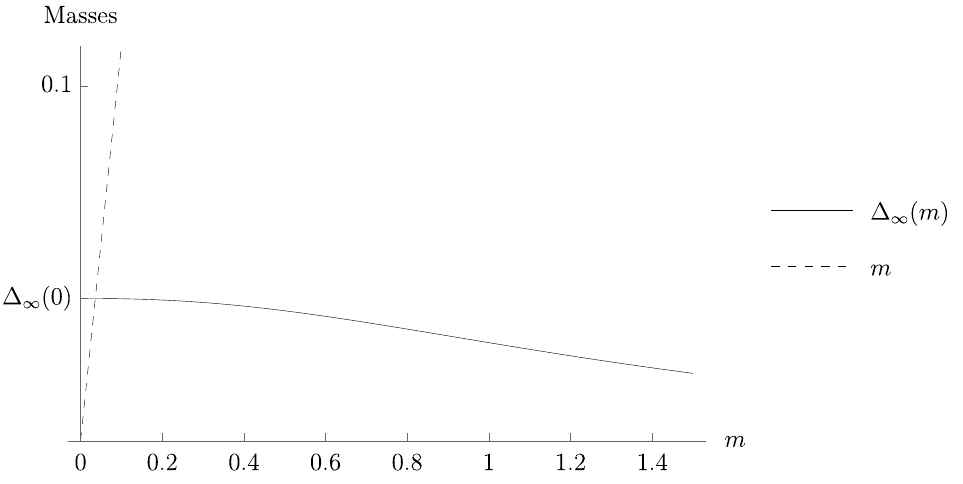}
        \caption{Crossover between the monopole and the fermion mass gap (in the thermodynamic limit) as a function of $m$ for $0\leq m \leq 1.5$ and $t=1$, $U=0$. See \autoref{delta} for the behavior of $\Delta_{\infty}(m)$ for a larger range of $m$.}
        \label{crossover}
    \end{figure}   
\end{remark}

\subsection{Loop operators and braiding properties}\label{qaf}

We have established that the Hamiltonian has a gapped, topologically ordered ground state space. The elementary excitations come in two types: magnetic monopoles obtained by carving out a $0$-flux in the uniform $\pi$-flux background on the one hand, and fermionic excitations upon the half-filled fermionic ground state. 

We now turn to adiabatic insertion of a flux or a holonomy through the system. We rely on Hastings' quasi-adiabatic evolution~\cite{Bachmann_2011, Hastings_2005} used in a general context of flux threading in~\cite{Bachmann_2019,Bachmann_2021}. This will allow us to describe the mapping of different ground states onto each other. As is typical in the presence of anyons, this is closely related to the process of creating a pair of monopoles, moving one of them along a non-contractible loop and fusing them back again cycles through the ground states.

The intersection number $\mathcal{I}$ between two edges $e  \in \text{E}(\Gamma_L)$ and $e^* \in \text{E}(\Gamma_L^*)$ is the map:
\begin{equation}
    \mathcal{I}: \text{E}(\Gamma_L) \times \text{E}(\Gamma_L^{*}) \to \{-1, 0, 1\}
\end{equation}
defined by the convention in Figure \ref{Intersectionform}.
\begin{figure}[h]
    \centering
   \begin{tikzpicture}[scale=0.8]
   \draw[black, ->>>>-] (0,0) -- (0,1);
       \draw[black, ->>>>-] (0,0) -- (1,0);
       \draw[black, ->>>>-] (1,0) -- (2,0);
       \draw[black, ->>>>-] (1,0) -- (1,1);
       \draw[black, ->>>>-] (2,0) -- (2,1);
       \draw[blue, ->>>>-, thick, line width = 0.04 cm] (0,1) -- (1,1);
       \draw[black, ->>>>-] (0,1) -- (0,2);
       \draw[black, ->>>>-] (1,1) -- (2,1);
       \draw[black, ->>>>-] (1,1) -- (1,2);
       \draw[black, ->>>>-] (0,2) -- (1,2);
       \draw[black, ->>>>-] (1,2) -- (2,2);
       \draw[black, ->>>>-] (2,1) -- (2,2);
       
       \draw[black, dotted] (-0.5, -0.5) -- (2.5, -0.5) -- (2.5,2.5) -- (-0.5, 2.5) -- (-0.5,-0.5);
       \node[red, above right] (a) at (0,1.25) {{$e^*$}};
       \node[blue, above left] (c) at (1.1,1) {{$e$}};
       \draw[black, ->-] (-0.5,0) -- (0,0);
        \draw[black, ->-] (-0.5,1) -- (0,1);
         \draw[black, ->-] (-0.5,2) -- (0,2);
       \draw[black, ->-] (0,-0.5) -- (0,0);
       \draw[black, ->-] (1,-0.5) -- (1,0);
       \draw[black, ->-] (2,-0.5) -- (2,0);
       \draw[black, ->>>-] (2,0) -- (2.5,0);
       \draw[black, ->>>-] (2,1) -- (2.5,1);
       \draw[black, ->>>-] (2,2) -- (2.5,2);
        \draw[black, ->>>-] (2,2) -- (2,2.5);
        \draw[black, ->>>-] (1,2) -- (1,2.5);
        \draw[black, ->>>-] (0,2) -- (0,2.5);
       \draw[black,fill=black] (0,1) circle (.1 cm);
        \draw[black,fill=black] (1,0) circle (.1 cm);
           \draw[black,fill=black] (2,1) circle (.1 cm);
        \draw[black,fill=black] (1,2) circle (.1 cm);
        \draw[black,fill=white] (0,0) circle (.1 cm);
        \draw[black,fill=white] (2,2) circle (.1 cm);
        \draw[black,fill=white] (2,0) circle (.1 cm);
        \draw[black,fill=white] (0,2) circle (.1 cm);
        \draw[black,fill=white] (1,1) circle (.1 cm);
        \draw[black, dotted] (0.5,-0.5)--(0.5, 2.5);
         \draw[black, dotted] (1.5,-0.5)--(1.5, 2.5);
          \draw[black, dotted] (-0.5,1.5)--(2.5, 1.5);
           \draw[black, dotted] (-0.5,0.5)--(2.5, 0.5);
           \draw[red, thick, line width = 0.04 cm, ->>>>>-] (0.5,0.5) -- (0.5,1.5);  
           \node[below] (e) at (1, -1) 
           {{$\mathcal{I}(e,e^*)=+1$}};
   \begin{scope}[xshift=5cm]
       \draw[black, ->>>>-] (0,0) -- (0,1);
       \draw[black, ->>>>-] (0,0) -- (1,0);
       \draw[black, ->>>>-] (1,0) -- (2,0);
       \draw[black, ->>>>-] (1,0) -- (1,1);
       \draw[black, ->>>>-] (2,0) -- (2,1);
       \draw[blue, ->>>>-, thick, line width = 0.04 cm] (0,1) -- (1,1);
       \draw[black, ->>>>-] (0,1) -- (0,2);
       \draw[black, ->>>>-] (1,1) -- (2,1);
       \draw[black, ->>>>-] (1,1) -- (1,2);
       \draw[black, ->>>>-] (0,2) -- (1,2);
       \draw[black, ->>>>-] (1,2) -- (2,2);
       \draw[black, ->>>>-] (2,1) -- (2,2);
       
       \draw[black, dotted,] (-0.5, -0.5) -- (2.5, -0.5) -- (2.5,2.5) -- (-0.5, 2.5) -- (-0.5,-0.5);
      \node[red, below right] (a) at (0,0.6) {{$e^*$}};
        \node[blue, above left] (c) at (1,1) {{$e$}};
       \draw[black, ->-] (-0.5,0) -- (0,0);
        \draw[black, ->-] (-0.5,1) -- (0,1);
         \draw[black, ->-] (-0.5,2) -- (0,2);
       \draw[black, ->-] (0,-0.5) -- (0,0);
       \draw[black, ->-] (1,-0.5) -- (1,0);
       \draw[black, ->-] (2,-0.5) -- (2,0);
       \draw[black, ->>>-] (2,0) -- (2.5,0);
       \draw[black, ->>>-] (2,1) -- (2.5,1);
       \draw[black, ->>>-] (2,2) -- (2.5,2);
        \draw[black, ->>>-] (2,2) -- (2,2.5);
        \draw[black, ->>>-] (1,2) -- (1,2.5);
        \draw[black, ->>>-] (0,2) -- (0,2.5);
       \draw[black,fill=black] (0,1) circle (.1 cm);
        \draw[black,fill=black] (1,0) circle (.1 cm);
           \draw[black,fill=black] (2,1) circle (.1 cm);
        \draw[black,fill=black] (1,2) circle (.1 cm);
        \draw[black,fill=white] (0,0) circle (.1 cm);
        \draw[black,fill=white] (2,2) circle (.1 cm);
        \draw[black,fill=white] (2,0) circle (.1 cm);
        \draw[black,fill=white] (0,2) circle (.1 cm);
        \draw[black,fill=white] (1,1) circle (.1 cm);
        \draw[black, dotted] (0.5,-0.5)--(0.5, 2.5);
         \draw[black, dotted] (1.5,-0.5)--(1.5, 2.5);
          \draw[black, dotted] (-0.5,1.5)--(2.5, 1.5);
           \draw[black, dotted] (-0.5,0.5)--(2.5, 0.5);
           \draw[red, thick, line width = 0.04 cm, ->>>>>-] (0.5,1.5) -- (0.5,0.5);  
           \node[below] (e) at (1, -1) 
           {{$\mathcal{I}(e,e^*)=-1$}};
           \end{scope}
            \begin{scope}[xshift=10cm]
       \draw[black, ->>>>-] (0,0) -- (0,1);
       \draw[black, ->>>>-] (0,0) -- (1,0);
       \draw[black, ->>>>-] (1,0) -- (2,0);
       \draw[black, ->>>>-] (1,0) -- (1,1);
       \draw[black, ->>>>-] (2,0) -- (2,1);
       \draw[blue, ->>>>-, thick, line width = 0.04 cm] (0,1) -- (1,1);
       \draw[black, ->>>>-] (0,1) -- (0,2);
       \draw[black, ->>>>-] (1,1) -- (2,1);
       \draw[black, ->>>>-] (1,1) -- (1,2);
       \draw[black, ->>>>-] (0,2) -- (1,2);
       \draw[black, ->>>>-] (1,2) -- (2,2);
       \draw[black, ->>>>-] (2,1) -- (2,2);
       
       \draw[black, dotted,] (-0.5, -0.5) -- (2.5, -0.5) -- (2.5,2.5) -- (-0.5, 2.5) -- (-0.5,-0.5);
        \node[red, above right] (a) at (1,1.45) {{$e^*$}};
       \node[blue, below left] (c) at (1,1) {{$e$}};
       \draw[black, ->-] (-0.5,0) -- (0,0);
        \draw[black, ->-] (-0.5,1) -- (0,1);
         \draw[black, ->-] (-0.5,2) -- (0,2);
       \draw[black, ->-] (0,-0.5) -- (0,0);
       \draw[black, ->-] (1,-0.5) -- (1,0);
       \draw[black, ->-] (2,-0.5) -- (2,0);
       \draw[black, ->>>-] (2,0) -- (2.5,0);
       \draw[black, ->>>-] (2,1) -- (2.5,1);
       \draw[black, ->>>-] (2,2) -- (2.5,2);
        \draw[black, ->>>-] (2,2) -- (2,2.5);
        \draw[black, ->>>-] (1,2) -- (1,2.5);
        \draw[black, ->>>-] (0,2) -- (0,2.5);
       \draw[black,fill=black] (0,1) circle (.1 cm);
        \draw[black,fill=black] (1,0) circle (.1 cm);
        \draw[black,fill=black] (2,1) circle (.1 cm);
        \draw[black,fill=black] (1,2) circle (.1 cm);
        \draw[black,fill=white] (0,0) circle (.1 cm);
        \draw[black,fill=white] (2,2) circle (.1 cm);
        \draw[black,fill=white] (2,0) circle (.1 cm);
        \draw[black,fill=white] (0,2) circle (.1 cm);
        \draw[black,fill=white] (1,1) circle (.1 cm);
        \draw[black, dotted] (0.5,-0.5)--(0.5, 2.5);
         \draw[black, dotted] (1.5,-0.5)--(1.5, 2.5);
          \draw[black, dotted] (-0.5,1.5)--(2.5, 1.5);
           \draw[black, dotted] (-0.5,0.5)--(2.5, 0.5);
           \draw[red, thick, line width = 0.04 cm, ->>>>>-] (0.5,1.5) -- (1.5,1.5);  
           \node[below] (e) at (1, -1) 
           {{$\mathcal{I}(e,e^*)=0$}};
           \end{scope}
   \end{tikzpicture}
   \caption{Convention for the definition of the intersection number.}
    \label{Intersectionform}
\end{figure}
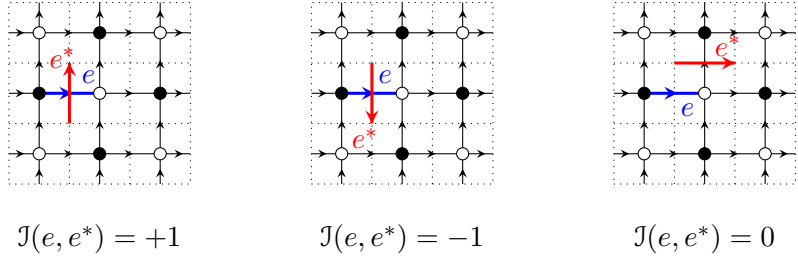
The intersection number can then be extended to a map on $C_1(\Gamma_L)\times C_1(\Gamma^*_L)$ by adding up the contributions coming from every single edge.

Let now $\mathcal{C}^*\in Z_1(\Gamma_L^*)$. The twisted Hamiltonian is defined as (see \autoref{fig:twistedhami}):
\begin{align}
    H_{\mathcal{C}^*}(\phi) &= -t \sum_{(i,j) \in \text{E}(\Gamma_L)} \sum_{\eta = \uparrow,\downarrow} \left(e^{-i \phi \mathcal{I}[(i,j), \mathcal{C}^*]} a^+_{i,\eta} \hat{\sigma}^z_{ij} a^-_{j,\eta} + \mathrm{h.c.}\right)
    + m \sum_{i \in \Gamma_L}  (-1)^{\vert i_1+i_2 \vert} n_{i}  \\
    &\quad + U \sum_{i\in \Gamma_{L}} \Big( n_{i,\uparrow} - \frac{1}{2} \Big)\Big( n_{i,\downarrow} - \frac{1}{2} \Big).
    \label{twistedhami}
\end{align}
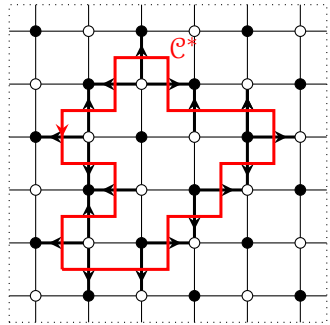
\begin{figure}
\centering
\begin{tikzpicture}[scale=0.7]
\draw[] (1,1) grid (6,6);
\draw[black] (0.5,1) -- (1,1);
\draw[black] (0.5,2) -- (1,2);
\draw[black] (0.5,3) -- (1,3);
\draw[black] (0.5,4) -- (1,4);
\draw[black] (0.5,5) -- (1,5);
\draw[black] (0.5,6) -- (1,6);
\draw[black] (1,6) -- (1,6.5);
\draw[black] (2,6) -- (2,6.5);
\draw[black] (3,6) -- (3,6.5);
\draw[black] (4,6) -- (4,6.5);
\draw[black] (5,6) -- (5,6.5);
\draw[black] (6,6) -- (6,6.5);
\draw[black] (6,6) -- (6.5,6);
\draw[black] (6,5) -- (6.5,5);

\draw[black] (1,1) -- (1,0.5);
\draw[black] (2,1) -- (2,0.5);
\draw[black] (3,1) -- (3,0.5);
\draw[black] (4,1) -- (4,0.5);
\draw[black] (5,1) -- (5,0.5);
\draw[black] (6,1) -- (6,0.5);
\draw[black] (6,1) -- (6.5,1);
\draw[black] (6,2) -- (6.5,2);
\draw[black] (6,3) -- (6.5,3);
\draw[black] (6,4) -- (6.5,4);

\draw[black, dotted] (0.5,0.5)--(6.5,0.5) --(6.5, 6.5) -- (0.5,6.5)--(0.5,0.5);
\draw[black, ->>>-, thick, line width = 0.04 cm] (2,2)--(2,1);
\draw[black, ->>>-, thick, line width = 0.04 cm] (3,2)--(3,1);
\draw[black, ->>>-, thick, line width = 0.04 cm] (3,2)--(4,2);
\draw[black, ->>>-, thick, line width = 0.04 cm] (4,3)--(4,2);
\draw[black, ->>>-, thick, line width = 0.04 cm] (4,3)--(5,3);
\draw[black, ->>>-, thick, line width = 0.04 cm] (5,4)--(5,5);
\draw[black, ->>>-, thick, line width = 0.04 cm] (5,4)--(5,5);
\draw[black, ->>>-, thick, line width = 0.04 cm] (5,4)--(6,4);
\draw[black, ->>>-, thick, line width = 0.04 cm] (4,4)--(4,5);
\draw[black, ->>>-, thick, line width = 0.04 cm] (3,5)--(4,5);
\draw[black, ->>>-, thick, line width = 0.04 cm] (3,5)--(3,6);
\draw[black, ->>>-, thick, line width = 0.04 cm] (3,5)--(2,5);
\draw[black, ->>>-, thick, line width = 0.04 cm] (2,4)--(2,5);
\draw[black, ->>>-, thick, line width = 0.04 cm] (2,4)--(1,4);
\draw[black, ->>>-, thick, line width = 0.04 cm] (2,4)--(2,3);
\draw[black, ->>>-, thick, line width = 0.04 cm] (3,3)--(2,3);
\draw[black, ->>>-, thick, line width = 0.04 cm] (2,2)--(2,3);
\draw[black, ->>>-, thick, line width = 0.04 cm] (2,2)--(1,2);
\draw[black, ->>>-, thick, line width = 0.04 cm] (5,4)--(5,3);
\draw[red, ->>>-, thick, line width = 0.04 cm] (1.5,1.5) -- (3.5,1.5) -- (3.5, 2.5) -- (4.5, 2.5) -- (4.5,3.5) -- (5.5, 3.5) -- (5.5,4.5) -- (3.5,4.5) -- (3.5,5.5) -- (2.5,5.5) -- (2.5,4.5) -- (1.5,4.5) -- (1.5, 3.5) -- (2.5, 3.5) -- (2.5,2.5) -- (1.5,2.5) -- (1.5,1.5);
\draw[black,fill=white] (1,1) circle (.1 cm);
\draw[black,fill=white] (5,1) circle (.1 cm);
\draw[black,fill=white] (4,2) circle (.1 cm);
\draw[black,fill=white] (6,2) circle (.1 cm);
\draw[black,fill=white] (5,3) circle (.1 cm);
\draw[black,fill=white] (4,4) circle (.1 cm);
\draw[black,fill=white] (6,4) circle (.1 cm);
\draw[black,fill=white] (2,4) circle (.1 cm);
\draw[black,fill=white] (1,5) circle (.1 cm);
\draw[black,fill=white] (3,5) circle (.1 cm);
\draw[black,fill=white] (5,5) circle (.1 cm);
\draw[black,fill=white] (2,6) circle (.1 cm);
\draw[black,fill=white] (4,6) circle (.1 cm);
\draw[black,fill=white] (6,6) circle (.1 cm);
\draw[black,fill=white] (3,1) circle (.1 cm);
\draw[black,fill=white] (3,3) circle (.1 cm);
\draw[black,fill=white] (2,2) circle (.1 cm);
\draw[black,fill=white] (1,3) circle (.1 cm);
\draw[black,fill=black] (1,2) circle (.1 cm);
\draw[black,fill=black] (2,1) circle (.1 cm);
\draw[black,fill=black] (3,2) circle (.1 cm);
\draw[black,fill=black] (2,3) circle (.1 cm);
\draw[black,fill=black] (1,4) circle (.1 cm);
\draw[black,fill=black] (1,6) circle (.1 cm);
\draw[black,fill=black] (2,5) circle (.1 cm);
\draw[black,fill=black] (3,4) circle (.1 cm);
\draw[black,fill=black] (3,6) circle (.1 cm);
\draw[black,fill=black] (5,6) circle (.1 cm);
\draw[black,fill=black] (4,5) circle (.1 cm);
\draw[black,fill=black] (6,5) circle (.1 cm);
\draw[black,fill=black] (5,4) circle (.1 cm);
\draw[black,fill=black] (4,3) circle (.1 cm);
\draw[black,fill=black] (6,3) circle (.1 cm);
\draw[black,fill=black] (5,2) circle (.1 cm);
\draw[black,fill=black] (4,1) circle (.1 cm);
\draw[black,fill=black] (6,1) circle (.1 cm);
\node[red, above] (c) at (3.8,5.3) {{$\mathcal{C}^*$}};
\end{tikzpicture}
\caption{Graphical representation of twisting the hamiltonian inserting a flux $\phi$ inside $\mathcal{C}^*$, oriented counterclockwise. Each bold edge is twisted with a phase of $\ep{-\iu \phi}$.}
    \label{fig:twistedhami}
\end{figure}
    As we shall see later, see \hyperref[claim1]{Lemma \ref*{claim1}}, the flow $H\mapsto H_{\caC^*}(\phi)$ is just a unitary transformation whenever $\caC^* = \partial\Lambda$ is a trivial cocycle, namely it is a coboundary. As such, it does not affect the spectrum. The case of a non-trivial cocycle $\mathcal{C}^*$ is different and it is described by \autoref{thm3}: the ground state manifold undergoes a non-trivial spectral flow, along which the spectral gap does not close. This `parallel transport' of the ground state space is given by a unitary propagator, see~(\ref{eq:QAFlow}, \ref{eq:parallel}) below, generated by the so-called quasi-adiabatic generator.

\begin{definition}[Quasi-adiabatic generator]\label{def:QAG}
Let $\mathcal{C}^*\in Z_1(\Gamma_L^*)$. The quasi-adiabatic generator $\mathcal{K}_{\mathcal{C}^*}$ is the family 
\begin{equation}
   \mathcal{K}_{\mathcal{C}^*}(\phi) 
   =\tau_f^\phi\left(\dot{H}_{\mathcal{C}^*}(\phi)\right)
   = \int_{\mathbb{R}}  f(s) \ep{\iu s H_{\mathcal{C}^*}(\phi)} \dot{H}_{\mathcal{C}^*}(\phi) \ep{-i s H_{\mathcal{C}^*}(\phi)} ds,
\end{equation}
where $f \in L^\infty(\mathbb{R})$ such that $\|f\|_{L^1}=1$, $\widehat{f}(\omega) = \frac{1}{i \omega}$ for $|\omega|>g$ and $|f(t)|\leq \frac{C_k}{1+|t|^k}$ for any $k \in \mathbb{N}$. Here $g$ is gap of $H$, see \eqref{eq:gap}.
\end{definition}

\begin{remark}\label{remexploc}
By definition, the operator $\dot{H}_{\mathcal{C}^*}(\phi)$ is supported on the edges intersecting $\mathcal{C}^*$. With this, it is a standard argument that $\mathcal{K}_{\mathcal{C}^*}(\phi)$ is almost supported on $\mathcal{C}^*$ in the sense that $\mathcal{K}_{\mathcal{C}^*}(\phi)$ can be approximated by an observable stricly supported on a ribbon of width $R$ around $\mathcal{C}^*$, with an error $O(R^{-\infty})$. Indeed, this is true for $\ep{\iu s H_{\mathcal{C}^*}(\phi)} \dot{H}_{\mathcal{C}^*}(\phi) \ep{-i s H_{\mathcal{C}^*}(\phi)}$ for times $s$ of order $1$ by the Lieb-Robinson bound, while the contribution to the integral for longer times is small by the decay of $f$. We refer to~\emph{\cite[Lemma~4.7]{Bachmann_2011}} and the review~\emph{\cite{Nachtergaele_2019}} for additional details on the Lieb-Robinson bound and the spectral flow and to~\emph{\cite{Bachmann_2019}} specifically for the locality of $\mathcal{K}_{\mathcal{C}^*}(\phi)$.
\end{remark}

Besides its locality, the operator $\mathcal{K}_{\mathcal{C^*}}(\phi)$ generates a parallel transport on the bundle of ground state projections $P_{\mathcal{C}^*}(\phi)$ of $H_{\mathcal{C}^*}(\phi)$:
\begin{equation}\label{gsode}
   \dot P_{\mathcal{C}^*}(\phi) = \iu [\mathcal{K}_{\mathcal{C}^*}(\phi), P_{\mathcal{C}^*}(\phi)]
\end{equation}
provided the spectral gap remains open, see~\cite[Proposition~2.4]{Bachmann_2011}. This is a consequence of the more general fact that the map $O\mapsto\tau^\phi_f(O)$ is an inverse of $-\iu[H_{\mathcal{C}^*}(\phi),O]$ on the set of off-diagonal operators $O = POP^\perp + P^\perp O P$, applied here to $O = \dot{P}_{\mathcal{C}^*}(\phi)$.

Let now $V_{\caC^*}(\phi)$ be the solution of 
\begin{equation}\label{eq:QAFlow}
    -\iu \dot V_{\caC^*}(\phi) = \mathcal{K}_{\mathcal{C}^*}(\phi) V_{\caC^*}(\phi)
\end{equation}
such that $V_{\caC^*}(0) = \mathbbm{1}$. It is the propagator of parallel transport, namely
\begin{equation}\label{eq:parallel}
    P_{\mathcal{C}^*}(\phi) = V_{\caC^*}(\phi)P_{\mathcal{C}^*}(0)V_{\caC^*}(\phi)^*.
\end{equation}

\begin{definition}[Loop Operators] \label{def:loop operators} Let $\mathcal{C}^*\in Z_1(\Gamma_L^*)$. We define 
    \begin{equation}
    W_{\mathcal{C}^*} = \hat X_{\mathcal{C}^*} V_{\caC^*}(\pi).
\end{equation}
\end{definition}

\begin{theorem}[Braiding]\label{thm3} Under the same assumptions of \autoref{thm2} the following is true. Let $\{\ket{\Omega_{ab}}:(a,b)\in\mathbb{Z}_2\times\mathbb{Z}_2\}$ be the four ground states of \autoref{thm2}, let $P_{ab}$ be the corresponding orthogonal projectors, and let $P = \sum_{(a,b)\in \mathbb{Z}_2\times\mathbb{Z}_2} P_{ab}$. 
\begin{enumerate}
\item \label{thm3.1} For any $\mathcal{C}^*\in Z_1(\Gamma_L^*)$
\begin{equation}
[W_{\mathcal{C}^*}, P] \overset{L}{=} 0.
\end{equation} 
where $\overset{L}{=}$ means that the equality holds up to terms that are smaller, in operator norm, than any power of the system size. 
\item \label{thm3.2} For any $\mathcal{C} \in Z_1(\Gamma_L), \mathcal{C}^*\in Z_1(\Gamma_L^*)$,
\begin{equation}\label{eq:ZW braiding}
\hat Z_{\mathcal{C}} W_{\mathcal{C}^*} = \ep{\iu \pi \mathcal{I}(\mathcal{C},\mathcal{C}^*)} W_{\mathcal{C}^*} \hat Z_{\mathcal{C}}.
\end{equation}
\item \label{thm3.3}  If $\{0,\caC_1^*,\caC_2^*, \caC_1^*+\caC_2^*\}$ are representatives of each of the cohomology classes, then
\begin{equation}
    P_{a'b'}(W_{\mathcal{C}_1^*})^{p_1}(W_{\mathcal{C}_2^*})^{p_2} \overset{L}{=}   (W_{\mathcal{C}_1^*})^{p_1}(W_{\mathcal{C}_2^*})^{p_2}P_{ab},
\end{equation}
where $a' = (-1)^{p_1}a, b' = (-1)^{p_2}b$, for any $p_1,p_2\in\mathbb{Z}$.
\item \label{thm3.3bis} For any $(a,b)$,
\begin{equation}\label{bosons}
    (W_{\mathcal{C}_1^*})^*(W_{\mathcal{C}_2^*})^*W_{\mathcal{C}_1^*}W_{\mathcal{C}_2^*} \ket{\Omega_{ab}} \stackrel{L}{=} \ket{\Omega_{ab}}.
\end{equation}
\item \label{thm3.4} Let $i,j\in\Lambda$ be such that $\mathrm{dist}(i,j)$ is of order $L$ and let $\caC_{i,j}\in C_1(\Gamma_L)$ be a $1$-chain such that $(i,j) = \partial\caC_{i,j}$. Denote $\ket{\xi_{ab}^{ij}} = \mathcal{N}^{-1} a^{+}_{i} \hat Z_{\caC_{i,j}} a^{+}_{j}\ket{\Omega_{ab}}$, where $\mathcal{N}^{-1}$ ensures that $\Vert \xi_{ab}^{ij}\Vert = \Vert \Omega_{ab}\Vert$. Then
\begin{equation}
\frac{\langle \xi_{ab}^{ij}\vert W_{\mathcal{C}^*} \xi_{ab}^{ij} \rangle}{\langle \Omega_{ab} \vert W_{\mathcal{C}^*} \Omega_{ab}\rangle } \overset{L}{=} -1
\end{equation}
where $\mathcal{C}^*$ is the circle of radius $\frac{1}{2}\mathrm{dist}(i,j)$ centred at $i$.
\end{enumerate}
\end{theorem}

With topological order, all local observables act trivially on the ground state space. The theorem \autoref{thm3.3} above exhibits the only non-trivial, necessarily global, observables, namely the large loop operators $W_{\mathcal{C}_j^*},j=1,2$. When the loops are opened up, a string operator $W_{\mathcal{C}^*}$ creates a pair of monopoles while a string operator $a^{+}_{i} \hat Z_{\caC_{i,j}} a^{+}_{j}$ creates a pair of fermions at their endpoints, and~\autoref{thm3.4} shows explicitly that they are mutual anyons, see again~\autoref{braiding1}. The commutation~(\ref{bosons}) is independent of any choice of phases in the definition of the loop operators or ground state vector: It proves that monopoles are, as expected, bosons.

For any even local observable $\mathcal{O}$ supported away from sites $i$ and $j$,
\begin{equation}\label{eq:localization}
    \langle \xi_{ab}^{ij} \vert \mathcal{O} \xi_{ab}^{ij}\rangle \simeq \langle \Omega_{ab}\vert \mathcal{O} \Omega_{ab}\rangle
\end{equation}
up to errors that decay superpolynomially in the distance between the support of $\mathcal{O}$ and the sites $i,j$, and in the system size $L$. This should be interpreted as the fact that $\vert \xi_{ab}^{ij} \rangle$ describes (fermionic) excitations that are localized at the sites $i,j$. In order to derive~(\ref{eq:localization}), we first note that for any path $\mathcal{\tilde C}_{ij}$ from $i$ to $j$ and such that $\mathcal{C}_{ij} - \mathcal{\tilde C}_{ij}$ is a boundary,
\begin{equation}
    Z_{\mathcal{C}_{ij}}\Omega_{ab} = z Z_{\mathcal{\tilde C}_{ij}}\Omega_{ab}
\end{equation}
where $\vert z\vert = 1$ and it depends on the paths. If $\mathcal{O}$ is a local observable, we can pick $\mathcal{\tilde C}_{ij}$ that lies away from its support. The corresponding string operator then commutes with $\mathcal{O}$ and so
\begin{equation}
    \langle \xi_{ab}^{ij}\vert \mathcal{O} \xi_{ab}^{ij}\rangle
    = \mathcal{N}^{-2}\langle \Omega_{ab}\vert a_ja_i \mathcal{O} a_i^+a_j^+\Omega_{ab}\rangle.
\end{equation}
If $\mathcal{O}$ is even, then $\langle \Omega_{ab}\vert a_ja_i \mathcal{O} a_i^+a_j^+\Omega_{ab}\rangle = \langle \Omega_{ab}\vert \mathcal{O} a_ja_i a_i^+a_j^+\Omega_{ab}\rangle$. Since the spectral gap implies fast decay of correlations~\cite[Appendix~A]{Bachmann_2021} and $\mathcal{O}$ is supported away from the sites $i,j$, 
\begin{equation}
    \langle \xi_{ab}^{ij}\vert \mathcal{O} \xi_{ab}^{ij}\rangle
    \simeq \mathcal{N}^{-2}\langle \Omega_{ab}\vert \mathcal{O} P a_ja_i a_i^+a_j^+\Omega_{ab}\rangle
\end{equation}
up to an error that decays with the distance between the support of $\mathcal{O}$ and $\{i,j\}$. Finally, local topological order implies that $P$ can be replaced by the one-dimensional projection onto the range of $\Omega_{ab}$ up to an error that decays fast in $L$, namely
\begin{equation}
    \langle \xi_{ab}^{ij}\vert \mathcal{O} \xi_{ab}^{ij}\rangle
    \simeq \mathcal{N}^{-2}\langle \Omega_{ab}\vert \mathcal{O} \Omega_{ab}\rangle \Vert a_i^+a_j^+\vert \Omega_{ab}\rangle\Vert^2.
\end{equation}
This is the claim since the last factor is equal to $\mathcal{N}^2$.

Overall, the model has the full toric code structure. If we let $m$ denote a monopole excitation and $\epsilon$ denote a fermionic excitation, then $m\times m = 1 = \epsilon\times\epsilon$, where $1$ denotes the vacuum. Moreover, the composite particle $e = m\times \epsilon$ is a boson, and the braiding of $e$ with $m$ yields again $-1$. Although our model is fundamentally fermionic, it therefore exhibits the exact toric code topological order, unlike the fermionic model proposed in~\cite{gu2014lattice}. It is a fermionic reflection positive representative of the toric code phase, see the discussion in~\cite{sopenkoRP}.

\section{Proof of \autoref{thm1}}\label{sec2}
\subsection{Proof of \autoref{thm1} -  \autoref{thm1.1}: non-interacting fermions}\label{specpiflux}

We start by diagonalizing the $\pi$-flux Hamiltonian at $U=0$ with $(1,1)$ holonomies. A particularly simple gauge field configuration associated with this phase is represented in \autoref{fig:pi1}. Clearly, the corresponding Hamiltonian is not translationally invariant with respect to all lattice translations. In order to recover translation invariance, we introduce the appropriate fundamental cell formed by four lattice sites, labelled by $A, B, C, D$, as in \autoref{fig:pi1}. The position of the cell is defined by the coordinate of the $B$-lattice site. We shall denote by $\Gamma_{L}^{\text{red}}$ the two-dimensional lattice formed by the positions of the fundamental cells. The lattice spacing between nearest-neighbour fundamental cells in $\Gamma_{L}^{\text{red}}$ is $2$ in both directions.

For vanishing Hubbard interaction, the spin does not play any role and we discuss the computation in the spinless case. If spin was included, all energy levels would simply be doubly degenerate.

Using these coordinates described in \autoref{fig:pi1}, the Hamiltonian can be rewritten as:
\begin{align}
    H(\bm{-1}; 1,1) 
&=-t\sum_{x \in \Gamma_{L}^{\text{red}}}
\left( a^{+}_{B,x} a^{-}_{A,x} + a^{+}_{B,x} a^{-}_{C,x} + a^{+}_{A,x} a^{-}_{B,x+2e_2} - a^{+}_{A,x} a^{-}_{D ,x} \right.\\ 
&\phantom{=\sum_{x \in \Gamma_{L}^{\text{red}}}}\left.\quad- a^{+}_{D,x} a^{-}_{A,x+2e_1} + a^{+}_{D,x} a^{-}_{C,x+2e_2} +a^{+}_{C,x} a^{-}_{D,x} + a^{+}_{C,x} a^{-}_{B,x+2e_1} + \mathrm{h.c.}\right) \\
&\quad +m \sum_{x \in \Gamma^{\text{red}}_L} \left(n_{A,x} + n_{C,x} - n_{B,x} - n_{D,x}\right).
\label{eq:H11}
\end{align}
\begin{figure}
\centering
\begin{tikzpicture}[scale=0.6]
\begin{scope}[xshift=-1cm, yshift = 2cm]
\draw (-0.5,-0.5) grid (7.5,7.5);
\draw[line width = 0.05 cm] (-0.5,0) -- (7.5,0);
\draw[line width = 0.05 cm] (-0.5,2) -- (7.5,2);
\draw[line width = 0.05 cm] (-0.5,4) -- (7.5,4);
\draw[line width = 0.05 cm] (-0.5,6) -- (7.5,6);

\draw[dotted, fill = gray, fill opacity = 0.2] (3.5,2.5)--(5.5,2.5) -- (5.5,4.5) -- (3.5,4.5) -- (3.5,2.5);
\draw[black,fill=white] (0,0) circle (.1 cm);
\draw[black,fill=white] (1,1) circle (.1 cm);
\draw[black,fill=white] (3,3) circle (.1 cm);
\draw[black,fill=white] (2,2) circle (.1 cm);
\draw[black,fill=white] (4,4) circle (.1 cm);
\draw[black,fill=white] (5,5) circle (.1 cm);
\draw[black,fill=white] (6,6) circle (.1 cm);
\draw[black,fill=white] (7,7) circle (.1 cm);

\draw[black,fill=white] (3,5) circle (.1 cm);

\draw[black,fill=white](2,0) circle (.1 cm);
\draw[black,fill=white] (0,2) circle (.1 cm);
\draw[black,fill=white] (4,0) circle (.1 cm);
\draw[black,fill=white] (0,4) circle (.1 cm);
\draw[black,fill=white] (6,0) circle (.1 cm);
\draw[black,fill=white] (0,6) circle (.1 cm);
\draw[black,fill=white] (1,3) circle (.1 cm);
\draw[black,fill=white] (1,5) circle (.1 cm);
\draw[black,fill=white] (1,7) circle (.1 cm);
\draw[black,fill=white] (7,1) circle (.1 cm);
\draw[black,fill=white] (5,1) circle (.1 cm);
\draw[black,fill=white] (3,1) circle (.1 cm);

\draw[black,fill=white] (4,2) circle (.1 cm);
\draw[black,fill=white] (6,2) circle (.1 cm);
\draw[black,fill=white] (2,6) circle (.1 cm);
\draw[black,fill=white] (2,4) circle (.1 cm);
\draw[black,fill=white] (6,4) circle (.1 cm);
\draw[black,fill=white] (5,3) circle (.1 cm);
\draw[black,fill=white] (7,3) circle (.1 cm);
\draw[black,fill=white] (3,5) circle (.1 cm);
\draw[black,fill=white] (3,7) circle (.1 cm);
\draw[black,fill=white] (7,5) circle (.1 cm);
\draw[black,fill=white] (4,6) circle (.1 cm);
\draw[black,fill=white] (5,7) circle (.1 cm);

\draw[black,fill=black] (1,0) circle (.1 cm);
\draw[black,fill=black] (0,1) circle (.1 cm);
\draw[black,fill=black] (3,0) circle (.1 cm);
\draw[black,fill=black] (0,3) circle (.1 cm);
\draw[black,fill=blue] (5,0) circle (.1 cm);
\draw[blue,fill=black] (0,5) circle (.1 cm);
\draw[black,fill=black] (7,0) circle (.1 cm);
\draw[black,fill=black] (0,7) circle (.1 cm);
\draw[black,fill=black] (1,0) circle (.1 cm);
\draw[black,fill=black] (0,1) circle (.1 cm);
\draw[black,fill=black] (3,0) circle (.1 cm);
\draw[black,fill=black] (0,3) circle (.1 cm);
\draw[black,fill=black] (5,0) circle (.1 cm);
\draw[black,fill=black] (0,5) circle (.1 cm);
\draw[black,fill=black] (7,0) circle (.1 cm);
\draw[black,fill=black] (0,7) circle (.1 cm);

\draw[black,fill=black] (1,2) circle (.1 cm);
\draw[black,fill=black] (2,1) circle (.1 cm);
\draw[black,fill=black] (3,2) circle (.1 cm);
\draw[black,fill=black] (2,3) circle (.1 cm);
\draw[black,fill=black] (5,2) circle (.1 cm);
\draw[black,fill=black] (2,5) circle (.1 cm);
\draw[black,fill=black] (7,2) circle (.1 cm);
\draw[black,fill=black] (2,7) circle (.1 cm);
\draw[black,fill=black] (1,4) circle (.1 cm);
\draw[black,fill=black] (4,1) circle (.1 cm);
\draw[black,fill=black] (3,4) circle (.1 cm);
\draw[black,fill=black] (4,3) circle (.1 cm);
\draw[black,fill=black] (5,4) circle (.1 cm);
\draw[black,fill=black] (4,5) circle (.1 cm);
\draw[black,fill=black] (7,4) circle (.1 cm);
\draw[black,fill=black] (4,7) circle (.1 cm);

\draw[black,fill=black] (1,6) circle (.1 cm);
\draw[black,fill=black] (6,1) circle (.1 cm);
\draw[black,fill=black] (3,6) circle (.1 cm);
\draw[black,fill=black] (6,3) circle (.1 cm);
\draw[black,fill=black] (5,6) circle (.1 cm);
\draw[black,fill=black] (6,5) circle (.1 cm);
\draw[black,fill=black] (7,6) circle (.1 cm);
\draw[black,fill=black] (6,7) circle (.1 cm);

\draw[dotted, ->-] (-0.5,-0.5) -- (7.5,-0.5);
\draw[dotted, ->-] (-0.5,7.5) -- (7.5,7.5);
\draw[dotted, ->>-] (-0.5,-0.5) -- (-0.5,7.5);
\draw[dotted, ->>-] (7.5,-0.5) -- (7.5,7.5);
\end{scope}
\begin{scope}[xshift =11.5 cm, yshift=3.5cm]
\draw[dotted,fill = gray, fill opacity = 0.2] (0,0) -- (0,4) -- (4,4)--(4,0)--(0,0);
\draw[] (1,0) -- (1,4);
\draw[] (3,0) -- (3,4);
\draw[] (0,1) -- (4,1);
\draw[thick, line width = 0.05 cm] (0,3) -- (4,3);
\node[below left] (a) at (1,1) {{$B$}};
\node[above left] (b) at (1,3) {{$A$}};
\node[below right] (a) at (3,1) {{$C$}};
\node[above right] (b) at (3,3) {{$D$}};
\draw[line width = 0.02 cm, ->] (0,0)--(0,4);
\draw[line width = 0.02 cm, ->] (0,0)--(4,0);
\draw[black,fill=black] (1,1) circle (.1 cm);
\draw[black,fill=black] (3,3) circle (.1 cm);
\draw[black,fill=white] (1,3) circle (.1 cm);
\draw[black,fill=white] (3,1) circle (.1 cm);
\end{scope}
\end{tikzpicture}
\caption{On the left, graphical representation of the gauge field configurations associated with the $(1,1)$ $\pi$ flux phase. The thin bonds correspond to $\sigma^z = +1$, while the solid bonds correspond to spin $\sigma^z = -1$. On the right, definition of the fundamental cell.}
\label{fig:pi1}
\end{figure}
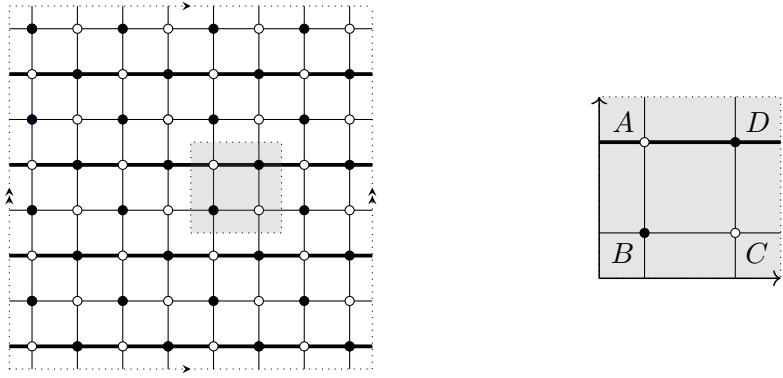
The Brillouin zone $B_{L}(1,1)$ is defined as
\begin{equation}\label{eq:BL11}
B_{L}(1,1) := \Big\{ k\in \frac{2\pi}{L}(n_{1}, n_{2}) \mid 0\leq n_{1} \leq L/2-1,\; 0\leq n_{2} \leq L/2 - 1 \Big\}\;,
\end{equation}
and the momentum-space creation/annihilation operators as:
\begin{equation}\label{eq:afou}
\hat a^{\pm}_{\alpha, k} = \sum_{x\in \Gamma_{L}^{\text{red}}} e^{\pm ik\cdot x} a^{\pm}_{\alpha, x} \iff a^{\pm}_{\alpha, x} = \frac{1}{|\Gamma^{\text{red}}_{L}|} \sum_{k \in B_{L}(1,1)} e^{\mp ik\cdot x} \hat a^{\pm}_{\alpha, k}\;,
\end{equation}
for $\alpha \in \{A,B,C,D\}$. The Hamiltonian \eqref{eq:H11} becomes:
\begin{equation}\label{eq:blochHam}
H(\bm{-1}; 1,1) = \frac{1}{|\Gamma^{\text{red}}_{L}|
} \sum_{k \in B_{L}(1,1)} (\hat a^{+}_{k}, h(k) \hat a^{-}_{k})\;,
\end{equation}
where $(f,g) = \sum_{\alpha} f_{\alpha} g_{\alpha}$, and the Bloch Hamiltonian $h(k)$ is given by
 \begin{equation}\label{eq:BlochHam}
     h(k) = \begin{pmatrix}
         m &-t(1 +  e^{2i k_2}) &0 &t(1+ e^{-2i k_1})\\
         -t(1 +  e^{-2i k_2}) &-m &-t(1+ e^{-2i k_1}) &0\\
         0 &-t(1+e^{2 i k_1}) &m &-t(1+ e^{-2 i k_2})\\
        t( 1+ e^{2i k_1}) &0 &-t(1+ e^{2ik_2}) &-m
     \end{pmatrix}\;.
 \end{equation}
The above Bloch Hamiltonian coincides with the one found in \cite{PhysRevB.50.7526}. Its (doubly degenerate) eigenvalues are
\begin{equation}\label{eq:epi}
e_{\pm}(k) = \pm 2t \sqrt{\Big(\frac{m}{2t}\Big)^2 + 1 + \frac{1}{2}\cos(2k_{1}) + \frac{1}{2} \cos(2k_{2})},
\end{equation}
and they vanish nowhere in the Brillouin zone whenever $m \neq 0$ (see \autoref{fig:disp}), leaving a gap equal to $2m$.
\begin{figure}
\centering
\includegraphics[trim={0cm 0cm 0cm 0cm}, clip, width=.7\textwidth]{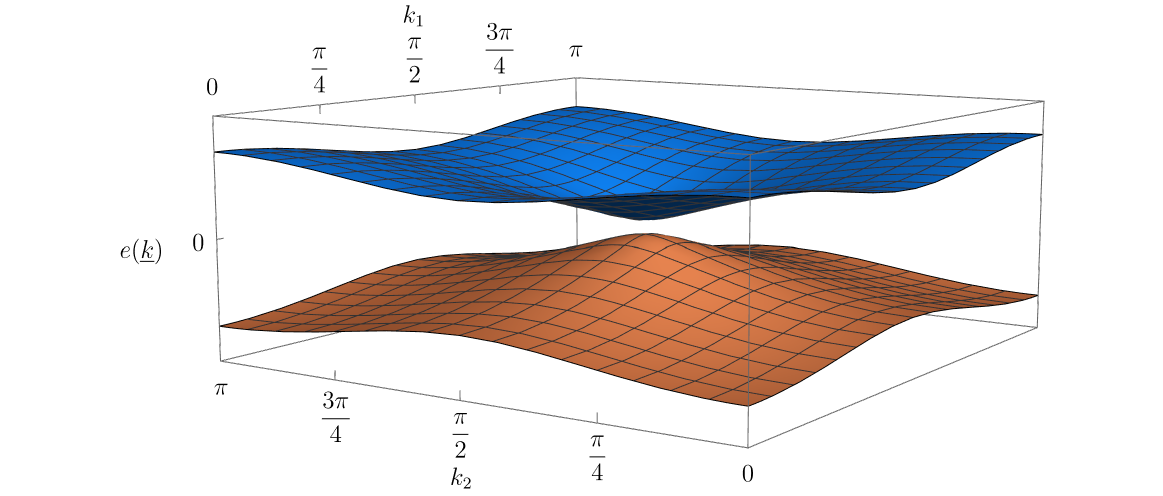}
\caption{Energy bands associated with the $\pi$-flux phase with a staggered mass \eqref{eq:epi}.}
    \label{fig:disp}
\end{figure}
The ground state of the system is the Slater determinant obtained by occupying all the Bloch states below energy zero. Since the cardinality of Brillouin zone is equal to $L^2/4$, see \eqref{eq:BL11}, the rank of the Fermi projection, namely the number of particles, is equal to $L^2/2$. If the two spin states are reintroduced, the number of particles equals $L^2$. In both cases, this corresponds to half-filling. 

The case of more general holonomies around non-contractible cycles can be analysed via a change of the boundary conditions, and thus of the allowed momenta in the Brillouin zone. Namely, the Brillouin zone is given by
\begin{equation}
    B_{L}(e^{i\theta},e^{i\phi}) = \Big\{ k \in \frac{2\pi}{L}\big(n_1+ \frac{\theta}{2\pi}, n_2 + \frac{\phi}{2\pi}\big) \mid 0 \leq n_1 \leq L/2-1,\,0 \leq n_2 \leq L/2-1\Big\}\;.
\end{equation}
From there, the discussion proceeds exactly as in the $(1,1)$ case. This concludes the proof of \autoref{thm1} - \autoref{thm1.1}, for non-interacting fermions.

\subsection{Proof of \autoref{thm1} - \autoref{thm1.2}: non-interacting fermions}\label{sec:expdeg}

The ground state energy of the system for general holonomies is, using the short-hand notation $B(\theta,\phi) \equiv B_{L}(e^{i\theta}, e^{i\phi})$, $E_{0}(\bm{-1},\theta,\phi) \equiv E_{0,L}(\bm{-1},e^{i\theta},e^{i\phi})$:
\begin{equation}\label{eq:gsgen}
    E_0({\bf -1}, \theta, \phi) = 2 \sum_{k \in B(\theta,\phi)} e_-(k)\;.
\end{equation}
In order to prove the bound \eqref{eq:expdeg}, we will rewrite \eqref{eq:gsgen} using Poisson's summation formula:
\begin{equation}\label{PSF}
    \frac{1}{L^2} \sum_{k \in B(\theta,\phi)} e_-(k) = \sum_{\ell_1, \ell_2 \in \mathbb{Z}} e^{i \theta \ell_1} e^{i \phi \ell_2} e_{\infty}(\ell_1 L, \ell_2 L)
\end{equation}
where
\begin{equation}
    e_{\infty}(\ell_1 L, \ell_2 L) = \int_{\mathbb{T}^2} \frac{d^2k}{(2\pi)^2} e^{i k_1 \ell_1 L + i k_2 \ell_2 L} e_-(k)
\end{equation}
is the infinite volume ground state energy density (which is independent of the boundary conditions). In this way, we can conveniently rewrite the energy difference as:
\begin{equation}
    E_0(\bm{-1},\theta_1, \phi_1)- E_0(\bm{-1},\theta_2, \phi_2) = 2L^{2}\sum_{(\ell_1, \ell_2) \in \mathbb{Z}^2\setminus(0,0)} (e^{i \ell_1 \theta_1 + i \ell_2 \phi_1} - e^{i \ell_1 \theta_2 + i \ell_2 \phi_2})  e_{\infty}(\ell_{1}L, \ell_{2}L).
\end{equation}
Note that the $\ell_1 = \ell_2 = 0$ term vanishes because the exponentials cancel out.

\begin{proposition}\label{prp:complex}
   There exist two constants $C,c>0$ depending only on $m$ and $t$, such that
\begin{equation}\label{eq:complex}
    \bigg|\int_{\mathbb{T}^2} \frac{d^2k}{(2\pi)^2} e^{i k x} e_-(k) \bigg| \leq C e^{-c(\vert x_1\vert + \vert x_2\vert)},
\end{equation}
\end{proposition}
Hence, by Poisson summation formula \eqref{PSF},
\begin{equation}
    | E_0(\bm{-1},\theta_1, \phi_1)- E_0(\bm{-1},\theta_2, \phi_2)| \leq 4L^2 \sum_{(\ell_1, \ell_2) \in \mathbb{Z}^2\setminus(0,0)} C e^{-c(|\ell_1| + |\ell_2|) L } \leq K L^2 e^{-2 cL}
\end{equation}
uniformly in $\phi_1, \theta_1, \phi_2, \theta_2$. This concludes the proof of \autoref{thm1.2} of \autoref{thm1} for non-interacting fermions. It remains to prove Proposition~\ref{prp:complex}.
\begin{proof}
    We first note that
    \begin{equation}
        e_-(k) = \frac{1}{2}\tr(h(k)p_-(k))
    \end{equation}
    where the trace is over the fundamental cell and the factor $2$ stands for the rank of $p_-(k)$. The Bloch Hamiltonian $h(k)$, which is explicitly given in~(\ref{eq:BlochHam}), and the ground state projections $p(k)$ are by definition 
\begin{equation}
        h(k) = \sum_{x \in \Gamma_{L}^{\text{red}}} h(x,0)\ep{-\iu k\cdot x},
        \qquad
        p_{-}(k) = \sum_{x \in \Gamma_{L}^{\text{red}}} p_-(x,0)\ep{-\iu k\cdot x},
    \end{equation}
    where $h(x,y),p_-(x,y)$ are the $4\times4$-matrix-valued functions given by 
    \begin{equation}
        H(\bm{-1}; 1,1) = \sum_{x,y\in\Gamma_L^{\rm{red}}} (a^+_{x},h(x,y)a^-_y),
    \end{equation}
    see~(\ref{eq:H11}). With this, the integral~(\ref{eq:complex}) can be computed as
    \begin{equation}\label{eq:e-P-}
        \int_{\mathbb{T}^2} \frac{d^2k}{(2\pi)^2} \ep{\iu k x} e_-(k) 
        = \frac{1}{2}  \sum_{y \in \Gamma_{L}^{\text{red}}} \tr\big(h(x-y,0)p_-(y,0)\big).
    \end{equation}
    By translation invariance, $h(x-y,0) = h(x,y)$ and it vanishes whenever $\vert x-y\vert> 4$. Since the Fermi energy lies in a spectral gap, the Riesz formula for the spectral projector $p_- = \frac{1}{2\pi\iu}\int_{\gamma}(h-z)^{-1}dz$ together with the Combes-Thomas estimate for the resolvent ~\cite[Theorem~10.5]{AizenmanWarzel} imply:
    \begin{equation}
        \vert p_-(y,0)\vert \leq C\ep{-c\vert y\vert}.
    \end{equation}
This bound combined with the identity (\ref{eq:e-P-}) yield the claim (\ref{eq:complex}). Finally, the constant $C,c>0$ can be estimated explicitly in terms of the spectral gap, which is proportional to $m$ and the `hopping amplitude' $t$ of $h(x,y)$:
     \begin{equation}\label{eq:CT C and c}
         C \propto \frac{1}{m},\qquad c\propto \frac{m}{t},
     \end{equation}
    see again~\cite[Theorem~10.5]{AizenmanWarzel}.
\end{proof}

\subsection{Proof of \autoref{thm1} - \autoref{thm1.1} and \autoref{thm1.2}: weakly interacting fermions}\label{sec:int}

We now extend the results of \hyperref[specpiflux]{Subsection \ref*{specpiflux}} and \hyperref[sec:expdeg]{Subsection \ref*{sec:expdeg}} to the case of weakly interacting fermions. This will be done using fermionic cluster expansion. The method provides a convergent expansion for the ground state energy,
\begin{equation}
E_{0}({\bf -1}, \theta, \phi) = -\lim_{\beta \to \infty} \frac{1}{\beta} \log \tr_{\mathcal{F}_{L}} e^{-\beta H({\bf -1}, e^{i\theta}, e^{i\phi})}
\end{equation}
for $|U|$ small enough, uniformly in $L$. Inspection of the convergent expansion will allow us to prove items $1$ and $2$ of Theorem \ref{thm1} for weakly interacting fermions. In the context of topological phases of matter, these methods have been used in \cite{GMPhall} to prove the universality of the Hall conductivity. See also \cite{GJMP, GMPhald, FGR} for the analysis of Hall transitions, via cluster expansion and renormalization group methods. In our setting, inspection of the expansion will allow us to prove the exponential closeness of the approximate ground state energies (\ref{eq:expdeg}). Furthermore, the same technique allows to prove the stability of the fermionic spectral gap, following \cite[Theorem 1]{DRS}.

\paragraph{Analyticity of the free energy.} Here we will discuss the convergent expansion for the free energy, which will be the starting point for the proofs of item $1$ and item $2$ of Theorem~\ref{thm1} for interacting systems. From now on, we shall write $H({\bf -1}, e^{i\theta}, e^{i\phi}) = H_{0} + UV + (U/4)L^{2}$, where:
\begin{equation}
H_{0} = H({\bf -1}, e^{i\theta}, e^{i\phi})\Big|_{U=0} - \frac{U}{2} \sum_{i,\sigma} n_{i,\sigma}\;,\qquad
V = \sum_{i} n_{i,\uparrow} n_{i,\downarrow}\;;
\end{equation}
observe that we included the interaction-dependent quadratic term due to the presence of the $-1/2$ factors in the Hubbard interaction in the definition of the quadratic Hamiltonian $H_0$. The additive constant term in the Hamiltonian does not have any relevance for the Gibbs state of the model, and it will be omitted in what follows. The starting point is the Duhamel expansion for the free energy:
\begin{equation}\label{eq:duhamel}
\begin{split}
&\frac{\tr_{\mathcal{F}_{L}} e^{-\beta (H_{0} + U V)}}{\tr_{\mathcal{F}_{L}} e^{-\beta H_{0} }} \\
&\qquad = 1 + \sum_{n\geq 1} (-U)^{n} \int_{0}^{\beta} dt_{1} \int_{0}^{t_{1}} dt_{2} \cdots \int_{0}^{t_{n-1}} dt_{n} \langle \gamma_{t_{1}}(V) \gamma_{t_{2}}(V) \cdots \gamma_{t_{n}}(V) \rangle^{0}_{\beta,L}
\end{split}
\end{equation}
where $\langle \cdot \rangle_{\beta,L}^{0}$ is the Gibbs state of $H_{0}$ at half-filling, and $\gamma_{t}(\cdot)$ is the imaginary-time evolution,
\begin{equation}
\gamma_{t}(V) = e^{t H_{0}} V e^{-t H_{0}}\;.
\end{equation}
Eq. \eqref{eq:duhamel} can be rewritten as:
\begin{equation}\label{eq:duhamel2}
\begin{split}
&\frac{\tr_{\mathcal{F}_{L}} e^{-\beta (H_{0} + U V)}}{\tr_{\mathcal{F}_{L}} e^{-\beta H_{0} }} \\
&\qquad = 1 + \sum_{n\geq 1} \frac{(-U)^{n}}{n!} \int_{[0,\beta]^{n}} dt_{1}\cdots dt_{n} \langle {\bf T} \gamma_{t_{1}}(V) \gamma_{t_{2}}(V) \cdots \gamma_{t_{n}}(V) \rangle^{0}_{\beta,L}
\end{split}
\end{equation}
where ${\bf T}$ is the time-ordering operator. It acts on fermionic monomials as, at non-equal times:
\begin{equation}
{\bf T} \gamma_{t_{1}}(a^{\varepsilon_{1}}_{i_{1},\sigma_{1}}) \cdots \gamma_{t_{n}}(a^{\varepsilon_{n}}_{i_{n},\sigma_{n}}) = \text{sgn}(\pi) \gamma_{t_{1}}(a^{\varepsilon_{\pi(1)}}_{i_{\pi(1)},\sigma_{\pi(1)}}) \cdots \gamma_{t_{\pi(n)}}(a^{\varepsilon_{\pi(n)}}_{i_{\pi(n)},\sigma_{\pi(n)}})
\end{equation}
where $\pi$ is the permutation such that $t_{\pi(i)} > t_{\pi(i+1)}$. Whenever two times coincide, the time-ordering operator acts as normal-ordering. Taking the logarithm of the left-hand side and of the right-hand side of (\ref{eq:duhamel2}) one gets:
\begin{equation}\label{eq:cumu}
\log \frac{\tr_{\mathcal{F}_{L}} e^{-\beta (H_{0} + U V)}}{\tr_{\mathcal{F}_{L}} e^{-\beta H_{0} }} = \sum_{n\geq 1} \frac{(-U)^{n}}{n!} \int_{[0,\beta]^{n}} dt_{1}\cdots dt_{n} \langle {\bf T} \gamma_{t_{1}}(V)\,;\, \gamma_{t_{2}}(V)\,;\, \cdots\,;\, \gamma_{t_{n}}(V) \rangle^{0}_{\beta,L}
\end{equation}
where the argument of the integral is the $n$-th order cumulant. It is well-known that the cumulant expansion can be represented as an expansion over connected Feynman diagrams, by application of the fermionic Wick's rule (see {\it e.g.} \cite{Salmhofer} for a review of fermionic perturbation theory). The main ingredient of the expansion is the fermionic two-point function, defined as, for $0\leq t,t'<\beta$:
\begin{equation}\label{eq:2pt}
\begin{split}
g((t,i,\sigma), (t',i',\sigma')) &= \langle {\bf T} \gamma_{t}(a^{-}_{i,\sigma}) \gamma_{t'}(a^{+}_{i',\sigma'}) \rangle^{0}_{\beta,L} \\
&= \theta(t>t') \frac{e^{-(t-t') h_{0}}}{1 + e^{-\beta h_{0}}}(i,\sigma; i',\sigma') - \theta(t\leq t') \frac{e^{-(t-t')h_{0}}}{1 + e^{\beta h_{0}}}(i,\sigma; i',\sigma')
\end{split}
\end{equation}
with $h_{0}$ the single-particle Hamiltonian associated with $H_{0}$. It is not difficult to see that  $g((0,i), (t',j)) = -g((\beta,i), (t',j))$, which allows to extend antiperiodically the two-point function to all times $t$ and $t'$ in $\mathbb{R}$.

The problem with the naive diagrammatic expansion of (\ref{eq:cumu}) is that it involves too many addends. In fact, the number of diagrams contributing to the $n$-th order grows as $(n!)^{2}$, which beats the $1/n!$ in (\ref{eq:cumu}), and does not allow to prove absolute summability of the series. Observe that this apparent factorial divergence of the series is insensitive to the fermionic nature of the particles. 

This issue can be avoided using a different expansion, that allows to exploit the anticommutativity of the fermionic operators. Let us rewrite:
\begin{equation}
\gamma_{s}(V) = \sum_{i} \gamma_{t}(V_{i})\qquad \text{with $V_{i} = n_{i,\uparrow}n_{i,\downarrow}$.}
\end{equation}
The Battle-Brydges-Federbush (BBF) formula states that:
\begin{equation}\label{eq:BBF}
\langle {\bf T} \gamma_{t_{1}}(V_{i_{1}})\,;\, \gamma_{t_{2}}(V_{i_{2}})\,;\, \cdots\,;\, \gamma_{t_{n}}(V_{i_{n}}) \rangle^{0}_{\beta,L}  = \sum_{T} \alpha_{T} \Big[ \prod_{\ell\in T} g_{\ell}\Big] \int d\mu_{T}(\underline{s}) \det\big[ s_{b(f),b(f')}g_{(f,f')}\big]\;;
\end{equation}
see \cite[Theorem 6]{DRS}, \cite[Appendix D]{GMR}, \cite[Appendix A3.2]{GeMa} for reviews, proofs, and further references; see also \cite[Theorem 4.4]{chen2025} for a recent application in the context of computational quantum physics. Let us explain the meaning of the various objects involved in the identity (\ref{eq:BBF}). Every fermionic operator appearing in $\gamma_{t}(V_{i})$ is graphically represented by an oriented line, which exits or enters a vertex labelled by the space-time coordinates $(t,i)$, depending on whether the fermionic operator creates or destroys a particle, respectively. The lines are further decorated by the spin labels of the corresponding operators. In this way, every operator $\gamma_{t}(V_{i})$ is graphically represented by a vertex with four incident lines, two incoming and two outgoing. The sum in the right-hand side of (\ref{eq:BBF}) is over the spanning trees $T$ connecting the $n$ vertices associated with the interaction terms, obtained contracting lines with opposite orientations. The pairing of an incoming line associated with the vertex $(t,i)$ and of an outgoing line associated with the vertex $(t',i')$, labelled respectively by spin labels $\sigma$ and $\sigma'$, forms an edge $\ell = ((t,i,\sigma), (t',i',\sigma'))$. Associated with every edge, we have a propagator $g_{\ell}$, 
\begin{equation}
g_{\ell} \equiv g((t,i,\sigma),(t',i',\sigma'))
\end{equation}
as given by (\ref{eq:2pt}). The product in (\ref{eq:BBF}) involves the propagators associated with the edges of the tree. In the determinant, $g_{(f,f')}$ is the entry of a $(n+1)\times(n+1)$ dimensional matrix representing the contraction of lines that do not belong to $T$, where $f$ is the label for a generic incoming line, and $f'$ is the label for a generic outgoing line. From the point of view of Feynman diagrams, these are the loop lines. This matrix element is then multiplied by a number $s_{b(f),b(f')}$ between $0$ and $1$, an interpolation parameter, where $b(f), b(f')$ are integers in $[1,n]$, which return the labels of the vertices associated with the lines $f,f'$. The measure $d\mu_{T}(\underline{s})$ is a probability measure, supported on sequences $\underline{s}$ whose entries $s_{b(f),b(f')}$ can be written as the scalar product of two vectors in $\mathbb{R}^{n}$, $u_{b(f)}$ and $u_{b(f')}$, with unit norm. Finally, $\alpha_{T}$ is either $+1$ or $-1$, and it will not play any role in what follows.

Thus, thanks to the BBF formula (\ref{eq:BBF}), we have:
\begin{equation}\label{eq:BBFbd}
\Big|\langle {\bf T} \gamma_{t_{1}}(V_{i_{1}})\,;\, \gamma_{t_{2}}(V_{i_{2}})\,;\, \cdots\,;\, \gamma_{t_{n}}(V_{i_{n}}) \rangle^{0}_{\beta,L}\Big| \leq \sum_{T}  \Big[ \prod_{\ell\in T} | g_{\ell}| \Big] \int d\mu_{T}(\underline{s})\, \Big| \det\big[ s_{b(f),b(f')}g_{(f,f')}\big] \Big|\;;
\end{equation}
we will use this bound to show that
\begin{equation}\label{eq:bdBBF}
\frac{1}{\beta}\sum_{i_{1}, \ldots, i_{n}}\int_{[0,\beta]^{n}} d\underline{t}\, \Big|\langle {\bf T} \gamma_{t_{1}}(V_{i_{1}})\,;\, \gamma_{t_{2}}(V_{i_{2}})\,;\, \cdots\,;\, \gamma_{t_{n}}(V_{i_{n}}) \rangle^{0}_{\beta,L}\Big| \leq L^{2} C^{n}  n!
\end{equation}
an estimate which allows to prove analiticity of the specific free energy (\ref{eq:cumu}) for small $|U|$ (observe that the two-point function depends analytically on $U$), uniformly in $\beta$ and $L$. The existence of the limits as $\beta \to \infty, L \to \infty$ follows from the uniform convergence of the series, and from the existence of the pointwise limit of the two-point function, which can be proved easily.

The factorial growth of the bound comes from counting the number of trees connecting the $n$ vertices (observe that this number is much smaller than the number of $n$-th order Feynman diagrams). The proof of (\ref{eq:bdBBF}) is a consequence of two main ingredients: a good bound for the $\ell^{1}$ norm of the two-point function, and a good bound for the $\ell^{\infty}$ norm of the determinant. Let us first discuss the bound for the two-point function. For $L$ fixed, and for any $\theta, \phi$, the two-point function associated with $H_{0}$ can be represented as, using Poisson summation formula:
\begin{equation}\label{eq:poisson2pt}
g((t,i), (t',i')) = \sum_{m_{1}, m_{2} \in \mathbb{Z}} e^{i \theta m_{1}} e^{i \phi m_{2}} g_{\infty}((t,i + e_{1}m_{1}L + e_{2}m_{2}L), (t',i'))
\end{equation}
where $e_{1},e_{2}$ is the standard basis of $\mathbb{R}^{2}$ and $g_{\infty}$ is the two-point function on $\mathbb{Z}^{2}$. By the spectral gap of $h_{0}$, a standard Combes-Thomas estimate (see {\it e.g.} \cite[Appendix B]{DRS}) allows to prove that, for all $t,t'$ in $\mathbb{R}$ and for all $i,i'$ in $\mathbb{Z}^{2}$:
\begin{equation}
\| g_{\infty}((t,i); (t',i')) \| \leq C e^{-c\| (t,i) - (t',i') \|_{\beta}}
\end{equation}
with the time-periodic distance:
\begin{equation}
\| (t,i) - (t',i') \|_{\beta}^{2} = \min_{n\in \mathbb{Z}} | t - t' + n \beta |^{2} + \|i-i'\|^{2}\;.
\end{equation}
Thus, this bound, combined with Eq. (\ref{eq:poisson2pt}) allows to prove that the finite-volume two-point function satisfies the estimate (with different constants):
\begin{equation}\label{eq:decper}
\| g((t,i); (t',i')) \| \leq C e^{-c\| (t,i) - (t',i') \|_{\beta,L}}
\end{equation}
with the space-time periodic distance:
\begin{equation}
\| (t,i) - (t',i') \|_{\beta,L}^{2} = \min_{n\in \mathbb{Z}} | t - t' + n \beta |^{2} + \min_{n_{1},n_{2} \in \mathbb{Z}} \|i-i' + n_{1}e_{1} L + n_{2} e_{2} L \|^{2}\;.
\end{equation}
This decay estimate immediately implies the finiteness of the $\ell^{1}$ norm of the two-point function,
\begin{equation}\label{eq:l1bd}
\max_{t',i'} \int_{0}^{\beta} dt \sum_{i} \| g((t,i); (t',i')) \| \leq K
\end{equation}
uniformly in $\beta, L$. Let us now discuss the bound for the determinant in (\ref{eq:BBFbd}). The key tool we shall use is the Gram-Hadamard inequality: if $M = (m_{ij})_{1\leq i,j\leq K}$ is a Gram matrix, that is if all matrix entries are of the form $m_{ij} = (u_{i}, w_{j})$ for $u_{i}, w_{i}$ in a Hilbert space with scalar product $( \cdot, \cdot)$, then:
\begin{equation}
|\det[M]| \leq \prod_{i=1}^{K} \|u_{i}\| \|w_{i}\|
\end{equation}
with $\|\cdot\|$ the norm induced by the scalar product. This bound is useful if the vectors $u_{i}$ and $w_{j}$ have norm independent of $K$. It is well-known however that the two-point function does not have a good Gram representation, due to the discontinuity of the indicator functions in (\ref{eq:2pt}), introduced by the time-ordering. The easiest way to solve this issue is to observe that the two factors multiplying the theta functions in (\ref{eq:2pt}) do have a good Gram representation, separately; see {\it e.g.} \cite[Lemma 4.1]{PS} and \cite[Lemma 10]{DRS}. As shown in \cite[Theorem 2.4]{PS}, this fact allows to prove a bound on the determinant  in (\ref{eq:BBFbd}), of the form:
\begin{equation}\label{eq:detbd}
\Big| \det\big[ s_{b(f),b(f')}g_{(f,f')}\big] \Big| \leq C^{n}
\end{equation}
where the constant $C$ depends on the Hamiltonian $h_{0}$ (and it is finite uniformly in $\beta, L$, thanks to the spectral gap). Putting together (\ref{eq:BBFbd}), (\ref{eq:l1bd}), (\ref{eq:detbd}), the bound (\ref{eq:bdBBF}) easily follows.

\paragraph{Proof of item $1$ of Theorem \ref{thm1}.} Following \cite[Theorem 1]{DRS}, it is known that the convergence of the above expansion, combined with the spectral gap of the single-particle Hamiltonian $h_{0}$, implies that the many-body Hamiltonian $H$ has a spectral gap, bounded below by $2m - K|U|^{1/3}$ for some $K>0$, uniformly in $L$. The same methods allow to prove that, for small $U$, the many-body ground state is unique and it is at half-filling. These claims are true at $U=0$, as proved in \hyperref[specpiflux]{Subsection \ref*{specpiflux}}. The uniqueness of the ground state follows from the continuity of the eigenvalues of the Hamiltonian as a function of $U$, and from the stability of the spectral gap for $U$ small. Let us now prove that the ground state is at half-filling. Since the Hamiltonian commutes with the number operator, the unique ground state of $H$ must be an eigenstate of the number operator $N$. Furthermore, by the convergence of the cluster expansion, $\langle \psi_{0}, N \psi_{0}\rangle$ is continuous in $U$ for $|U|$ small (in fact, analytic). Thus, since $\langle \psi_{0}, N \psi_{0}\rangle$ is integer valued, it must be constant in $U$, and hence equal to its value at $U=0$, which is $L^{2}$. This concludes the proof of \autoref{thm1.1} of \autoref{thm1} for weakly interacting fermions.

\paragraph{Proof of item $2$ of Theorem \ref{thm1}} Let us now apply the previous strategy to prove the exponential closeness of the approximate ground state energies, in presence of many-body interactions, Eq. (\ref{eq:expdeg}). The proof is based on the BBF formula (\ref{eq:BBF}) combined with the Poisson formula for the two-point function (\ref{eq:poisson2pt}). Let us write:
\begin{equation}\label{eq:poissonbbf}
\begin{split}
&\frac{1}{\beta}\sum_{i_{1},\ldots, i_{n}} \int_{[0,\beta]^{n}} d\underline{t}\, \langle {\bf T} \gamma_{t_{1}}(V_{i_{1}})\,;\, \gamma_{t_{2}}(V_{i_{2}})\,;\, \cdots\,;\, \gamma_{t_{n}}(V_{i_{n}}) \rangle^{0}_{\beta,L} \\
&\quad = \frac{1}{\beta}\sum_{i_{1},\ldots, i_{n}} \int_{[0,\beta]^{n}} d\underline{t}\, \sum_{T} \alpha_{T} \Big[ \prod_{\ell\in T} g_{\ell}\Big] \int d\mu_{T}(\underline{s}) \det\big[ s_{b(f),b(f')}g_{(f,f')}\big] \\
&\quad = \sum_{T} \alpha_{T} \frac{1}{\beta}\sum_{i_{1},\ldots, i_{n}} \int_{[0,\beta]^{n}} d\underline{t}\, \Big[ \prod_{\ell\in T} g_{\ell}\Big] \int d\mu_{T}(\underline{s}) \det\big[ s_{b(f),b(f')}g_{(f,f')}\big]\;.
\end{split}
\end{equation}
Using the translation-invariance of the $\pi$-flux phase, we can rewrite (\ref{eq:poissonbbf}) as:
\begin{equation}
(L^{2} / |\mathcal{C}|)\sum_{T} \alpha_{T} \frac{1}{\beta}\sum_{\substack{i_{1},\ldots, i_{n} \\ i_{1} \in \mathcal{C}(L/2,L/2)}} \int_{[0,\beta]^{n}} d\underline{t}\,\Big[ \prod_{\ell\in T} g_{\ell}\Big] \int d\mu_{T}(\underline{s}) \det\big[ s_{b(f),b(f')}g_{(f,f')}\big]
\end{equation}
where $|\mathcal{C}|$ is the number of sites in the fundamental cell, and $\mathcal{C}(L/2,L/2)$ is the fundamental cell containing the site $(L/2,L/2)$. We then break the innermost sum as:
\begin{equation}\label{eq:splitexp}
\begin{split}
&\sum_{\substack{i_{1},\ldots, i_{n} \\ i_{1} \in \mathcal{C}(L/2,L/2)}} \int_{[0,\beta]^{n}} d\underline{t}\,\Big[ \prod_{\ell\in T} g_{\ell}\Big] \int d\mu_{T}(\underline{s}) \det\big[ s_{b(f),b(f')}g_{(f,f')}\big] \\
&\quad = \sum^{*}_{\substack{i_{1},\ldots, i_{n} \\ i_{1} \in \mathcal{C}(L/2,L/2)}}\int_{[0,\beta]^{n}} d\underline{t}\, \Big[ \prod_{\ell\in T} g_{\ell}\Big] \int d\mu_{T}(\underline{s}) \det\big[ s_{b(f),b(f')}g_{(f,f')}\big] \\
&\qquad + \sum^{**}_{\substack{i_{1},\ldots, i_{n} \\ i_{1} \in \mathcal{C}(L/2,L/2)}} \int_{[0,\beta]^{n}} d\underline{t}\,\Big[ \prod_{\ell\in T} g_{\ell}\Big] \int d\mu_{T}(\underline{s}) \det\big[ s_{b(f),b(f')}g_{(f,f')}\big]\;,
\end{split}
\end{equation}
where the first sum involves vertices with coordinates such that $\|i_{j} - i_{1}\| \leq L/3$ for all $j=2,\ldots, n$, while in the second sum at least one vertex is such that $\|i_{j} - i_{1}\|>L/3$. Observe that\footnote{This is due to the fact that, as vector in $\mathbb{R}^{2}$, $i_{j} - i_{1}$ has norm bounded by $L/\sqrt{2}$. The norm of this vector does not decrease after adding $n_{1} e_{1} L + n_{2} e_{2} L$ for integer $n_{1}, n_{2}$.} this also implies $\|i_{j} - i_{1}\|_{L} > L/3$. Consider the second term in (\ref{eq:splitexp}). Let $T'\subseteq T$ be the subtree of $T$ which connects $i_{1}$ to $i_{j}$. Consider the product of propagators associated with $T'$. Denoting by $(i_{1}, i_{f_{2}}, \ldots, i_{f_{r}}, i_{j})$ the path connecting $i_{1}$ to $i_{j}$ in the tree, we have the following:
\begin{equation}
\| i_{1} - i_{f_{2}} \|_{\beta,L} + \ldots + \| i_{f_{r}} - i_{f} \|_{\beta,L} \geq \|i_{j} - i_{1}\|_{L} \geq L/3\;.
\end{equation}
Combined with the exponential decay (\ref{eq:decper}) of the two-point function, this allows to extract an exponentially small factor after summing over all vertex coordinates compatible with the constraint in the sum. Repeating the argument used to prove the analiticity of the free energy discussed in the previous paragraph, we obtain:
\begin{equation}
\Big| \sum_{T} \alpha_{T} \sum^{**}_{\substack{i_{1},\ldots, i_{n} \\ i_{1} \in \mathcal{C}(L/2,L/2)}} \int_{[0,\beta]^{n}} d\underline{t}\,\Big[ \prod_{\ell\in T} g_{\ell}\Big] \int d\mu_{T}(\underline{s}) \det\big[ s_{b(f),b(f')}g_{(f,f')}\big] \Big| \leq C^{n} n! \beta  e^{-cL}\;.
\end{equation}
Consider now the first term in (\ref{eq:splitexp}). Observe that the constraint in the sum implies that $\|i_{f} - i_{f'}\| \leq 2L/3$ for all branches of the tree. The idea is to replace every propagator in the sum with its infinite volume limit, $g_{\infty}$, and controlling the error using the Poisson summation formula (\ref{eq:poisson2pt}). By (\ref{eq:poisson2pt}), we have:
\begin{equation}
g((t,i), (t',i')) = g_{\infty}((t,i), (t',i')) + r((t,i), (t',i'))
\end{equation}
where, if $\|i - i'\| \leq 2L/3$,
\begin{equation}\label{eq:rest}
\| r((t,i), (t',i')) \| \leq Ce^{-(c/6)L - (c/6)\| (t,i) - (t',i') \|_{\beta,L}}\;. 
\end{equation}
Next, let us define the interpolating propagator as, for $\lambda \in [0;1]$:
\begin{equation}
g_{\lambda}((t,i), (t',i')) = \lambda g((t,i), (t',i')) + (1-\lambda) g_{\infty}((t,i), (t',i'))\;.
\end{equation}
Let:
\begin{equation}
f_{T}(\lambda) := \sum^{**}_{\substack{i_{1},\ldots, i_{n} \\ i_{1} \in \mathcal{C}(L/2,L/2)}} \int_{[0,\beta]^{n}} d\underline{t}\,\Big[ \prod_{\ell\in T} g_{\lambda;\ell}\Big] \int d\mu_{T}(\underline{s}) \det\big[ s_{b(f),b(f')}g_{\lambda;(f,f')}\big]\;;
\end{equation}
for $\lambda = 1$, this is our starting point, computed with a certain choice of holonomies. For $\lambda = 0$, the function $f_{T}(\lambda)$ does not depend on the boundary conditions. We are thus interested in quantifying the difference:
\begin{equation}\label{eq:intf}
f_{T}(1) - f_{T}(0) = \int_{0}^{1}d\lambda\, \frac{d}{d\lambda} f_{T}(\lambda)\;.
\end{equation}
Whenever the derivative hits the product of propagators on the spanning tree, we get:
\begin{equation}
\sum_{\ell \in T} \Big[\prod_{\ell'<\ell} g_{\lambda;\ell}\Big] r_{\ell} \Big[ \prod_{\ell'>\ell} g_{\lambda;\ell} \Big]\;.
\end{equation}
The sum over all space-time coordinates of this expression can be performed as before, using (\ref{eq:decper}), (\ref{eq:rest}). The only difference is the presence of the extra factor $e^{-(c/6)L}$ in the final estimate, coming from (\ref{eq:rest}). Suppose now that the derivative hits the determinant. To begin, observe preliminarily that, as for $g$ and $g_{\infty}$, the interpolating propagator $g_{\lambda}$ can also be represented as the sum of two terms admitting a Gram representation. To see this, we write:
\begin{equation}
g_{\lambda}((t,i), (t',i')) = \theta(t>t') A^{+}_{\lambda}((t,i), (t',i')) - \theta(t\leq t') A^{-}_{\lambda}((t,i), (t',i'))
\end{equation}
where:
\begin{equation}\label{eq:intA}
A^{\pm}_{\lambda}((t,i), (t',i')) = \lambda A^{\pm}((t,i), (t',i')) + (1-\lambda) A_{\infty}^{\pm}((t,i), (t',i'))\;.
\end{equation}
Following \cite[Lemma 4.1]{PS} or \cite[Lemma 10]{DRS} (for a more general setting), we have:
\begin{equation}
A^{\pm}((t,i), (t',i')) = \big( u^{\pm}_{(t,i)}, w^{\pm}_{(t',i')} \big)_{\mathfrak{h}}\;,\qquad A_{\infty}^{\pm}((t,i), (t',i')) = \big( u^{\pm}_{\infty,(t,i)}, w^{\pm}_{\infty,(t',i')} \big)_{\mathfrak{h}_{\infty}}
\end{equation}
for vectors of unit norm in suitable Hilbert spaces $\mathfrak{h}$ and $\mathfrak{h}_{\infty}$, whose definitions will not be important in what follows. Then, the interpolating functions (\ref{eq:intA}) inherit the Gram representation:
\begin{equation}
A_{\lambda}^{\pm}((t,i), (t',i')) = \Big( \sqrt{\lambda} u^{\pm}_{(t,i)} \oplus \sqrt{1-\lambda} u^{\pm}_{\infty,(t,i)}\, , \sqrt{\lambda} w^{\pm}_{(t',i')} \oplus \sqrt{1-\lambda} w^{\pm}_{\infty,(t',i')}\Big)_{\mathfrak{h}\oplus \mathfrak{h}_{\infty}}\;.
\end{equation}
From now on, one can proceed as in \cite[Theorem 2.4]{PS} to show that:
\begin{equation}\label{eq:GHint}
\Big|\det\big[ s_{b(f),b(f')}g_{\lambda;(f,f')}\big]\Big|\leq C^{n}\;.
\end{equation}
Next, we have to estimate the derivative of the determinant.  Let us denote by $G_{\lambda}(\underline{s})$ the argument of the determinant. By Jacobi's formula:
\begin{equation}
\frac{d}{d\lambda} \det G_{\lambda}(\underline{s}) = \tr \Big( \text{adj}(G_{\lambda}(\underline{s})) \frac{d}{d\lambda} G_{\lambda}(\underline{s})\Big)\;,
\end{equation}
where $\text{adj}(\cdot)$ denotes the adjugate matrix, that is the transpose of the matrix of the cofactors. We have:
\begin{equation}
\Big(\frac{d}{d\lambda} G_{\lambda}(\underline{s})\Big)_{f,f'} = s_{b(f), b(f')} r_{f,f'}\;,
\end{equation}
with $r_{f,f'}$ satisfying the bound (\ref{eq:rest}). Therefore:
\begin{equation}
\Big|\Big(\frac{d}{d\lambda} G_{\lambda}(\underline{s})\Big)_{f,f'}\Big| \leq Ce^{-(c/6)L}\;.
\end{equation}
Concerning $\text{adj}(G_{\lambda}(\underline{s}))$, its matrix entries are, up to a sign, the determinants of the minors of the original matrix $G_{\lambda}(\underline{s})$ after deleting a column and a row. They can all be estimated as in (\ref{eq:GHint}). All in all, we have:
\begin{equation}
\Big|\frac{d}{d\lambda} \det G_{\lambda}(\underline{s})\Big|\leq n^{2}C^{n}e^{-(c/6)L}\;.
\end{equation}
We are now in the position to estimate (\ref{eq:intf}). From the above considerations we easily get, for suitable constants $C,c>0$:
\begin{equation}
\begin{split}
\Big|f_{T}(1) - f_{T}(0)\Big| &= \int_{0}^{1}d\lambda\, \Big|\frac{d}{d\lambda} f_{T}(\lambda)\Big| \\
&\leq n! \beta C^{n} e^{-cL}\;.
\end{split}
\end{equation}
Coming back to (\ref{eq:splitexp}), we proved that:
\begin{equation}\label{eq:infty}
\begin{split}
&\frac{1}{\beta}\sum_{i_{1},\ldots, i_{n}} \int_{[0,\beta]^{n}} d\underline{t}\, \langle {\bf T} \gamma_{t_{1}}(V_{i_{1}})\,;\, \gamma_{t_{2}}(V_{i_{2}})\,;\, \cdots\,;\, \gamma_{t_{n}}(V_{i_{n}}) \rangle^{0}_{\beta,L} \\
&\quad = (L^{2} / |\mathcal{C}|) \frac{1}{\beta}\sum^{*}_{\substack{i_{1},\ldots, i_{n} \\ i_{1} \in \mathcal{C}(L/2,L/2)}} \int_{[0,\beta]^{n}} d\underline{t}\, \langle {\bf T} \gamma_{t_{1}}(V_{i_{1}})\,;\, \gamma_{t_{2}}(V_{i_{2}})\,;\, \cdots\,;\, \gamma_{t_{n}}(V_{i_{n}}) \rangle^{0}_{\beta,\infty} + {\mathfrak e}_{\beta,L}(n)
\end{split}
\end{equation}
where the first term does not depend on the holonomies, while the error term is bounded as:
\begin{equation}\label{eq:errinfty}
|{\mathfrak e}_{\beta,L}(n)| \leq C^{n} n! L^{2} e^{-cL}\;.
\end{equation}
Thus, from (\ref{eq:duhamel2}):
\begin{equation}\label{eq:diffen}
\begin{split}
&\log \tr_{\mathcal{F}_{L}} e^{-\beta H({\bf -1}, e^{i\theta}, e^{i\phi})} - \log \tr_{\mathcal{F}_{L}} e^{-\beta H({\bf -1}, e^{i\theta'}, e^{i\phi'})} \\
&\quad = \log \tr_{\mathcal{F}_{L}} e^{-\beta H_{0}({\bf -1}, e^{i\theta}, e^{i\phi})} - \log \tr_{\mathcal{F}_{L}} e^{-\beta H_{0}({\bf -1}, e^{i\theta'}, e^{i\phi'})}\\
&\qquad + \sum_{n\geq 1} \frac{(-U)^{n}}{n!} \int_{[0,\beta]^{n}} dt_{1}\cdots dt_{n} \Big( \langle {\bf T} \gamma_{t_{1}}(V)\,;\, \gamma_{t_{2}}(V)\,;\, \cdots\,;\, \gamma_{t_{n}}(V) \rangle^{0;\theta,\phi}_{\beta,L}\\&\qquad\qquad - \langle {\bf T} \gamma_{t_{1}}(V)\,;\, \gamma_{t_{2}}(V)\,;\, \cdots\,;\, \gamma_{t_{n}}(V) \rangle^{0;\theta',\phi'}_{\beta,L}\Big)
\end{split}
\end{equation}
where $\langle \cdot \rangle_{\beta,L}^{0;\theta,\phi}$ is the Gibbs state of the quadratic Hamiltonian $H_{0}({\bf -1}, \theta, \phi)$. By (\ref{eq:infty}), (\ref{eq:errinfty}), the sum in (\ref{eq:diffen}) can be estimated as:
\begin{equation}\label{eq:diff}
\begin{split}&
\Big|\sum_{n\geq 1} \frac{(-U)^{n}}{n!} \int_{[0,\beta]^{n}} dt_{1}\cdots dt_{n} \Big( \langle {\bf T} \gamma_{t_{1}}(V)\,;\, \gamma_{t_{2}}(V)\,;\, \cdots\,;\, \gamma_{t_{n}}(V) \rangle^{0;\theta,\phi}_{\beta,L}\\&\qquad\qquad - \langle {\bf T} \gamma_{t_{1}}(V)\,;\, \gamma_{t_{2}}(V)\,;\, \cdots\,;\, \gamma_{t_{n}}(V) \rangle^{0;\theta',\phi'}_{\beta,L}\Big)\Big| \leq CU\beta L^{2}e^{-cL}\;,
\end{split}
\end{equation}
where we used that the main term in \eqref{eq:infty} cancels in the difference. Furthermore, the difference of the non-interacting free energies has been studied, in the limit $\beta \to \infty$, in \hyperref[sec:expdeg]{Subsection \ref*{sec:expdeg}}:
\begin{equation}\label{eq:nnint}
\Big|\lim_{\beta \to \infty} \frac{1}{\beta} \Big(\log \tr_{\mathcal{F}_{L}} e^{-\beta H_{0}({\bf -1}, \theta,\phi)} - \log \tr_{\mathcal{F}_{L}} e^{-\beta H_{0}({\bf -1}, \theta',\phi')}\Big)\Big| \leq CL^{2} e^{-cL}\;.
\end{equation}
In conclusion, the exponential closeness of the many-body ground state energies, Eq. (\ref{eq:expdeg}), follows from (\ref{eq:diff}) (in the $\beta \to \infty$ limit) and from (\ref{eq:nnint}). This concludes the proof of part of Theorem \ref{thm1} for weakly interacting fermions.
\subsection{Proof of \autoref{thm1} - \autoref{thm1.3}}
The proof of Eqs. (\ref{eq:chess})-(\ref{bound2}) relies on reflection positivity, following the original insight of Lieb \cite{Lieb}. The lower bound for the energetic cost of the monopoles' excitations follows from the chessboard estimate, adapting \cite[Proposition 3.10]{GP}.
\subsubsection{Reflection positivity}\label{sec:RP}
Let us denote by $P$ a hyperplane cutting perpendicularly the torus $\Gamma_{L}$ in two halves as in \autoref{fig:cut}.
\begin{figure}
\begin{tikzpicture}[scale=1.6]
\begin{axis}[axis equal image,
        axis line style={draw=none},
        tick style={draw=none},
        xmax=18,ymax=22,zmax=7,xmin=-15,ymin=-20,zmin=-7,
        ticks=none,
        clip bounding box=upper bound,
        colormap/blackwhite,
        samples=23,
        view={45}{45}]
        \begin{pgfonlayer}{background layer}
            \addplot3[domain=0:360,
                    y domain=90:270,
                    surf,
                    z buffer=sort,
                    color = white!30,
                    mesh/ordering=y varies,
                    shader=flat, 
                    draw=black,
                    line width=0.1pt
                ]
                ({(12 + 3 * cos(x)) * cos(y)} ,
                {(12 + 3 * cos(x)) * sin(y)},
                {3 * sin(x)});
        \end{pgfonlayer} 
        \begin{pgfonlayer}{main}
            \draw[fill=gray, fill opacity=0.4, draw=black, line width=0.1pt] (0,18,-7) -- (0,18,7) -- (0,-20,7)  -- (0,-20,-7) -- (0,18,-7);
        \end{pgfonlayer}

        \begin{pgfonlayer}{foreground layer}
            \addplot3[domain=0:360,y 
                domain=-90:90,
                surf,
                z buffer=sort,
                color=gray, 
                mesh/ordering=y varies,
                shader=flat, 
                draw=black,
                line width=0.1pt]
                ({(12 + 3 * cos(x)) * cos(y)} ,
                {(12 + 3 * cos(x)) * sin(y)},
                {3 * sin(x)});
            \addplot3[line width = 0.1pt,domain=0:360,surf, z buffer=sort, black,samples=71, samples y=1]
                (0,
                {(12 + 3 * cos(x)) * sin(-90)},
                {3 * sin(x)});
            \addplot3[line width=0.1pt,domain=0:360,surf, z buffer=sort, black,samples=71, samples y=1]
                (0,
                {(12 + 3 * cos(x)) * sin(90)},
                {3 * sin(x)});

            \node[] at (21,0,3) {$\Gamma_L^\text{r}$};
            \node[] at (-12,0,10) {$\Gamma_L^\text{l}$};
        \end{pgfonlayer}
    \end{axis}
\begin{scope}[xshift =6cm,yshift=1cm, scale =0.4]
\draw[black, ->>>>-] (1,1) -- (2,1);
\draw[black, ->>>>-] (0,0) -- (0,1);
\draw[black, ->>>>-] (0,0) -- (1,0);
\draw[black, ->>>>-] (1,0) -- (1,1);
\draw[black, ->>>>-] (1,0) -- (2,0);
\draw[black, ->>>>-] (2,0) -- (2,1);
\draw[black, ->>>>-] (2,0) -- (2,1);
\draw[black, ->>>>-] (3,0) -- (4,0);
\draw[black, ->>>>-] (3,0) -- (3,1);
\draw[black, ->>>>-] (4,0) -- (5,0);
\draw[black, ->>>>-] (2,0) -- (3,0); 
\draw[black, ->>>>-] (4,0) -- (4,1);
\draw[black, ->>>>-] (5,0) -- (6,0);
\draw[black, ->>>>-] (5,0) -- (5,1);
\draw[black, ->>>>-] (6,0) -- (7,0);
\draw[black, ->>>>-] (6,0) -- (6,1);
\draw[black, ->>>>-] (7,0) -- (7,1);

\draw[black, ->>>>-] (0,1) -- (1,1);
\draw[black, ->>>>-] (0,1) -- (0,2);
\draw[black, ->>>>-] (1,1) -- (1,2);
\draw[black, ->>>>-] (2,1) -- (2,2);

\draw[black, ->>>>-] (2,1) -- (3,1);
\draw[black, ->>>>-] (3,1) -- (3,2);
\draw[black, ->>>>-] (3,1) -- (4,1);
\draw[black, ->>>>-] (4,1) -- (4,2);

\draw[black, ->>>>-] (4,1) -- (5,1);
\draw[black, ->>>>-] (5,1) -- (5,2);
\draw[black, ->>>>-] (5,1) -- (6,1);
\draw[black, ->>>>-] (6,1) -- (6,2);

\draw[black, ->>>>-] (6,1) -- (7,1);
\draw[black, ->>>>-] (7,1) -- (7,2);
\draw[black, ->>>>-] (1,2) -- (2,2);
\draw[black, ->>>>-] (1,2) -- (1,3);

\draw[black, ->>>>-] (0,2) -- (1,2);
\draw[black, ->>>>-] (0,2) -- (0,3);

\draw[black, ->>>>-] (2,2) -- (3,2);
\draw[black, ->>>>-] (2,2) -- (2,3);
\draw[black, ->>>>-] (3,2) -- (4,2);
\draw[black, ->>>>-] (3,2) -- (3,3);
\draw[black, ->>>>-] (4,2) -- (5,2);
\draw[black, ->>>>-] (4,2) -- (4,3);
\draw[black, ->>>>-] (5,2) -- (6,2);
\draw[black, ->>>>-] (5,2) -- (5,3);
\draw[black, ->>>>-] (5,2) -- (6,2);
\draw[black, ->>>>-] (5,2) -- (5,3);
\draw[black, ->>>>-] (6,2) -- (7,2);
\draw[black, ->>>>-] (6,2) -- (6,3);
\draw[black, ->>>>-] (4,2) -- (5,2);
\draw[black, ->>>>-] (5,3) -- (6,3);
\draw[black, ->>>>-] (7,2) -- (7,3);
\draw[black, ->>>>-] (0,3) -- (1,3);
\draw[black, ->>>>-] (0,3) -- (0,4);
\draw[black, ->>>>-] (1,3) -- (2,3);
\draw[black, ->>>>-] (2,3) -- (3,3);
\draw[black, ->>>>-] (3,3) -- (4,3);
\draw[black, ->>>>-] (4,3) -- (5,3);
\draw[black, ->>>>-] (0,4) -- (0,5);
\draw[black, ->>>>-] (5,3) -- (6,3);
\draw[black, ->>>>-] (6,3) -- (7,3);
\draw[black, ->>>>-] (0,4) -- (0,5);
\draw[black, ->>>>-] (0,4) -- (1,4);
\draw[black, ->>>>-] (1,4) -- (2,4);
\draw[black, ->>>>-] (2,4) -- (3,4);
\draw[black, ->>>>-] (3,4) -- (4,4);
\draw[black, ->>>>-] (4,4) -- (5,4);
\draw[black, ->>>>-] (0,5) -- (0,6);
\draw[black, ->>>>-] (5,4) -- (6,4);
\draw[black, ->>>>-] (6,4) -- (7,4);
\draw[black, ->>>>-] (0,5) -- (1,5);
\draw[black, ->>>>-] (1,5) -- (2,5);
\draw[black, ->>>>-] (2,5) -- (3,5);
\draw[black, ->>>>-] (3,5) -- (4,5);
\draw[black, ->>>>-] (4,5) -- (5,5);
\draw[black, ->>>>-] (0,6) -- (0,7);
\draw[black, ->>>>-] (5,5) -- (6,5);
\draw[black, ->>>>-] (6,5) -- (7,5);
\draw[black, ->>>>-] (0,6) -- (1,6);
\draw[black, ->>>>-] (1,6) -- (2,6);
\draw[black, ->>>>-] (2,6) -- (3,6);
\draw[black, ->>>>-] (3,6) -- (4,6);
\draw[black, ->>>>-] (4,6) -- (5,6);
\draw[black, ->>>>-] (7,6) -- (7,7);
\draw[black, ->>>>-] (5,6) -- (6,6);
\draw[black, ->>>>-] (6,6) -- (7,6);
\draw[black, ->>>>-] (0,7) -- (1,7);
\draw[black, ->>>>-] (1,7) -- (2,7);
\draw[black, ->>>>-] (2,7) -- (3,7);
\draw[black, ->>>>-] (3,7) -- (4,7);
\draw[black, ->>>>-] (4,7) -- (5,7);
\draw[black, ->>>>-] (7,5) -- (7,6);
\draw[black, ->>>>-] (5,7) -- (6,7);
\draw[black, ->>>>-] (6,7) -- (7,7);
\draw[black, ->>>>-] (7,3) -- (7,4);
\draw[black, ->>>>-] (7,4) -- (7,5);
\draw[black, ->>>>-] (1,4) -- (1,5);
\draw[black, ->>>>-] (1,3) -- (1,4);
\draw[black, ->>>>-] (1,5) -- (1,6);
\draw[black, ->>>>-] (1,6) -- (1,7);
\draw[black, ->>>>-] (2,4) -- (2,5);
\draw[black, ->>>>-] (2,3) -- (2,4);
\draw[black, ->>>>-] (2,5) -- (2,6);
\draw[black, ->>>>-] (2,6) -- (2,7);
\draw[black, ->>>>-] (3,4) -- (3,5);
\draw[black, ->>>>-] (3,3) -- (3,4);
\draw[black, ->>>>-] (3,5) -- (3,6);
\draw[black, ->>>>-] (3,6) -- (3,7);
\draw[black, ->>>>-] (4,4) -- (4,5);
\draw[black, ->>>>-] (4,3) -- (4,4);
\draw[black, ->>>>-] (4,5) -- (4,6);
\draw[black, ->>>>-] (4,6) -- (4,7);
\draw[black, ->>>>-] (5,4) -- (5,5);
\draw[black, ->>>>-] (5,3) -- (5,4);
\draw[black, ->>>>-] (5,5) -- (5,6);
\draw[black, ->>>>-] (5,6) -- (5,7);
\draw[black, ->>>>-] (6,4) -- (6,5);
\draw[black, ->>>>-] (6,3) -- (6,4);
\draw[black, ->>>>-] (6,5) -- (6,6);
\draw[black, ->>>>-] (6,6) -- (6,7);
\draw[black, dotted, ->-] (-0.5,-0.5) -- (7.5,-0.5);
\draw[black, dotted, ->-] (-0.5,7.5) -- (7.5,7.5);
\draw[black, dotted, ->>-] (-0.5,-0.5) -- (-0.5,7.5);
\draw[black, dotted, ->>-] (7.5,-0.5) -- (7.5,7.5);
\draw[black, ->>>>-] (0, -0.5) --(0,0);
\draw[black, ->>>>-] (-0.5,0) -- (0,0);
\draw[black, ->>>>-] (1, -0.5) --(1,0);
\draw[black, ->>>>-] (-0.5,1) -- (0,1);
\draw[black, ->>>>-] (2, -0.5) --(2,0);
\draw[black, ->>>>-] (-0.5,2) -- (0,2);
\draw[black, ->>>>-] (3, -0.5) --(3,0);
\draw[black, ->>>>-] (-0.5,3) -- (0,3);
\draw[black, ->>>>-] (4, -0.5) --(4,0);
\draw[black, ->>>>-] (-0.5,4) -- (0,4);
\draw[black, ->>>>-] (5, -0.5) --(5,0);
\draw[black, ->>>>-] (-0.5,6) -- (0,6);
\draw[black, ->>>>-] (6,-0.5) -- (6,0);
\draw[black, ->>>>-] (-0.5,5) -- (0,5);
\draw[black, ->>>>-] (7, -0.5) --(7,0);
\draw[black, ->>>>-] (-0.5,7) -- (0,7);
\draw[fill = black,  fill opacity =0.6] (1.5,-0.5) -- (5.5,-0.5) -- (5.5,7.5) -- (1.5,7.5) -- (1.5, -0.5);
\draw[black, ->>>-] (0,7) -- (0,7.5);
\draw[black, ->>>-] (7,7) -- (7.5,7);
\draw[black, ->>>-] (1,7) -- (1,7.5);
\draw[black, ->>>-] (7,6) -- (7.5,6);
\draw[black, ->>>-] (2,7) -- (2,7.5);
\draw[black, ->>>-] (7,5) -- (7.5,5);
\draw[black, ->>>-] (3,7) -- (3,7.5);
\draw[black, ->>>-] (7,4) -- (7.5,4);
\draw[black, ->>>-] (4,7) -- (4,7.5);
\draw[black, ->>>-] (7,0) -- (7.5,0);
\draw[black, ->>>-] (5,7) -- (5,7.5);
\draw[black, ->>>-] (7,3) -- (7.5,3);
\draw[black, ->>>-] (6,7) -- (6,7.5);
\draw[black, ->>>-] (7,2) -- (7.5,2);
\draw[black, ->>>-] (7,7) -- (7,7.5);
\draw[black, ->>>-] (7,1) -- (7.5,1);

\draw[black,fill=white] (0,0) circle (.1 cm);
\draw[black,fill=white] (1,1) circle (.1 cm);
\draw[black,fill=white] (3,3) circle (.1 cm);
\draw[black,fill=white] (2,2) circle (.1 cm);
\draw[black,fill=white] (4,4) circle (.1 cm);
\draw[black,fill=white] (5,5) circle (.1 cm);
\draw[black,fill=white] (6,6) circle (.1 cm);
\draw[black,fill=white] (7,7) circle (.1 cm);
\draw[black,fill=white] (3,5) circle (.1 cm);
\draw[black,fill=white](2,0) circle (.1 cm);
\draw[black,fill=white] (0,2) circle (.1 cm);
\draw[black,fill=white] (4,0) circle (.1 cm);
\draw[black,fill=white] (0,4) circle (.1 cm);
\draw[black,fill=white] (6,0) circle (.1 cm);
\draw[black,fill=white] (0,6) circle (.1 cm);
\draw[black,fill=white] (1,3) circle (.1 cm);
\draw[black,fill=white] (1,5) circle (.1 cm);
\draw[black,fill=white] (1,7) circle (.1 cm);
\draw[black,fill=white] (7,1) circle (.1 cm);
\draw[black,fill=white] (5,1) circle (.1 cm);
\draw[black,fill=white] (3,1) circle (.1 cm);

\draw[black,fill=white] (4,2) circle (.1 cm);
\draw[black,fill=white] (6,2) circle (.1 cm);
\draw[black,fill=white] (2,6) circle (.1 cm);
\draw[black,fill=white] (2,4) circle (.1 cm);
\draw[black,fill=white] (6,4) circle (.1 cm);
\draw[black,fill=white] (5,3) circle (.1 cm);
\draw[black,fill=white] (7,3) circle (.1 cm);
\draw[black,fill=white] (3,5) circle (.1 cm);
\draw[black,fill=white] (3,7) circle (.1 cm);
\draw[black,fill=white] (7,5) circle (.1 cm);
\draw[black,fill=white] (4,6) circle (.1 cm);
\draw[black,fill=white] (5,7) circle (.1 cm);

\draw[black,fill=black] (1,0) circle (.1 cm);
\draw[black,fill=black] (0,1) circle (.1 cm);
\draw[black,fill=black] (3,0) circle (.1 cm);
\draw[black,fill=black] (0,3) circle (.1 cm);
\draw[black,fill=blue] (5,0) circle (.1 cm);
\draw[blue,fill=black] (0,5) circle (.1 cm);
\draw[black,fill=black] (7,0) circle (.1 cm);
\draw[black,fill=black] (0,7) circle (.1 cm);
\draw[black,fill=black] (1,0) circle (.1 cm);
\draw[black,fill=black] (0,1) circle (.1 cm);
\draw[black,fill=black] (3,0) circle (.1 cm);
\draw[black,fill=black] (0,3) circle (.1 cm);
\draw[black,fill=black] (5,0) circle (.1 cm);
\draw[black,fill=black] (0,5) circle (.1 cm);
\draw[black,fill=black] (7,0) circle (.1 cm);
\draw[black,fill=black] (0,7) circle (.1 cm);

\draw[black,fill=black] (1,2) circle (.1 cm);
\draw[black,fill=black] (2,1) circle (.1 cm);
\draw[black,fill=black] (3,2) circle (.1 cm);
\draw[black,fill=black] (2,3) circle (.1 cm);
\draw[black,fill=black] (5,2) circle (.1 cm);
\draw[black,fill=black] (2,5) circle (.1 cm);
\draw[black,fill=black] (7,2) circle (.1 cm);
\draw[black,fill=black] (2,7) circle (.1 cm);
\draw[black,fill=black] (1,4) circle (.1 cm);
\draw[black,fill=black] (4,1) circle (.1 cm);
\draw[black,fill=black] (3,4) circle (.1 cm);
\draw[black,fill=black] (4,3) circle (.1 cm);
\draw[black,fill=black] (5,4) circle (.1 cm);
\draw[black,fill=black] (4,5) circle (.1 cm);
\draw[black,fill=black] (7,4) circle (.1 cm);
\draw[black,fill=black] (4,7) circle (.1 cm);

\draw[black,fill=black] (1,6) circle (.1 cm);
\draw[black,fill=black] (6,1) circle (.1 cm);
\draw[black,fill=black] (3,6) circle (.1 cm);
\draw[black,fill=black] (6,3) circle (.1 cm);
\draw[black,fill=black] (5,6) circle (.1 cm);
\draw[black,fill=black] (6,5) circle (.1 cm);
\draw[black,fill=black] (7,6) circle (.1 cm);
\draw[black,fill=black] (6,7) circle (.1 cm);

\node[below] (a) at (3.5,-0.7) {$\Gamma_L^{\text{r}}$};
\node[below] (b) at (0.5,-0.5) {$\Gamma_L^\text{l}$};
\end{scope}
\end{tikzpicture}
 \caption{Graphical representation of the cut torus.}
    \label{fig:cut}
\end{figure}
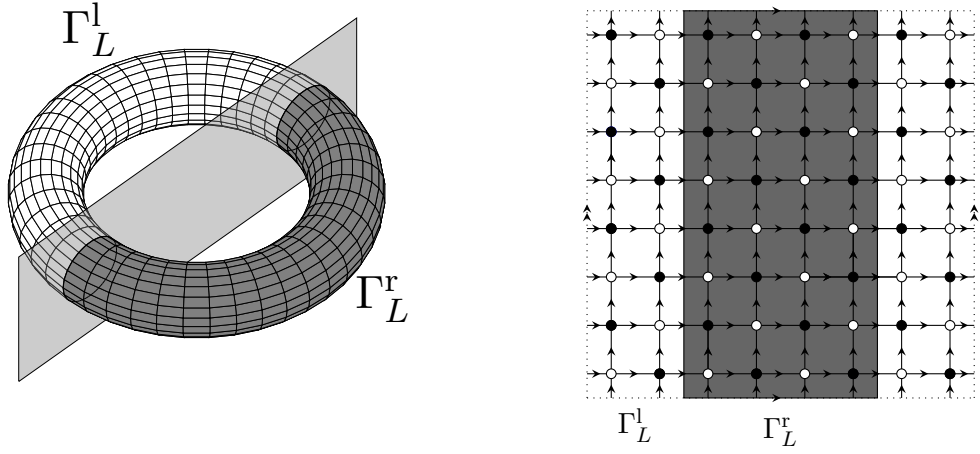
Let us denote by $\Gamma^{\text{l}}_{L}$ and $\Gamma_{L}^{\text{r}}$ the left and right portion of the lattice $\Gamma_{L}$, with respect to the cut introduced by the hyperplane. Let $\theta(i)$ be the geometrical reflection of $i$ across the hyperplane $P$.

We define the operator that implements the reflection across $P$. First of all, let $\mathcal{R}$ be the unitary operator implementing the geometric reflection on the fermionic algebra,
\begin{equation}
\mathcal{R}^{*} a^{\pm}_{i,\eta} \mathcal{R} = a^{\pm}_{\theta(i ),\eta},
\end{equation}
and let $\tau$ be the unitary operator implementing the particle-hole transformation,
\begin{equation}
\tau^{*} a^{\pm}_{i,\eta} \tau = a^{\mp}_{i,\eta}\;.
\end{equation}
\begin{definition}[Reflection operator]\label{refl}
The reflection operator $\Theta$ is the antilinear, unitary operator acting on the fermionic algebra as:
\begin{equation}
\Theta(\mathcal{O}) = \overline{\tau^{*} \mathcal{R}^{*} \mathcal{O} \mathcal{R} \tau}
\end{equation}
where the complex conjugation acts on the coefficients of the fermionic monomials.
\end{definition}
\begin{remark} This is the notion of reflection operator of \emph{\cite{Lieb}}. In the case of $\mathbb{Z}_{2}$ gauge theory, the complex conjugation could actually be dropped (all terms in the Hamiltonian are real in the sense that the first quantized hamiltonian is real).
\end{remark}
Let now
\begin{equation}
    M_{\text{l/r}} = m \sum_{i \in \Gamma_L^{\text{l/r}}} (-1)^{i_1+i_2} n_{i}
\end{equation}
be the left and right mass terms; we recall that $n_{i} = \sum_{\eta= \uparrow, \downarrow}  n_{i, \eta}$. The left and right mass terms are connected by the reflection operator, as the next lemma shows.
\begin{lemma}
    $\Theta(M_{\emph{\text{l/r}}}) = M_{\emph{\text{r/l}}}$.
\end{lemma}
\begin{proof}
    Since $\Theta^2 = \mathrm{id}$, it is sufficient to prove that $\Theta(M_{\text{l}}) = M_{\text{r}}$. By construction, the reflection exchanges the two sublattice (white/black vertices) of the bipartition. Furthermore,
    \begin{equation}
        \Theta(n_{i,\eta}) =  a^-_{\theta(i),\eta} a^+_{\theta(i),\eta} = 1 - n_{\theta(i),\eta}.
    \end{equation}
Thus:
\begin{equation}
\begin{aligned}
    \Theta(M_{\text{l}}) &= m \sum_{i \in \Gamma_L^{\text{l}}} \sum_{\eta = \uparrow, \downarrow} (-1)^{i_1+i_2} (1-n_{\theta(i), \eta})\\
    &=2m \sum_{i \in \Gamma_L^{\text{l}}} (-1)^{i_1+i_2} + m \sum_{i \in \Gamma_L^{\text{l}}} \sum_{\eta=\uparrow, \downarrow} (-1)^{i_1+i_2+1} n_{\theta(i), \eta}.
    \end{aligned}
\end{equation}
The first term vanishes because $\Gamma_L^{\text{l}}$ has an even number of vertices equally partitioned between the two sublattices. The second term equals $M_{\text{r}}$ since $(-1)^{\theta(i)_1} = -(-1)^{i_1}$ while $(-1)^{\theta(i)_2} =(-1)^{i_2}$ and $\theta(\Gamma_L^{\text{l}}) = \Gamma_L^{\text{r}}$.
\end{proof}
From here on, we proceed as in \cite[Section 3]{GP}.  We rewrite the Hamiltonian as:
\begin{equation}\label{eq:matterLR}
H(\bm{\sigma}) = H^{\text{r}}(\bm{\sigma}) + H^{\text{l}}(\bm{\sigma}) + V_{\text{int}}(\bm{\sigma})\;,
\end{equation}
where $H^{\text{r},\text{l}}$ collects fermionic monomials that are parametrized by space coordinates fully contained in $\Gamma_{L}^{\text{r},\text{l}}$, while $V_{\text{int}}$ takes into account the hopping terms that connect the two halves of the cut torus. In particular, the mass terms and the Hubbard interactions are all in $H^{\text{r},\text{l}}$. Without loss of generality, we can assume that $\sigma^z_{ij} = 1$ for all bonds intersecting the cut, namely
\begin{equation}\label{eq:Vint}
V_{\text{int}}(\bm{\sigma}) = -t \sum_{\substack{i
\in \Gamma^{\text{l}}_{L}:\\ i + e_{1} \in \Gamma_{L}^{\text{r}}}} (a^{+}_{i} a^{-}_{i + e_{1}} + a^{+}_{i+e_{1}} a^{-}_{i}).
\end{equation} 
If not, the Hamiltonian can be brought in this form after a gauge transformation.

The next result states a key estimate for the partition function. It is a direct adaptation of the argument of \cite{Lieb} to our setting, and we refer the reader to \cite[Section 3.1]{GP}, for the proof. In what follows, we denote by $Z_{\beta, L}(H^{\text{l}}, H^{\text{r}})$ the partition function corresponding to the Hamiltonian $H = H^{\text{l}} + H^{\text{r}} + V_{\text{int}}$, with $H^{\text{l}}, H^{\text{r}}$ in the left, resp. right fermionic algebras, and $V_{\text{int}}$ as in (\ref{eq:Vint}).
\begin{lemma}[Lieb]\label{lem:lieb} For any $\beta \geq 0$ and for $L = 2\ell$, the following inequality holds true:
\begin{equation}\label{eq:RP}
Z_{\beta,L}(H^{\textnormal{l}}, H^{\textnormal{r}})^{2} \leq Z_{\beta,L}(H^{\textnormal{l}}, \Theta(H^{\textnormal{l}})) Z_{\beta,L}(H^{\textnormal{r}}, \Theta(H^{\textnormal{r}}))\;. 
\end{equation}
\end{lemma}
Using the above lemma, one can infer with a standard argument \cite{Lieb} that the partition function of \eqref{hami} is maximized in a uniform $\pi$-flux background even in presence of a staggered mass term.

Reflection positivity can also be used to quantify the energetic excitations above the $\pi$-flux phase \cite{GP}. For holonomies $(1,1)$, $(1,-1)$, $(-1,1)$, $(-1,-1)$, Eqs. (\ref{eq:chess}), (\ref{eq:Delta}) are proven exactly as in \cite[Proposition 3.10]{GP}. It remains to compute $\Delta_{\beta,L}$ in the presence of the Hubbard interaction. This will be done in the next section. Later, we will comment about the case of more general holonomies.

\subsubsection{Lower bound on the monopole mass}\label{mmono}
As proven in \cite[Proposition 3.10]{GP}, the monopoles' mass can be bounded below in terms of the free energy of a suitable staggered flux configuration; this is achieved using the chessboard estimate, see \cite{Tasaki} for a review, and \cite[Lemma 3.9]{GP} for the application to our context. Let $\Phi^{*}$ denote the chessboard flux configuration depicted in \autoref{fig:chessboard}. Picking the gauge field ${\bm\sigma^*}$ with holonomies $(1,1)$ as in \autoref{fig:conf2}, we again chose a unit cell that makes the system translation invariant (observe that this is possible thanks to the fact that $L\in 4\mathbb{N}$). If the corresponding lattice is denoted $ \widetilde{\Gamma}^{\text{red}}_{L}$, the Hamiltonian reads
\begin{figure}
\centering
\begin{tikzpicture}[scale=0.6]
\draw (-0.5,-0.5) grid (7.5,7.5);
\draw[dotted, ->-] (-0.5,-0.5) -- (7.5,-0.5);
\draw[dotted, ->-] (-0.5,7.5) -- (7.5,7.5);
\draw[dotted, ->>-] (-0.5,-0.5) -- (-0.5,7.5);
\draw[dotted, ->>-] (7.5,-0.5) -- (7.5,7.5);
\draw[line width = 0.07 cm] (-0.5,0) -- (7.5,0);
\draw[line width = 0.07 cm] (-0.5,2) -- (7.5,2);
\draw[line width = 0.07 cm] (-0.5,4) -- (7.5,4);
\draw[line width = 0.07 cm] (-0.5,6) -- (7.5,6);

\draw[line width = 0.07 cm] (1,0) -- (1,-0.5);
\draw[line width = 0.07 cm] (2,0) -- (2,-0.5);
\draw[line width = 0.07 cm] (5,0) -- (5,-0.5);
\draw[line width = 0.07 cm] (6,0) -- (6,-0.5);

\draw[line width = 0.07 cm] (1,7.5) -- (1,7);
\draw[line width = 0.07 cm] (2,7.5) -- (2,7);
\draw[line width = 0.07 cm] (5,7.5) -- (5,7);
\draw[line width = 0.07 cm] (6,7.5) -- (6,7);
\draw[line width = 0.07 cm] (1,6) -- (1,5);
\draw[line width = 0.07 cm] (2,6) -- (2,5);
\draw[line width = 0.07 cm] (5,6) -- (5,5);
\draw[line width = 0.07 cm] (6,6) -- (6,5);
\draw[line width = 0.07 cm] (1,4) -- (1,3);
\draw[line width = 0.07 cm] (2,4) -- (2,3);
\draw[line width = 0.07 cm] (5,4) -- (5,3);
\draw[line width = 0.07 cm] (6,4) -- (6,3);
\draw[line width = 0.07 cm] (1,2) -- (1,1);
\draw[line width = 0.07 cm] (2,2) -- (2,1);
\draw[line width = 0.07 cm] (5,2) -- (5,1);
\draw[line width = 0.07 cm] (6,2) -- (6,1);
\draw[line width = 0.07 cm] (5,7.5) -- (5,7);
\draw[line width = 0.07 cm] (6,7.5) -- (6,7);

\draw[dotted, fill = gray, fill opacity = 0.1] (3.5,2.5)--(7.5,2.5) -- (7.5,4.5) -- (3.5,4.5) -- (3.5,2.5);
\draw[black,fill=white] (0,0) circle (.1 cm);
\draw[black,fill=white] (1,1) circle (.1 cm);
\draw[black,fill=white] (3,3) circle (.1 cm);
\draw[black,fill=white] (2,2) circle (.1 cm);
\draw[black,fill=white] (4,4) circle (.1 cm);
\draw[black,fill=white] (5,5) circle (.1 cm);
\draw[black,fill=white] (6,6) circle (.1 cm);
\draw[black,fill=white] (7,7) circle (.1 cm);
\draw[black,fill=white] (3,5) circle (.1 cm);
\draw[black,fill=white](2,0) circle (.1 cm);
\draw[black,fill=white] (0,2) circle (.1 cm);
\draw[black,fill=white] (4,0) circle (.1 cm);
\draw[black,fill=white] (0,4) circle (.1 cm);
\draw[black,fill=white] (6,0) circle (.1 cm);
\draw[black,fill=white] (0,6) circle (.1 cm);
\draw[black,fill=white] (1,3) circle (.1 cm);
\draw[black,fill=white] (1,5) circle (.1 cm);
\draw[black,fill=white] (1,7) circle (.1 cm);
\draw[black,fill=white] (7,1) circle (.1 cm);
\draw[black,fill=white] (5,1) circle (.1 cm);
\draw[black,fill=white] (3,1) circle (.1 cm);

\draw[black,fill=white] (4,2) circle (.1 cm);
\draw[black,fill=white] (6,2) circle (.1 cm);
\draw[black,fill=white] (2,6) circle (.1 cm);
\draw[black,fill=white] (2,4) circle (.1 cm);
\draw[black,fill=white] (6,4) circle (.1 cm);
\draw[black,fill=white] (5,3) circle (.1 cm);
\draw[black,fill=white] (7,3) circle (.1 cm);
\draw[black,fill=white] (3,5) circle (.1 cm);
\draw[black,fill=white] (3,7) circle (.1 cm);
\draw[black,fill=white] (7,5) circle (.1 cm);
\draw[black,fill=white] (4,6) circle (.1 cm);
\draw[black,fill=white] (5,7) circle (.1 cm);

\draw[black,fill=black] (1,0) circle (.1 cm);
\draw[black,fill=black] (0,1) circle (.1 cm);
\draw[black,fill=black] (3,0) circle (.1 cm);
\draw[black,fill=black] (0,3) circle (.1 cm);
\draw[black,fill=blue] (5,0) circle (.1 cm);
\draw[blue,fill=black] (0,5) circle (.1 cm);
\draw[black,fill=black] (7,0) circle (.1 cm);
\draw[black,fill=black] (0,7) circle (.1 cm);
\draw[black,fill=black] (1,0) circle (.1 cm);
\draw[black,fill=black] (0,1) circle (.1 cm);
\draw[black,fill=black] (3,0) circle (.1 cm);
\draw[black,fill=black] (0,3) circle (.1 cm);
\draw[black,fill=black] (5,0) circle (.1 cm);
\draw[black,fill=black] (0,5) circle (.1 cm);
\draw[black,fill=black] (7,0) circle (.1 cm);
\draw[black,fill=black] (0,7) circle (.1 cm);

\draw[black,fill=black] (1,2) circle (.1 cm);
\draw[black,fill=black] (2,1) circle (.1 cm);
\draw[black,fill=black] (3,2) circle (.1 cm);
\draw[black,fill=black] (2,3) circle (.1 cm);
\draw[black,fill=black] (5,2) circle (.1 cm);
\draw[black,fill=black] (2,5) circle (.1 cm);
\draw[black,fill=black] (7,2) circle (.1 cm);
\draw[black,fill=black] (2,7) circle (.1 cm);
\draw[black,fill=black] (1,4) circle (.1 cm);
\draw[black,fill=black] (4,1) circle (.1 cm);
\draw[black,fill=black] (3,4) circle (.1 cm);
\draw[black,fill=black] (4,3) circle (.1 cm);
\draw[black,fill=black] (5,4) circle (.1 cm);
\draw[black,fill=black] (4,5) circle (.1 cm);
\draw[black,fill=black] (7,4) circle (.1 cm);
\draw[black,fill=black] (4,7) circle (.1 cm);

\draw[black,fill=black] (1,6) circle (.1 cm);
\draw[black,fill=black] (6,1) circle (.1 cm);
\draw[black,fill=black] (3,6) circle (.1 cm);
\draw[black,fill=black] (6,3) circle (.1 cm);
\draw[black,fill=black] (5,6) circle (.1 cm);
\draw[black,fill=black] (6,5) circle (.1 cm);
\draw[black,fill=black] (7,6) circle (.1 cm);
\draw[black,fill=black] (6,7) circle (.1 cm);

\begin{scope}[xshift =10 cm, yshift=1.5cm]
\draw[dotted,fill = gray, fill opacity = 0.1] (0,0) -- (8,0) -- (8,4)--(0,4)--(0,0);
\draw[thick] (1,0) -- (1,4);
\draw[thick] (3,0) -- (3,4);
\draw[thick] (5,0) -- (5,4);
\draw[thick] (7,0) -- (7,4);
\draw[thick] (0,1) -- (8,1);
\draw[thick, line width= 0.07 cm] (3,1) -- (3,3);
\draw[thick, line width= 0.07 cm] (5,1) -- (5,3);
\draw[thick, line width = 0.07 cm] (0,3) -- (8,3);
\node[below left] (a) at (1,1) {{$a$}};
\node[above left] (b) at (1,3) {{$A$}};
\node[below left] (a) at (3,1) {{$b$}};
\node[above left] (b) at (3,3) {{$B$}};
\node[below right] (a) at (5,1) {{$c$}};
\node[above right] (b) at (5,3) {{$C$}};
\node[below right] (a) at (7,1) {{$d$}};
\node[above right] (b) at (7,3) {{$D$}};
\draw[line width = 0.04 cm, ->] (0,0)--(0,4);
\draw[line width = 0.04 cm, ->] (0,0)--(8,0);
\draw[black,fill=black] (1,1) circle (.1 cm);
\draw[black,fill=black] (3,3) circle (.1 cm);
\draw[black,fill=black] (5,1) circle (.1 cm);
\draw[black,fill=black] (7,3) circle (.1 cm);
\draw[black,fill=white] (3,1) circle (.1 cm);
\draw[black,fill=white] (1,3) circle (.1 cm);
\draw[black,fill=white] (5,3) circle (.1 cm);
\draw[black,fill=white] (7,1) circle (.1 cm);
\end{scope}
\end{tikzpicture}
\caption{Left: gauge field configuration ${\bm\sigma^*}$ associated with the chessboard flux arrangement $\Phi^{*}$. Solid bonds correspond to $\sigma^z=-1$, while light bonds correspond to $\sigma^z=+1$. Right: fundamental cell associated with a translation-invariant configuration.}
\label{fig:conf2}
\end{figure}
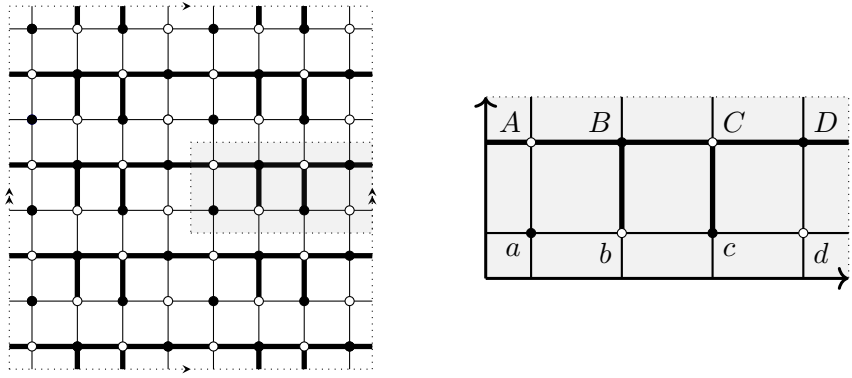
\begin{align}
H({\bm \sigma^{*}}; 1,1) &= -t \sum_{{\bf i} \in \widetilde{\Gamma}^{\text{red}}_{L} \times \{\uparrow,\downarrow\}}  \Big(a^+_{a, {\bf i}} a^-_{b, {\bf i}} + a^+_{a, {\bf i}} a^-_{A, {\bf i}} + a^+_{b, {\bf i}} a^-_{c, {\bf i}} - a^+_{b, {\bf i}} a^-_{B, {\bf i}} + a^+_{c, {\bf i}} a^-_{d, {\bf i}} - a^+_{c, {\bf i}} a^-_{C, {\bf i}} \\
&\phantom{=-t} 
+ a^+_{d, {\bf i}} a^-_{D, {\bf i}}  + a^+_{d, {\bf i}} a^-_{a, {\bf i} + 4 e_{1} }-a^+_{A, {\bf i}} a^-_{B, {\bf i}} + a^+_{A, {\bf i}} a^-_{a, {\bf i}+2 e_{2}} - a^+_{B, {\bf i}} a^-_{C, {\bf i}}  \\
&\phantom{=-t}
+ a^+_{B, {\bf i}} a^-_{b, {\bf i}+2 e_{2}}  - a^+_{C, {\bf i}} a^-_{D, {\bf i}} + a^+_{C, {\bf i}} a^-_{c, {\bf i}+2 e_{2}} - a^+_{D, {\bf i}} a^-_{A, {\bf i}+4 e_{1}} + a^+_{D, {\bf i}} a^-_{d, {\bf i}+2 e_{2}} + \text{h.c.} \Big)  \\
&\quad +m \sum_{{\bf i} \in \widetilde{\Gamma}^{\text{red}}_L} \big(n_{A,{\bf i}} - n_{B, {\bf i}} + n_{C, {\bf i}} - n_{D,{\bf i}}
- n_{a, {\bf i}} + n_{b, {\bf i}} - n_{c, {\bf i}} + n_{d, {\bf i}} \big) \\
&\quad + U \sum_{i \in \widetilde{\Gamma}^{\text{red}}_{L}} \sum_{\rho \in I} \Big(n_{\rho,i,\uparrow} - \frac{1}{2}\Big)\Big( n_{\rho,i,\uparrow} - \frac{1}{2} \Big) 
\end{align}
where $I$ collects the labels of the fundamental cell and we introduced the notation ${\bf i} = (i,\eta)$. The Brillouin zone is:
\begin{equation}
\widetilde{B}_{L}(1,1) := \Big\{ k\in \frac{2\pi}{L}(n_{1}, n_{2}) \mid 0\leq n_{1} \leq L/4-1,\; 0\leq n_{2} \leq L/2-1 \Big\}\;.
\end{equation}

For $U=0$, the quantity $\Delta_{\beta,L}$ can be computed by exact diagonalization. Let us consider the spinless case; the presence of the spin will eventually amount to a factor $2$ in the final expression. The Bloch Hamiltonian is:
\begin{equation}
\begin{split}
&h(k) = \\
&\quad  \scriptsize{\begin{pmatrix}
           -m &1 &0 & e^{-4ik_1} &1+e^{-2ik_2} &0 &0 &0\\
           1 &m &1 &0 &0 &-1+e^{-2ik_2} &0 &0\\
           0 &1 &-m &1 &0 &0 &-1+e^{-2ik_2} &0\\
           e^{4 i k_1} &0 &1 &
           m&0 &0 &0 &1+e^{-2ik_2} \\
           1+e^{2ik_2} &0 &0 &0 &m &-1 &0 &-e^{-4ik_1}\\
           0 &-1+e^{2ik_2} &0 &0 &-1 &-m &-1 &0\\
           0 &0 &-1 + e^{2ik_2} &0 &0 &-1 &m &-1\\
           0 &0 &0 &1+e^{2ik_2} &-e^{4ik_1} &0 &-1 &-m
       \end{pmatrix}} \;.
       \end{split}
\end{equation}
Its eigenvalues are doubly degenerate, and are given by:
\begin{equation}\label{eq:deltabands}
\begin{split}
       e_{1;\pm}(k) &= \pm 2t \sqrt{\bigg( \frac{m}{2t}\bigg)^2+1+\frac{1}{2} \sqrt{1 + \cos(2k_1)^2 +\cos(2k_2)^2}}\\
        e_{2;\pm}(k) &= \pm 2t \sqrt{\bigg( \frac{m}{2t}\bigg)^2+1-\frac{1}{2}\sqrt{1 + \cos(2k_1)^2 +\cos(2k_2)^2}}\;.
        \end{split}
\end{equation}
We are now ready to compute $\Delta_{\beta,L}$, as given by (\ref{eq:Delta}). Let us denote by $e^{\pi}_{\pm}(k)$ the eigenvalues of the Bloch Hamiltonian associated with the $\pi$-flux phase, as given by (\ref{eq:epi}). Proceeding as in the proof of \cite[Proposition 3.13]{GP}, we have:
\begin{equation}\label{boundint}
\begin{split}
&-\frac{1}{\beta L^{2}} \log \frac{Z_{\beta,L}(\Phi^{*}; a,b)}{Z_{\beta,L}({\bf -1}; a,b)} \\
&\qquad = \frac{t}{4\pi^2} \int_0^{2\pi} \int_0^{2 \pi} dk_{1} dk_{2}\, \sqrt{\bigg( \frac{m}{2t}\bigg)^2+1+\frac{1}{2} \cos(k_{1}) +\frac{1}{2}\cos(k_{2})} \\&\qquad \quad - \frac{t}{8\pi^2} \int_0^{2\pi} \int_0^{2 \pi} dk_{1} dk_{2}\, \sqrt{\bigg( \frac{m}{2t}\bigg)^2+1+\frac{1}{2} \sqrt{1+\cos^2(k_{1}) +\cos^2(k_{2})}}\\
        &\qquad\quad - \frac{t}{8\pi^2} \int_0^{2\pi} \int_0^{2 \pi} dk_{1} dk_{2}\, \sqrt{\bigg( \frac{m}{2t}\bigg)^2+1-\frac{1}{2} \sqrt{1+\cos^2(k_{1}) +\cos^2(k_{2})}} + o(1)\;.
\end{split}
\end{equation}
Numerical evaluation shows that the sum of these three dominant terms is positive (see \autoref{delta}).

\begin{figure}
    \centering
    \includegraphics[width=0.5\linewidth]{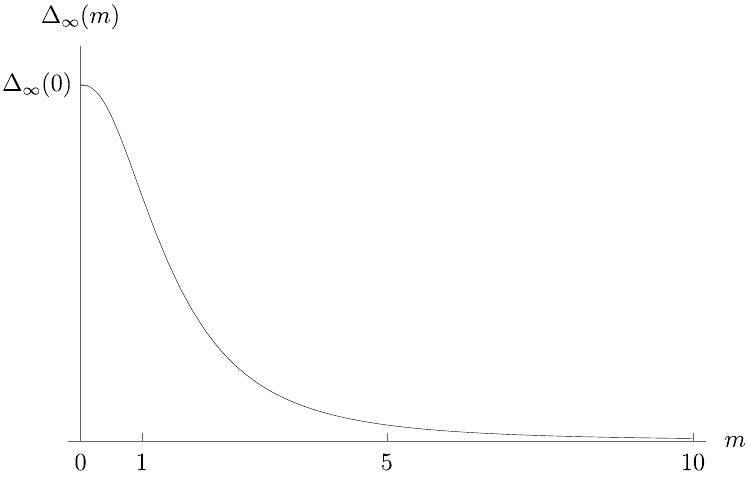}
    \caption{Numerical plot of  the monopole's mass $\Delta_{\infty}(m)$ as a function of $m$ ($t=1$, $U=0$).}
    \label{delta}
\end{figure}

Finally, let us discuss the effect of many-body interactions. By cluster expansion, one could actually prove that $\Delta_{\beta,L}$ is analytic in $U$ for $|U|$ small. This is actually not needed to prove the stability of the monopoles' gap. It is enough to use the general estimate:
\begin{equation}\label{eq:partfcnbd}
 e^{-\beta \|V\| |U| } Z^{0}_{\beta, L}(\Phi; a,b) \leq  Z_{\beta, L}(\Phi;a,b) \leq e^{\beta \|V\| |U| } Z^{0}_{\beta, L}(\Phi;a,b)
\end{equation}
where $Z^{0}_{\beta, L}$ is the partition function for $U=0$ and $UV$ is the Hubbard interaction. This bound follows from the operator inequalities $H_0({\bm \sigma}) - |U| \|V\| \leq H({\bm \sigma}) \leq H_{0}({\bm \sigma}) + |U| \|V\|$ and from the fact that, for any two self-adjoint operators $A, B$
\begin{equation}
\Tr e^{A} \leq \Tr e^{B}\qquad \text{if $A\leq B$,}
\end{equation}
a bound that can be proved using the min-max characterization of the eigenvalues; see {\it e.g.} \cite[Corollary A.8]{Tasaki}. The bound (\ref{eq:partfcnbd}), combined with the expression (\ref{eq:Delta}), immediately implies that:
\begin{equation}\label{intmono}
\Delta_{\beta,L} \geq \Delta_{\beta,L}^{0} - C|U|\;.
\end{equation}
This concludes the proof of \autoref{thm1.3} of \autoref{thm1}, for holonomies $a,b$ equal to $\pm 1$.

\begin{remark}
    The bound \eqref{intmono} is useful for small $U$ only. For $U \to \infty$, we expect that  $\Delta_{\beta, L}(U) \to 0$, as the dominant part of the hamiltonian is background independent (similarly to what happens in $m \to \infty$ limit, where the computation can be done exactly).
\end{remark}
\subsubsection{Extension to twisted Hamiltonians}\label{esgt}
To conclude the proof of Theorem \ref{thm1}, it remains to prove \autoref{thm1.3} for general holonomies. We shall adapt the reflection positivity argument in presence of a general holonomy around the non-contractible cycles, and $\mathbb{Z}_2$ fluxes in the plaquettes\footnote{Notice that such condition are compatible as they correspond to independent generators of $C_1(\Gamma_L)$}. Let $H(\bm{\sigma}, e^{i\phi},e^{i\theta})$ be the Hamiltonian with holonomies $e^{i \phi}$ on $\mathcal{C}_1$ and $e^{i \theta}$ on $\mathcal{C}_2$ added on a $\mathbb{Z}_{2}$ valued $\bm{\sigma}$ background; see \autoref{background}. Up to unitary equivalences, such a Hamiltonian can be explicitly constructed in the following way (recall that Hamiltonians with the same fluxes across the lattice plaquettes and the same holonomies are isospectral \cite{LL}). Given the background $\bm{\sigma}$, we can compute the holonomy of the background around $\mathcal{C}_1$ and $\mathcal{C}_2$ and it is either $\pm 1$. If the holonomy around a cycle $\mathcal{C}_1$ or $\mathcal{C}_2$ is $+1$, we multiply the hoppings on the edges crossed respectively by $\mathcal{C}_2^*$ and $\mathcal{C}_1^{*}$ (see \autoref{background}) by $e^{i \phi}$ and $e^{i \theta}$. Observe that this procedure does not change the fluxes through the plaquettes (which is computed taking into account the orientation of the edges). Similarly, if one of such holonomy is $-1$, we will twist by $-e^{i \theta}$ or $-e^{i \phi}$. We will apply chessboard estimates on such Hamiltonian following the same procedure as in \cite[Section 3.1]{GP}, tracking the fate of holonomies after reflections. The key observation is that after an horizontal reflection (across a plane which does not intersect $\Gamma_L$ in $\mathcal{C}_2^*$), the holonomy along $\mathcal{C}_1$ becomes $-1$ leaving the holonomy around $\mathcal{C}_2$ unchanged (due to the complex conjugation in the definition of the reflection operator see  \hyperref[refl]{Definition \ref*{refl}}). If we then perform a vertical reflection (across a plane which does not intersect $\Gamma_L$ in $\mathcal{C}_1^*$), the holonomy along $\mathcal{C}_1$ becomes $-1$, leaving the holonomy around $\mathcal{C}_2$ to $-1$. After this pair of reflections, we are left with $4$ possible partition functions with all holonomies valued in $\mathbb{Z}_2$ so that they can be represented with choice of a $\bm{\sigma}$ background only. From this point on, the proof follows as in \cite[Section 3.1]{GP}. This concludes the proof of \autoref{thm1}. \qed

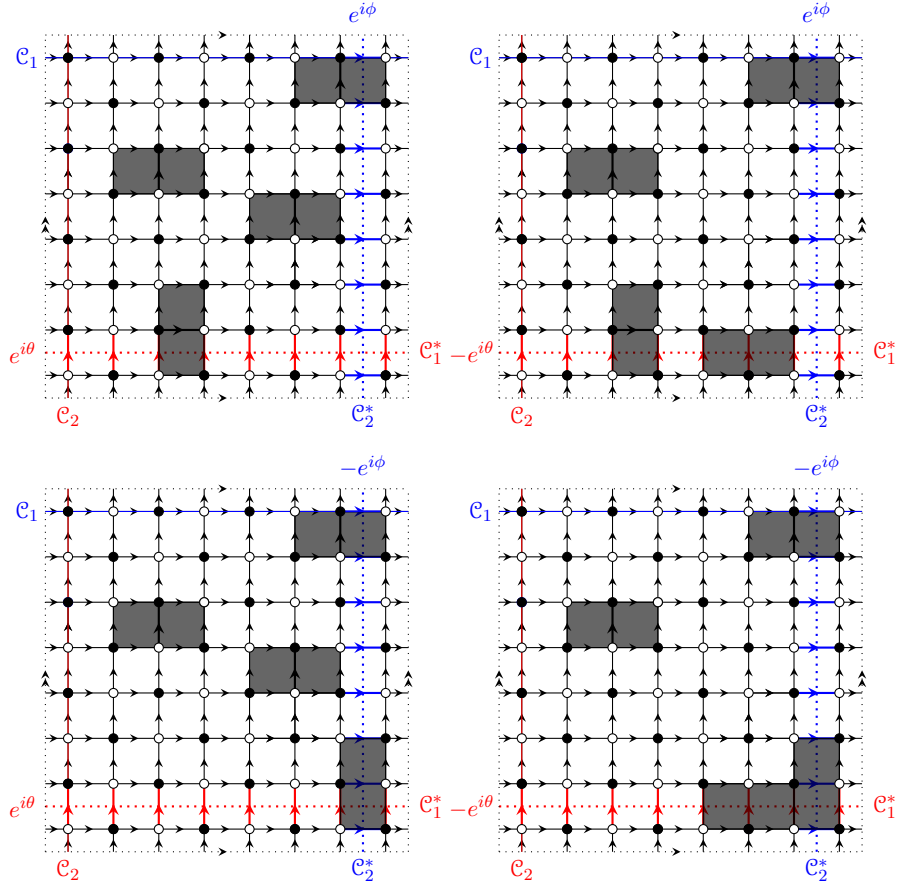
\begin{figure}
    \centering
   \begin{tikzpicture}[scale=0.6]

\draw[red, thick, dotted] (-0.5,0.5) -- (7.5,0.5);
\draw[blue, thick, dotted] (6.5,-0.5) -- (6.5,7.5);
\draw[red, thick,->>>>-] (0,0) -- (0,1);
\draw[red, thick,->>>>-] (1,0) -- (1,1);
\draw[red, thick,->>>>-] (2,0) -- (2,1);
\draw[red, thick,->>>>-] (3,0) -- (3,1);
\draw[red, thick,->>>>-] (4,0) -- (4,1);
\draw[red, thick, ->>>>-] (5,0) -- (5,1);
\draw[red, thick, ->>>>-] (6,0) -- (6,1);
\draw[red, thick, ->>>>-] (7,0) -- (7,1);
 \node[red] (a) at (-1, 0.5) 
           {\scalebox{0.8}{$e^{i \theta}$}};
\draw[blue, thick,->>>>-] (6,0) -- (7,0);
\draw[blue, thick,->>>>-] (6,1) -- (7,1);
\draw[blue, thick,->>>>-] (6,2) -- (7,2);
\draw[blue, thick,->>>>-] (6,3) -- (7,3);
\draw[blue, thick, ->>>>-] (6,4) -- (7,4);
\draw[blue, thick, ->>>>-] (6,5) -- (7,5);
\draw[blue, thick, ->>>>-] (6,6) -- (7,6);
\draw[blue, thick, ->>>>-] (6,7) -- (7,7);

\node[blue] (a) at (6.5,8) 
           {\scalebox{0.8}{$e^{i \phi}$}};

            \node[red] (a) at (0, -0.9) 
           {\scalebox{0.8}{$\mathcal{C}_2$}};
            \node[blue] (a) at (-0.9, 7) 
           {\scalebox{0.8}{$\mathcal{C}_1$}};
           \node[blue] (a) at (6.5, -0.9) 
           {\scalebox{0.8}{$\mathcal{C}_2^*$}};
            \node[red] (a) at (8, 0.5) 
           {\scalebox{0.8}{$\mathcal{C}_1^*$}};
   \draw[fill = black, fill opacity = 0.6] 
(2,0) -- (3, 0) -- (3, 2) -- (2, 2) -- (2,0);
\draw[fill = black, fill opacity = 0.6] 
(1,4) -- (3, 4) -- (3, 5) -- (1, 5) -- (1,4);
\draw[fill = black, fill opacity = 0.6] 
(4,3) -- (6, 3) -- (6, 4) -- (4, 4) -- (4,3);
\draw[fill = black, fill opacity = 0.6] 
(5,6) -- (7, 6) -- (7, 7) -- (5, 7) -- (5,6);
\draw[black, ->>>>-] (1,1) -- (2,1);
\draw[black, ->>>>-] (0,0) -- (1,0);
\draw[black, ->>>>-] (1,0) -- (2,0);
\draw[black, ->>>>-] (3,0) -- (4,0);
\draw[black, ->>>>-] (4,0) -- (5,0);
\draw[black, ->>>>-,] (2,0) -- (3,0); 
\draw[black, ->>>>-] (5,0) -- (6,0);

\draw[black, ->>>>-] (0,1) -- (1,1);
\draw[black, ->>>>-] (0,1) -- (0,2);
\draw[black, ->>>>-] (1,1) -- (1,2);
\draw[black, ->>>>-] (2,1) -- (2,2);

\draw[black, ->>>>-, thick] (2,1) -- (3,1);
\draw[black, ->>>>-] (3,1) -- (3,2);
\draw[black, ->>>>-] (3,1) -- (4,1);
\draw[black, ->>>>-] (4,1) -- (4,2);

\draw[black, ->>>>-] (4,1) -- (5,1);
\draw[black, ->>>>-] (5,1) -- (5,2);
\draw[black, ->>>>-] (5,1) -- (6,1);
\draw[black, ->>>>-] (6,1) -- (6,2);

\draw[black, ->>>>-] (7,1) -- (7,2);
\draw[black, ->>>>-] (1,2) -- (2,2);
\draw[black, ->>>>-] (1,2) -- (1,3);

\draw[black, ->>>>-] (0,2) -- (1,2);
\draw[black, ->>>>-] (0,2) -- (0,3);

\draw[black, ->>>>-] (2,2) -- (3,2);
\draw[black, ->>>>-] (2,2) -- (2,3);
\draw[black, ->>>>-] (3,2) -- (4,2);
\draw[black, ->>>>-] (3,2) -- (3,3);
\draw[black, ->>>>-] (4,2) -- (5,2);
\draw[black, ->>>>-] (4,2) -- (4,3);
\draw[black, ->>>>-] (5,2) -- (6,2);
\draw[black, ->>>>-] (5,2) -- (5,3);
\draw[black, ->>>>-] (5,2) -- (6,2);
\draw[black, ->>>>-] (5,2) -- (5,3);
\draw[black, ->>>>-] (6,2) -- (6,3);
\draw[black, ->>>>-] (4,2) -- (5,2);
\draw[black, ->>>>-] (5,3) -- (6,3);
\draw[black, ->>>>-] (7,2) -- (7,3);
\draw[black, ->>>>-] (0,3) -- (1,3);
\draw[black, ->>>>-] (0,3) -- (0,4);
\draw[black, ->>>>-] (1,3) -- (2,3);
\draw[black, ->>>>-] (2,3) -- (3,3);
\draw[black, ->>>>-] (3,3) -- (4,3);
\draw[black, ->>>>-] (4,3) -- (5,3);
\draw[black, ->>>>-] (0,4) -- (0,5);
\draw[black, ->>>>-] (5,3) -- (6,3);
\draw[black, ->>>>-] (0,4) -- (0,5);
\draw[black, ->>>>-] (0,4) -- (1,4);
\draw[black, ->>>>-] (1,4) -- (2,4);
\draw[black, ->>>>-] (2,4) -- (3,4);
\draw[black, ->>>>-] (3,4) -- (4,4);
\draw[black, ->>>>-] (4,4) -- (5,4);
\draw[black, ->>>>-] (0,5) -- (0,6);
\draw[black, ->>>>-] (5,4) -- (6,4);
\draw[black, ->>>>-] (0,5) -- (1,5);
\draw[black, ->>>>-] (1,5) -- (2,5);
\draw[black, ->>>>-] (2,5) -- (3,5);
\draw[black, ->>>>-] (3,5) -- (4,5);
\draw[black, ->>>>-] (4,5) -- (5,5);
\draw[black, ->>>>-] (0,6) -- (0,7);
\draw[black, ->>>>-] (5,5) -- (6,5);
\draw[black, ->>>>-] (0,6) -- (1,6);
\draw[black, ->>>>-] (1,6) -- (2,6);
\draw[black, ->>>>-] (2,6) -- (3,6);
\draw[black, ->>>>-] (3,6) -- (4,6);
\draw[black, ->>>>-] (4,6) -- (5,6);
\draw[black, ->>>>-] (7,6) -- (7,7);
\draw[black, ->>>>-] (5,6) -- (6,6);
\draw[black, ->>>>-] (0,7) -- (1,7);
\draw[black, ->>>>-] (1,7) -- (2,7);
\draw[black, ->>>>-] (2,7) -- (3,7);
\draw[black, ->>>>-] (3,7) -- (4,7);
\draw[black, ->>>>-] (4,7) -- (5,7);
\draw[black, ->>>>-] (7,5) -- (7,6);
\draw[black, ->>>>-] (5,7) -- (6,7);
\draw[black, ->>>>-] (7,3) -- (7,4);
\draw[black, ->>>>-] (7,4) -- (7,5);
\draw[black, ->>>>-] (1,4) -- (1,5);
\draw[black, ->>>>-] (1,3) -- (1,4);
\draw[black, ->>>>-] (1,5) -- (1,6);
\draw[black, ->>>>-] (1,6) -- (1,7);
\draw[black, ->>>>-, thick] (2,4) -- (2,5);
\draw[black, ->>>>-] (2,3) -- (2,4);
\draw[black, ->>>>-] (2,5) -- (2,6);
\draw[black, ->>>>-] (2,6) -- (2,7);
\draw[black, ->>>>-] (3,4) -- (3,5);
\draw[black, ->>>>-] (3,3) -- (3,4);
\draw[black, ->>>>-] (3,5) -- (3,6);
\draw[black, ->>>>-] (3,6) -- (3,7);
\draw[black, ->>>>-] (4,4) -- (4,5);
\draw[black, ->>>>-] (4,3) -- (4,4);
\draw[black, ->>>>-] (4,5) -- (4,6);
\draw[black, ->>>>-] (4,6) -- (4,7);
\draw[black, ->>>>-] (5,4) -- (5,5);
\draw[black, ->>>>-, thick] (5,3) -- (5,4);
\draw[black, ->>>>-] (5,5) -- (5,6);
\draw[black, ->>>>-] (5,6) -- (5,7);
\draw[black, ->>>>-] (6,4) -- (6,5);
\draw[black, ->>>>-] (6,3) -- (6,4);
\draw[black, ->>>>-] (6,5) -- (6,6);
\draw[black, ->>>>-, thick] (6,6) -- (6,7);
\draw[black, dotted, ->-] (-0.5,-0.5) -- (7.5,-0.5);
\draw[black, dotted, ->-] (-0.5,7.5) -- (7.5,7.5);
\draw[black, dotted, ->>-] (-0.5,-0.5) -- (-0.5,7.5);
\draw[black, dotted, ->>-] (7.5,-0.5) -- (7.5,7.5);
\draw[black, ->>>>-] (0, -0.5) --(0,0);
\draw[black, ->>>>-] (-0.5,0) -- (0,0);
\draw[black, ->>>>-] (1, -0.5) --(1,0);
\draw[black, ->>>>-] (-0.5,1) -- (0,1);
\draw[black, ->>>>-] (2, -0.5) --(2,0);
\draw[black, ->>>>-] (-0.5,2) -- (0,2);
\draw[black, ->>>>-] (3, -0.5) --(3,0);
\draw[black, ->>>>-] (-0.5,3) -- (0,3);
\draw[black, ->>>>-] (4, -0.5) --(4,0);
\draw[black, ->>>>-] (-0.5,4) -- (0,4);
\draw[black, ->>>>-] (5, -0.5) --(5,0);
\draw[black, ->>>>-] (-0.5,6) -- (0,6);
\draw[black, ->>>>-] (6,-0.5) -- (6,0);
\draw[black, ->>>>-] (-0.5,5) -- (0,5);
\draw[black, ->>>>-] (7, -0.5) --(7,0);
\draw[black, ->>>>-] (-0.5,7) -- (0,7);

\draw[black, ->>>-] (0,7) -- (0,7.5);
\draw[black, ->>>-] (7,7) -- (7.5,7);
\draw[black, ->>>-] (1,7) -- (1,7.5);
\draw[black, ->>>-] (7,6) -- (7.5,6);
\draw[black, ->>>-] (2,7) -- (2,7.5);
\draw[black, ->>>-] (7,5) -- (7.5,5);
\draw[black, ->>>-] (3,7) -- (3,7.5);
\draw[black, ->>>-] (7,4) -- (7.5,4);
\draw[black, ->>>-] (4,7) -- (4,7.5);
\draw[black, ->>>-] (7,0) -- (7.5,0);
\draw[black, ->>>-] (5,7) -- (5,7.5);
\draw[black, ->>>-] (7,3) -- (7.5,3);
\draw[black, ->>>-] (6,7) -- (6,7.5);
\draw[black, ->>>-] (7,2) -- (7.5,2);
\draw[black, ->>>-] (7,7) -- (7,7.5);
\draw[black, ->>>-] (7,1) -- (7.5,1);
\draw[red] (0,-0.5) -- (0,7.5);
\draw[blue] (-0.5,7) -- (7.5,7);
\draw[black,fill=white] (0,0) circle (.1 cm);
\draw[black,fill=white] (1,1) circle (.1 cm);
\draw[black,fill=white] (3,3) circle (.1 cm);
\draw[black,fill=white] (2,2) circle (.1 cm);
\draw[black,fill=white] (4,4) circle (.1 cm);
\draw[black,fill=white] (5,5) circle (.1 cm);
\draw[black,fill=white] (6,6) circle (.1 cm);
\draw[black,fill=white] (7,7) circle (.1 cm);
\draw[black,fill=white] (3,5) circle (.1 cm);
\draw[black,fill=white](2,0) circle (.1 cm);
\draw[black,fill=white] (0,2) circle (.1 cm);
\draw[black,fill=white] (4,0) circle (.1 cm);
\draw[black,fill=white] (0,4) circle (.1 cm);
\draw[black,fill=white] (6,0) circle (.1 cm);
\draw[black,fill=white] (0,6) circle (.1 cm);
\draw[black,fill=white] (1,3) circle (.1 cm);
\draw[black,fill=white] (1,5) circle (.1 cm);
\draw[black,fill=white] (1,7) circle (.1 cm);
\draw[black,fill=white] (7,1) circle (.1 cm);
\draw[black,fill=white] (5,1) circle (.1 cm);
\draw[black,fill=white] (3,1) circle (.1 cm);

\draw[black,fill=white] (4,2) circle (.1 cm);
\draw[black,fill=white] (6,2) circle (.1 cm);
\draw[black,fill=white] (2,6) circle (.1 cm);
\draw[black,fill=white] (2,4) circle (.1 cm);
\draw[black,fill=white] (6,4) circle (.1 cm);
\draw[black,fill=white] (5,3) circle (.1 cm);
\draw[black,fill=white] (7,3) circle (.1 cm);
\draw[black,fill=white] (3,5) circle (.1 cm);
\draw[black,fill=white] (3,7) circle (.1 cm);
\draw[black,fill=white] (7,5) circle (.1 cm);
\draw[black,fill=white] (4,6) circle (.1 cm);
\draw[black,fill=white] (5,7) circle (.1 cm);

\draw[black,fill=black] (1,0) circle (.1 cm);
\draw[black,fill=black] (0,1) circle (.1 cm);
\draw[black,fill=black] (3,0) circle (.1 cm);
\draw[black,fill=black] (0,3) circle (.1 cm);
\draw[black,fill=blue] (5,0) circle (.1 cm);
\draw[blue,fill=black] (0,5) circle (.1 cm);
\draw[black,fill=black] (7,0) circle (.1 cm);
\draw[black,fill=black] (0,7) circle (.1 cm);
\draw[black,fill=black] (1,0) circle (.1 cm);
\draw[black,fill=black] (0,1) circle (.1 cm);
\draw[black,fill=black] (3,0) circle (.1 cm);
\draw[black,fill=black] (0,3) circle (.1 cm);
\draw[black,fill=black] (5,0) circle (.1 cm);
\draw[black,fill=black] (0,5) circle (.1 cm);
\draw[black,fill=black] (7,0) circle (.1 cm);
\draw[black,fill=black] (0,7) circle (.1 cm);

\draw[black,fill=black] (1,2) circle (.1 cm);
\draw[black,fill=black] (2,1) circle (.1 cm);
\draw[black,fill=black] (3,2) circle (.1 cm);
\draw[black,fill=black] (2,3) circle (.1 cm);
\draw[black,fill=black] (5,2) circle (.1 cm);
\draw[black,fill=black] (2,5) circle (.1 cm);
\draw[black,fill=black] (7,2) circle (.1 cm);
\draw[black,fill=black] (2,7) circle (.1 cm);
\draw[black,fill=black] (1,4) circle (.1 cm);
\draw[black,fill=black] (4,1) circle (.1 cm);
\draw[black,fill=black] (3,4) circle (.1 cm);
\draw[black,fill=black] (4,3) circle (.1 cm);
\draw[black,fill=black] (5,4) circle (.1 cm);
\draw[black,fill=black] (4,5) circle (.1 cm);
\draw[black,fill=black] (7,4) circle (.1 cm);
\draw[black,fill=black] (4,7) circle (.1 cm);

\draw[black,fill=black] (1,6) circle (.1 cm);
\draw[black,fill=black] (6,1) circle (.1 cm);
\draw[black,fill=black] (3,6) circle (.1 cm);
\draw[black,fill=black] (6,3) circle (.1 cm);
\draw[black,fill=black] (5,6) circle (.1 cm);
\draw[black,fill=black] (6,5) circle (.1 cm);
\draw[black,fill=black] (7,6) circle (.1 cm);
\draw[black,fill=black] (6,7) circle (.1 cm);

   \begin{scope}[xshift=10cm]
      \draw[red, thick, dotted] (-0.5,0.5) -- (7.5,0.5);
\draw[blue, thick, dotted] (6.5,-0.5) -- (6.5,7.5);
\draw[red, thick,->>>>-] (0,0) -- (0,1);
\draw[red, thick,->>>>-] (1,0) -- (1,1);
\draw[red, thick,->>>>-] (2,0) -- (2,1);
\draw[red, thick,->>>>-] (3,0) -- (3,1);
\draw[red, thick,->>>>-] (4,0) -- (4,1);
\draw[red, thick, ->>>>-] (5,0) -- (5,1);
\draw[red, thick, ->>>>-] (6,0) -- (6,1);
\draw[red, thick, ->>>>-] (7,0) -- (7,1);
 \node[red] (a) at (-1.1, 0.5) 
           {\scalebox{0.8}{$-e^{i \theta}$}};
\draw[blue, thick,->>>>-] (6,0) -- (7,0);
\draw[blue, thick,->>>>-] (6,1) -- (7,1);
\draw[blue, thick,->>>>-] (6,2) -- (7,2);
\draw[blue, thick,->>>>-] (6,3) -- (7,3);
\draw[blue, thick, ->>>>-] (6,4) -- (7,4);
\draw[blue, thick, ->>>>-] (6,5) -- (7,5);
\draw[blue, thick, ->>>>-] (6,6) -- (7,6);
\draw[blue, thick, ->>>>-] (6,7) -- (7,7);

\node[blue] (a) at (6.5,8) 
           {\scalebox{0.8}{$e^{i \phi}$}};
             \node[red] (a) at (0, -0.9) 
           {\scalebox{0.8}{$\mathcal{C}_2$}};
            \node[blue] (a) at (-0.9, 7) 
           {\scalebox{0.8}{$\mathcal{C}_1$}};
           \node[blue] (a) at (6.5, -0.9) 
           {\scalebox{0.8}{$\mathcal{C}_2^*$}};
            \node[red] (a) at (8, 0.5) 
           {\scalebox{0.8}{$\mathcal{C}_1^*$}};
   \draw[fill = black, fill opacity = 0.6] 
(2,0) -- (3, 0) -- (3, 2) -- (2, 2) -- (2,0);
\draw[fill = black, fill opacity = 0.6] 
(1,4) -- (3, 4) -- (3, 5) -- (1, 5) -- (1,4);
\draw[fill = black, fill opacity = 0.6] 
(4,0) -- (6, 0) -- (6, 1) -- (4, 1) -- (4,0);
\draw[fill = black, fill opacity = 0.6] 
(5,6) -- (7, 6) -- (7, 7) -- (5, 7) -- (5,6);
\draw[black, ->>>>-] (1,1) -- (2,1);
\draw[black, ->>>>-] (0,0) -- (1,0);
\draw[black, ->>>>-] (1,0) -- (2,0);
\draw[black, ->>>>-] (3,0) -- (4,0);
\draw[black, ->>>>-] (4,0) -- (5,0);
\draw[black, ->>>>-,] (2,0) -- (3,0); 
\draw[black, ->>>>-] (5,0) -- (6,0);

\draw[black, ->>>>-] (0,1) -- (1,1);
\draw[black, ->>>>-] (0,1) -- (0,2);
\draw[black, ->>>>-] (1,1) -- (1,2);
\draw[black, ->>>>-] (2,1) -- (2,2);

\draw[black, ->>>>-, thick] (2,1) -- (3,1);
\draw[black, ->>>>-] (3,1) -- (3,2);
\draw[black, ->>>>-] (3,1) -- (4,1);
\draw[black, ->>>>-] (4,1) -- (4,2);

\draw[black, ->>>>-] (4,1) -- (5,1);
\draw[black, ->>>>-] (5,1) -- (5,2);
\draw[black, ->>>>-] (5,1) -- (6,1);
\draw[black, ->>>>-] (6,1) -- (6,2);

\draw[black, ->>>>-] (7,1) -- (7,2);
\draw[black, ->>>>-] (1,2) -- (2,2);
\draw[black, ->>>>-] (1,2) -- (1,3);

\draw[black, ->>>>-] (0,2) -- (1,2);
\draw[black, ->>>>-] (0,2) -- (0,3);

\draw[black, ->>>>-] (2,2) -- (3,2);
\draw[black, ->>>>-] (2,2) -- (2,3);
\draw[black, ->>>>-] (3,2) -- (4,2);
\draw[black, ->>>>-] (3,2) -- (3,3);
\draw[black, ->>>>-] (4,2) -- (5,2);
\draw[black, ->>>>-] (4,2) -- (4,3);
\draw[black, ->>>>-] (5,2) -- (6,2);
\draw[black, ->>>>-] (5,2) -- (5,3);
\draw[black, ->>>>-] (5,2) -- (6,2);
\draw[black, ->>>>-] (5,2) -- (5,3);
\draw[black, ->>>>-] (6,2) -- (6,3);
\draw[black, ->>>>-] (4,2) -- (5,2);
\draw[black, ->>>>-] (5,3) -- (6,3);
\draw[black, ->>>>-] (7,2) -- (7,3);
\draw[black, ->>>>-] (0,3) -- (1,3);
\draw[black, ->>>>-] (0,3) -- (0,4);
\draw[black, ->>>>-] (1,3) -- (2,3);
\draw[black, ->>>>-] (2,3) -- (3,3);
\draw[black, ->>>>-] (3,3) -- (4,3);
\draw[black, ->>>>-] (4,3) -- (5,3);
\draw[black, ->>>>-] (0,4) -- (0,5);
\draw[black, ->>>>-] (5,3) -- (6,3);
\draw[black, ->>>>-] (0,4) -- (0,5);
\draw[black, ->>>>-] (0,4) -- (1,4);
\draw[black, ->>>>-] (1,4) -- (2,4);
\draw[black, ->>>>-] (2,4) -- (3,4);
\draw[black, ->>>>-] (3,4) -- (4,4);
\draw[black, ->>>>-] (4,4) -- (5,4);
\draw[black, ->>>>-] (0,5) -- (0,6);
\draw[black, ->>>>-] (5,4) -- (6,4);
\draw[black, ->>>>-] (0,5) -- (1,5);
\draw[black, ->>>>-] (1,5) -- (2,5);
\draw[black, ->>>>-] (2,5) -- (3,5);
\draw[black, ->>>>-] (3,5) -- (4,5);
\draw[black, ->>>>-] (4,5) -- (5,5);
\draw[black, ->>>>-] (0,6) -- (0,7);
\draw[black, ->>>>-] (5,5) -- (6,5);
\draw[black, ->>>>-] (0,6) -- (1,6);
\draw[black, ->>>>-] (1,6) -- (2,6);
\draw[black, ->>>>-] (2,6) -- (3,6);
\draw[black, ->>>>-] (3,6) -- (4,6);
\draw[black, ->>>>-] (4,6) -- (5,6);
\draw[black, ->>>>-] (7,6) -- (7,7);
\draw[black, ->>>>-] (5,6) -- (6,6);
\draw[black, ->>>>-] (0,7) -- (1,7);
\draw[black, ->>>>-] (1,7) -- (2,7);
\draw[black, ->>>>-] (2,7) -- (3,7);
\draw[black, ->>>>-] (3,7) -- (4,7);
\draw[black, ->>>>-] (4,7) -- (5,7);
\draw[black, ->>>>-] (7,5) -- (7,6);
\draw[black, ->>>>-] (5,7) -- (6,7);
\draw[black, ->>>>-] (7,3) -- (7,4);
\draw[black, ->>>>-] (7,4) -- (7,5);
\draw[black, ->>>>-] (1,4) -- (1,5);
\draw[black, ->>>>-] (1,3) -- (1,4);
\draw[black, ->>>>-] (1,5) -- (1,6);
\draw[black, ->>>>-] (1,6) -- (1,7);
\draw[black, ->>>>-, thick] (2,4) -- (2,5);
\draw[black, ->>>>-] (2,3) -- (2,4);
\draw[black, ->>>>-] (2,5) -- (2,6);
\draw[black, ->>>>-] (2,6) -- (2,7);
\draw[black, ->>>>-] (3,4) -- (3,5);
\draw[black, ->>>>-] (3,3) -- (3,4);
\draw[black, ->>>>-] (3,5) -- (3,6);
\draw[black, ->>>>-] (3,6) -- (3,7);
\draw[black, ->>>>-] (4,4) -- (4,5);
\draw[black, ->>>>-] (4,3) -- (4,4);
\draw[black, ->>>>-] (4,5) -- (4,6);
\draw[black, ->>>>-] (4,6) -- (4,7);
\draw[black, ->>>>-] (5,4) -- (5,5);
\draw[black, ->>>>-] (5,3) -- (5,4);
\draw[black, ->>>>-] (5,5) -- (5,6);
\draw[black, ->>>>-] (5,6) -- (5,7);
\draw[black, ->>>>-] (6,4) -- (6,5);
\draw[black, ->>>>-] (6,3) -- (6,4);
\draw[black, ->>>>-] (6,5) -- (6,6);
\draw[black, ->>>>-, thick] (6,6) -- (6,7);
\draw[black, dotted, ->-] (-0.5,-0.5) -- (7.5,-0.5);
\draw[black, dotted, ->-] (-0.5,7.5) -- (7.5,7.5);
\draw[black, dotted, ->>-] (-0.5,-0.5) -- (-0.5,7.5);
\draw[black, dotted, ->>-] (7.5,-0.5) -- (7.5,7.5);
\draw[black, ->>>>-] (0, -0.5) --(0,0);
\draw[black, ->>>>-] (-0.5,0) -- (0,0);
\draw[black, ->>>>-] (1, -0.5) --(1,0);
\draw[black, ->>>>-] (-0.5,1) -- (0,1);
\draw[black, ->>>>-] (2, -0.5) --(2,0);
\draw[black, ->>>>-] (-0.5,2) -- (0,2);
\draw[black, ->>>>-] (3, -0.5) --(3,0);
\draw[black, ->>>>-] (-0.5,3) -- (0,3);
\draw[black, ->>>>-] (4, -0.5) --(4,0);
\draw[black, ->>>>-] (-0.5,4) -- (0,4);
\draw[black, ->>>>-] (5, -0.5) --(5,0);
\draw[black, ->>>>-] (-0.5,6) -- (0,6);
\draw[black, ->>>>-] (6,-0.5) -- (6,0);
\draw[black, ->>>>-] (-0.5,5) -- (0,5);
\draw[black, ->>>>-] (7, -0.5) --(7,0);
\draw[black, ->>>>-] (-0.5,7) -- (0,7);

\draw[black, ->>>-] (0,7) -- (0,7.5);
\draw[black, ->>>-] (7,7) -- (7.5,7);
\draw[black, ->>>-] (1,7) -- (1,7.5);
\draw[black, ->>>-] (7,6) -- (7.5,6);
\draw[black, ->>>-] (2,7) -- (2,7.5);
\draw[black, ->>>-] (7,5) -- (7.5,5);
\draw[black, ->>>-] (3,7) -- (3,7.5);
\draw[black, ->>>-] (7,4) -- (7.5,4);
\draw[black, ->>>-] (4,7) -- (4,7.5);
\draw[black, ->>>-] (7,0) -- (7.5,0);
\draw[black, ->>>-] (5,7) -- (5,7.5);
\draw[black, ->>>-] (7,3) -- (7.5,3);
\draw[black, ->>>-] (6,7) -- (6,7.5);
\draw[black, ->>>-] (7,2) -- (7.5,2);
\draw[black, ->>>-] (7,7) -- (7,7.5);
\draw[black, ->>>-] (7,1) -- (7.5,1);
\draw[red] (0,-0.5) -- (0,7.5);
\draw[blue] (-0.5,7) -- (7.5,7);
\draw[black,fill=white] (0,0) circle (.1 cm);
\draw[black,fill=white] (1,1) circle (.1 cm);
\draw[black,fill=white] (3,3) circle (.1 cm);
\draw[black,fill=white] (2,2) circle (.1 cm);
\draw[black,fill=white] (4,4) circle (.1 cm);
\draw[black,fill=white] (5,5) circle (.1 cm);
\draw[black,fill=white] (6,6) circle (.1 cm);
\draw[black,fill=white] (7,7) circle (.1 cm);
\draw[black,fill=white] (3,5) circle (.1 cm);
\draw[black,fill=white](2,0) circle (.1 cm);
\draw[black,fill=white] (0,2) circle (.1 cm);
\draw[black,fill=white] (4,0) circle (.1 cm);
\draw[black,fill=white] (0,4) circle (.1 cm);
\draw[black,fill=white] (6,0) circle (.1 cm);
\draw[black,fill=white] (0,6) circle (.1 cm);
\draw[black,fill=white] (1,3) circle (.1 cm);
\draw[black,fill=white] (1,5) circle (.1 cm);
\draw[black,fill=white] (1,7) circle (.1 cm);
\draw[black,fill=white] (7,1) circle (.1 cm);
\draw[black,fill=white] (5,1) circle (.1 cm);
\draw[black,fill=white] (3,1) circle (.1 cm);

\draw[black,fill=white] (4,2) circle (.1 cm);
\draw[black,fill=white] (6,2) circle (.1 cm);
\draw[black,fill=white] (2,6) circle (.1 cm);
\draw[black,fill=white] (2,4) circle (.1 cm);
\draw[black,fill=white] (6,4) circle (.1 cm);
\draw[black,fill=white] (5,3) circle (.1 cm);
\draw[black,fill=white] (7,3) circle (.1 cm);
\draw[black,fill=white] (3,5) circle (.1 cm);
\draw[black,fill=white] (3,7) circle (.1 cm);
\draw[black,fill=white] (7,5) circle (.1 cm);
\draw[black,fill=white] (4,6) circle (.1 cm);
\draw[black,fill=white] (5,7) circle (.1 cm);

\draw[black,fill=black] (1,0) circle (.1 cm);
\draw[black,fill=black] (0,1) circle (.1 cm);
\draw[black,fill=black] (3,0) circle (.1 cm);
\draw[black,fill=black] (0,3) circle (.1 cm);
\draw[black,fill=blue] (5,0) circle (.1 cm);
\draw[blue,fill=black] (0,5) circle (.1 cm);
\draw[black,fill=black] (7,0) circle (.1 cm);
\draw[black,fill=black] (0,7) circle (.1 cm);
\draw[black,fill=black] (1,0) circle (.1 cm);
\draw[black,fill=black] (0,1) circle (.1 cm);
\draw[black,fill=black] (3,0) circle (.1 cm);
\draw[black,fill=black] (0,3) circle (.1 cm);
\draw[black,fill=black] (5,0) circle (.1 cm);
\draw[black,fill=black] (0,5) circle (.1 cm);
\draw[black,fill=black] (7,0) circle (.1 cm);
\draw[black,fill=black] (0,7) circle (.1 cm);

\draw[black,fill=black] (1,2) circle (.1 cm);
\draw[black,fill=black] (2,1) circle (.1 cm);
\draw[black,fill=black] (3,2) circle (.1 cm);
\draw[black,fill=black] (2,3) circle (.1 cm);
\draw[black,fill=black] (5,2) circle (.1 cm);
\draw[black,fill=black] (2,5) circle (.1 cm);
\draw[black,fill=black] (7,2) circle (.1 cm);
\draw[black,fill=black] (2,7) circle (.1 cm);
\draw[black,fill=black] (1,4) circle (.1 cm);
\draw[black,fill=black] (4,1) circle (.1 cm);
\draw[black,fill=black] (3,4) circle (.1 cm);
\draw[black,fill=black] (4,3) circle (.1 cm);
\draw[black,fill=black] (5,4) circle (.1 cm);
\draw[black,fill=black] (4,5) circle (.1 cm);
\draw[black,fill=black] (7,4) circle (.1 cm);
\draw[black,fill=black] (4,7) circle (.1 cm);

\draw[black,fill=black] (1,6) circle (.1 cm);
\draw[black,fill=black] (6,1) circle (.1 cm);
\draw[black,fill=black] (3,6) circle (.1 cm);
\draw[black,fill=black] (6,3) circle (.1 cm);
\draw[black,fill=black] (5,6) circle (.1 cm);
\draw[black,fill=black] (6,5) circle (.1 cm);
\draw[black,fill=black] (7,6) circle (.1 cm);
\draw[black,fill=black] (6,7) circle (.1 cm);

   \end{scope}

\begin{scope}[yshift=-10cm]
      \draw[red, thick, dotted] (-0.5,0.5) -- (7.5,0.5);
\draw[blue, thick, dotted] (6.5,-0.5) -- (6.5,7.5);
\draw[red, thick,->>>>-] (0,0) -- (0,1);
\draw[red, thick,->>>>-] (1,0) -- (1,1);
\draw[red, thick,->>>>-] (2,0) -- (2,1);
\draw[red, thick,->>>>-] (3,0) -- (3,1);
\draw[red, thick,->>>>-] (4,0) -- (4,1);
\draw[red, thick, ->>>>-] (5,0) -- (5,1);
\draw[red, thick, ->>>>-] (6,0) -- (6,1);
\draw[red, thick, ->>>>-] (7,0) -- (7,1);
 \node[red] (a) at (-1, 0.5) 
           {\scalebox{0.8}{$e^{i \theta}$}};
\draw[blue, thick,->>>>-] (6,0) -- (7,0);
\draw[blue, thick,->>>>-] (6,1) -- (7,1);
\draw[blue, thick,->>>>-] (6,2) -- (7,2);
\draw[blue, thick,->>>>-] (6,3) -- (7,3);
\draw[blue, thick, ->>>>-] (6,4) -- (7,4);
\draw[blue, thick, ->>>>-] (6,5) -- (7,5);
\draw[blue, thick, ->>>>-] (6,6) -- (7,6);
\draw[blue, thick, ->>>>-] (6,7) -- (7,7);

\node[blue] (a) at (6.5,8) 
           {\scalebox{0.8}{$-e^{i \phi}$}};
             \node[red] (a) at (0, -0.9) 
           {\scalebox{0.8}{$\mathcal{C}_2$}};
            \node[blue] (a) at (-0.9, 7) 
           {\scalebox{0.8}{$\mathcal{C}_1$}};
           \node[blue] (a) at (6.5, -0.9) 
           {\scalebox{0.8}{$\mathcal{C}_2^*$}};
            \node[red] (a) at (8, 0.5) 
           {\scalebox{0.8}{$\mathcal{C}_1^*$}};
\draw[fill = black, fill opacity = 0.6] 
(6,0) -- (7, 0) -- (7, 2) -- (6, 2) -- (6,0);
\draw[fill = black, fill opacity = 0.6] 
(1,4) -- (3, 4) -- (3, 5) -- (1, 5) -- (1,4);
\draw[fill = black, fill opacity = 0.6] 
(4,3) -- (6, 3) -- (6, 4) -- (4, 4) -- (4,3);
\draw[fill = black, fill opacity = 0.6] 
(5,6) -- (7, 6) -- (7, 7) -- (5, 7) -- (5,6);
\draw[black, ->>>>-] (1,1) -- (2,1);
\draw[black, ->>>>-] (0,0) -- (1,0);
\draw[black, ->>>>-] (1,0) -- (2,0);
\draw[black, ->>>>-] (3,0) -- (4,0);
\draw[black, ->>>>-] (4,0) -- (5,0);
\draw[black, ->>>>-,] (2,0) -- (3,0); 
\draw[black, ->>>>-] (5,0) -- (6,0);

\draw[black, ->>>>-] (0,1) -- (1,1);
\draw[black, ->>>>-] (0,1) -- (0,2);
\draw[black, ->>>>-] (1,1) -- (1,2);
\draw[black, ->>>>-] (2,1) -- (2,2);

\draw[black, ->>>>-] (2,1) -- (3,1);
\draw[black, ->>>>-] (3,1) -- (3,2);
\draw[black, ->>>>-] (3,1) -- (4,1);
\draw[black, ->>>>-] (4,1) -- (4,2);

\draw[black, ->>>>-] (4,1) -- (5,1);
\draw[black, ->>>>-] (5,1) -- (5,2);
\draw[black, ->>>>-] (5,1) -- (6,1);
\draw[black, ->>>>-] (6,1) -- (6,2);

\draw[black, ->>>>-] (7,1) -- (7,2);
\draw[black, ->>>>-] (1,2) -- (2,2);
\draw[black, ->>>>-] (1,2) -- (1,3);

\draw[black, ->>>>-] (0,2) -- (1,2);
\draw[black, ->>>>-] (0,2) -- (0,3);

\draw[black, ->>>>-] (2,2) -- (3,2);
\draw[black, ->>>>-] (2,2) -- (2,3);
\draw[black, ->>>>-] (3,2) -- (4,2);
\draw[black, ->>>>-] (3,2) -- (3,3);
\draw[black, ->>>>-] (4,2) -- (5,2);
\draw[black, ->>>>-] (4,2) -- (4,3);
\draw[black, ->>>>-] (5,2) -- (6,2);
\draw[black, ->>>>-] (5,2) -- (5,3);
\draw[black, ->>>>-] (5,2) -- (6,2);
\draw[black, ->>>>-] (5,2) -- (5,3);
\draw[black, ->>>>-] (6,2) -- (6,3);
\draw[black, ->>>>-] (4,2) -- (5,2);
\draw[black, ->>>>-] (5,3) -- (6,3);
\draw[black, ->>>>-] (7,2) -- (7,3);
\draw[black, ->>>>-] (0,3) -- (1,3);
\draw[black, ->>>>-] (0,3) -- (0,4);
\draw[black, ->>>>-] (1,3) -- (2,3);
\draw[black, ->>>>-] (2,3) -- (3,3);
\draw[black, ->>>>-] (3,3) -- (4,3);
\draw[black, ->>>>-] (4,3) -- (5,3);
\draw[black, ->>>>-] (0,4) -- (0,5);
\draw[black, ->>>>-] (5,3) -- (6,3);
\draw[black, ->>>>-] (0,4) -- (0,5);
\draw[black, ->>>>-] (0,4) -- (1,4);
\draw[black, ->>>>-] (1,4) -- (2,4);
\draw[black, ->>>>-] (2,4) -- (3,4);
\draw[black, ->>>>-] (3,4) -- (4,4);
\draw[black, ->>>>-] (4,4) -- (5,4);
\draw[black, ->>>>-] (0,5) -- (0,6);
\draw[black, ->>>>-] (5,4) -- (6,4);
\draw[black, ->>>>-] (0,5) -- (1,5);
\draw[black, ->>>>-] (1,5) -- (2,5);
\draw[black, ->>>>-] (2,5) -- (3,5);
\draw[black, ->>>>-] (3,5) -- (4,5);
\draw[black, ->>>>-] (4,5) -- (5,5);
\draw[black, ->>>>-] (0,6) -- (0,7);
\draw[black, ->>>>-] (5,5) -- (6,5);
\draw[black, ->>>>-] (0,6) -- (1,6);
\draw[black, ->>>>-] (1,6) -- (2,6);
\draw[black, ->>>>-] (2,6) -- (3,6);
\draw[black, ->>>>-] (3,6) -- (4,6);
\draw[black, ->>>>-] (4,6) -- (5,6);
\draw[black, ->>>>-] (7,6) -- (7,7);
\draw[black, ->>>>-] (5,6) -- (6,6);
\draw[black, ->>>>-] (0,7) -- (1,7);
\draw[black, ->>>>-] (1,7) -- (2,7);
\draw[black, ->>>>-] (2,7) -- (3,7);
\draw[black, ->>>>-] (3,7) -- (4,7);
\draw[black, ->>>>-] (4,7) -- (5,7);
\draw[black, ->>>>-] (7,5) -- (7,6);
\draw[black, ->>>>-] (5,7) -- (6,7);
\draw[black, ->>>>-] (7,3) -- (7,4);
\draw[black, ->>>>-] (7,4) -- (7,5);
\draw[black, ->>>>-] (1,4) -- (1,5);
\draw[black, ->>>>-] (1,3) -- (1,4);
\draw[black, ->>>>-] (1,5) -- (1,6);
\draw[black, ->>>>-] (1,6) -- (1,7);
\draw[black, ->>>>-, thick] (2,4) -- (2,5);
\draw[black, ->>>>-] (2,3) -- (2,4);
\draw[black, ->>>>-] (2,5) -- (2,6);
\draw[black, ->>>>-] (2,6) -- (2,7);
\draw[black, ->>>>-] (3,4) -- (3,5);
\draw[black, ->>>>-] (3,3) -- (3,4);
\draw[black, ->>>>-] (3,5) -- (3,6);
\draw[black, ->>>>-] (3,6) -- (3,7);
\draw[black, ->>>>-] (4,4) -- (4,5);
\draw[black, ->>>>-] (4,3) -- (4,4);
\draw[black, ->>>>-] (4,5) -- (4,6);
\draw[black, ->>>>-] (4,6) -- (4,7);
\draw[black, ->>>>-] (5,4) -- (5,5);
\draw[black, ->>>>-, thick] (5,3) -- (5,4);
\draw[black, ->>>>-] (5,5) -- (5,6);
\draw[black, ->>>>-] (5,6) -- (5,7);
\draw[black, ->>>>-] (6,4) -- (6,5);
\draw[black, ->>>>-] (6,3) -- (6,4);
\draw[black, ->>>>-] (6,5) -- (6,6);
\draw[black, ->>>>-, thick] (6,6) -- (6,7);
\draw[black, dotted, ->-] (-0.5,-0.5) -- (7.5,-0.5);
\draw[black, dotted, ->-] (-0.5,7.5) -- (7.5,7.5);
\draw[black, dotted, ->>-] (-0.5,-0.5) -- (-0.5,7.5);
\draw[black, dotted, ->>-] (7.5,-0.5) -- (7.5,7.5);
\draw[black, ->>>>-] (0, -0.5) --(0,0);
\draw[black, ->>>>-] (-0.5,0) -- (0,0);
\draw[black, ->>>>-] (1, -0.5) --(1,0);
\draw[black, ->>>>-] (-0.5,1) -- (0,1);
\draw[black, ->>>>-] (2, -0.5) --(2,0);
\draw[black, ->>>>-] (-0.5,2) -- (0,2);
\draw[black, ->>>>-] (3, -0.5) --(3,0);
\draw[black, ->>>>-] (-0.5,3) -- (0,3);
\draw[black, ->>>>-] (4, -0.5) --(4,0);
\draw[black, ->>>>-] (-0.5,4) -- (0,4);
\draw[black, ->>>>-] (5, -0.5) --(5,0);
\draw[black, ->>>>-] (-0.5,6) -- (0,6);
\draw[black, ->>>>-] (6,-0.5) -- (6,0);
\draw[black, ->>>>-] (-0.5,5) -- (0,5);
\draw[black, ->>>>-] (7, -0.5) --(7,0);
\draw[black, ->>>>-] (-0.5,7) -- (0,7);

\draw[black, ->>>-] (0,7) -- (0,7.5);
\draw[black, ->>>-] (7,7) -- (7.5,7);
\draw[black, ->>>-] (1,7) -- (1,7.5);
\draw[black, ->>>-] (7,6) -- (7.5,6);
\draw[black, ->>>-] (2,7) -- (2,7.5);
\draw[black, ->>>-] (7,5) -- (7.5,5);
\draw[black, ->>>-] (3,7) -- (3,7.5);
\draw[black, ->>>-] (7,4) -- (7.5,4);
\draw[black, ->>>-] (4,7) -- (4,7.5);
\draw[black, ->>>-] (7,0) -- (7.5,0);
\draw[black, ->>>-] (5,7) -- (5,7.5);
\draw[black, ->>>-] (7,3) -- (7.5,3);
\draw[black, ->>>-] (6,7) -- (6,7.5);
\draw[black, ->>>-] (7,2) -- (7.5,2);
\draw[black, ->>>-] (7,7) -- (7,7.5);
\draw[black, ->>>-] (7,1) -- (7.5,1);
\draw[red] (0,-0.5) -- (0,7.5);
\draw[blue] (-0.5,7) -- (7.5,7);
\draw[black,fill=white] (0,0) circle (.1 cm);
\draw[black,fill=white] (1,1) circle (.1 cm);
\draw[black,fill=white] (3,3) circle (.1 cm);
\draw[black,fill=white] (2,2) circle (.1 cm);
\draw[black,fill=white] (4,4) circle (.1 cm);
\draw[black,fill=white] (5,5) circle (.1 cm);
\draw[black,fill=white] (6,6) circle (.1 cm);
\draw[black,fill=white] (7,7) circle (.1 cm);
\draw[black,fill=white] (3,5) circle (.1 cm);
\draw[black,fill=white](2,0) circle (.1 cm);
\draw[black,fill=white] (0,2) circle (.1 cm);
\draw[black,fill=white] (4,0) circle (.1 cm);
\draw[black,fill=white] (0,4) circle (.1 cm);
\draw[black,fill=white] (6,0) circle (.1 cm);
\draw[black,fill=white] (0,6) circle (.1 cm);
\draw[black,fill=white] (1,3) circle (.1 cm);
\draw[black,fill=white] (1,5) circle (.1 cm);
\draw[black,fill=white] (1,7) circle (.1 cm);
\draw[black,fill=white] (7,1) circle (.1 cm);
\draw[black,fill=white] (5,1) circle (.1 cm);
\draw[black,fill=white] (3,1) circle (.1 cm);

\draw[black,fill=white] (4,2) circle (.1 cm);
\draw[black,fill=white] (6,2) circle (.1 cm);
\draw[black,fill=white] (2,6) circle (.1 cm);
\draw[black,fill=white] (2,4) circle (.1 cm);
\draw[black,fill=white] (6,4) circle (.1 cm);
\draw[black,fill=white] (5,3) circle (.1 cm);
\draw[black,fill=white] (7,3) circle (.1 cm);
\draw[black,fill=white] (3,5) circle (.1 cm);
\draw[black,fill=white] (3,7) circle (.1 cm);
\draw[black,fill=white] (7,5) circle (.1 cm);
\draw[black,fill=white] (4,6) circle (.1 cm);
\draw[black,fill=white] (5,7) circle (.1 cm);

\draw[black,fill=black] (1,0) circle (.1 cm);
\draw[black,fill=black] (0,1) circle (.1 cm);
\draw[black,fill=black] (3,0) circle (.1 cm);
\draw[black,fill=black] (0,3) circle (.1 cm);
\draw[black,fill=blue] (5,0) circle (.1 cm);
\draw[blue,fill=black] (0,5) circle (.1 cm);
\draw[black,fill=black] (7,0) circle (.1 cm);
\draw[black,fill=black] (0,7) circle (.1 cm);
\draw[black,fill=black] (1,0) circle (.1 cm);
\draw[black,fill=black] (0,1) circle (.1 cm);
\draw[black,fill=black] (3,0) circle (.1 cm);
\draw[black,fill=black] (0,3) circle (.1 cm);
\draw[black,fill=black] (5,0) circle (.1 cm);
\draw[black,fill=black] (0,5) circle (.1 cm);
\draw[black,fill=black] (7,0) circle (.1 cm);
\draw[black,fill=black] (0,7) circle (.1 cm);

\draw[black,fill=black] (1,2) circle (.1 cm);
\draw[black,fill=black] (2,1) circle (.1 cm);
\draw[black,fill=black] (3,2) circle (.1 cm);
\draw[black,fill=black] (2,3) circle (.1 cm);
\draw[black,fill=black] (5,2) circle (.1 cm);
\draw[black,fill=black] (2,5) circle (.1 cm);
\draw[black,fill=black] (7,2) circle (.1 cm);
\draw[black,fill=black] (2,7) circle (.1 cm);
\draw[black,fill=black] (1,4) circle (.1 cm);
\draw[black,fill=black] (4,1) circle (.1 cm);
\draw[black,fill=black] (3,4) circle (.1 cm);
\draw[black,fill=black] (4,3) circle (.1 cm);
\draw[black,fill=black] (5,4) circle (.1 cm);
\draw[black,fill=black] (4,5) circle (.1 cm);
\draw[black,fill=black] (7,4) circle (.1 cm);
\draw[black,fill=black] (4,7) circle (.1 cm);

\draw[black,fill=black] (1,6) circle (.1 cm);
\draw[black,fill=black] (6,1) circle (.1 cm);
\draw[black,fill=black] (3,6) circle (.1 cm);
\draw[black,fill=black] (6,3) circle (.1 cm);
\draw[black,fill=black] (5,6) circle (.1 cm);
\draw[black,fill=black] (6,5) circle (.1 cm);
\draw[black,fill=black] (7,6) circle (.1 cm);
\draw[black,fill=black] (6,7) circle (.1 cm);

   \end{scope}

   \begin{scope}[xshift=10cm, yshift= -10cm]
      \draw[red, thick, dotted] (-0.5,0.5) -- (7.5,0.5);
\draw[blue, thick, dotted] (6.5,-0.5) -- (6.5,7.5);
\draw[red, thick,->>>>-] (0,0) -- (0,1);
\draw[red, thick,->>>>-] (1,0) -- (1,1);
\draw[red, thick,->>>>-] (2,0) -- (2,1);
\draw[red, thick,->>>>-] (3,0) -- (3,1);
\draw[red, thick,->>>>-] (4,0) -- (4,1);
\draw[red, thick, ->>>>-] (5,0) -- (5,1);
\draw[red, thick, ->>>>-] (6,0) -- (6,1);
\draw[red, thick, ->>>>-] (7,0) -- (7,1);
 \node[red] (a) at (-1.1, 0.5) 
           {\scalebox{0.8}{$-e^{i \theta}$}};
\draw[blue, thick,->>>>-] (6,0) -- (7,0);
\draw[blue, thick,->>>>-] (6,1) -- (7,1);
\draw[blue, thick,->>>>-] (6,2) -- (7,2);
\draw[blue, thick,->>>>-] (6,3) -- (7,3);
\draw[blue, thick, ->>>>-] (6,4) -- (7,4);
\draw[blue, thick, ->>>>-] (6,5) -- (7,5);
\draw[blue, thick, ->>>>-] (6,6) -- (7,6);
\draw[blue, thick, ->>>>-] (6,7) -- (7,7);

\node[blue] (a) at (6.5,8) 
           {\scalebox{0.8}{$-e^{i \phi}$}};
             \node[red] (a) at (0, -0.9) 
           {\scalebox{0.8}{$\mathcal{C}_2$}};
            \node[blue] (a) at (-0.9, 7) 
           {\scalebox{0.8}{$\mathcal{C}_1$}};
           \node[blue] (a) at (6.5, -0.9) 
           {\scalebox{0.8}{$\mathcal{C}_2^*$}};
            \node[red] (a) at (8, 0.5) 
           {\scalebox{0.8}{$\mathcal{C}_1^*$}};
  \draw[fill = black, fill opacity = 0.6] 
(6,0) -- (7, 0) -- (7, 2) -- (6, 2) -- (6,0);
\draw[fill = black, fill opacity = 0.6] 
(1,4) -- (3, 4) -- (3, 5) -- (1, 5) -- (1,4);
\draw[fill = black, fill opacity = 0.6] 
(4,0) -- (6, 0) -- (6, 1) -- (4, 1) -- (4,0);
\draw[fill = black, fill opacity = 0.6] 
(5,6) -- (7, 6) -- (7, 7) -- (5, 7) -- (5,6);
\draw[black, ->>>>-] (1,1) -- (2,1);
\draw[black, ->>>>-] (0,0) -- (1,0);
\draw[black, ->>>>-] (1,0) -- (2,0);
\draw[black, ->>>>-] (3,0) -- (4,0);
\draw[black, ->>>>-] (4,0) -- (5,0);
\draw[black, ->>>>-,] (2,0) -- (3,0); 
\draw[black, ->>>>-] (5,0) -- (6,0);

\draw[black, ->>>>-] (0,1) -- (1,1);
\draw[black, ->>>>-] (0,1) -- (0,2);
\draw[black, ->>>>-] (1,1) -- (1,2);
\draw[black, ->>>>-] (2,1) -- (2,2);

\draw[black, ->>>>-] (2,1) -- (3,1);
\draw[black, ->>>>-] (3,1) -- (3,2);
\draw[black, ->>>>-] (3,1) -- (4,1);
\draw[black, ->>>>-] (4,1) -- (4,2);

\draw[black, ->>>>-] (4,1) -- (5,1);
\draw[black, ->>>>-] (5,1) -- (5,2);
\draw[black, ->>>>-] (5,1) -- (6,1);
\draw[black, ->>>>-] (6,1) -- (6,2);

\draw[black, ->>>>-] (7,1) -- (7,2);
\draw[black, ->>>>-] (1,2) -- (2,2);
\draw[black, ->>>>-] (1,2) -- (1,3);

\draw[black, ->>>>-] (0,2) -- (1,2);
\draw[black, ->>>>-] (0,2) -- (0,3);

\draw[black, ->>>>-] (2,2) -- (3,2);
\draw[black, ->>>>-] (2,2) -- (2,3);
\draw[black, ->>>>-] (3,2) -- (4,2);
\draw[black, ->>>>-] (3,2) -- (3,3);
\draw[black, ->>>>-] (4,2) -- (5,2);
\draw[black, ->>>>-] (4,2) -- (4,3);
\draw[black, ->>>>-] (5,2) -- (6,2);
\draw[black, ->>>>-] (5,2) -- (5,3);
\draw[black, ->>>>-] (5,2) -- (6,2);
\draw[black, ->>>>-] (5,2) -- (5,3);
\draw[black, ->>>>-] (6,2) -- (6,3);
\draw[black, ->>>>-] (4,2) -- (5,2);
\draw[black, ->>>>-] (5,3) -- (6,3);
\draw[black, ->>>>-] (7,2) -- (7,3);
\draw[black, ->>>>-] (0,3) -- (1,3);
\draw[black, ->>>>-] (0,3) -- (0,4);
\draw[black, ->>>>-] (1,3) -- (2,3);
\draw[black, ->>>>-] (2,3) -- (3,3);
\draw[black, ->>>>-] (3,3) -- (4,3);
\draw[black, ->>>>-] (4,3) -- (5,3);
\draw[black, ->>>>-] (0,4) -- (0,5);
\draw[black, ->>>>-] (5,3) -- (6,3);
\draw[black, ->>>>-] (0,4) -- (0,5);
\draw[black, ->>>>-] (0,4) -- (1,4);
\draw[black, ->>>>-] (1,4) -- (2,4);
\draw[black, ->>>>-] (2,4) -- (3,4);
\draw[black, ->>>>-] (3,4) -- (4,4);
\draw[black, ->>>>-] (4,4) -- (5,4);
\draw[black, ->>>>-] (0,5) -- (0,6);
\draw[black, ->>>>-] (5,4) -- (6,4);
\draw[black, ->>>>-] (0,5) -- (1,5);
\draw[black, ->>>>-] (1,5) -- (2,5);
\draw[black, ->>>>-] (2,5) -- (3,5);
\draw[black, ->>>>-] (3,5) -- (4,5);
\draw[black, ->>>>-] (4,5) -- (5,5);
\draw[black, ->>>>-] (0,6) -- (0,7);
\draw[black, ->>>>-] (5,5) -- (6,5);
\draw[black, ->>>>-] (0,6) -- (1,6);
\draw[black, ->>>>-] (1,6) -- (2,6);
\draw[black, ->>>>-] (2,6) -- (3,6);
\draw[black, ->>>>-] (3,6) -- (4,6);
\draw[black, ->>>>-] (4,6) -- (5,6);
\draw[black, ->>>>-] (7,6) -- (7,7);
\draw[black, ->>>>-] (5,6) -- (6,6);
\draw[black, ->>>>-] (0,7) -- (1,7);
\draw[black, ->>>>-] (1,7) -- (2,7);
\draw[black, ->>>>-] (2,7) -- (3,7);
\draw[black, ->>>>-] (3,7) -- (4,7);
\draw[black, ->>>>-] (4,7) -- (5,7);
\draw[black, ->>>>-] (7,5) -- (7,6);
\draw[black, ->>>>-] (5,7) -- (6,7);
\draw[black, ->>>>-] (7,3) -- (7,4);
\draw[black, ->>>>-] (7,4) -- (7,5);
\draw[black, ->>>>-] (1,4) -- (1,5);
\draw[black, ->>>>-] (1,3) -- (1,4);
\draw[black, ->>>>-] (1,5) -- (1,6);
\draw[black, ->>>>-] (1,6) -- (1,7);
\draw[black, ->>>>-, thick] (2,4) -- (2,5);
\draw[black, ->>>>-] (2,3) -- (2,4);
\draw[black, ->>>>-] (2,5) -- (2,6);
\draw[black, ->>>>-] (2,6) -- (2,7);
\draw[black, ->>>>-] (3,4) -- (3,5);
\draw[black, ->>>>-] (3,3) -- (3,4);
\draw[black, ->>>>-] (3,5) -- (3,6);
\draw[black, ->>>>-] (3,6) -- (3,7);
\draw[black, ->>>>-] (4,4) -- (4,5);
\draw[black, ->>>>-] (4,3) -- (4,4);
\draw[black, ->>>>-] (4,5) -- (4,6);
\draw[black, ->>>>-] (4,6) -- (4,7);
\draw[black, ->>>>-] (5,4) -- (5,5);
\draw[black, ->>>>-] (5,3) -- (5,4);
\draw[black, ->>>>-] (5,5) -- (5,6);
\draw[black, ->>>>-] (5,6) -- (5,7);
\draw[black, ->>>>-] (6,4) -- (6,5);
\draw[black, ->>>>-] (6,3) -- (6,4);
\draw[black, ->>>>-] (6,5) -- (6,6);
\draw[black, ->>>>-, thick] (6,6) -- (6,7);
\draw[black, dotted, ->-] (-0.5,-0.5) -- (7.5,-0.5);
\draw[black, dotted, ->-] (-0.5,7.5) -- (7.5,7.5);
\draw[black, dotted, ->>-] (-0.5,-0.5) -- (-0.5,7.5);
\draw[black, dotted, ->>-] (7.5,-0.5) -- (7.5,7.5);
\draw[black, ->>>>-] (0, -0.5) --(0,0);
\draw[black, ->>>>-] (-0.5,0) -- (0,0);
\draw[black, ->>>>-] (1, -0.5) --(1,0);
\draw[black, ->>>>-] (-0.5,1) -- (0,1);
\draw[black, ->>>>-] (2, -0.5) --(2,0);
\draw[black, ->>>>-] (-0.5,2) -- (0,2);
\draw[black, ->>>>-] (3, -0.5) --(3,0);
\draw[black, ->>>>-] (-0.5,3) -- (0,3);
\draw[black, ->>>>-] (4, -0.5) --(4,0);
\draw[black, ->>>>-] (-0.5,4) -- (0,4);
\draw[black, ->>>>-] (5, -0.5) --(5,0);
\draw[black, ->>>>-] (-0.5,6) -- (0,6);
\draw[black, ->>>>-] (6,-0.5) -- (6,0);
\draw[black, ->>>>-] (-0.5,5) -- (0,5);
\draw[black, ->>>>-] (7, -0.5) --(7,0);
\draw[black, ->>>>-] (-0.5,7) -- (0,7);

\draw[black, ->>>-] (0,7) -- (0,7.5);
\draw[black, ->>>-] (7,7) -- (7.5,7);
\draw[black, ->>>-] (1,7) -- (1,7.5);
\draw[black, ->>>-] (7,6) -- (7.5,6);
\draw[black, ->>>-] (2,7) -- (2,7.5);
\draw[black, ->>>-] (7,5) -- (7.5,5);
\draw[black, ->>>-] (3,7) -- (3,7.5);
\draw[black, ->>>-] (7,4) -- (7.5,4);
\draw[black, ->>>-] (4,7) -- (4,7.5);
\draw[black, ->>>-] (7,0) -- (7.5,0);
\draw[black, ->>>-] (5,7) -- (5,7.5);
\draw[black, ->>>-] (7,3) -- (7.5,3);
\draw[black, ->>>-] (6,7) -- (6,7.5);
\draw[black, ->>>-] (7,2) -- (7.5,2);
\draw[black, ->>>-] (7,7) -- (7,7.5);
\draw[black, ->>>-] (7,1) -- (7.5,1);
\draw[red] (0,-0.5) -- (0,7.5);
\draw[blue] (-0.5,7) -- (7.5,7);
\draw[black,fill=white] (0,0) circle (.1 cm);
\draw[black,fill=white] (1,1) circle (.1 cm);
\draw[black,fill=white] (3,3) circle (.1 cm);
\draw[black,fill=white] (2,2) circle (.1 cm);
\draw[black,fill=white] (4,4) circle (.1 cm);
\draw[black,fill=white] (5,5) circle (.1 cm);
\draw[black,fill=white] (6,6) circle (.1 cm);
\draw[black,fill=white] (7,7) circle (.1 cm);
\draw[black,fill=white] (3,5) circle (.1 cm);
\draw[black,fill=white](2,0) circle (.1 cm);
\draw[black,fill=white] (0,2) circle (.1 cm);
\draw[black,fill=white] (4,0) circle (.1 cm);
\draw[black,fill=white] (0,4) circle (.1 cm);
\draw[black,fill=white] (6,0) circle (.1 cm);
\draw[black,fill=white] (0,6) circle (.1 cm);
\draw[black,fill=white] (1,3) circle (.1 cm);
\draw[black,fill=white] (1,5) circle (.1 cm);
\draw[black,fill=white] (1,7) circle (.1 cm);
\draw[black,fill=white] (7,1) circle (.1 cm);
\draw[black,fill=white] (5,1) circle (.1 cm);
\draw[black,fill=white] (3,1) circle (.1 cm);

\draw[black,fill=white] (4,2) circle (.1 cm);
\draw[black,fill=white] (6,2) circle (.1 cm);
\draw[black,fill=white] (2,6) circle (.1 cm);
\draw[black,fill=white] (2,4) circle (.1 cm);
\draw[black,fill=white] (6,4) circle (.1 cm);
\draw[black,fill=white] (5,3) circle (.1 cm);
\draw[black,fill=white] (7,3) circle (.1 cm);
\draw[black,fill=white] (3,5) circle (.1 cm);
\draw[black,fill=white] (3,7) circle (.1 cm);
\draw[black,fill=white] (7,5) circle (.1 cm);
\draw[black,fill=white] (4,6) circle (.1 cm);
\draw[black,fill=white] (5,7) circle (.1 cm);

\draw[black,fill=black] (1,0) circle (.1 cm);
\draw[black,fill=black] (0,1) circle (.1 cm);
\draw[black,fill=black] (3,0) circle (.1 cm);
\draw[black,fill=black] (0,3) circle (.1 cm);
\draw[black,fill=blue] (5,0) circle (.1 cm);
\draw[blue,fill=black] (0,5) circle (.1 cm);
\draw[black,fill=black] (7,0) circle (.1 cm);
\draw[black,fill=black] (0,7) circle (.1 cm);
\draw[black,fill=black] (1,0) circle (.1 cm);
\draw[black,fill=black] (0,1) circle (.1 cm);
\draw[black,fill=black] (3,0) circle (.1 cm);
\draw[black,fill=black] (0,3) circle (.1 cm);
\draw[black,fill=black] (5,0) circle (.1 cm);
\draw[black,fill=black] (0,5) circle (.1 cm);
\draw[black,fill=black] (7,0) circle (.1 cm);
\draw[black,fill=black] (0,7) circle (.1 cm);

\draw[black,fill=black] (1,2) circle (.1 cm);
\draw[black,fill=black] (2,1) circle (.1 cm);
\draw[black,fill=black] (3,2) circle (.1 cm);
\draw[black,fill=black] (2,3) circle (.1 cm);
\draw[black,fill=black] (5,2) circle (.1 cm);
\draw[black,fill=black] (2,5) circle (.1 cm);
\draw[black,fill=black] (7,2) circle (.1 cm);
\draw[black,fill=black] (2,7) circle (.1 cm);
\draw[black,fill=black] (1,4) circle (.1 cm);
\draw[black,fill=black] (4,1) circle (.1 cm);
\draw[black,fill=black] (3,4) circle (.1 cm);
\draw[black,fill=black] (4,3) circle (.1 cm);
\draw[black,fill=black] (5,4) circle (.1 cm);
\draw[black,fill=black] (4,5) circle (.1 cm);
\draw[black,fill=black] (7,4) circle (.1 cm);
\draw[black,fill=black] (4,7) circle (.1 cm);

\draw[black,fill=black] (1,6) circle (.1 cm);
\draw[black,fill=black] (6,1) circle (.1 cm);
\draw[black,fill=black] (3,6) circle (.1 cm);
\draw[black,fill=black] (6,3) circle (.1 cm);
\draw[black,fill=black] (5,6) circle (.1 cm);
\draw[black,fill=black] (6,5) circle (.1 cm);
\draw[black,fill=black] (7,6) circle (.1 cm);
\draw[black,fill=black] (6,7) circle (.1 cm);

   \end{scope}
   \end{tikzpicture}
   \caption{Possible $\bm{\sigma}$ background with $e^{i \theta}$ and $e^{i \phi}$ holonomies around $a$ and $b$ cycles}
   \label{background}
\end{figure}
\section{Proof of \autoref{thm2}}\label{sec:topord}
\paragraph{Ground states in the physical Hilbert space.}
The Hamiltonian \eqref{hami} has four-fold almost degenerate ground states characterized by the four possible $\pi$-flux backgrounds of the torus. Each ground state is obtained by acting with the projection that imposes the Gauss' law on the ground state of the many-body Hamiltonian in the $\pi$-flux phase, for a given choice of holonomies:
\begin{equation}\label{eq:GS two expressions}
    \ket{\Omega_{ab}} = 
    \prod_{i \in \text{V}(\Gamma_L)} \bigg(\frac{1+Q_i}{2}\bigg) \ket{\psi_{\bm{-1},a,b}}\otimes \ket{\bm{-1},a,b}
    =\frac{1}{2^{|\Gamma|}} \sum_{\Lambda \subset \Gamma} A_{\Lambda} (-1)^{N_{\Lambda}}\ket{\psi_{\bm{-1},a,b}}\otimes \ket{\bm{-1},a,b}
\end{equation}
where $\ket{\bm{-1},a,b}$ are a set of common eigenstates of all $\hat{\sigma}^z_{ij}$ satisfying: $B_p=-1$ for any plaquette $p \in \text{F}(\Gamma_L)$, $\hat Z_{\caC_1}=a$ and $\hat Z_{\caC_2} = b$ and $\psi_{\bm{-1},a,b}$ is the ground state of $H(\bm{-1},a,b)$acting on the fermionic Fock space.

The following proposition states some basic properties about these states. Recall that $\mathcal{H}^{\rm phys}$ is the physical Hilbert space, where Gauss' law is satisfied on each vertex, see \hyperref[def:gauss]{Definition \ref*{def:gauss}}.

\begin{proposition}\label{gs}
    The states $\ket{\Omega_{ab}}$ satisfy the following properties:
    \begin{enumerate}
        \item $\ket{\Omega_{ab}}\in\mathcal{H}^{\rm phys}$
        \item $\ket{\Omega_{ab}}$ depends only on the choice of the background holonomies
        \item $\ket{\Omega_{ab}}\neq 0$
    \end{enumerate}
\end{proposition}

\begin{proof}
    \begin{enumerate}
        \item This is immediate since $\prod_{i \in \text{V}(\Gamma_L)} \frac{1+Q_i}{2}$ projects onto the eigenspace $+1$ of each~$Q_i$.
        \item Let $\bm{\sigma}, \bm{\sigma'}$ be two different gauge equivalent backgrounds. There exist a subset $\Lambda \subset \Gamma$ such that $A_{\Lambda} \ket{\bm{\sigma}} = \ket{\bm{\sigma'}}$, and the Hamiltonian in the new background is given by $H(\bm{\sigma'}) = (-1)^{N_{\Lambda}} H(\bm{\sigma}) (-1)^{N_{\Lambda}}$. Thus the ground state of $H(\bm{\sigma'})$ is
        \begin{equation}
            \ket{\psi_{{\bm\sigma'}}} =(-1)^{N_{\Lambda}} \ket{\psi_{{\bm\sigma}}}
        \end{equation}
Hence,
\begin{align}
\ket{\Omega'_{ab}} &= \frac{1}{2^{|\Gamma|}} \sum_{\Upsilon \subset \Gamma} (-1)^{N_{\Upsilon}} A_{\Upsilon} \ket{\psi_{{\bm\sigma'}}}\otimes \ket{\bm{\sigma'},a,b}\\
&= \frac{1}{2^{|\Gamma|}} \sum_{\Upsilon \subset \Gamma} (-1)^{N_{\Upsilon}+N_\Lambda} A_{\Upsilon} \ket{\psi_{{\bm\sigma}}}\otimes A_\Lambda\ket{\bm{\sigma},a,b}
=\ket{\Omega_{ab}}
\end{align}
        since $A_{\Upsilon}  A_\Lambda = A_{\Upsilon + \Lambda}$, $N_{\Upsilon}+N_\Lambda = N_{\Upsilon + \Lambda}$ (where the sum of subsets is understood in the $\mathbb{Z}_2$-valued chain group) and the sum is translation invariant. 
\item Using again the right expression of the ground state in~(\ref{eq:GS two expressions}), we have
\begin{align}
\Vert\Omega_{ab}\Vert^2 
    & = \frac{1}{2^{|\Gamma|}} \sum_{\Lambda \subset \Gamma} \braket{\bm{-1},a,b} {A_{\Lambda}\big| \bm{-1},a,b } \braket{\psi_{\bm{-1},a,b}}{(-1)^{N_{\Lambda}}\bigg| \psi_{\bm{-1},a,b} } \\
    &=\frac{1}{2^{|\Gamma|}} \braket{\psi_{\bm{-1},a,b}}{\Big(1+(-1)^{N_{\Gamma}}\Big)\bigg|\psi_{\bm{-1},a,b}}
\end{align}
where we have used that the diagonal matrix elements of $A_\Lambda$ in the $\hat\sigma_{ij}^z$-basis vanish unless $\Lambda = \emptyset$ or $\Lambda=\Gamma$. Since the ground states are at half-filling by Theorem~\ref{thm1} and $L\in 4\mathbb{N}$, $\ket{\psi_{\bm{-1},a,b}}$ has an even number of particles and so the matrix element is $2$. This shows that $\ket{\Omega_{ab}}\neq 0$.
    \end{enumerate}
\end{proof}
\begin{remark} The argument used in the computation of the norm of $\Omega_{ab}$ can be repeated to rewrite in a convenient way the energy of $\Omega_{ab}$. We have:
\begin{equation}\label{eq:en}
\begin{split}
\frac{\braket{\Omega_{ab}}{H \big|\Omega_{ab}}}{\|\Omega_{ab}\|^{2}} &= \braket{ \psi_{\bm{-1},a,b}}{H(\bm{-1},a,b) \big|\psi_{\bm{-1},a,b}} \\
&= E_{0,L}(\bm{-1},a,b)\;.
\end{split}
\end{equation}
\end{remark}

\paragraph{Spectral gap estimate.} Let us prove the spectral gap estimate \eqref{eq:gap}. In view of (\ref{eq:en}),
\begin{equation}\label{eq:max}
\sup \text{Spec}(PHP) = \max_{a,b \in \mathbb{Z}_{2} \times \mathbb{Z}_{2}} E_{0,L}({\bf -1}, a,b)\;.
\end{equation}
Then, we write, for any normalized vector $\ket{\xi}$ in the total Hilbert space:
\begin{equation}\label{eq:xi}\braket{ \xi}{P^{\perp} H P^{\perp}\bigg|\xi } = \braket{\xi_{{\bf -1}}}{ P^{\perp} H P^{\perp}\bigg|\xi_{{\bf -1}}} + \braket{\tilde \xi}{P^{\perp} H P^{\perp}\bigg| \tilde\xi}
\end{equation}
where $\ket{\xi_{\bf -1}}$ is the orthogonal projection of $\ket{\xi}$ on the subspace associated with flux $\pi$ in every plaquette, and $\ket{\tilde \xi} = \ket{\xi} - \ket{\xi_{{\bf -1}}}$. Consider the second term. By \eqref{bound2}, we have:
\begin{equation}
\braket{ \tilde \xi}{ P^{\perp} H P^{\perp}\bigg| \tilde\xi} \geq \|\tilde \xi\|^{2} \Big(2\Delta_{L} + \min_{a,b\in \mathbb{Z}_{2} \times \mathbb{Z}_{2}} E_{0,L}(\bm{-1}, a,b)\Big)\;.
\end{equation}
Consider now the first term in \eqref{eq:xi}. We have:
\begin{equation}\label{eq:xi2}
\braket{\xi_{{\bf -1}}}{P^{\perp} H P^{\perp}\bigg|\xi_{{\bf -1}}} \geq \|\xi_{{\bf -1}}\|^{2} \min_{a,b \in \mathbb{Z}_{2} \times \mathbb{Z}_{2}} E_{1,L}({\bf -1},a,b)\;,
\end{equation}
where $E_{1,L}({\bf -1},a,b)$ is the first eigenvalue above the ground state for the Hamiltonian $H({\bf -1}, a,b)$. For $U=0$, the Hamiltonian is gapped, and the spectral gap is $4m$, recall \eqref{eq:epi}. For $U\neq 0$, the stability of the spectral gap follows from the convergence of the cluster expansion method reviewed in \hyperref[sec:int]{Subsection \ref*{sec:int}}, see discussion after Eq. (\ref{eq:detbd}). One has:
\begin{equation}\label{eq:xi3}
E_{1,L}({\bf -1},a,b) \geq E_{0,L}({\bf -1},a,b) + 2\delta_{L}\;,\qquad \delta_{L} \geq m - C|U|^{1/3}\;.
\end{equation}
Therefore, putting together (\ref{eq:max}), (\ref{eq:xi}), (\ref{eq:xi2}), (\ref{eq:xi3}), we have, using that $\|\xi\|=1$:
\begin{equation}
\begin{split}
&\braket{\xi}{ P^{\perp} H P^{\perp}\bigg|\xi} - \sup \text{Spec}(PHP) \\
&\quad \geq 2\min(\delta_{L}, \Delta_{L}) + \min_{a,b} E_{0,L}({\bf -1},a,b) - \max_{a,b} E_{0,L}({\bf -1},a,b)\\
&\quad \geq 2\min(\delta_{L}, \Delta_{L}) - CL^{2}e^{-cL}\;,
\end{split}
\end{equation}
where in the last step we used the exponential closeness of the approximate ground state energies, \eqref{eq:expdeg}. This concludes the proof of \eqref{eq:gap}. 

\paragraph{Local topological order.} Let us now prove topological order defined by \eqref{eq:Exp LTQO}. Recalling \hyperref[def:gi]{Definition \ref*{def:gi}}, the algebra of gauge-invariant observables is the centralizer~$\cal C_{\cal A}(\cal Q)$ of $\cal Q$. It is generated by
\begin{equation}
    \{\hat\sigma^x_{ij}:(i,j)\in \text{E}(\Gamma_L)\}\cup\{\hat Z_{\caC}:\caC\in Z_1(\Gamma_L)\}\cup\{a^\pm_{i,\eta} \hat Z_{\caC_{i,j}} a^\pm_{j,\eta'}:i,j\in\Gamma_{L},\, \eta,\eta'\in \{\uparrow,\downarrow\}\},
\end{equation}
where $\caC_{i,j}$ is any cycle such that $\partial \caC_{i,j} = \{i,j\}$, with the convention that $\caC_{i,i} = 0$ for all $i\in\Gamma$. We will now show that all these local gauge-invariant observables are first of all diagonal on the ground state space, and secondly that they have the same expectation value in all of the ground states, up to exponentially small errors. But then,
\begin{equation}
 \langle\Omega_{ab}\vert \mathcal{O}\Omega_{ab}\rangle \stackrel{L}{=} \frac{1}{4}\sum_{(a,b)\in\mathbb{Z}_2\times\mathbb{Z}_2}\langle\Omega_{ab}\vert \mathcal{O}\Omega_{ab}\rangle = \frac{1}{\Tr P }\Tr(P\mathcal{O})   
\end{equation}
since $\{\vert \Omega_{ab}\rangle:(a,b)\in\mathbb{Z}_2\times\mathbb{Z}_2\}$ form an orthonormal basis of the range of $P$.

A monomial in $\hat\sigma^x_{ij}$ is naturally identified with $\hat X_{\caC^*}$ where $\caC^*$ is the sum of all edges in the monomial (seen in the $\mathbb{Z}_2$-valued cohomology group). If $\partial \caC^* = 0$, then $\caC^*$ is a coboundary $\caC^* = \partial\Lambda$ since otherwise it would not be a local observable, and so
\begin{equation}
    \hat X_{\caC^*} \ket{\Omega_{ab}} = (-1)^{N_\Lambda}\ket{\Omega_{ab}}
\end{equation}
namely, it acts as a purely fermionic observable, and it is therefore diagonal, on the ground state space. As proved below, local observables that are diagonal in the $\hat \sigma^{z}_{ij}$ basis have diagonal elements in the ground-state space that are approximately independent by $(a,b)$. 
If $\partial\caC^*$ is not empty, then the observable changes the flux from $\pi$ to $0$ at each boundary plaquette and so $\langle \Omega_{a'b'}\vert \hat X_{\caC^*} \Omega_{ab}\rangle = 0$
for all $(a',b'),(a,b)\in\mathbb{Z}_2\times\mathbb{Z}_2$.

We are left with discussing gauge invariant observables that are diagonal in the $\hat\sigma^z_{ij}$ basis. These observables do not change the background, and hence they are diagonal in the ground state space. We must only prove that the expectation values in all four ground states are equal as $L\to\infty$. For the $\hat Z_\caC$ observables with $\caC$ being a boundary, we write $\caC = \partial\sum_{i}p_i$ for a set of plaquettes $\{p_i:i\in\{1,\ldots,M\}\}$ and note that $Z_\caC\ket{\Omega_{ab}} = \prod_i B_p \ket{\Omega_{ab}} = (-1)^{M}\ket{\Omega_{ab}}$ for all $(a,b)\in\mathbb{Z}_2$. 

It remains to consider the algebra generated by the `open lines' $a^\pm_{i,\eta} \hat Z_{\caC_{i,j}} a^\pm_{j,\eta'}$. Since the states $\Omega_{ab}$ have a definite number of fermions, the operators $a^+_{i,\eta} \hat Z_{\caC_{i,j}} a^+_{j,\eta'}$ vanish on the ground state space, and so do their adjoints. So we turn to $a^+_{i,\eta} \hat Z_{\caC_{i,j}} a^-_{j,\eta'}$. Commuting the operator through the Gauss' law projection, the operator $\hat Z_{\caC_{i,j}}$ then becomes a phase $\pm 1$ that depends on the background, see \autoref{fig:pifluxes}. Hamiltonians in the $\pi$-flux phase, differing by the value of the holonomies, can be viewed as being the same Hamiltonian but endowed with different boundary conditions (periodic or antiperiodic). The $\pm 1$ phase produced by the observable is compensated by the change of boundary conditions.

We are left with the evaluation of fermionic monomials in the $\pi$-flux state, with $(a,b)$ holonomies. For $U=0$, this can be done via the application of the fermionic Wick's rule, whose outcome is completely specified by the two-point function (\ref{eq:2pt}) as $\beta \to \infty$, which is the Fermi projector. The exponential closeness of the expectation values for different $(a,b)$ follows the representation of the finite-volume Fermi projector via Poisson summation formula:
\begin{equation}\label{eq:PoissonF}
P^{(a,b)}_{L}(x;y) = \sum_{m_{1}, m_{2}\in \mathbb{Z}^{2}} a^{m_{1}} b^{m_{2}} P_{\infty}(x+m_{1} e_{1} L + m_{2} e_{2} L;y)\;,
\end{equation}
where $P^{(a,b)}_{L}$ is the Fermi projector for the model in a finite volume and with $(a,b)$ holonomies, while $P_{\infty}$ is the Fermi projector for the model on $\mathbb{Z}^{2}$. As in the proof for the exponential closeness of the ground state energies, expectation values for different holonomies are then compared by inspection of the corresponding Wick's rules, using that the term with $m_{1} = m_{2} = 0$ in (\ref{eq:PoissonF}) is independent of $(a,b)$, and using that the terms with $m_{i}\neq 0$ are exponentially suppressed, thanks to the exponential decay of the Fermi projector.

Let us now discuss the case of weakly interacting fermions, $U\neq 0$. The starting point is the Duhamel expansion for the expectation value of the observable,
\begin{equation}\label{eq:Oexp}
\langle \mathcal{O} \rangle_{\beta,L} = \langle \mathcal{O} \rangle^{0}_{\beta,L} + \sum_{n\geq 1} \frac{(-U)^{n}}{n!} \int_{[0,\beta]^{n}} dt_{1}\cdots dt_{n} \langle {\bf T} \gamma_{t_{1}}(V)\,;\, \gamma_{t_{2}}(V)\,;\, \cdots\,;\, \gamma_{t_{n}}(V)\,; \mathcal{O} \rangle^{0}_{\beta,L}
\end{equation}
The ground state expectation value is the $\beta \to \infty$ limit of the above expression. The convergence of the series for small $|U|$, and uniformly in $L$, follows from fermionic cluster expansion, analogously to what has been reviewed in \hyperref[sec:int]{Subsection \ref*{sec:int}} for the free energy. Furthermore, similarly to what has been done in \hyperref[sec:int]{Subsection \ref*{sec:int}} to compare energies associated with different boundary conditions, the term-by-term comparison of the cumulants in \eqref{eq:Oexp}, computed using the BBF formula \eqref{eq:BBF}, combined with the Poisson formula for the two-point function \eqref{eq:poisson2pt}, allows to compare Gibbs states of Hamiltonians with different boundary conditions:
\begin{equation}\label{eq:bdtopo}
| \langle \mathcal{O} \rangle^{(a,b)}_{\beta,L} - \langle \mathcal{O} \rangle^{(a',b')}_{\beta,L} | \leq C_{\mathcal{O}} e^{-cL}
\end{equation}
for small $|U|$ uniformly in $\beta,L$. Taking the $\beta \to \infty$ limit, the claim (\ref{eq:Exp LTQO}) follows. This concludes the proof of \autoref{thm2}.\qed

\begin{remark} The bound (\ref{eq:bdtopo}) is effective for observables that are $L$-independent, or more generally for which the constant $C_{\mathcal{O}}$ depends on $L$ sub-exponentially; {\it e.g.} extensive sums of local operators. It does not allow us to consider observables $\mathcal{O}$ that are exponentials of extensive observables (not surprisingly).
\end{remark}

\begin{remark}
Excited eigenstates of $H$ can be constructed by replacing $\ket{\psi_{\bm{-1},a,b}}$ in (\ref{eq:GS two expressions}) with an excited state of $H(\bm{-1},a,b)$. If $U=0$, these excited states are explicit. For instance,
 \begin{equation}
        \ket{\zeta^{k,k'}_{ab}} = \prod_{i \in \Gamma_L} \bigg(\frac{1+Q_i}{2} \bigg) a^+_{k,\eta} a^+_{k',\eta'} \ket{\psi_{\bm{-1},a,b}}\otimes \ket{\bm{-1},a,b}
    \end{equation}
are exact many-body excitations, where $a^+_{k,\eta}$ creates a particle in an eigenstate of the Bloch Hamiltonian (\ref{eq:blochHam}) with quasi-momentum $k$, with appropriate boundary conditions. Due to the fact that the vacuum state is obtained filling all the negative energy eigenstates of the Bloch Hamiltonian, the Pauli principle implies that the operator $a^+_{k,\eta}$ creates a particle with positive energy. More generally, one can generate a family of eigenstates of $H$ by acting on $\ket{\psi_{\bm{-1},a,b}}$ with an even number of momentum-space fermionic operators. These states are completely delocalized in configuration space. Instead, if one acts with an odd number of fermionic operators on $\ket{\psi_{\bm{-1},a,b}}$, the resulting vector is annihilated by the action of the projector imposing the Gauss' law.
\end{remark}

\section{Proof of \autoref{thm3}}\label{sec:TQO}

\subsection{Adiabatic flux insertion and braiding}\label{picontr}

We turn to \autoref{thm3}, starting with the properties of the loop operators $W_{\caC^*}$, see \hyperref[def:loop operators]{Definition \ref*{def:loop operators}}.

\subsubsection{Threading a $\pi$-flux through a contactible cycle}

We consider the twisted Hamiltonians~(\ref{twistedhami}), first in the case of $\mathcal{C}^* = \partial\Lambda$ being a $1$-coboundary.
\begin{lemma}\label{claim1}
    Let $\mathcal{C}^* = \partial\Lambda$ be a $1$-coboundary. Then
\begin{equation}\label{twist}
 H_{\mathcal{C}^*}(\phi) =\ep{\iu \phi N_{\Lambda}} H \ep{-\iu \phi N_{\Lambda}}.
\end{equation}
\end{lemma}

\begin{proof}
Since both left and right hand sides equal $H$ at $\phi = 0$, it suffices to prove that they satisfy the same differential equation. On the right hand side, we observe that for any $(j,\ell)\in \text{E}(\Gamma_L)$,
\begin{equation}
    -\iu\partial_\phi\ep{\iu\phi n_{i,\eta}}a_{j,\eta}^+a_{\ell,\eta}\ep{\iu\phi n_{i,\eta}} = \epsilon(i;j,k)\ep{\iu\phi n_{i,\eta}}a_{j,\eta}^+a_{\ell,\eta}\ep{\iu\phi n_{i,\eta}},\qquad
    \epsilon(i;j,k) = 
    \begin{cases} 1 & \text{if }i = j, \\ -1& \text{if }i = \ell, \\ 0 &\text{otherwise.}
    \end{cases}
\end{equation}
Summing these up, we conclude that the only non-vanishing terms in the derivative of the right hand side are those such that $\mathcal{I}[(i,j), \mathcal{C}^*] = \pm 1$, which is exactly what arises from $\partial_\phi H_{\mathcal{C}^*}(\phi)$ as defined by~(\ref{twistedhami}). 
\end{proof}

The observation~(\ref{twist}) implies that $\dot P_{\mathcal{C}^*}(\phi) = \iu[N_\Lambda , P_{\mathcal{C}^*}(\phi)]$ and therefore
    \begin{equation}\label{commut}
        [\mathcal{K}_{\mathcal{C}^*}(\phi) - N_\Lambda, P_{\mathcal{C^*}}(\phi)] = 0
    \end{equation}
by~(\ref{gsode}).

\begin{lemma}\label{linecontractible}
 Let $\mathcal{C}^* = \partial\Lambda$ be a $1$-coboundary. Then $W_{\mathcal{C}^*} $ preserves the ground state manifold of $H$ exactly:
    \begin{equation}
    W_{\mathcal{C}^*}  P  W_{\mathcal{C}^*}^* = P.
\end{equation}
\end{lemma}
\begin{proof}
Definition~\ref{def:QAG} and Lemma~\ref{claim1} imply that 
\begin{equation}
    \mathcal{K}_{\mathcal{C}^*}(\phi) = \ep{\iu \phi N_{\Lambda}} \mathcal{K}_{\mathcal{C}^*}(0) \ep{-\iu \phi N_{\Lambda}}.
\end{equation}
From now on, we will use the notation $\mathcal{K}_{\mathcal{C}^*} \equiv \mathcal{K}_{\mathcal{C}^*}(0)$. It follows that the unitary family $\tilde V_{\caC^*}(\phi)= \ep{\iu \phi N_{\Lambda}}\ep{-\iu \phi (N_{\Lambda}-\mathcal{K}_{\mathcal{C}^*})}$ solves the initial value problem
\begin{equation}
    -\iu \partial_\phi \tilde V_{\caC^*}(\phi) = \mathcal{K}_{\mathcal{C}^*}(\phi) \tilde V_{\caC^*}(\phi),\qquad \tilde V_{\caC^*}(0) = \mathbbm{1},
\end{equation}
and therefore $\tilde V_{\caC^*}(\phi) = V_{\caC^*}(\phi)$ by uniqueness, see~(\ref{eq:QAFlow}). Hence $W_{\mathcal{C}^*} = A_\Lambda \ep{\iu \pi N_{\Lambda}}\ep{-\iu \pi (N_{\Lambda}-\mathcal{K}_{\mathcal{C}^*})}$. Moreover, $[N_{\Lambda}-\mathcal{K}_{\mathcal{C}^*}, P]=0$, see \eqref{commut}, implies that
\begin{equation}
    [\ep{-\iu \pi (N_{\Lambda}-\mathcal{K}_{\mathcal{C}^*})}, P]=0.
\end{equation}
With this, the identity $\ep{\iu \pi N_{\Lambda}} = (-1)^{N_\Lambda}$ and the definition~(\ref{eq:Def of charge}) of charge yield
\begin{equation}
     W_{\mathcal{C}^*}  P  W_{\mathcal{C}^*}^* = Q_{\Lambda} P Q_{\Lambda} = P
\end{equation}
since $P$ is a spectral projector of a gauge invariant Hamiltonian, see~(\ref{eq:H gauge invariance}).
\end{proof}

While the usefulness of the identity $V_{\caC^*}(\phi) = \ep{\iu \phi N_{\Lambda}}\ep{-\iu \phi (N_{\Lambda}-\mathcal{K}_{\mathcal{C}^*})}$ was pointed out in~\cite{PhysRevB.101.085138} because the unitary simplifies at $\phi = 2\pi$ by integrality of the spectrum of $N_\Lambda$, we are interested in applying it here up to only $\phi= \pi$. Physically, this corresponds to a process in which we insert a $\pi$-flux with an external $U(1)$ gauge field which is absorbed by the system as a $\mathbb{Z}_2$ gauge transformation, leaving thus the ground state space invariant. This justifies the inclusion of a $\mathbb{Z}_2$ pure gauge transformation $A_{\Lambda} = \hat X_{\caC^*}$ to $V_{\caC^*}(\phi)$ in \hyperref[def:loop operators]{Definition \ref*{def:loop operators}}. Of course, this cancellation of the two types of `gauge transformation' is possible only at $\phi = \pi$, since the gauge field is $\mathbb{Z}_2$-valued, see also the proof just below.

\noindent This concludes the proof of \autoref{thm3.1} of \autoref{thm3} in this case. 

 Since the ground state manifold is degenerate, an adiabatic process preserving such subspace may in principle shuffle its basis in a non universal way. However this does not happen as a consequence of the protection given by topology. As \autoref{thm2} shows, see Remark~\ref{rem: after thm2}(2), ground states of $H$ are labelled by pairs $(a,b)\in\mathbb{Z}_2\times\mathbb{Z}_2$ of holonomies around two fixed cycles $\mathcal{C}_1,\mathcal{C}_2 \in Z_1(\Gamma_L)$ that are not boundaries:
 \begin{equation}\label{eq:Characterize GSS}
     \hat Z_{\mathcal{C}_1} \ket{\Omega_{ab}} = a \ket{\Omega_{ab}},\qquad
     \hat Z_{\mathcal{C}_2} \ket{\Omega_{ab}} = b \ket{\Omega_{ab}}.
 \end{equation}
With this, we first recall that $\hat Z_{\mathcal{C}_j}$ is gauge invariant because $\mathcal{C}_j$ are cycles, see~(\ref{eq:Z cycle}), and note further that they commute with $\ep{-\iu \pi (N_{\Lambda}-\mathcal{K}_{\mathcal{C}^*})}$ since this term is diagonal in $\hat{\sigma}^z_{ij}$ basis. Hence they commute with $W_{\mathcal{C}^*}$. This proves \autoref{thm3.2} of \autoref{thm3} in this case of `contractible cycles' since the intersection number between $\caC$ and $\caC^*$ must be even, see \autoref{fig:intersection number}. 
     \begin{figure}
\centering
\begin{tikzpicture}[scale=0.7]
\draw (-0.5,-0.5) grid (7.5,7.5);
\draw[blue, very thick] (1,3) -- (0.5,3);
\draw[blue, very thick] (1,3) -- (2,3);
\draw[blue, very thick] (2,3) -- (2,5);
\draw[blue, very thick] (2,5) -- (4,5);
\draw[blue, very thick] (4,5) -- (4,4);
\draw[blue, very thick] (4,4) -- (6.5,4);
\draw[blue, very thick] (-0.5,3) -- (0.5,3);
\draw[blue, very thick] (6.5,4) -- (7,4);
\draw[blue, very thick] (7,4) -- (7,3);
\draw[blue, very thick] (7.5,3) -- (7,3);
\draw[red, very thick] (1.5,1.5) -- (3.5,1.5) -- (3.5, 2.5) -- (4.5, 2.5) -- (4.5,3.5) -- (5.5, 3.5) -- (5.5,4.5) -- (3.5,4.5) -- (3.5,5.5) -- (2.5,5.5) -- (2.5,4.5) -- (1.5,4.5) -- (1.5, 3.5) -- (2.5, 3.5) -- (2.5,2.5) -- (1.5,2.5) -- (1.5,1.5);
\draw[black,fill=white] (0,0) circle (.1 cm);
\draw[black,fill=white] (1,1) circle (.1 cm);
\draw[black,fill=white] (3,3) circle (.1 cm);
\draw[black,fill=white] (2,2) circle (.1 cm);
\draw[black,fill=white] (4,4) circle (.1 cm);
\draw[black,fill=white] (5,5) circle (.1 cm);
\draw[black,fill=white] (6,6) circle (.1 cm);
\draw[black,fill=white] (7,7) circle (.1 cm);

\draw[black,fill=white] (3,5) circle (.1 cm);

\draw[black,fill=white](2,0) circle (.1 cm);
\draw[black,fill=white] (0,2) circle (.1 cm);
\draw[black,fill=white] (4,0) circle (.1 cm);
\draw[black,fill=white] (0,4) circle (.1 cm);
\draw[black,fill=white] (6,0) circle (.1 cm);
\draw[black,fill=white] (0,6) circle (.1 cm);
\draw[black,fill=white] (1,3) circle (.1 cm);
\draw[black,fill=white] (1,5) circle (.1 cm);
\draw[black,fill=white] (1,7) circle (.1 cm);
\draw[black,fill=white] (7,1) circle (.1 cm);
\draw[black,fill=white] (5,1) circle (.1 cm);
\draw[black,fill=white] (3,1) circle (.1 cm);

\draw[black,fill=white] (4,2) circle (.1 cm);
\draw[black,fill=white] (6,2) circle (.1 cm);
\draw[black,fill=white] (2,6) circle (.1 cm);
\draw[black,fill=white] (2,4) circle (.1 cm);
\draw[black,fill=white] (6,4) circle (.1 cm);
\draw[black,fill=white] (5,3) circle (.1 cm);
\draw[black,fill=white] (7,3) circle (.1 cm);
\draw[black,fill=white] (3,5) circle (.1 cm);
\draw[black,fill=white] (3,7) circle (.1 cm);
\draw[black,fill=white] (7,5) circle (.1 cm);
\draw[black,fill=white] (4,6) circle (.1 cm);
\draw[black,fill=white] (5,7) circle (.1 cm);

\draw[black,fill=black] (1,0) circle (.1 cm);
\draw[black,fill=black] (0,1) circle (.1 cm);
\draw[black,fill=black] (3,0) circle (.1 cm);
\draw[black,fill=black] (0,3) circle (.1 cm);
\draw[black,fill=blue] (5,0) circle (.1 cm);
\draw[blue,fill=black] (0,5) circle (.1 cm);
\draw[black,fill=black] (7,0) circle (.1 cm);
\draw[black,fill=black] (0,7) circle (.1 cm);
\draw[black,fill=black] (1,0) circle (.1 cm);
\draw[black,fill=black] (0,1) circle (.1 cm);
\draw[black,fill=black] (3,0) circle (.1 cm);
\draw[black,fill=black] (0,3) circle (.1 cm);
\draw[black,fill=black] (5,0) circle (.1 cm);
\draw[black,fill=black] (0,5) circle (.1 cm);
\draw[black,fill=black] (7,0) circle (.1 cm);
\draw[black,fill=black] (0,7) circle (.1 cm);

\draw[black,fill=black] (1,2) circle (.1 cm);
\draw[black,fill=black] (2,1) circle (.1 cm);
\draw[black,fill=black] (3,2) circle (.1 cm);
\draw[black,fill=black] (2,3) circle (.1 cm);
\draw[black,fill=black] (5,2) circle (.1 cm);
\draw[black,fill=black] (2,5) circle (.1 cm);
\draw[black,fill=black] (7,2) circle (.1 cm);
\draw[black,fill=black] (2,7) circle (.1 cm);
\draw[black,fill=black] (1,4) circle (.1 cm);
\draw[black,fill=black] (4,1) circle (.1 cm);
\draw[black,fill=black] (3,4) circle (.1 cm);
\draw[black,fill=black] (4,3) circle (.1 cm);
\draw[black,fill=black] (5,4) circle (.1 cm);
\draw[black,fill=black] (4,5) circle (.1 cm);
\draw[black,fill=black] (7,4) circle (.1 cm);
\draw[black,fill=black] (4,7) circle (.1 cm);

\draw[black,fill=black] (1,6) circle (.1 cm);
\draw[black,fill=black] (6,1) circle (.1 cm);
\draw[black,fill=black] (3,6) circle (.1 cm);
\draw[black,fill=black] (6,3) circle (.1 cm);
\draw[black,fill=black] (5,6) circle (.1 cm);
\draw[black,fill=black] (6,5) circle (.1 cm);
\draw[black,fill=black] (7,6) circle (.1 cm);
\draw[black,fill=black] (6,7) circle (.1 cm);

\draw[dotted, ->-] (-0.5,-0.5) -- (7.5,-0.5);
\draw[dotted, ->-] (-0.5,7.5) -- (7.5,7.5);
\draw[dotted, ->>-] (-0.5,-0.5) -- (-0.5,7.5);
\draw[dotted, ->>-] (7.5,-0.5) -- (7.5,7.5);

\node[red,above] (c) at (3.8,5.3) {{$\mathcal{C}^*$}};
\node[blue,above] (d) at (0.5,2.9) {{$\mathcal{C}$}};

\begin{scope}[xshift = 12cm]
\draw (-0.5,-0.5) grid (7.5,7.5);
\draw[blue, very thick] (1,3) -- (0.5,3);
\draw[blue, very thick] (1,3) -- (2,3);
\draw[blue, very thick] (2,3) -- (2,5);
\draw[blue, very thick] (2,5) -- (4,5);
\draw[blue, very thick] (4,5) -- (4,4);
\draw[blue, very thick] (4,4) -- (6.5,4);
\draw[blue, very thick] (-0.5,3) -- (0.5,3);
\draw[blue, very thick] (6.5,4) -- (7,4);
\draw[blue, very thick] (7,4) -- (7,3);
\draw[blue, very thick] (7.5,3) -- (7,3);
 
\draw[red, very thick] (4.5,-0.5) -- (4.5,0.5)--(3.5,0.5) -- (3.5, 2.5) -- (4.5, 2.5) -- (4.5,3.5) -- (5.5, 3.5) -- (5.5,4.5) -- (3.5,4.5) -- (3.5,5.5) -- (2.5,5.5) -- (2.5,6.5) -- (4.5,6.5) -- (4.5,7.5);
 
\node[red,above] (c) at (3.8,5.3) {{$\mathcal{C}^*$}};
\node[blue,above] (d) at (1.5,2.9) {{$\mathcal{C}$}};

\draw[black,fill=white] (0,0) circle (.1 cm);
\draw[black,fill=white] (1,1) circle (.1 cm);
\draw[black,fill=white] (3,3) circle (.1 cm);
\draw[black,fill=white] (2,2) circle (.1 cm);
\draw[black,fill=white] (4,4) circle (.1 cm);
\draw[black,fill=white] (5,5) circle (.1 cm);
\draw[black,fill=white] (6,6) circle (.1 cm);
\draw[black,fill=white] (7,7) circle (.1 cm);

\draw[black,fill=white] (3,5) circle (.1 cm);

\draw[black,fill=white](2,0) circle (.1 cm);
\draw[black,fill=white] (0,2) circle (.1 cm);
\draw[black,fill=white] (4,0) circle (.1 cm);
\draw[black,fill=white] (0,4) circle (.1 cm);
\draw[black,fill=white] (6,0) circle (.1 cm);
\draw[black,fill=white] (0,6) circle (.1 cm);
\draw[black,fill=white] (1,3) circle (.1 cm);
\draw[black,fill=white] (1,5) circle (.1 cm);
\draw[black,fill=white] (1,7) circle (.1 cm);
\draw[black,fill=white] (7,1) circle (.1 cm);
\draw[black,fill=white] (5,1) circle (.1 cm);
\draw[black,fill=white] (3,1) circle (.1 cm);

\draw[black,fill=white] (4,2) circle (.1 cm);
\draw[black,fill=white] (6,2) circle (.1 cm);
\draw[black,fill=white] (2,6) circle (.1 cm);
\draw[black,fill=white] (2,4) circle (.1 cm);
\draw[black,fill=white] (6,4) circle (.1 cm);
\draw[black,fill=white] (5,3) circle (.1 cm);
\draw[black,fill=white] (7,3) circle (.1 cm);
\draw[black,fill=white] (3,5) circle (.1 cm);
\draw[black,fill=white] (3,7) circle (.1 cm);
\draw[black,fill=white] (7,5) circle (.1 cm);
\draw[black,fill=white] (4,6) circle (.1 cm);
\draw[black,fill=white] (5,7) circle (.1 cm);

\draw[black,fill=black] (1,0) circle (.1 cm);
\draw[black,fill=black] (0,1) circle (.1 cm);
\draw[black,fill=black] (3,0) circle (.1 cm);
\draw[black,fill=black] (0,3) circle (.1 cm);
\draw[black,fill=blue] (5,0) circle (.1 cm);
\draw[blue,fill=black] (0,5) circle (.1 cm);
\draw[black,fill=black] (7,0) circle (.1 cm);
\draw[black,fill=black] (0,7) circle (.1 cm);
\draw[black,fill=black] (1,0) circle (.1 cm);
\draw[black,fill=black] (0,1) circle (.1 cm);
\draw[black,fill=black] (3,0) circle (.1 cm);
\draw[black,fill=black] (0,3) circle (.1 cm);
\draw[black,fill=black] (5,0) circle (.1 cm);
\draw[black,fill=black] (0,5) circle (.1 cm);
\draw[black,fill=black] (7,0) circle (.1 cm);
\draw[black,fill=black] (0,7) circle (.1 cm);

\draw[black,fill=black] (1,2) circle (.1 cm);
\draw[black,fill=black] (2,1) circle (.1 cm);
\draw[black,fill=black] (3,2) circle (.1 cm);
\draw[black,fill=black] (2,3) circle (.1 cm);
\draw[black,fill=black] (5,2) circle (.1 cm);
\draw[black,fill=black] (2,5) circle (.1 cm);
\draw[black,fill=black] (7,2) circle (.1 cm);
\draw[black,fill=black] (2,7) circle (.1 cm);
\draw[black,fill=black] (1,4) circle (.1 cm);
\draw[black,fill=black] (4,1) circle (.1 cm);
\draw[black,fill=black] (3,4) circle (.1 cm);
\draw[black,fill=black] (4,3) circle (.1 cm);
\draw[black,fill=black] (5,4) circle (.1 cm);
\draw[black,fill=black] (4,5) circle (.1 cm);
\draw[black,fill=black] (7,4) circle (.1 cm);
\draw[black,fill=black] (4,7) circle (.1 cm);

\draw[black,fill=black] (1,6) circle (.1 cm);
\draw[black,fill=black] (6,1) circle (.1 cm);
\draw[black,fill=black] (3,6) circle (.1 cm);
\draw[black,fill=black] (6,3) circle (.1 cm);
\draw[black,fill=black] (5,6) circle (.1 cm);
\draw[black,fill=black] (6,5) circle (.1 cm);
\draw[black,fill=black] (7,6) circle (.1 cm);
\draw[black,fill=black] (6,7) circle (.1 cm);

\draw[dotted, ->-] (-0.5,-0.5) -- (7.5,-0.5);
\draw[dotted, ->-] (-0.5,7.5) -- (7.5,7.5);
\draw[dotted, ->>-] (-0.5,-0.5) -- (-0.5,7.5);
\draw[dotted, ->>-] (7.5,-0.5) -- (7.5,7.5);
\end{scope}

\end{tikzpicture}
\caption{If $\caC$ is a non-trivial cycle, the  number of intersections between $\mathcal{C}^*$ and $\mathcal{C}$ is even if $\caC^*$ is a boundary (left) but odd if it is a non-trivial cocycle (right). Notice that the curves' orientation is irrelevant here since $\ep{\iu \pi} = \ep{-\iu \pi}$ so positive and negative intersections give the same contribution$\mod 2\pi$}\label{fig:intersection number}
\end{figure}
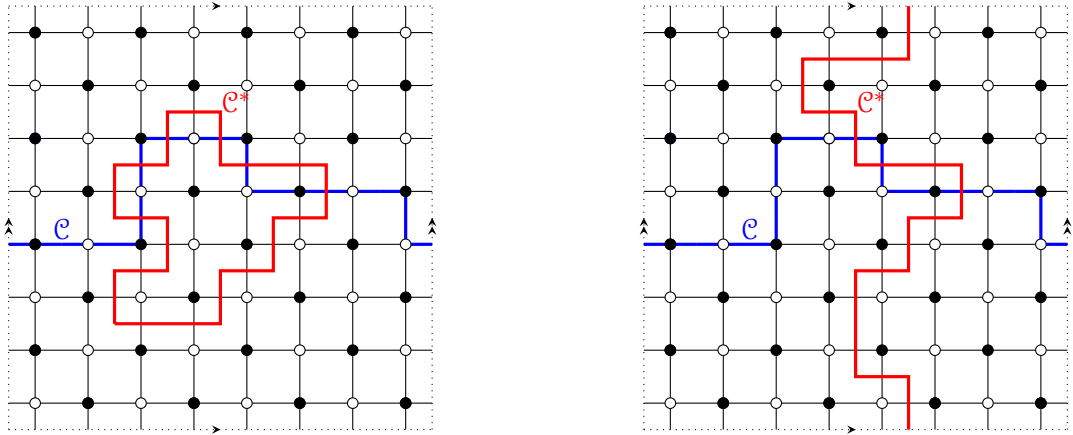
Moreover,
 \begin{equation}
    \hat Z_{\mathcal{C}_1} W_{\mathcal{C}^*} \ket{\Omega_{ab}} 
    = a W_{\mathcal{C}^*} \ket{\Omega_{ab}},\qquad
    \hat Z_{\mathcal{C}_2} W_{\mathcal{C}^*} \ket{\Omega_{ab}} 
    = b W_{\mathcal{C}^*} \ket{\Omega_{ab}}.
 \end{equation}
Since $W_{\mathcal{C}^*}$ preserves the ground state manifold and the eigenvalues $(a,b)$ uniquely determine the state $\ket{\Omega_{ab}}$, we conclude that 
 \begin{equation}\label{eq:trivial loops on GS}
     W_{\mathcal{C}^*} \ket{\Omega_{ab}} = \ep{\iu \omega_{\mathcal{C}^*}^{ab}} \ket{\Omega_{ab}}
 \end{equation}
for some phase $\omega_{\mathcal{C}^*}^{ab}$ which is in principle dependent on the curve and specific ground state we are acting on, and on the system's size. This concludes the proof of \autoref{thm3.3} of the \autoref{thm3} for the trivial cohomology class.

 With the triviality of $W_{\mathcal{C}^*}$ associated with $1$-coboundaries in hand, we now turn to the case of where $\mathcal{C}^*$ are cocycles that are not coboundaries. This will allow us to derive braiding relations.
 
\subsubsection{Threading a $\pi$-holonomy through a non-contractible cycle}

In differential geometric terms, the equation~\eqref{gsode} along non-trivial cocycles describes parallel transport of the Fermi projector around the Jacobian manifold of the (twisted) boundary conditions on the torus since inserting an holonomy around a non-contractible cycle corresponds to changing boundary conditions to the fermions (see \autoref{fig:pifluxes}).

In this case, \hyperref[claim1]{Lemma \ref*{claim1}} cannot be used and the proof of Lemma~\ref{linecontractible} does not hold. We claim nonetheless that $W_{\mathcal{C}^*}$ still preserve the ground state space: This is the content of the lemma just below. Although its proof is a little technical, the idea is simple: We consider another ${\mathcal{C}'}^*\in C_1(\Gamma_L^*)$ such that $\mathcal{C}^* - {\mathcal{C}'}^*$ is a coboundary and such that $\mathrm{dist}(\mathcal{C}^*, {\mathcal{C}'}^*)$ is of order $L$. Then, by locality, we prove that (see Eq. (\ref{eq:WWW}) below)$W_{\mathcal{C}^* - \bar{\mathcal{C}'}^*} \overset{L}{=} W_{\mathcal{C}^*}W_{{\mathcal{C}'}^*}$. Since the left hand side preserves the ground state space, and all operators are unitary, decay of correlations yields the claim, which is that of \autoref{thm3.1} of \autoref{thm3}:
 \begin{lemma}\label{lemmatattico}
    For any $\mathcal{C}^* \in C_1(\Gamma_L^*)$, the operator $W_{\mathcal{C}^*}$ satisfies $\|[W_{\mathcal{C}^*}, P] \| \overset{L}{=}0$. 
\end{lemma}
\begin{proof}
    Let $\mathcal{C}'^*$ be another cycle such that 
    \begin{enumerate}
        \item $\mathrm{dist}(\mathcal{C}^*, \mathcal{C}'^*) = \rho L$ for some $0<\rho<1$,
        \item $\mathcal{C}^*-\mathcal{C}'^*$ is a coboundary, namely $\mathcal{C}^*-\mathcal{C}'^* = \partial \Lambda$ for some $\Lambda\subset\Gamma_L$.
    \end{enumerate}
see \autoref{fig:linegs2}. \hyperref[linecontractible]{Lemma \ref*{linecontractible}} implies that 
\begin{equation}
[W_{\mathcal{C}^*-\mathcal{C}'^*}, P]=0.
\end{equation}
Since $\dot H_{\mathcal{C}^*-\mathcal{C}'^*}(\phi) = \dot H_{\mathcal{C}^*}(\phi) + \dot H_{-\mathcal{C}'^*}(\phi)$, \hyperref[def:QAG]{Definition~\ref*{def:QAG}} of the quasi-adiabatic generator implies that
\begin{equation}\label{eq:splitting of Ks}
    \caK_{\mathcal{C}^*-\mathcal{C}'^*}(\phi)
    =\tau_f^\phi(\dot H_{\mathcal{C}^*}(\phi)) + \tau_f^\phi(\dot H_{-\mathcal{C}'^*})(\phi)
    \stackrel{L}{=} \caK_{\mathcal{C}^*}(\phi) + \caK_{-\mathcal{C}'^*}(\phi).
\end{equation}
The second identity amounts to replacing $\tau_s^{H_{\mathcal{C}^*-\mathcal{C}'^*}(\phi)} (\dot{H}_{\mathcal{C}^*}(\phi))$ with $\tau_s^{H_{\mathcal{C}^*}(\phi)}(\dot{H}_{\mathcal{C}^*}(\phi))$ and similarly with the second term. This is a consequence of the fact that $\dot{H}_{\mathcal{C}^*}$ is supported on the sites adjacent to $\mathcal{C}^*$ and of the following corollary of the Lieb-Robinson bound:
\begin{equation}
    \Vert \tau_t^{H_{X^r}}(A) - \tau_t^{H_{\Gamma_L}}(A) \Vert\leq 
    \int_0^t\Big \Vert [(H_{\Gamma_L}-H_{X^r}),\tau^{H_{\Gamma_L}}_s(A)]\Big\Vert ds 
    \leq C(A) \ep{-c(r-vt)},
\end{equation}
where $C(A)$ depends on the norm of $A$ and the size of its support $X\subset\Gamma_L$, for any $A$ and any $r\geq 0$, where $X^r = \{x\in\Gamma_L:\mathrm{dist}(x,X)\leq r\}$. In the current application, $A = \dot{H}_{\mathcal{C}^*}$ whose support has volume $2L$. This linear growth is however controlled by the exponential decay obtained by choosing $r$ proportional to $L$, which is the reason for Property~1 of $\mathcal{C}'^*$. By the Lieb-Robinson bound again, the operator $\caK_{\mathcal{C}^*}$ is almost localized on $\caC^*$ which implies that
\begin{equation}\label{eq:KCKC}
    [\caK_{\mathcal{C}^*},\caK_{-\mathcal{C}'^*}]\stackrel{L}{=}0.
\end{equation}
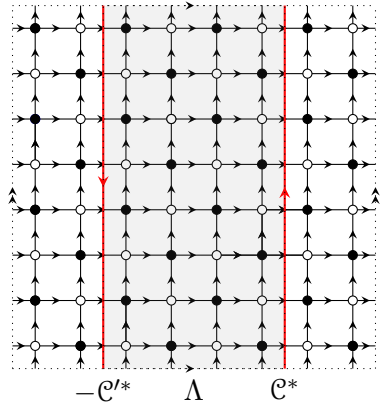
\begin{figure}
    \centering
   \begin{tikzpicture}[scale=0.6]
\draw[black, ->>>>-] (1,1) -- (2,1);
\draw[black, ->>>>-] (0,0) -- (0,1);
\draw[black, ->>>>-] (0,0) -- (1,0);
\draw[black, ->>>>-] (1,0) -- (1,1);
\draw[black, ->>>>-] (1,0) -- (2,0);
\draw[black, ->>>>-] (2,0) -- (2,1);
\draw[black, ->>>>-] (2,0) -- (2,1);
\draw[black, ->>>>-] (3,0) -- (4,0);
\draw[black, ->>>>-] (3,0) -- (3,1);
\draw[black, ->>>>-] (4,0) -- (5,0);
\draw[black, ->>>>-] (2,0) -- (3,0); 
\draw[black, ->>>>-] (4,0) -- (4,1);
\draw[black, ->>>>-] (5,0) -- (6,0);
\draw[black, ->>>>-] (5,0) -- (5,1);
\draw[black, ->>>>-] (6,0) -- (7,0);
\draw[black, ->>>>-] (6,0) -- (6,1);
\draw[black, ->>>>-] (7,0) -- (7,1);

\draw[black, ->>>>-] (0,1) -- (1,1);
\draw[black, ->>>>-] (0,1) -- (0,2);
\draw[black, ->>>>-] (1,1) -- (1,2);
\draw[black, ->>>>-] (2,1) -- (2,2);

\draw[black, ->>>>-] (2,1) -- (3,1);
\draw[black, ->>>>-] (3,1) -- (3,2);
\draw[black, ->>>>-] (3,1) -- (4,1);
\draw[black, ->>>>-] (4,1) -- (4,2);

\draw[black, ->>>>-] (4,1) -- (5,1);
\draw[black, ->>>>-] (5,1) -- (5,2);
\draw[black, ->>>>-] (5,1) -- (6,1);
\draw[black, ->>>>-] (6,1) -- (6,2);

\draw[black, ->>>>-] (6,1) -- (7,1);
\draw[black, ->>>>-] (7,1) -- (7,2);
\draw[black, ->>>>-] (1,2) -- (2,2);
\draw[black, ->>>>-] (1,2) -- (1,3);

\draw[black, ->>>>-] (0,2) -- (1,2);
\draw[black, ->>>>-] (0,2) -- (0,3);

\draw[black, ->>>>-] (2,2) -- (3,2);
\draw[black, ->>>>-] (2,2) -- (2,3);
\draw[black, ->>>>-] (3,2) -- (4,2);
\draw[black, ->>>>-] (3,2) -- (3,3);
\draw[black, ->>>>-] (4,2) -- (5,2);
\draw[black, ->>>>-] (4,2) -- (4,3);
\draw[black, ->>>>-] (5,2) -- (6,2);
\draw[black, ->>>>-] (5,2) -- (5,3);
\draw[black, ->>>>-] (5,2) -- (6,2);
\draw[black, ->>>>-] (5,2) -- (5,3);
\draw[black, ->>>>-] (6,2) -- (7,2);
\draw[black, ->>>>-] (6,2) -- (6,3);
\draw[black, ->>>>-] (4,2) -- (5,2);
\draw[black, ->>>>-] (5,3) -- (6,3);
\draw[black, ->>>>-] (7,2) -- (7,3);
\draw[black, ->>>>-] (0,3) -- (1,3);
\draw[black, ->>>>-] (0,3) -- (0,4);
\draw[black, ->>>>-] (1,3) -- (2,3);
\draw[black, ->>>>-] (2,3) -- (3,3);
\draw[black, ->>>>-] (3,3) -- (4,3);
\draw[black, ->>>>-] (4,3) -- (5,3);
\draw[black, ->>>>-] (0,4) -- (0,5);
\draw[black, ->>>>-] (5,3) -- (6,3);
\draw[black, ->>>>-] (6,3) -- (7,3);
\draw[black, ->>>>-] (0,4) -- (0,5);
\draw[black, ->>>>-] (0,4) -- (1,4);
\draw[black, ->>>>-] (1,4) -- (2,4);
\draw[black, ->>>>-] (2,4) -- (3,4);
\draw[black, ->>>>-] (3,4) -- (4,4);
\draw[black, ->>>>-] (4,4) -- (5,4);
\draw[black, ->>>>-] (0,5) -- (0,6);
\draw[black, ->>>>-] (5,4) -- (6,4);
\draw[black, ->>>>-] (6,4) -- (7,4);
\draw[black, ->>>>-] (0,5) -- (1,5);
\draw[black, ->>>>-] (1,5) -- (2,5);
\draw[black, ->>>>-] (2,5) -- (3,5);
\draw[black, ->>>>-] (3,5) -- (4,5);
\draw[black, ->>>>-] (4,5) -- (5,5);
\draw[black, ->>>>-] (0,6) -- (0,7);
\draw[black, ->>>>-] (5,5) -- (6,5);
\draw[black, ->>>>-] (6,5) -- (7,5);
\draw[black, ->>>>-] (0,6) -- (1,6);
\draw[black, ->>>>-] (1,6) -- (2,6);
\draw[black, ->>>>-] (2,6) -- (3,6);
\draw[black, ->>>>-] (3,6) -- (4,6);
\draw[black, ->>>>-] (4,6) -- (5,6);
\draw[black, ->>>>-] (7,6) -- (7,7);
\draw[black, ->>>>-] (5,6) -- (6,6);
\draw[black, ->>>>-] (6,6) -- (7,6);
\draw[black, ->>>>-] (0,7) -- (1,7);
\draw[black, ->>>>-] (1,7) -- (2,7);
\draw[black, ->>>>-] (2,7) -- (3,7);
\draw[black, ->>>>-] (3,7) -- (4,7);
\draw[black, ->>>>-] (4,7) -- (5,7);
\draw[black, ->>>>-] (7,5) -- (7,6);
\draw[black, ->>>>-] (5,7) -- (6,7);
\draw[black, ->>>>-] (6,7) -- (7,7);
\draw[black, ->>>>-] (7,3) -- (7,4);
\draw[black, ->>>>-] (7,4) -- (7,5);
\draw[black, ->>>>-] (1,4) -- (1,5);
\draw[black, ->>>>-] (1,3) -- (1,4);
\draw[black, ->>>>-] (1,5) -- (1,6);
\draw[black, ->>>>-] (1,6) -- (1,7);
\draw[black, ->>>>-] (2,4) -- (2,5);
\draw[black, ->>>>-] (2,3) -- (2,4);
\draw[black, ->>>>-] (2,5) -- (2,6);
\draw[black, ->>>>-] (2,6) -- (2,7);
\draw[black, ->>>>-] (3,4) -- (3,5);
\draw[black, ->>>>-] (3,3) -- (3,4);
\draw[black, ->>>>-] (3,5) -- (3,6);
\draw[black, ->>>>-] (3,6) -- (3,7);
\draw[black, ->>>>-] (4,4) -- (4,5);
\draw[black, ->>>>-] (4,3) -- (4,4);
\draw[black, ->>>>-] (4,5) -- (4,6);
\draw[black, ->>>>-] (4,6) -- (4,7);
\draw[black, ->>>>-] (5,4) -- (5,5);
\draw[black, ->>>>-] (5,3) -- (5,4);
\draw[black, ->>>>-] (5,5) -- (5,6);
\draw[black, ->>>>-] (5,6) -- (5,7);
\draw[black, ->>>>-] (6,4) -- (6,5);
\draw[black, ->>>>-] (6,3) -- (6,4);
\draw[black, ->>>>-] (6,5) -- (6,6);
\draw[black, ->>>>-] (6,6) -- (6,7);
\draw[black, dotted, ->-] (-0.5,-0.5) -- (7.5,-0.5);
\draw[black, dotted, ->-] (-0.5,7.5) -- (7.5,7.5);
\draw[black, dotted, ->>-] (-0.5,-0.5) -- (-0.5,7.5);
\draw[black, dotted, ->>-] (7.5,-0.5) -- (7.5,7.5);
\draw[black, ->>>>-] (0, -0.5) --(0,0);
\draw[black, ->>>>-] (-0.5,0) -- (0,0);
\draw[black, ->>>>-] (1, -0.5) --(1,0);
\draw[black, ->>>>-] (-0.5,1) -- (0,1);
\draw[black, ->>>>-] (2, -0.5) --(2,0);
\draw[black, ->>>>-] (-0.5,2) -- (0,2);
\draw[black, ->>>>-] (3, -0.5) --(3,0);
\draw[black, ->>>>-] (-0.5,3) -- (0,3);
\draw[black, ->>>>-] (4, -0.5) --(4,0);
\draw[black, ->>>>-] (-0.5,4) -- (0,4);
\draw[black, ->>>>-] (5, -0.5) --(5,0);
\draw[black, ->>>>-] (-0.5,6) -- (0,6);
\draw[black, ->>>>-] (6,-0.5) -- (6,0);
\draw[black, ->>>>-] (-0.5,5) -- (0,5);
\draw[black, ->>>>-] (7, -0.5) --(7,0);
\draw[black, ->>>>-] (-0.5,7) -- (0,7);

\draw[black, ->>>-] (0,7) -- (0,7.5);
\draw[black, ->>>-] (7,7) -- (7.5,7);
\draw[black, ->>>-] (1,7) -- (1,7.5);
\draw[black, ->>>-] (7,6) -- (7.5,6);
\draw[black, ->>>-] (2,7) -- (2,7.5);
\draw[black, ->>>-] (7,5) -- (7.5,5);
\draw[black, ->>>-] (3,7) -- (3,7.5);
\draw[black, ->>>-] (7,4) -- (7.5,4);
\draw[black, ->>>-] (4,7) -- (4,7.5);
\draw[black, ->>>-] (7,0) -- (7.5,0);
\draw[black, ->>>-] (5,7) -- (5,7.5);
\draw[black, ->>>-] (7,3) -- (7.5,3);
\draw[black, ->>>-] (6,7) -- (6,7.5);
\draw[black, ->>>-] (7,2) -- (7.5,2);
\draw[black, ->>>-] (7,7) -- (7,7.5);
\draw[black, ->>>-] (7,1) -- (7.5,1);

\draw[black,fill=white] (0,0) circle (.1 cm);
\draw[black,fill=white] (1,1) circle (.1 cm);
\draw[black,fill=white] (3,3) circle (.1 cm);
\draw[black,fill=white] (2,2) circle (.1 cm);
\draw[black,fill=white] (4,4) circle (.1 cm);
\draw[black,fill=white] (5,5) circle (.1 cm);
\draw[black,fill=white] (6,6) circle (.1 cm);
\draw[black,fill=white] (7,7) circle (.1 cm);
\draw[black,fill=white] (3,5) circle (.1 cm);
\draw[black,fill=white](2,0) circle (.1 cm);
\draw[black,fill=white] (0,2) circle (.1 cm);
\draw[black,fill=white] (4,0) circle (.1 cm);
\draw[black,fill=white] (0,4) circle (.1 cm);
\draw[black,fill=white] (6,0) circle (.1 cm);
\draw[black,fill=white] (0,6) circle (.1 cm);
\draw[black,fill=white] (1,3) circle (.1 cm);
\draw[black,fill=white] (1,5) circle (.1 cm);
\draw[black,fill=white] (1,7) circle (.1 cm);
\draw[black,fill=white] (7,1) circle (.1 cm);
\draw[black,fill=white] (5,1) circle (.1 cm);
\draw[black,fill=white] (3,1) circle (.1 cm);

\draw[black,fill=white] (4,2) circle (.1 cm);
\draw[black,fill=white] (6,2) circle (.1 cm);
\draw[black,fill=white] (2,6) circle (.1 cm);
\draw[black,fill=white] (2,4) circle (.1 cm);
\draw[black,fill=white] (6,4) circle (.1 cm);
\draw[black,fill=white] (5,3) circle (.1 cm);
\draw[black,fill=white] (7,3) circle (.1 cm);
\draw[black,fill=white] (3,5) circle (.1 cm);
\draw[black,fill=white] (3,7) circle (.1 cm);
\draw[black,fill=white] (7,5) circle (.1 cm);
\draw[black,fill=white] (4,6) circle (.1 cm);
\draw[black,fill=white] (5,7) circle (.1 cm);

\draw[black,fill=black] (1,0) circle (.1 cm);
\draw[black,fill=black] (0,1) circle (.1 cm);
\draw[black,fill=black] (3,0) circle (.1 cm);
\draw[black,fill=black] (0,3) circle (.1 cm);
\draw[black,fill=blue] (5,0) circle (.1 cm);
\draw[blue,fill=black] (0,5) circle (.1 cm);
\draw[black,fill=black] (7,0) circle (.1 cm);
\draw[black,fill=black] (0,7) circle (.1 cm);
\draw[black,fill=black] (1,0) circle (.1 cm);
\draw[black,fill=black] (0,1) circle (.1 cm);
\draw[black,fill=black] (3,0) circle (.1 cm);
\draw[black,fill=black] (0,3) circle (.1 cm);
\draw[black,fill=black] (5,0) circle (.1 cm);
\draw[black,fill=black] (0,5) circle (.1 cm);
\draw[black,fill=black] (7,0) circle (.1 cm);
\draw[black,fill=black] (0,7) circle (.1 cm);

\draw[black,fill=black] (1,2) circle (.1 cm);
\draw[black,fill=black] (2,1) circle (.1 cm);
\draw[black,fill=black] (3,2) circle (.1 cm);
\draw[black,fill=black] (2,3) circle (.1 cm);
\draw[black,fill=black] (5,2) circle (.1 cm);
\draw[black,fill=black] (2,5) circle (.1 cm);
\draw[black,fill=black] (7,2) circle (.1 cm);
\draw[black,fill=black] (2,7) circle (.1 cm);
\draw[black,fill=black] (1,4) circle (.1 cm);
\draw[black,fill=black] (4,1) circle (.1 cm);
\draw[black,fill=black] (3,4) circle (.1 cm);
\draw[black,fill=black] (4,3) circle (.1 cm);
\draw[black,fill=black] (5,4) circle (.1 cm);
\draw[black,fill=black] (4,5) circle (.1 cm);
\draw[black,fill=black] (7,4) circle (.1 cm);
\draw[black,fill=black] (4,7) circle (.1 cm);

\draw[black,fill=black] (1,6) circle (.1 cm);
\draw[black,fill=black] (6,1) circle (.1 cm);
\draw[black,fill=black] (3,6) circle (.1 cm);
\draw[black,fill=black] (6,3) circle (.1 cm);
\draw[black,fill=black] (5,6) circle (.1 cm);
\draw[black,fill=black] (6,5) circle (.1 cm);
\draw[black,fill=black] (7,6) circle (.1 cm);
\draw[black,fill=black] (6,7) circle (.1 cm);

\draw[red, ->-, thick] (1.5,7.5) -- (1.5,-0.5);
\draw[red, ->-, thick] (5.5,-0.5) -- (5.5,7.5);

\draw[dotted,fill = gray, fill opacity = 0.1] (1.5,-0.5) -- (5.5,-0.5) -- (5.5,7.5)--(1.5,7.5)--(1.5,-0.5);
\node[below] (a) at (1.5,-0.5) {\scalebox{1}{$-\caC'^*$}};
\node[below] (b) at (5.5,-0.5) {\scalebox{1}{$\caC^*$}};
\node[below] (c) at (3.5,-0.5) {\scalebox{1}{$\Lambda$}};
   \end{tikzpicture}
   \caption{The coboundary $\partial\Lambda = \caC^* - \caC'^*$. The sign is given by the orientation of the dual lattice $\Gamma_L^*$.}
   \label{fig:linegs2}
\end{figure}
Recalling~\eqref{eq:QAFlow} that $V_{\mathcal{C}^*}(\phi)$ is the propagator generated by $\mathcal{K}_{\mathcal{C}^*}(\phi)$, Eq. (\ref{eq:KCKC}) yields that $[V_{\mathcal{C}^*},\caK_{-\mathcal{C}'^*}]\stackrel{L}{=}0$ and in turn that
\begin{equation}
    -\iu \frac{d}{d\phi}V^{*}_{\mathcal{C}^*-\mathcal{C}'^*}V_{\mathcal{C}^*} V_{-\mathcal{C}'^*}
    \stackrel{L}{=}V^{*}_{\mathcal{C}^*-\mathcal{C}'^*}\big(\mathcal{K}_{\mathcal{C}^*-\mathcal{C}'^*} - \mathcal{K}_{\mathcal{C}^*} - \mathcal{K}_{-\mathcal{C}'^*}\big)V_{\mathcal{C}^*} V_{-\mathcal{C}'^*}.
\end{equation}
The identity~(\ref{eq:splitting of Ks}) and Duhamel's principle then imply
\begin{equation}
    V_{\mathcal{C}^*-\mathcal{C}'^*}(\phi)
    \stackrel{L}{=} V_{\mathcal{C}^*}(\phi) V_{-\mathcal{C}'^*}(\phi).
\end{equation}
Since, moreover, $\hat X_{\mathcal{C}^*-\mathcal{C}'^*} = \hat X_{\mathcal{C}^*}\hat X_{-\mathcal{C}'^*}$,we immediately conclude that
\begin{equation}\label{eq:WWW}
    W_{\mathcal{C}^*-\mathcal{C}'^*} 
    \stackrel{L}{=} \hat X_{\mathcal{C}^*}\hat X_{-\mathcal{C}'^*} V_{\mathcal{C}^*}(\pi) V_{-\mathcal{C}'^*}(\pi)
    \stackrel{L}{=} W_{\mathcal{C}^*} W_{-\mathcal{C}'^*}.
\end{equation}
In the second equality, we used that $\hat X_{-\mathcal{C}'^*}^*$ and $V_{\mathcal{C}^*}(\pi)$ almost commutes because of their almost disjoint supports. 

It remains to prove that this and $[W_{\mathcal{C}^*-\mathcal{C}'^*} ,P]=0$ implies that each of the unitary factor $ W_{\mathcal{C}^*}, W_{-\mathcal{C}^*}$ almost commutes with $P$. As we shall explain just below, this follows from the fact that $\mathcal{C}^*,\mathcal{C}^*$ are separated by distance of order $L$ and by clustering:
\begin{equation}\label{eq:clustering}
\Vert P AB\vert \Omega\rangle - P APB\vert \Omega\rangle\Vert 
\leq C(A,B) f(d(X,Y))
\end{equation}
which holds for any normalized ground state $\Omega$ in the range of $P$ and any observables $A$ and $B$ supported on $X$ and $Y$, respectively. Here, $d(X,Y)$ is the distance between the two sets and $f$ is a superpolynomially decaying function
see~\cite[Proposition~3.1]{Bachmann_2021}. Since $[P,W_{\mathcal{C}^*}] = \left((1-P)W_{\mathcal{C}^*}^* P\right)^* - (1-P)W_{\mathcal{C}^*} P$, it suffices to show that $\Vert (1-P)W_{\mathcal{C}^*} P\Vert \stackrel{L}{=} 0$ and similarly for the adjoint. For this, we note that
\begin{align}
1 &\geq \|P W_{\mathcal{C}^*}\ket{\Omega}\| \geq \|P  W_{-\mathcal{C}'^*}\| \|P W_{\mathcal{C}^*}\ket{\Omega}\| \geq \|P W_{-\mathcal{C}'^*} P W_{\mathcal{C}^*} \ket{\Omega}\| \\
&\overset{L}{=} \|P  W_{-\mathcal{C}'^*} W_{\mathcal{C}^*} \ket{\Omega}\| \overset{L}{=} 1,
\end{align}
where we used clustering, Eq. (\ref{eq:clustering}), in the second line. More precisely, we applied (\ref{eq:clustering}) for  $A = W_{-\mathcal{C}'^*}$, $B = W_{\mathcal{C}^*}$, approximately localized in $\mathcal{C}'^*$, $\mathcal{C}^*$. This implies that (\ref{eq:clustering}) holds with a constant $C(A,B)$ that grows quadratically with $L$, which is however controlled by the superpolynomial decay of $f(\cdot)$ and by the order $L$ distance of the two sets. Since this holds for all $\ket{\Omega} = P\ket{\Omega}$, we conclude that  $\Vert (1-P)W_{\mathcal{C}^*} P\Vert \stackrel{L}{=} 0$. Repeating the argument with the adjoint yields the claim. 
\end{proof}

We turn to the commutation relations. Let $\caC\in Z_1(\Gamma_L)$. Since
\begin{equation}\label{eq:ZX braiding}
    \hat Z_{\mathcal{C}}\hat X_{\caC^*} 
    = \ep{\iu\pi\mathcal{I}(\mathcal{C},\caC^*)}\hat X_{\caC^*} \hat Z_{\mathcal{C}},
\end{equation}
we have that the operator $B_{\mathcal{C},\caC^*}(\phi) = \hat Z_{\mathcal{C}} W_{\mathcal{C}^*}(\phi)^* \hat Z_{\mathcal{C}}W_{\mathcal{C}^*}(\phi)$ satisfies the equation
\begin{equation}
   -\iu\frac{d}{d\phi} B_{\mathcal{C},\caC^*}(\phi) = 0
\end{equation}
since $\hat Z_{\mathcal{C}}$ commutes with $\caK_{\caC^*}$ because the latter contains only $\hat\sigma^z$ operators. Moreover, $B_{\mathcal{C},\caC^*}(0) = \hat Z_{\mathcal{C}}\hat X_{\caC^*} \hat Z_{\mathcal{C}}\hat X_{\caC^*} = \ep{\iu\pi\mathcal{I}(\mathcal{C},\caC^*)}$, so that
\begin{equation}
    \hat Z_{\mathcal{C}} W_{\caC^*}(\phi) 
    = \ep{\iu\pi\mathcal{I}(\mathcal{C},\caC^*)} W_{\caC^*}(\phi)\hat Z_{\mathcal{C}}
\end{equation}
for all $\phi$, a fortiori at $\phi = \pi$ (which is the only case of interest here since $ W_{\caC^*}(\phi)$ preserves the ground state space only at that particular value). This concludes the proof of \autoref{thm3.2} of \autoref{thm3}. The remaining cases of \autoref{thm3.3} of \autoref{thm3} follow immediately from this and from the fact that the intersection number of a non-trivial cycle and a non-trivial cocycle is necessarily odd, see \autoref{fig:intersection number}. Therefore
\begin{equation}
    \hat Z_{\mathcal{C}_1} W_{\caC_2^*}\ket{\Omega_{ab}} = -a  W_{\caC_2^*} \ket{\Omega_{ab}}
\end{equation}
from which we conclude that 
\begin{equation}
    W_{\caC_2^*}\ket{\Omega_{ab}} \stackrel{L}{=} \ep{\iu\omega^{ab}_{\caC_2^*}}\vert\Omega_{(-a)b}\rangle.
\end{equation}
The argument is similar with the other possible combinations of cycles and cocycles.

\subsubsection{Monopoles are bosons}

For this section, we denote $H(\phi_1,\phi_2)$ the Hamiltonian twisted by $\phi_1$ along $\caC^*_1$ and by $\phi_2$ along $\caC^*_2$. This family is uniformly gapped for $(\phi_1,\phi_2)\in\mathbb{T}^2$, with ground state projection $P(\phi_1,\phi_2)$. We let
\begin{equation}
\mathcal{K}_j(\phi_1,\phi_2) = \int f(t)\ep{\iu t H(\phi_1,\phi_2)}\partial_{\phi_j}H(\phi_1,\phi_2)\ep{-\iu t H(\phi_1,\phi_2)}\,dt.
\end{equation}
where $f$ is the function introduced in Definition~\ref*{def:QAG}.
We denote $\hat X_j \equiv \hat X_{\caC_j^*}$. Then $\hat X_1 H(\phi_1,\phi_2) \hat X_1 = H(\phi_1+\pi,\phi_2)$ so that 
\begin{equation}\label{eq: XK commutation}
\hat X_1 \mathcal{K}_j(\phi_1,\phi_2) \hat X_1 = \mathcal{K}_j(\phi_1+\pi,\phi_2)
\end{equation}
and similarly for $\hat  X_2$ with respect to $\phi_2$. Let $\gamma:[0,1]\to\mathbb{T}^2$, $s\mapsto(\gamma_1(s),\gamma_2(s))$ be a path and let $V_\gamma(s)$ be the unitary solution of
\begin{equation}
V_\gamma(s) = \mathbbm{1} + \iu \int_0^s \mathcal{K}(\gamma(r)) \cdot \gamma'(r) V_\gamma(r)\,dr
\end{equation}
where
\begin{equation}
\mathcal{K}(\gamma(s))\cdot\gamma'(s) = \mathcal{K}_1(\gamma(s))\gamma_1'(s) + \mathcal{K}_2(\gamma(s))\gamma_2'(s).
\end{equation}
Moreover, (\ref{eq: XK commutation}) implies that 
\begin{equation}\label{eq: XV CR}
\hat X_j V_\gamma(s) = V_{\pi e_j + \gamma}(s)\hat X_j
\end{equation}
where $(\pi e_1 + \gamma)(s) = (\pi,0) + \gamma(s)$ and similarly for $j=2$.

For any path $\gamma$, we shall denote $V_\gamma = V_\gamma(1)$. With this, we recall the definition~\ref{def:loop operators} of loop operators, observe that $W_{\mathcal{C}_j^*} = \hat X_j V_{\alpha_j}$, where $\alpha_1(s) = (s\pi,0)$ and $\alpha_2(s) = (0,s\pi)$ and conclude that
\begin{align}
    (W_{\mathcal{C}_2^*})^*(W_{\mathcal{C}_1^*})^*W_{\mathcal{C}_2^*}W_{\mathcal{C}_1^*}
    &= V_{\alpha_2}^* \hat X_2 V_{\alpha_1}^*\hat X_1\hat X_2 V_{\alpha_2}\hat X_1V_{\alpha_1} \\
    &= V_{\alpha_2}^* V_{\pi e_2 + \alpha_1}^* V_{\pi e_1 + \alpha_2}V_{\alpha_1} 
    = V_{\partial\Gamma}
\end{align}
where we used~(\ref{eq: XV CR}) and observed that $V^*_\gamma = V_{-\gamma}$ as well as $V_\gamma V_{\gamma'} = V_{\gamma + \gamma'}$ to conclude that the product reduces to the propagator along the boundary of the square $\Gamma$ drawn in~\autoref{fig:path}.

\begin{figure}
\centering
  \includegraphics[width=0.3\linewidth]{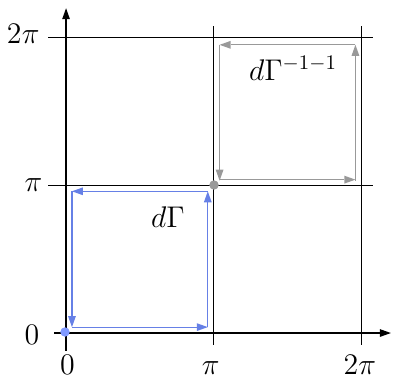}
  \caption{In blue, the path $\partial\Gamma$ corresponding to the braiding of monopoles on the ground state space~(\ref{eq:braiding monopoles}). Changing the holonomies amounts to changing the initial condition with the grey path corresponding to the vector $\Omega_{-1,-1}$, starting at $(\phi_1,\phi_2) = (\pi,\pi)$.}
  \label{fig:path}
\end{figure}

By the gauge invariance of the Hamiltonian,
\begin{equation}
    V_\gamma(s) \vert \Omega_{ab}\rangle = \prod_{i \in \Gamma_L} \bigg(\frac{\mathbbm{1} + Q_i}{2}\bigg) V_{\gamma^{ab}}^{\bm\sigma}(s)|\psi_{{\bm \sigma},a,b}\rangle\otimes |{\bm\sigma},a,b\rangle,
\end{equation}
where $\gamma^{ab}(s) = (\frac{1}{2}(1-a)\pi,\frac{1}{2}(1-b)\pi) + \gamma(s)$ is the shifted path, and $\bm \sigma$ is any configuration in the class $[{\bf-1}]$. Hence
\begin{equation}\label{eq:braiding monopoles}
    \ep{\iu\omega_{ab}} \stackrel{L}{=} \langle \Omega_{ab}\vert (W_{\mathcal{C}_2^*})^*(W_{\mathcal{C}_1^*})^*W_{\mathcal{C}_2^*}W_{\mathcal{C}_1^*} \Omega_{ab}\rangle
    = \langle \psi_{{\bm \sigma},a,b}\vert V^{\bm \sigma}_{\partial\Gamma^{ab}}\psi_{{\bm \sigma},a,b}\rangle
\end{equation}

Since $V_\gamma(s)$ is the spectral flow and $|\psi_{{\bm \sigma},a,b}\rangle$ is the unique, gapped, fermionic ground state, the vector $V_{\gamma^{ab}}^{\bm\sigma}(s)|\psi_{{\bm \sigma},a,b}\rangle$ is a fermionic ground state at fluxes $\gamma^{ab}(s)$. 

On the cut torus, there is a smooth family $\vert \psi^{\bm \sigma}(\phi_1,\phi_2)\rangle$ of vectors such that $P^{{\bm \sigma}}(\phi_1,\phi_2) = |\psi^{{\bm \sigma}}(\phi_1,\phi_2)\rangle\langle\psi^{{\bm \sigma}}(\phi_1,\phi_2)|$. Upon possibly redefining their phases, we can assume that $\vert \psi_{{\bm \sigma},a,b}\rangle  = \vert \psi^{{\bm \sigma}}(\frac{1}{2}(1-a)\pi,\frac{1}{2}(1-b)\pi)\rangle$. Let now
\begin{equation}
    \vert \varphi^{{\bm \sigma}}(s)\rangle  = V^{\bm \sigma}_{\partial\Gamma}(s) \vert \psi^{{\bm \sigma}}(0)\rangle ,
\end{equation}
where we have now dropped the index $ab$ since it amounts to a choice of one of $(0,0)$, $(\pi,0)$ $(0,\pi)$ or $(\pi,\pi)$ as initial condition $\partial \Gamma(0)$. The phase of~(\ref{eq:braiding monopoles}) is now given by
\begin{equation}
\ep{\iu\omega} = \langle\psi^{{\bm \sigma}}(0) \vert \varphi^{{\bm \sigma}}(1)\rangle.
\end{equation}

For any $s\in[0,1]$, there is a $c(s)\in\mathbb{C},\vert c(s) \vert =1$ such that 
\begin{equation}
\vert \varphi^{{\bm \sigma}}(s)\rangle = c(s) \vert \psi^{{\bm \sigma}}(s)\rangle.
\end{equation}
Taking the derivative along the path yields
\begin{equation}
\iu \mathcal{K}^{{\bm \sigma}}(s)\cdot \partial\Gamma(s)\vert \varphi^{{\bm \sigma}}(s)\rangle = c'(s)\vert \psi^{{\bm \sigma}}(s)\rangle + c(s)\frac{d}{ds}\vert \psi^{{\bm \sigma}}(s)\rangle.
\end{equation}
We now take $\langle \psi^{{\bm \sigma}}(s)\vert\cdot\rangle$ to get
\begin{equation}
0 = c'(s) + c(s)\left \langle \psi^{{\bm \sigma}}(s)\Bigg\vert\frac{d}{ds}\psi^{{\bm \sigma}}(s)\right\rangle,
\end{equation}
where we used $\langle \psi^{{\bm \sigma}}(s) \vert \mathcal{K}^{{\bm \sigma}}(s)\psi^{{\bm \sigma}}(s)\rangle = 0$. Equivalently,
\begin{equation}
\frac{d}{ds} \log c(s) = -\left \langle \psi^{{\bm \sigma}}(s)\Bigg\vert\frac{d}{ds}\psi^{{\bm \sigma}}(s)\right\rangle.
\end{equation}
Hence, since $\omega$ is the total change of argument of the function $c(s)$ along the path $\partial\Gamma$,
\begin{equation}
\iu \omega = -\int_{\partial\Gamma}\langle\psi^{{\bm \sigma}}\vert d\psi^{{\bm \sigma}}\rangle.
\end{equation}
We are now in the setting of~\cite{avron1983homotopy} and so an application of Stokes' theorem yields
\begin{align}
\int_{\partial \Gamma}\langle\psi^{{\bm \sigma}}\vert d\psi^{{\bm \sigma}}\rangle &= \int_{\Gamma} d\langle\psi^{{\bm \sigma}}\vert d\psi^{{\bm \sigma}}\rangle
= \int_{\Gamma} \langle d\psi^{{\bm \sigma}}\vert d\psi^{{\bm \sigma}}\rangle
= \int_{\Gamma} \Tr(dP^{{\bm \sigma}} P^{{\bm \sigma}} dP^{{\bm \sigma}}) \\
&= \int_{\Gamma} \Tr(P^{{\bm \sigma}} dP^{{\bm \sigma}}\wedge dP^{{\bm \sigma}})
\end{align}
namely the braiding phase is the integral of the adiabatic curvature, or equivalently the Hall conductance, over the quarter torus. Since $dP^{{\bm \sigma}} = \iu [\mathcal{K}^{{\bm \sigma}},P^{{\bm \sigma}}]$, we conclude that
\begin{equation}
\omega = -\iu \int_{\Gamma}\Tr(P^{{\bm \sigma}} [\mathcal{K}^{{\bm \sigma}}_1,\mathcal{K}^{{\bm \sigma}}_2])
\end{equation}
By~\cite{hastings2015quantization,Quantization}, the curvature is constant. Since the initial fermionic ground states are all time-reversal invariant, the Hall conductance vanishes and so $\omega = 0$, concluding the proof of \autoref{thm3}, \autoref{thm3.3bis}.

\begin{remark}
    The calculation above is in fact completely general under the assumption of a spectral gap. Interestingly, it relates the braiding phase of the dressed monopoles to the Hall conductance of the fermion sector, see also~\cite{PhysRevB.101.085138, kapustin2020hall}.
\end{remark}

\subsubsection{Braiding statistics between monopoles and fermions}

Let
\begin{equation}
    B = \frac{\langle  \xi_{ab}^{ij} \vert W_{\mathcal{C}^*} \xi_{ab}^{ij}\rangle }{ \langle\Omega_{ab}\vert W_{\mathcal{C}^*} \Omega_{ab}\rangle} .
\end{equation}
As discussed above $\vert \xi_{ab}^{ij}\rangle$ is obtained creating two fermions on $\Omega_{ab}$, and $B$ is the braiding between one of the fermions and a monopole, see also \autoref{fig:link}. 
    \begin{figure}[h]
    \centering
   \begin{tikzpicture}[scale=0.5]
   \draw[red, thick] (0,0) circle (3cm);
   \draw[blue, thick, shorten <= 0.25cm, shorten >= 0.25cm] (0,0) to [out=0,in=150] (6,2);
\draw[blue,fill=white] (6,2) circle (.1 cm);
\draw[blue,fill=white] (0,0) circle (.1 cm);
\node[left] (c) at (-3,0) {\scalebox{1}{$\mathcal{C}^*$}};
\node[below] (c) at (0,0) {\scalebox{1}{$a^+_i$}};
\node[below] (c) at (6,2) {\scalebox{1}{$a^+_j$}};
   \end{tikzpicture}
   \caption{A representation of the braiding $B$. Any choice of the curve connecting $i$ and $j$ crosses $\mathcal{C}^*$ an odd number of times giving a $-$ sign.}
   \label{fig:link}
\end{figure}
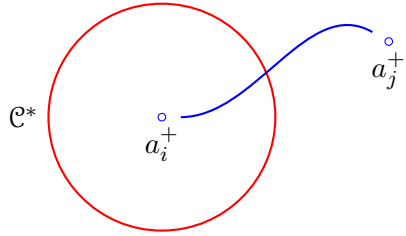

Since $\caC^*$ is a coboundary, the denominator has already been shown~(\ref{eq:trivial loops on GS}) to be a phase:
\begin{equation}
   \langle\Omega_{ab}\vert W_{\mathcal{C}^*}\Omega_{ab}\rangle = e^{i \omega^{ab}_{\mathcal{C}^*}}.
\end{equation}
We claim that
\begin{equation}\label{eq:em braiding}
   \left\Vert \left\{W_{\mathcal{C}^*},a^{+}_{i,\eta} \hat Z_{\caC_{i,j}} a^{+}_{j,\eta'}\right\}\right\Vert \overset{L}{=} 0.
\end{equation}
With this and $ W_{\mathcal{C}^*}\vert \Omega_{ab}\rangle = e^{i \omega^{ab}_{\mathcal{C}^*}}\vert \Omega_{ab}\rangle$, we have that 
\begin{equation}
    B \overset{L}{=} - e^{i \omega^{ab}_{\mathcal{C}^*}} \frac{\langle \xi_{ab}^{ij} \vert\xi_{ab}^{ij} \rangle }{\langle \Omega_{ab} \vert W_{\mathcal{C}^*}\Omega_{ab}\rangle} = -1
\end{equation}
indeed. To prove~(\ref{eq:em braiding}), we recall that $W_{\mathcal{C}^*}$ is a fermion-even operator and almost localized along $\mathcal{C}^*$, so that 
\begin{equation}
   \Vert  [W_{\mathcal{C}^*},a^{+}_{i,j}] \Vert \overset{L}{=} 0.
\end{equation}
Hence,
\begin{equation}
    W_{\mathcal{C}^*} a^{+}_{i,\eta} \hat Z_{\caC_{i,j}} a^{+}_{j,\eta'} 
    \overset{L}{=} a^{+}_{i,\eta} W_{\mathcal{C}^*}\hat Z_{\caC_{i,j}} a^{+}_{j,\eta'} 
    = - a^{+}_{i,\eta} \hat Z_{\caC_{i,j}} W_{\mathcal{C}^*} a^{+}_{j,\eta'} 
    \overset{L}{=} -a^{+}_{i,\eta} \hat Z_{\caC_{i,j}} a^{+}_{j,\eta'}W_{\mathcal{C}^*}
\end{equation}
where we used \eqref{eq:ZX braiding} and the fact that $Z_{\caC_{i,j}}$ commutes with $\caK_{\caC^*}(\phi)$ in the second equality. This concludes the proof \autoref{thm3}-\autoref{thm3.4}. \qed

\begin{remark}
\begin{enumerate}
\item Up to a phase, the vector $\vert\xi_{ab}^{ij}\rangle$ can be represented as:
\begin{equation}
    \bigg(\prod_{i \in \text{V}(\Gamma_L)} \frac{1+Q_i}{2}\bigg) a^{+}_{i,\eta} a^{+}_{j,\eta'}  \ket{\psi_{\bm{-1},a,b}}\otimes \ket{\bm{-1},a,b}.
\end{equation}
Indeed, $ \hat Z_{\caC_{i,j}} \ket{\psi_{\bm{-1},a,b}}\otimes \ket{\bm{-1},a,b} = Z_{\caC_{i,j}} \ket{\psi_{\bm{-1},a,b}}\otimes \ket{\bm{-1},a,b}$, where $Z_{\caC_{i,j}}$ is just a phase. While neither $a^{+}_{i} a^{+}_{j}$ nor $\hat Z_{\caC_{i,j}}$ commute with the $\mathbb{Z}_2$-charges $Q_i$, their combination does. Pulling $\hat Z_{\caC_{i,j}}$ through the projection onto the physical Hilbert space yields the claim. 
\item The braiding is the same if one or both of the creation operators are replaced by annihilation operators, because the braiding only takes into account the fermion parity of the excitation, not its $U(1)$ charge.
\item The proof shows explicitly that the braiding is independent of the choice of $(a,b)$ as one would have expected from the topological order condition and the locality of the braiding operation.
\end{enumerate}
\end{remark}

\appendix

 \section{Homology of the torus: basic definitions and notations}\label{homo}

 We recall that $\Gamma = \mathbb{Z}^2_L$ is the square lattice with periodic boundary conditions, and that $\text{V}(\Gamma),\text{E}(\Gamma),\text{F}(\Gamma)$ are the sets of all vertices, oriented edges and oriented faces of $\Gamma$. By construction $|\text{V}(\Gamma)| = |\text{F}(\Gamma)|$ and $|\text{E}(\Gamma)| = 2 |\text{V}(\Gamma)|$.

 This structure is naturally understood as a CW-complex, and we briefly recall the basic constructions of cellular homology. Since orientation does not play a role in the analysis of this paper, we don't insist on it here. 

 \begin{definition}[chains]
     For $i=0,1,2$, the $i$-chain group $C_i(\Gamma)$ is a free abelian group with coefficients in $\mathbb{Z}_2$ generated respectively by $\emph{V}(\Gamma)$, $\emph{E}(\Gamma)$ and $\emph{F}(\Gamma)$.
 \end{definition} 

 Geometrically, a $1$-chain $\caC$ is a string along the edges of the lattice, or a collection thereof. The fact that the coefficients are $\mathbb{Z}_2$ simply means (in the example of $i=1$) that an edge may appear at most once. 

 The boundary maps $\partial_i: C_{i}(\Gamma) \to C_{i-1}(\Gamma)$ associate to a $i$-chain the $(i-1)$-chain that forms its boundary. For example, if $\Lambda$ is a $2$-chain, namely a collection of plaquettes, then $\partial\Lambda$ is the $1$-chain made of all the boundary edges. As usual, $\partial_i\circ\partial_{i+1} = 0$, namely $\mathrm{Im} (\partial_{i+1}) \subset \mathrm{Ker} (\partial_{i})$. 

     \begin{definition}[cycles and boundaries]
The $i$-cycle group is $Z_i(\Gamma) = \mathrm{Ker}(\partial_i)$. The $i$-boundary group is $B_i(\Gamma) = \mathrm{Im}(\partial_{i+1})$. 
    \end{definition}
In other words, a $1$-cycle is geometrically a loop or a collection thereof. A $1$-boundary is the geometric boundary of a collection of plaquettes. As pointed out above, $B_i(\Gamma)\subset Z_i(\Gamma)$: For the case $i=1$, any boundary is indeed a collection of loops.

\begin{definition}[homology groups]
    The homology groups are defined as:
    \begin{equation}
        H_i(\Gamma) = \frac{Z_i(\Gamma)}{B_i(\Gamma)}.
    \end{equation}
\end{definition}

 The group $H_1(\Gamma)$ is made up of equivalence classes of loops, where two loops are equivalent if they differ by a boundary (one could say: they can be `deformed' into each other). On the torus, $H_1(\Gamma)$ has four elements: It is generated by the two types of loops that wind once around the torus in either direction. In this paper, they are associated with the holonomies, or large gauge transformations.

 \paragraph{Cohomology.} The regular square lattice has a natural dual obtained by identifying each face with a dual vertex, each edge with a dual edge (which is geometrically perpendicular to it) and each vertex with a dual face. The construction above can be repeated with this dual complex, yielding cochain groups $C_i(\Gamma^*)$, cocycles $Z_i(\Gamma^*)$ and coboundaries $B_i(\Gamma^*)$, and cohomology groups. For example, a $1$-cochain $\caC^*$ is a string in the dual lattice or a collection thereof.

\bigskip
\bigskip

 \paragraph{Data availability statement.} There is no data associated with this work.

\bibliographystyle{plain}
\bibliography{Bibliography}

\end{document}